\documentclass[11pt]{article} 
\usepackage{amsmath,amssymb,amsthm,mathrsfs}
\usepackage[margin=1.4in]{geometry}
\usepackage{enumerate}

\usepackage{graphics,geometry,epsfig}

\theoremstyle{plain}
\newtheorem{theorem}{Theorem}
\newtheorem{proposition}[theorem]{Proposition}
\newtheorem{lemma}[theorem]{Lemma}
\newtheorem{corollary}[theorem]{Corollary}

\theoremstyle{definition}
\newtheorem{definition}[theorem]{Definition}
\newtheorem{remark}[theorem]{Remark}

\numberwithin{theorem}{section} \numberwithin{equation}{section}

\newcommand{\nc}{\newcommand}

\nc{\be}{\begin{equation}}
\nc{\la}{\label}
\nc{\ba}{\begin{array}}
\nc{\ea}{\end{array}}
\nc{\bs}{\begin{split}}
\nc{\es}{\end{split}}

\newcommand{\TIMC}[1]{\textbf{#1}}

\newcommand{\R}{\mathbb{R}}
\newcommand{\C}{\mathbb{C}}
\newcommand{\bH}{\mathbb{H}}
\newcommand{\Z}{\mathbb{Z}}

\newcommand{\cC}{\mathcal{C}}

\newcommand{\cE}{\mathcal{E}}
\newcommand{\cF}{\mathcal{F}}

\newcommand{\cH}{\mathcal{H}}

\newcommand{\cK}{\mathcal{K}}
\newcommand{\cL}{\mathcal{L}}         
\newcommand{\cM}{\mathcal{M}}         
\newcommand{\cN}{\mathcal{N}}         

\newcommand{\cR}{\mathcal{R}}
\newcommand{\cS}{\mathcal{S}}
\newcommand{\cT}{\mathcal{T}}

\newcommand{\cV}{\mathcal{V}}

\newcommand{\E}{\mathcal{E}}

\nc{\G}{\Gamma}
\nc{\g}{\gamma}
\nc{\al}{\alpha}
\nc{\del}{\delta}
\nc{\Del}{\Delta}
\nc{\om}{\omega}
\nc{\Lam}{\Lambda}
\nc{\taul}{\lambda}
\nc{\lam}{\lambda}
\nc{\ka}{\kappa}

\nc{\w}{\tau}

\nc{\Om}{\Omega}
\nc{\ta}{\tau}
\nc{\io}{\iota}
\nc{\h}{\theta}
\nc{\z}{\zeta}
\nc{\s}{\sigma}
\nc{\Si}{\Sigma}
\nc{\vphi}{\varphi}

\renewcommand{\k}{\chi}

\renewcommand{\a}{\alpha}
\renewcommand{\b}{\beta}

\renewcommand{\d}{\delta}

\newcommand{\e}{\epsilon}
\newcommand{\eps}{\epsilon}

\nc{\si}{\sigma}

\nc{\Omt}{\Omega_\tau}
\newcommand{\Oms}{\Om^*_\tau}

\newcommand{\LAT}{\mathcal{L}} 
\newcommand{\LATt}{\mathcal{L}_\tau} 
\newcommand{\LATo}{\mathcal{L}_\omega}

\newcommand{\psio}{\psi_{\tau}}
\newcommand{\psit}{\psi_{\tau}}
\newcommand{\phik}{\phi_k}

\newcommand{\at}{a_{\tau}}
\newcommand{\lamt}{\lambda_{\tau}}

\newcommand{\rhoc}{\cR}

\newcommand{\ut}{u_{\tau}}

\newcommand{\Lt}{L_\tau}

\newcommand{\Lo}{L_\tau}

\newcommand{\Lk}{L_k}

\newcommand{\Lr}{L_{\tau}}

\newcommand{\Kt}{K_\tau}

\newcommand{\Kk}{K_k}

\newcommand{\cKk}{\cK_k}

\newcommand{\No}{N_\tau}

\newcommand{\Nt}{N_\tau}

\newcommand{\Pk}{P_k}

\newcommand{\hr}{h_{\tau}}

\newcommand{\hta}{h_{\tau}}

\newcommand{\Ggl}{G^{\rm glob}}

\renewcommand{\ta}{\tilde\alpha}

\newcommand{\ak}{a_k}
\newcommand{\bk}{b_k}
\newcommand{\ck}{c_k}

\newcommand{\vpp}{h}

\newcommand{\U}{\mathcal U}
\newcommand{\hrho}{\mathcal R}

\newcommand{\sH}{\mathscr{H}}

\newcommand{\Hil}{\mathcal H} 
\newcommand{\Hilk}{\mathcal H_k}
\newcommand{\cHk}{\mathcal H_k}

\newcommand{\fK}{\mathcal H^\oplus}

\newcommand{\hw}{\underbar v}

\newcommand{\uvk}{\underbar v_k}

\newcommand{\rf}{t^{\rm refl}}

\nc{\oP}{\overline P}

\nc{\bP}{\bar{P}}
\nc{\bQ}{\bar{Q}}
\nc{\bL}{\bar{L}}

\nc{\zb}{\underbar{z}}
\nc{\pb}{\underbar{p}}
\nc{\bb}{\underbar{b}}

\newcommand{\refl}{t^{\rm refl} }

\nc{\ran}{\rangle}
\nc{\lan}{\langle}
\nc{\rano}{\rangle_\Om}

\renewcommand{\Re}{\operatorname{Re}}
\renewcommand{\Im}{\operatorname{Im}}
\newcommand{\im}{\operatorname{Im}}
\newcommand{\re}{\operatorname{Re}}

\newcommand{\one}{\mathbf{1}}
\newcommand{\id}{{\bfone}}
\nc{\bfone}{{\bf 1}}

\newcommand{\p}{\partial}
\newcommand{\n}{{\nabla}}

\newcommand{\pa}{\partial_{a}} 
\newcommand{\pat}{{\partial_{a_{\tau}}}}

\newcommand{\nat}{{\nabla_{a_{\tau}}}}
\newcommand{\nao}{{\nabla_{a^0}}}

\newcommand{\db}{{\Delta_{a^0}}}

\newcommand{\Curl}{\operatorname{curl}}
\newcommand{\curl}{\operatorname{curl}}
\newcommand{\CURL}{\operatorname{curl}}
\newcommand{\divv}{\operatorname{div}}

\newcommand{\DIV}{\operatorname{div}}

\nc{\dA}{\nabla_A}

\newcommand{\COVGRAD}[1]{\nabla_{\!\!#1}}

\newcommand{\COVLAP}[1]{\Delta_{\!#1}}

\newcommand{\ra}{\rightarrow}

\newcommand{\NULL}{\operatorname{Null}}
\newcommand{\RANGE}{\operatorname{Ran}}
\newcommand{\Null}{\operatorname{Null}}
\newcommand{\Ran}{\operatorname{Ran}}

\newcommand{\diag}{\operatorname{diag}}
\newcommand{\dist}{\operatorname{dist}}

\newcommand{\Span}{\operatorname{Span}}

\newcommand{\ls}{\lesssim}
\newcommand{\gs}{\gtrsim}

\newcommand{\ThreeByOne}[3]{\left( \begin{array}{c} #1 \\ #2 \\ #3 \end{array} \right)}

\newcommand{\LA}[2]{\vec{\mathscr{L}}_{}^{}(\tau)}
\newcommand{\HA}[2]{\vec{\mathscr{H}}_{}^{}(\tau)}

\newcommand{\DETAILS}[1]{}

\numberwithin{equation}{section}

\begin{document}

\title{On stability of Abrikosov vortex lattices} 
\date{August 10, 2016}   

\author{Israel Michael Sigal
\thanks{ Dept.~of Math.,
Univ. of Toronto, Toronto, Canada; Supported by NSERC Grant No. NA7901.}\ 
\qquad \qquad
Tim Tzaneteas
\thanks{Math. Institut, Univ. T\"ubingen, T\"ubingen, Germany.} \\
}

\maketitle


\begin{abstract}
\DETAILS{  We consider Abrikosov-type vortex lattice solutions of the Ginzburg-Landau equations of superconductivity. 
We study stability of such solutions within the context of the time-dependent Ginzburg-Landau equations - the Gorkov-Eliashberg-Schmid equations. We consider  superconductors filling the entire space.  For  magnetic fields close to the second critical magnetic field  and for arbitrary lattice shapes, 

 The key groups solution of the Ginzburg-Landau equations of superconductivity  are  vortices and  Abrikosov vortex lattices. Existence theory for these solutions, as well as the stability theory - within the context of the time-dependent Ginzburg-Landau equations - the Gorkov-Eliashberg-Schmid equations and the standard cylindrical geometry - of vortices  is well developed. We study stability of Abrikosov vortex lattices. } 
 
  The Ginzburg-Landau equations 
play a key role in superconductivity and particle physics. 
They  inspired many imitations in other areas of physics.  
These equations  have two remarkable classes of solutions -   vortices and  (Abrikosov) vortex lattices. For the standard cylindrical geometry, the existence theory for these solutions, as well as the stability theory of vortices  are well developed. The latter is done within the context of the time-dependent Ginzburg-Landau equations - the Gorkov-Eliashberg-Schmid equations of superconductivity - and the abelian Higgs model of particle physics. 

We study stability of Abrikosov vortex lattices under finite energy perturbations satisfying a natural parity condition (both defined precisely in the text) for the dynamics given by the Gorkov-Eliashberg-Schmid equations.
  For  magnetic fields close to the second critical magnetic field  and for arbitrary lattice shapes, 
 we  prove that there exist two functions on the space of lattices, 
 such that  Abrikosov vortex lattice solutions are asymptotically stable, provided  
 the superconductor is of Type II and these functions are positive, and unstable, for  superconductors of Type I, or if one of these functions is negative. 
 \end{abstract}



\section{Introduction}\label{sec:introduction}

\subsection{Problem and results}\label{subsec:backgr} 
The macroscopic theory of superconductivity is a crown achievement of condensed matter physics, 
 presented in any book on superconductivity and solid state or condensed matter physics.  
It was developed 
along the lines of Landau's theory of the second order phase transitions before the microscopic theory was discovered.  At the foundation of this theory lie  the celebrated Ginzburg-Landau equations,
\begin{equation} \label{gle} \begin{cases}
    -\COVLAP{A}\Psi = \kappa^2(1 - |\Psi|^2)\Psi, \\
    \CURL^*\CURL A = \im(\bar{\Psi}\COVGRAD{A}\Psi),
\end{cases} \end{equation}
 which describe superconductors in thermodynamic equilibrium.  
Here $\Psi$ is a complex-valued function, called the order parameter, $A$ is a vector field (the magnetic potential), $\kappa$ is a positive constant,  called the  Ginzburg-Landau parameter,  $\COVGRAD{A} = \nabla - iA$ and $\COVLAP{A} = \COVGRAD{A}\cdot \COVGRAD{A}$ are the covariant gradient and Laplacian.
Physically, $|\Psi|^2$ gives the (local) density of superconducting electrons (Cooper pairs), $B=\CURL A$ is the magnetic field. The second equation is Amp\`ere's law with $J_S = \im(\bar{\Psi}\COVGRAD{A}\Psi)$ being the supercurrent associated to the electrons having formed  Cooper pairs.

We assume, as common, that superconductors fill in all of $\R^2$ (the cylindrical geometry in $\R^3$). In this case,  $\CURL A := \partial_{x_1}A_2 - \partial_{x_2}A_1$ and $\CURL^* f = (\partial_{x_2}f, -\partial_{x_1}f)$. 

By far, the most important and celebrated solutions of  the Ginzburg-Landau equations 
are magnetic vortex lattice solutions, discovered by Abrikosov (\cite{Abr}), and  known as (Abrikosov) vortex lattice solutions or simply {\it Abrikosov or vortex lattices}. Among other things, understanding these solutions is important for maintaining the superconducting current in Type II superconductors, i.e., for $\kappa> \frac{1}{\sqrt{2}}$.

Abrikosov lattices have been extensively studied in  the  physics literature. Among many rigorous results, we mention that the existence of these solutions was proven rigorously in  \cite{Odeh, BGT, Dut2, Al,  TS, ST2}. 
Moreover,  important and fairly detailed results on asymptotic behaviour of solutions, for $\kappa \to\infty$ and applied magnetic fields, $h$, satisfying $h\le \frac{1}{2}\ln \kappa$+const (the London limit), were obtained in \cite{as} (see this paper and the book \cite{ss} for references to earlier work). Further extensions to the Ginzburg-Landau equations for anisotropic and high temperature superconductors in the $\kappa \to\infty$ regime can be found in \cite{abs1, abs2}. (See \cite{GST, S} for reviews.) 

In this paper we are interested in dynamics of  of the Abrikosov lattices, as described by 
 the time-dependent generalization of  the Ginzburg-Landau equations proposed by Schmid (\cite{Sch}) and Gorkov and Eliashberg (\cite{GE})  (earlier versions are due to Bardeen and Stephen and Anderson, Luttinger and Werthamer, see \cite{Cyrot, Brandt} for reviews). These equations are of the form
\begin{equation}\label{GES}
\begin{cases}
    \chi \partial_{t,\Phi} \Psi = \COVLAP{A}\Psi + \kappa^2(1 - |\Psi|^2)\Psi, \\
 \sigma \partial_{t,\Phi} A   = -\CURL^*\CURL A  + \im(\bar{\Psi}\COVGRAD{A}\Psi).
\end{cases}
\end{equation}
Here $\Phi$ is the scalar (electric) potential, $\chi$, a complex number, and $\sigma$, a two-tensor, and $\partial_{t\Phi}$ is the covariant time derivative $\partial_{t,\Phi}(\Psi, A)  = ((\partial_t + i\Phi)\Psi, \partial_t A + \nabla\Phi)$.
The second equation is Amp\`ere's law,  $\CURL  B=J$, with $J=J_N +J_S,$ where $  J_N= -\sigma (\partial_t A + \nabla\Phi)$ (using Ohm's law) is the normal current associated to the electrons not having formed Cooper pairs, and $J_S = \Im(\bar{\Psi}\COVGRAD{A}\Psi)$, the supercurrent.  
\DETAILS{We use 
 the gauge transformation
$T_\eta$ (see \eqref{gauge-sym}), with  $\eta(x,t) = \int_0^t \Phi(x, s) ds, $ 
to achieve the gauge
\begin{equation}\label{Phi0}  \Phi (x, t)=0, \end{equation} 
which we assume from now on.

$$*********$$}
 \DETAILS{We use a gauge transformation (see \eqref{gauge-sym}), 
to achieve 
\begin{equation}\label{divA0}  \divv A (x, t)=0, \end{equation} 
which we assume from now on.}

Eqs \eqref{GES}, which we call the  Gorkov-Eliashberg-Schmid equations (also known as the Gorkov-Eliashberg or  the time-dependent Ginzburg-Landau equations), 
 have a much narrower range of applicability than  the Ginzburg-Landau equations (\cite{Tink}) and many refinements have been proposed. However, though improvements of these equations are, at least notationally, rather cumbersome, 
 they do not alter the mathematics involved 
 in an essential way.

The Abrikosov lattices are solutions, $(\Psi, A)$, to \eqref{gle}, i.e. static solutions to \eqref{GES}, whose physical characteristics, 
 $|\Psi|^2$, 
  $\CURL A$,  and 
  $J_S = \im(\bar{\Psi}\COVGRAD{A}\Psi)$ are double-periodic w.r. to a lattice $\LAT\subset \R^2$. 
   They are static solutions to \eqref{GES} and 
their stability w.r. to the dynamics induced by these equations is an important issue. 

In \cite{ST2},  we considered the stability of the Abrikosov lattices under the simplest perturbations, namely those preserving their periodicity 
 (we call such perturbations \emph{gauge-periodic}). For a lattice $\LAT$ 
of arbitrary shape,  
  and with the area, $|\Omega|$, of the fundamental domain, $\Om$ (or the torus $\R^2/\LAT$), close to $\frac{2\pi}{\kappa^2}$, 
we proved 
that, under gauge-periodic perturbations,
\begin{itemize}
\item[(i)] \emph{ Abrikosov vortex lattice solutions are asymptotically stable for}  $\kappa > \kappa_c(\LAT)$;
 \item[(ii)]   \emph{Abrikosov vortex lattice solutions are unstable for} $\kappa < \kappa_c(\LAT)$.
\end{itemize}
Here $\kappa_c(\LAT)$ is the function lattices given by $\kappa_c(\LAT) := \sqrt{\frac{1}{2}\left(1-\frac{1}{\beta (\LAT)}\right)}$, where $\beta(\LAT)$  is  the Abrikosov 'constant', defined in Remark 2)  below. 

Due to the magnetic flux quantization property - see Subsection \ref{sec:abr-lat},  the condition that $|\Om |$ is close to $\frac{2\pi}{\kappa^2}$ means that the average magnetic field (or magnetic flux),  per lattice cell, 
 $ b = \frac{1}{|\Omega|} \int_\Omega  \CURL A ,$ is close to the second critical magnetic field $h_{c2}=\kappa^2$.  
For the definitions of 
  various stability notions, see Subsection  \ref{sec:pert}. 

 This result 
 shows remarkable stability of  the Abrikosov vortex lattices and  it seems this is the first time the threshold 
$\kappa_c(\LAT)$  has been isolated.

Gauge-periodic perturbations are not a common type of perturbations occurring in superconductivity.
In this paper we address the problem of the stability of Abrikosov lattices under local, or  finite-energy, perturbations (defined in  Subsections \ref{sec:pert} below) satisfying a natural parity condition (see \eqref{parity0} below). 

 We consider  lattices, $\LAT$, of arbitrary shape and with the standard topology (see below) and denote by $u_\LAT=(\Psi_\LAT, A_\LAT)$ the Abrikosov lattice solution for a lattice $\LAT$.   
We denote by  $\LAT^*$  the lattice reciprocal to $\LAT$. 
  It consists of all vectors $s^* \in \R^2$ such that $s^* \cdot s \in 2\pi\Z$,  for all $s \in \LAT$. 
Assuming that  a lattice 
 $\LAT$ 
 has the area of fundamental cell 
 close to $\frac{2\pi}{\kappa^2}$, 
 we have
          \begin{itemize}
\item  There exist two families $\gamma_{k}(\LAT), k\in \R^2/\LAT^*,$ and $\mu (\LAT, \kappa)$ of real, smooth 
 functions on lattices $\LAT$,   s.t.  under  finite-energy perturbations, satisfying the parity condition, \eqref{parity0} below, the Abrikosov lattice solution $u_\LAT$   is 

\smallskip

\emph{asymptotically stable} for $\LAT$ and $\kappa$ satisfying   $\kappa> \frac{1}{\sqrt{2}}$, 
$\gamma_{k}(\LAT)>0,\ \forall k \ne 0,$ and $\mu (\LAT, \kappa)>0$; 
 and 
 
\smallskip

 \emph{energetically unstable} if either $\kappa < \frac{1}{\sqrt{2}}$, or $\inf_k\gamma_{k}(\LAT) <0$, or $\mu (\LAT, \kappa) < 0$. 
 \item  The functions  $\gamma_{k}(\LAT)$ and $\eta (\LAT, \kappa)$ 
 satisfy     $\gamma_{k}(\lam\LAT)=\gamma_{\lam^{-1} k}(\LAT),\ \forall\lam\ne 0,$ $\gamma_{k=0}(\LAT)= 0$,  and $\eta (\lam\LAT, \kappa)=\eta (\LAT, \kappa)$. 
 \end{itemize}
We have a decent understanding of the function $\gamma_{k}(\LAT)$, which is defined and discussed below, and only a partial understanding of the function $\eta (\LAT, \kappa)$. 
 By expressing $\g_k (\tau)$ as fast convergent series (see \eqref{gamk-series} below) and using numerical computations, we show that  (see \eqref{gamk-comp} for a precise statement) 
    
 $\gamma_{k}(\LAT)> 0\ \forall k \ne 0,$ for  $\LAT$ is hexagonal, 
  and $\inf_{k}\gamma_{k}(\LAT)<0$, if $\LAT$ is not hexagonal.

\smallskip
We can give an explicit form of  $ \eta (\LAT, \kappa)$ (see Remark \ref{rem:eta} below), but the derivation of the series representation for  $\eta (\LAT, \kappa)$ is substantially more complicated and is done in separate work \cite{OS}. Using these series and using numerical computations for $\LAT=\LAT_{\rm hex}$ hexagonal, it is shown in \cite{OS} that

$\eta (\LAT_{\rm hex}, \kappa)> 0$ for  $\kappa> \frac{1}{\sqrt{2}}$  
  and $\eta(\LAT_{\rm hex}, \kappa)<0$, if $\kappa < \frac{1}{\sqrt{2}}$. 
  
 

We explain the origin of 
the functions $\g_{k} (\LAT) $ and $\eta (\LAT, \kappa)$ entering the statement of our results above. 
%
%
 Let $ L^\LAT $ be the operator obtained by the (complex) linearization of the map on the r.h.s. of \eqref{GES} at the vortex lattice solution $u_\LAT = (\Psi_\LAT, A_\LAT)$. ($L^{\LAT}$ is the complex linear hessian, $\E''(u_\LAT)$, of the Ginzburg-Landau energy functional \eqref{gl-en} at $u_\LAT$, see Subsection \ref{sec:hessian}.) 
The key signature of stability of the static solution, $u_\LAT$, is  the behaviour of the low energy the spectrum of the operator $L^{\LAT}$:    
   $u_\LAT$ is likely to be unstable if  $ L^\LAT $ has some negative spectrum and stable, if $ L^\LAT \ge 0$, with $0$ being an isolated eigenvalue, i.e.  its continuous spectrum has a gap at $0$.  The difficult case is when  $ L^\LAT \ge 0$ and is gapless,  i.e.  its continuous spectrum begins at $0$.  In the latter case, the central role is played by the detailed nature of this continuous spectrum (the dispersion relation) 
    at $0$ (and its interaction with the nonlinearity).  
  
 We say that an operator $L$ has the ($\LAT-$) {\it band spectrum} iff there are functions $\nu_k^j$, $k\in \R^2/\LAT^*$,  s.t. $\s(L)=\cup_j \Ran \nu^j$. ($k$ is called quasimomentum.) Let   $\e$ be a small parameter proportional to $\sqrt{1-b/\kappa^2}$, with $b=\frac{2\pi}{|\Om|}$, defined in \eqref{eps} below. We show that the operator $L^{\LAT}$ 
 has the band spectrum 
 with two gapless 
bands of the form 
\begin{align} \label{gapless-branch1} 
\nu_k^1 (\LAT, \kappa) = & c_1 (\LAT)  \frac{\e ^2|k|^{2}}{\e ^2 + |k|^{2}} [(\kappa^2-\frac12) \g_k (\LAT)+ \eta(\LAT, \kappa)\e^2] + O(\e^4 |k|^{2}),\\
 \label{gapless-branch2} & |k|^2 \ls \nu_{k}^2  (\LAT, \kappa)\ls  |k|^2, 
\end{align} 
where 
 $c_1 (\LAT) \gtrsim 1$. 
 The remaining bands have a gap of the order, at least, $\e^2$. We see that the second branch, $\nu_{k}^2$, is always positive. 

 
 By the definition (see \eqref{gamk-LAT} below), $\g_k(\LAT)=O(|k|^2)$ and therefore $(\kappa^2-\frac12) \g_k(\LAT)$ is the leading term in \eqref{tilde-gam1-expan} for $|k|\gg \e$ and $\eta(\LAT, \kappa)\e^2$, for $|k|\ll \e$, provided $\kappa^2 > \frac12$.  (Numerics show that $\g_k(\LAT)=O(|k|^4)$, in which case the inequalities become, $|k|\gg \sqrt\e$ and $|k|\ll \sqrt\e$.)
Checking out \eqref{gapless-branch1}, we conclude that 

 \begin{itemize} \item $\nu^{1}_k>0,\ \forall k\ne 0$, if $ (\kappa^2-\frac12) \g_k(\LAT)>0, \forall k\ne 0,$ and $\eta(\LAT, \kappa)>0$; 
\item   $\inf_k\nu^{1}_k<0$, if either $ (\kappa^2-\frac12) \inf_k\g_k(\LAT)<0$, or/and $\eta(\LAT, \kappa)<0$. 
 \end{itemize}
When either $ \g_k(\LAT)=0 $, for some $k\ne 0$, or $\eta(\LAT, \kappa)=0$, one has to go to the higher order. 

  



The gapless spectral branches 
\eqref{gapless-branch1} - \eqref{gapless-branch2} are due to breaking of 
 translational symmetry (by $u_\LAT$) and represent, what is known in particle physics as, the  Goldstone excitation spectrum. (For more detail see Subsection \ref{sec:approach}.) 
Breaking of the global gauge 
 symmetry leads to the zero eigenvalue for all $k$. 
 
 %
\DETAILS{ The function  $\gamma_{\e \kappa}(\tau)$ is defined as 
 \begin{align} \label{gam-epskapp}\gamma_{\e \kappa}(\tau):=\inf_k \g_{k \e \kappa}^- (\tau).\end{align}} 
We define the function $\g_{k} (\LAT) $. Let $\lan f \ran_\LAT$ stand for the average, $\lan f \ran_\LAT := \frac{1}{|\Omega|} \int_{\Omega} f ,$ of a function $f$ over a fundamental domain $\Om$ of $\LAT$ (or  $\R^2/\LAT$) and $\phi_k\equiv \phi_k^\LAT$ be the unique (up to a factor, of course) solutions of the equation
\begin{align}\label{phik-eqs-LAT} (-\COVLAP{a^0}-1)\phi =0,\ \quad \phi (x+s) =e^{i g_s(x)} 
e^{i k\cdot s}\phi (x),\ \quad \forall s\in \LAT,
	\end{align}
 with 
$a^0(x):= \frac{\pi }{|\Om|}   x^\perp$, $x^\perp:= (- x_2,  x_1)$, and   $ g_s(x):= \frac{\pi }{|\Om|}  s^\perp\cdot x+ c_s,$  where $c_s$ are numbers satisfying $ c_{s+t} - c_s - c_t  + \frac{\pi }{|\Om|}  s \wedge t \in 2\pi\Z$,  normalized as $\lan|\phi_k|^2\ran_{\LAT} =1$.   
 We define  
\DETAILS{ \begin{align}\label{gam-del} \gamma_{\del}(\tau)\equiv \gamma_{\del}(\LAT_\tau):=\inf_{\dist(k, \LAT^*_\tau)\ge \del} \g_{k} (\tau),\end{align}
  where the function  $\gamma_{k}(\LAT), k\in \C/\LAT^*_\tau,$ $\im\tau>0$, is given by}
\begin{align}\label{gamk-LAT}\g_k (\LAT) :=& 2\lan|\phi_0|^2|\phi_k|^2\ran_{\LAT} - |\lan \phi_0^2 \bar{\phi}_{- k}\bar{\phi}_{k} \ran_{\LAT}| - \lan|\phi_0|^4\ran_{\LAT}. 
\end{align}

 The definition of $\phi_k$ and the periodicity of the character $\chi(s) = e^{ik\cdot s}$ in  \eqref{phik-eqs-LAT} imply that 
 \begin{align}\label{phik-per-k} \phi_{k+\xi} = \phi_k,\  \text{ and therefore }\ \g_{k+\xi}(\LAT)=\g_k(\LAT),\ \forall \xi\in \LAT^*. 
 	\end{align} 
Hence we can take $k$ in $\R^2$, or in $\R^2/\LAT^*$, or in 
an elementary cell of the dual lattice. 

Properties of the  functions $\g_{k} ( \LAT)$, as well a series representation for it, are described in Subsection \ref{sec:main-res} and Section \ref{sec:phik-etc}.

 We identify $\R^2$ with $\C$, via the map $(x_1, x_2)\ra x_1+i x_2$,  and, for $\tau\in \C,$ $\im\tau>0$, introduce the normalized lattice
   \begin{equation}\label{LATtau}\LATt:=\sqrt{\frac{2\pi}{\im\tau} }  (\Z+\tau\Z). 
     \end{equation} 
We call $\tau$ the shape parameter.
 We denote  
 \begin{align}\label{gamk-tau} \gamma_{k}(\tau)\equiv \gamma_{k}(\LATt)\ \text{ and }\  \eta (\tau, \kappa)\equiv \eta (\LATt, \kappa)).\end{align}  
 By expressing $\g_k (\tau)$ as a fast convergent series (see \eqref{gamk-series} below) and using numerical simulations (with MATLAB with  a meshwidth of $0.01$), we show that 
 \begin{align}\label{gamk-comp}\g_k (e^{i\pi/3}) > 0\ \quad \forall k\ne 0\ 
 \text{ and }\  \inf_k\g_k ( e^{i\pi\al}) < 0,\  \text{ for }\  
  \al \notin [0.28, 0.41],  
\end{align} 
 which 
  implies the second statement above. Moreover, we show that $\sup_k\g_k (\tau) > 0$. 

If we define $\g_k (\tau)$ on the entire Poincar\'e half plane $\bH$, then, since $\LAT_{g\tau}=\LATt, \forall g\in SL(2, \Z)$, it is invariant under the action of the modular group $SL(2, \Z)$,
\[\g_k (g\tau) = \g_k (\tau)\ \quad  \forall g \in SL(2, \Z).\]
 Hence, it can be defined on the fundamental domain, $\bH/SL(2, \Z)$, of this group acting on $\bH$. ($\bH/SL(2, \Z)$ is given explicitly as
     $\bH/SL(2, \Z)=\{\tau\in \C: \im\tau > 0,\ |\tau| \geq 1,\ -\frac{1}{2} < \re\tau \leq \frac{1}{2} \}$, 
   see Fig. \ref{fig:PoincareStrip}.) Apart for this, we show in Section \ref{sec:phik-etc} that 
   \begin{align}\label{gamk-prop}  \g_{\bar k}(- \bar\tau)= \g_{k}( \tau),\end{align}  
   so that it suffices to consider $\g_{k}( \tau)$ ion the half of the fundamental domain $\bH/SL(2, \Z)$. 

 %
 
   
%
\DETAILS{the lattice shape parameter,  $\tau\in \C$, $\Im\tau > 0$,  is defined by identifying $\R^2$ with $\C$, via the map $(x_1, x_2)\ra x_1+i x_2$ and  bringing a lattice (using translations,  rotations and rescaling, if necessary) into the form $\LAT_\tau=\sqrt{\frac{2\pi}{\im\tau} }  (\Z+\tau\Z)$,}

\begin{figure}[h!]
	\centering 

  \includegraphics[width=2.5in]{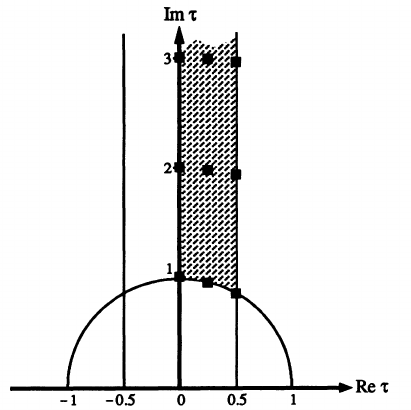} 
  \caption{Fundamental domain  of $\gamma(\tau)$. }\label{fig:PoincareStrip}
\end{figure}  

\begin{remark} \label{rem:eta}   (a)          
An explicit expression for $\eta (\tau, \kappa)$ is derived in Lemma \ref{lem:tilde-gam12-expan} in Appendix \ref{sec:prop-Lk-spec-pf}. While $\g_k (\tau)$ is expressed in terms of solutions of the linear problem \eqref{phik-eqs-LAT}, see \eqref{gamk-LAT}, 
  the function $\eta$ involves the subleading terms in the expansion of the solutions of the GLEs \eqref{gle} in $\e$. Again these can be expressed  in terms of $\theta-$series, but the expressions involved become  cumbersome, with some limited results announced in \cite{OS}.

(b) The  the Abrikosov constant, $\beta(\tau)$ (=$\beta(\LAT)$), is defined as  \begin{align} \label{AF} \beta(\tau) =  \lan |\phi_0|^4 \ran_{\LATt}\end{align} and is related to $\g_k ( \tau)$ as  $\beta(\tau)= \frac12 \g_0( \tau)$. 
   The term  \TIMC{} Abrikosov constant  comes from the physics literature, where one often considers only equilateral triangular or square lattices.
   
(c)
   The definition \eqref{AF} implies that 
    $\beta(\tau)$ is manifestly independent of $b$. Our definition differs from the standard one by rescaling: the standard definition uses the function   $ \phi_0(\sqrt b x)$, instead of  $\phi_0(x) $.

 (d) 
  While $\beta(\tau) $ is defined in terms of the standard theta function,  $\g_{k}(\tau) $ is defined in terms of theta functions with finite characteristics,  see Section \ref{sec:phik-etc} below.
 %
\end{remark} 

We believe that the methods we develop are fairly robust  and can be extended - at the expense of significantly more technicalities - to substantially wider classes of perturbation. 
 Moreover, the same techniques could be used in other problems of pattern formation, which are ubiquitous in applications.

In the rest of this section we  introduce some basic definitions, 
present our results and  sketch the approach and possible extensions.


\subsection{
Ginzburg-Landau 
energy}


 The Ginzburg-Landau equations are the Euler-Lagrange equations for the Ginzburg-Landau energy functional 
\begin{equation}
\label{gl-en}
    \E_Q(\Psi, A) = \frac{1}{2} \int_Q \left\{ |\COVGRAD{A}\Psi|^2 + |\CURL A|^2 + \frac{\kappa^2}{2}(1 - |\Psi|^2)^2 \right\},
\end{equation}
where $Q$ is any domain in $\R^2$.  ($\E_{Q}$ is the difference in (Helmhotz) free energy between the superconducting and normal states.)

 The Gorkov-Eliashberg-Schmidt equations have the structure of a gradient-flow equation for $\E_Q(\Psi, A)$.
Indeed, they can be put in the form
\begin{equation}\label{grad-flow}
    \partial_{t\Phi}( \Psi,  A) =  - \lam\E'(\Psi, A),
\end{equation}
where $\lam:=\diag (\g^{-1}, \sigma^{-1})$, $\partial_{t\Phi}( \Psi,  A) := ((\partial_t + i\Phi)\Psi, \partial_tA + \nabla\Phi)$ and $\E'$ is the $L^2-$gradient of $ \E(\Psi, A) \equiv \E_{\R^2}(\Psi, A) $, defined as $\lan v, \E'(u) \ran_{\R^2}=
d \E(u)v$, with  $d$ being  the G\^{a}teaux derivative, $d F(u)v:= \frac{\partial}{\partial s} F(u + s v )\big|_{s=0}$. This definition implies that
\begin{equation}\label{gradient-E}
    \E'(\Psi, A) = (-\COVLAP{A}\Psi - \kappa^2(1 - |\Psi|^2)\Psi,\ \CURL^*\CURL A - \Im(\bar{\Psi}\COVGRAD{A}\Psi)).
\end{equation}

 We note that the symmetries above restrict to symmetries of the Ginzburg-Landau equations by considering time-independent transformations.

\subsection{Symmetries}

The Gorkov-Eliashberg-Schmidt equations~\eqref{GES} admit several 
 symmetries, that is, transformations which map solutions to solutions.

{\it Gauge symmetry}:  for any sufficiently regular function $\gamma : \R^+ \times \R^2 \to \R$,
\begin{equation}\label{gauge-sym}
    (\Psi(t,x), A(t,x), \Phi(t,x)) \mapsto (e^{i\gamma(t,x)}\Psi(t,x), A(t,x) + \nabla\gamma(t,x), \Phi(t,x) - \p_t\gamma(t,x));
\end{equation}

{\it Translation symmetry}: for any $h \in \R^2$,
\begin{equation}\label{transl-sym}
    (\Psi(t,x), A(t,x), \Phi(t,x)) \mapsto (\Psi(t,x+h), A(t,x+h), \Phi(t,x+h));
\end{equation}

{\it Rotation symmetry}: for any $\rho \in SO(2)$, 
\begin{equation}\label{rot-sym}
   (\Psi(t,x), A(t,x), \Phi(t,x)) \mapsto (\Psi(t,\rho^{-1} x), \rho^{-1}A(t, (\rho^{-1})^T x), \Phi(t,\rho^{-1} x)),
\end{equation}

{\it Reflection symmetry}:
\begin{equation}\label{refl-sym}
    (\Psi(t,x), A(t,x), \Phi(t,x))  \mapsto (\Psi(t, -x), - A(t, - x), \Phi(t,- x)).
\end{equation}

{To show that the Gorkov-Eliashberg-Schmidt equations~\eqref{GES} (or  the Ginzburg-Landau equations,
\eqref{gle}) are covariant under the rotation symmetry, one first passes the Coulomb gauge $\divv A=0$, so that $\CURL^*\CURL A=-\Delta A$, and then, after the rotation transformation, reverts to the original $\Psi$ and $A$.)
\medskip


\subsection{Abrikosov lattices}\label{sec:abr-lat}

As was mentioned above,  Abrikosov vortex  lattices (or just Abrikosov lattices), are solutions, whose 
 physical characteristics, density of  Cooper pairs,  $|\Psi|^2$, the magnetic field,  $\CURL A$,  and the supercurrent,  $J_S = \Im(\bar{\Psi}\COVGRAD{A}\Psi)$, are double-periodic w.r. to a lattice $\LAT\subset \R^2$.

 We note that the symmetries \eqref{gauge-sym} -  \eqref{refl-sym} map Abrikosov lattices to Abrikosov lattices. 
Moreover, for Abrikosov states, for $(\Psi, A)$,  the magnetic flux, $
\int_\Omega \Curl A$, through a lattice cell, $\Om$, is quantized,
\begin{equation}\label{flux-quant}
    \int_\Omega \Curl A = 2\pi n,
\end{equation}
for some integer $n$.  Indeed, the periodicity of $n_s=|\Psi| ^2$ and $J=\im(\bar \Psi \n_A \Psi)$ imply that $\nabla\theta - A$, where $\Psi = |\Psi|e^{i\theta}$, is periodic, provided $\Psi \neq 0$ on $\partial\Omega$.  This, together with Stokes's theorem, $\int_\Omega \Curl A = \oint_{\partial\Omega} A = \oint_{\partial\Omega} \nabla\theta$ and the single-valuedness of $\Psi$,   implies \eqref{flux-quant}. 
 Using the reflection symmetry of the problem, one can easily check that we can always assume $n \geq 0$.

 Equation \eqref{flux-quant} implies the relation between the average magnetic field (or magnetic flux), $b$, per lattice cell, defined as,
 \begin{equation}\label{av-magn-flux}b = \frac{1}{|\Omega|} \int_\Omega  \CURL A , \end{equation}
  and the area of a fundamental cell, namely,
\begin{equation}\label{quant-cond}
     b = \frac{2\pi n}{|\Omega|}.
\end{equation}
Due to the quantization relation \eqref{quant-cond}, the parameters $\tau$, $b$, and $n$ determine the lattice $\LAT$ up  to a rotation and a translation. 
 Applying a rotation, if necessary, any lattice $\LAT$ can be brought to the form 
 \begin{equation}\label{LATom}\LAT_\om=r  (\Z+\tau\Z),\ \text{ where }\ \om=(\tau, r),\ r>0,\ \tau\in \C,\ \Im\tau > 0.\end{equation} 
Due to the condition, \eqref{quant-cond}, $r$ and $b$ are connected as $r^2 = \frac{2\pi n}{b\im \tau}$.
As the equations are invariant under rotations and translations, we can always assume that the underlying lattice is of the form \eqref{LATom}. 

In what follows, we {\it restrict ourselves to the case \eqref{flux-quant} with} $n=1$. 
\DETAILS{We will say that a gauge-periodic pair $(\Psi, A)$ is of type $\om=
(\tau, b)$, if the underlying lattice has the shape parameter $\tau$ and the scale parameter $r = \sqrt{\frac{2\pi }{b \Im\tau}}$ (and the average magnetic flux per lattice cell is equal to $b$, and there are $n$ quanta of magnetic flux per lattice cell $r^2 = \frac{2\pi }{b \Im\tau}$).}
We denote by $u_\om=(\Psi_\om, A_\om), \om=(\tau, r),$ the Abrikosov lattice solution for the lattice $\LAT_\om$.  Such a solution has the average magnetic flux per lattice cell equal to $b = \frac{2\pi }{r^2 \Im\tau}$.

Recall the definition of the Ginzburg - Landau parameter threshold $ \kappa_c(\tau)$ given in 
\begin{equation} 
     \kappa_c(\tau) \equiv \kappa_c(\LAT_\tau):= \sqrt{\frac{1}{2}\left(1-\frac{1}{\beta (\tau)}\right)} ,
     \end{equation}
where, recall, $\beta(\tau) \equiv \beta(\LAT_\tau)$ is the  the Abrikosov constant, defined as  $\beta(\tau) =  \lan |\phi_0|^4 \ran_{\LATt}=\frac12 \g_0( \tau)$. We have the following existence theorem (see \cite{TS2}, for the first result and its further elaborations see \cite{Odeh, BGT, Dut2, TS}).
\begin{theorem} \label{thm:existence}
 For any 
 $\tau\in \C,\ \im\tau>0,$ and 
  for any $b>0$, such that 
   \begin{equation} \label{b-cond}
	|1 - b/\kappa^2|\ll (2\kappa^2 - 1)\beta(\tau) + 1,	 
\end{equation} 
 and 
 \begin{equation}\label{b-cond-exist}
  ( \kappa - \kappa_c(\tau))( \kappa^2 -b ) >0, \end{equation}  
there exists a smooth Abrikosov lattice solution $u_\om = (\Psi_\om, A_\om)$, $\om =(\tau, r)$, with  $r = \sqrt{\frac{2\pi }{b \Im\tau}}$. 
\end{theorem}
More detailed properties of these solutions are given in Subsection \ref{sec:rescaling} below. As we deal only with the case $n = 1$, we now assume that this is so and drop $n$ from the notation. 

It is key to realize that a state $(\Psi, A)$ is an Abrikosov lattice if and only if  $(\Psi, A)$ is gauge-periodic or  gauge-equivariant (with respect to the lattice $\LATt$)in the sense that
there exist (possibly multivalued) functions $g_s : \R^2 \to \R$, $s \in \LAT$, such that
\begin{equation}\label{gauge-per'}
(\Psi(x+s), A(x+s))=(e^{i g_s(x)}\Psi(x), A(x) + \nabla g_s(x)).\end{equation}
 Indeed, if state  $(\Psi, A)$ satisfies \eqref{gauge-per'}, then all associated physical quantities are $\LAT-$periodic, i.e. $(\Psi, A)$ is an Abrikosov lattice. 
In the opposite direction, if $(\Psi, A)$ is an Abrikosov lattice, then $\curl A(x)$ is periodic w.r.to $\LAT$,  and therefore   $A(x + s) = A(x) +\n g_s(x)$,  for some functions $g_s(x)$. 
Next, we write $\Psi(x)=|\Psi(x)|e^{i \phi(x)}$. Since $|\Psi(x)|$ and $J(x)= |\Psi(x)|^2 (\n \phi(x)-  A(x))$ are  periodic w.r.to $\LAT$,  we have that   $\n \phi(x + s) = \n \phi(x) +\n \tilde g_s(x)$, which implies that   $\phi(x + s) =\phi(x) + g_s(x)$, where $g_s(x)=\tilde g_s(x)+ c_s$, for some constants $c_s$. 


\subsection{Finite-energy ($H^1-$) perturbations}\label{sec:pert}

We now wish to study the stability of these Abrikosov lattice solutions under a class of perturbations that have finite-energy.
More precisely, we fix an Abrikosov lattice solution $u_\om$ and consider perturbations $v : \R^2 \to \C \times \R^2$ that satisfy
\begin{equation}\label{Lambda}
\Lambda_{u_\om}(v) = \lim_{Q\ra\R^2} \big(\E_{Q}(u_\om + v) - \E_{Q}(u_\om)\big) < \infty.
\end{equation}
Clearly, $\Lambda_{u_\om}(v)<\infty$, for all vectors of the form $v=T^{\rm gauge}_\gamma  u_\om -  u_\om$, where $T^{\rm gauge}_\gamma : 
    (\Psi(x), A(x)) \mapsto (e^{i\gamma( x)}\Psi( x), A( x) + \nabla\gamma( x))$ and $\g \in  H^2(\R^2;\R)$.

 In fact, we will be dealing with the smaller class,  $ H^1_{\textrm{cov}} $,  of perturbations, where $ H^1_{\textrm{cov}} $ 
 is the Sobolev space of order $1$ defined by the covariant derivatives, i.e., $ H^1_{\textrm{cov}}:=\{v\in L^2(\R^2, \C\times \R^2)\ |\ \|v\|_{H^1}< \infty \}$, where the norm $\|v\|_{H^1}$ is determined by
 the covariant inner product
\begin{align*}
    \langle v, v' \rangle_{H^1}^{\rm Re} = \Re \int \bar{\xi}\xi' + \overline{\COVGRAD{A_\om}\xi} \cdot \COVGRAD{A_\om}\xi' + \alpha \cdot \alpha' + \sum_{i=1}^2 \nabla\alpha_i\cdot\nabla\alpha'_i, 
\end{align*}
where $v=( \xi, \al),\ v'=( \xi', \al')$, 
 while  the $L^2-$norm  is given by 
\begin{equation} \label{L2-inner-product}
    \langle v, v' \rangle_{L^2}^{\rm Re} = \Re \int \bar{\xi}\xi' + \alpha \cdot \alpha'.
\end{equation}
In Lemma \ref{lem:En-fluct-expr} 
of  Appendix \ref{sec:fluct-en}, we will find an explicit representation of  $\Lambda_{u_\om}(v)$. 

\medskip
 We define the gauge transformation
\begin{align}\label{gauge-transf-orig}
 \tau^{\rm gauge}_\chi :\ &  (\Psi( x),A( x)) 
  \mapsto (e^{i\chi(x)}\Psi( x),  A( x) + \nabla\chi( x)).
\end{align}
To formulate the notion of  asymptotically stability we define the manifold (equivalence class) $$\mathcal{M}_\om = \{ \tau^{\rm gauge}_\chi u_\om : \chi \in H^1 (\R^2, \R) \}$$ of gauge equivalent Abrikosov lattices and the $H^1-$distance, $\dist_{H^1}$, to this manifold.

  \begin{definition}\label{def:stability} We say that the Abrikosov lattice $u_\om$ is {\it asymptotically stable} under $H^1_{\textrm{cov}} -$ perturbations, 
  if there is $\del>0$ s.t. for any initial condition $u_0$ satisfying $\dist_{H^1} (u_0, \mathcal{M}_\om)\le \del$ 
there exists $g(t)\in H^1$, s.t.   the solution $u(t)=  (\Psi(t), A(t))$ of \eqref{GES}  satisfies 
   \[ \| u(t)- \tau^{\rm gauge}_{g(t)} u_\om\|_{H^1} \ra 0,\]    as $t \ra \infty$.  We say that $u_\om$ is {\it energetically} unstable if 
 $\inf_{v\in H^1_{\textrm{cov}}}\lan v, \E''(u_\om)v\ran<0$ for  the hessian, $\E''(u_\om)$, of $\E(u)$ at $u_\om$. 
 \end{definition}
The hessian, $ \E''(u)$ of the energy functional $\E$, -  at $u\in u_\om +  H^1_{\textrm{cov}}$ - is defined as $ \E''(u) = d\E'(u )$ (the G\^{a}teaux derivative of the $L^2-$gradient map), where $d$ and $'$ are  the G\^{a}teaux derivative and  $L^2-$gradient map defined in the paragraph preceding \eqref{gradient-E}. 
 Although $\E (u)$ is infinite 
on  $u_\om +  H^1_{\textrm{cov}} $, the hessian $\E''(u)$ is 
well defined as a differential operator 
explicitly and is given in \eqref{Lc-expl} of Appendix \ref{sec:hess-expl}. 
We restrict the initial conditions $(\Psi_0, A_0)$ for \eqref{GES} 
 satisfying 
\begin{equation} \label{parity0}
(\Psi_0 (-x), - A_0 ( - x))=(\Psi_0 (x), A_0 (x)).	 
\end{equation}
\DETAILS{T^{\rm refl} : (\Psi(t,x), A(t,x), \Phi(t,x))  \mapsto (\Psi(t, -x), - A(t, - x), \Phi(t,- x))

Note that, by \TIMC{} uniqueness,  the Abrikosov lattice solutions $u_\om = (\Psi_\om, A_\om)$ satisfy $T^{\rm refl} u_\om=u_\om$ and therefore so are the perturbations, $v_0:= u_0-u_\om$, where $u_0:= (\Psi_0, A_0)$:
\begin{equation} \label{v0-parity}
	T^{\rm refl} v_0=v_0.	 
\end{equation} }
%



\subsection{Main results}\label{sec:main-res}
Recall that the functions  $\gamma_{k}(\tau)$ and $\mu (\tau, \kappa) 
$ is defined in \eqref{gamk-tau}.
 \begin{theorem}\label{thm:stability} 
Let $b$ be sufficiently close to $\kappa^2$,  in the sense of   \eqref{b-cond}.
Then, under 
$H^1-$perturbations, satisfying \eqref{parity0},   the Abrikosov lattice $u_\om$, with  $\om =(\tau, r)$, $r = \sqrt{\frac{2\pi }{b \Im\tau}}$, is
  \begin{itemize} \item asymptotically stable 
 for all 
  for $\tau$ and $\kappa$ satisfying  
$\gamma_{k}(\tau)>0,\ \forall k \ne 0,$  $\kappa> \frac{1}{\sqrt{2}}$ and $\eta (\tau, \kappa)>0$; 
 and 
 
\smallskip

 \emph{energetically unstable} if either $\kappa < \frac{1}{\sqrt{2}}$, or $\inf_k\gamma_{k}(\tau) <0$, or $\eta (\tau, \kappa) < 0$. 
 \DETAILS{$\tau$ s.t.  
 $\g_{\del} (\tau) >0,\ 
 \forall \del > 0$;   
  \item {\it energetically} unstable 
   for all $\tau$  s.t.   $\gamma_{\del=0}(\tau) <0$.}  
  \end{itemize} 
 \end{theorem}

This theorem is proven in Sections \ref{sec:hess-spec}--\ref{sec:Pf-StabThm}, with some technical details given in Appendices \ref{U-prod-transf}--\ref{sec:nonlin}. 
The proof 
 consists of two parts: the linear and nonlinear analysis. In the next two subsections we sketch main steps of the proof. 
 
Next, we turn to the function $\gamma_{k}(\tau)$. 
 To prove the properties \eqref{gamk-comp} 
   of this function we use numerics. These numerics  are based on the  explicit expression for the functions $\g_k(  \tau) ,\   k\in \Om^*_\tau,$ 
which we describe now.  
Then we have the following explicit representation of  the function    $\g_k( \tau)$, as a fast convergent series   (cf \cite{ABN, NV}),  
\begin{theorem}\label{thm:gammak}  Let $\LAT_{\tau}':=\frac{1}{2}i \LAT_{\tau}^*=\sqrt{\frac{\pi}{\im\tau}}(\Z+\tau \Z)$. 
 For the function  $\g_k ( \tau) ,\ \im\tau>0, k\in \C/\LAT'_\tau,$ 
defined in  \eqref{gamk-tau},   have the explicit representation 
\DETAILS{ \begin{align}\label{gamk-series}
\g_k( \tau)& =2\sum_{t\in \tilde\LAT_{\tau}^*} e^{- 
\frac12  |t|^2  } \cos [  \im (  \bar k t)]  - |\sum_{t\in \tilde\LAT_{\tau}^*} e^{- 
\frac12   |t+k| ^2     +   i  \im (  \bar k t) }|- \sum_{t\in \tilde\LAT_{\tau}^*} e^{- 
\frac12 |t|^2  } .
\end{align}}
 \begin{align}\label{gamk-series}
\g_k( \tau)& =2\sum_{t\in \LAT_{\tau}'} e^{-  |t+q|^2  }  - |\sum_{t\in \LAT_{\tau}'} e^{-   |t+q| ^2     - 2 i  \im (  \bar q t) }|- \sum_{t\in \LAT_{\tau}'} e^{-  |t|^2  } ,
\end{align}
where $k$  and $q$ are related as $ k=-\frac{2\pi}{\im\tau}    i q$. \end{theorem}
This theorem  
is proven in Appendix \ref{sec:gammaq-series} using results of Section \ref{sec:phik-etc}  (cf \cite{ABN, NV}).
\DETAILS{{\bf Our computations show that ??} 
 
\begin{itemize}

\item $ \g_k( \tau)$ is minimized at  $k \approx \sqrt{\frac{2\pi}{\im\tau}}   (\frac{1}{2} - \frac{1}{2\sqrt{3}} i)$ at the point $\tau = e^{i \pi /3}$, and a value of
 $k \approx \sqrt{\frac{2\pi}{\im\tau}} (\frac{1}{2} + i\frac{1}{2})$ for $\tau = e^{i \pi/2}$, which corresponds to vertices of the corresponding Wigner-Seitz cells. 
\end{itemize}}
%
Interestingly,  in Proposition \ref{thm:gamk-cp2}  below, we show that  the points  $k\in \frac12\cL^*_\tau$ 
  are critical points of the function $\g_{k}( \tau)$ in $k$.

Moreover, 
the  functions $\gamma_{\del}(\tau)\equiv \gamma_{\del}(\LAT_\tau):=\inf_{\dist(k, \LAT^*_\tau)\ge \del} \g_{k} (\tau).$
have the following properties 
 proven in Section \ref{sec:phik-etc}. 
\begin{proposition}\label{prop:gamma}   
  \begin{itemize} 
\item  
$\g_{\del}( \tau)$ 
  is invariant under the action of  the modular group  $SL(2, \Z)$. 
 \item $\g_\del( \tau)$ 
is symmetric w.r.to the imaginary axis. 
 \item  $\g_\del( \tau)$ has critical points at   $\tau=e^{i\pi/2}$ and $\tau=e^{i\pi/3}$, provided it is differentiable at these points.
  \end{itemize}  
 \end{proposition}
 The first property says that $\g_\del( \tau)$ is independent of the choice of  
 a basis in $\LAT_\tau$ and $\LAT_\tau^*$ (see also Remark 5 below).  It implies that it suffices to consider  $\tau$ in  the fundamental domain,  $\bH/SL(2, \Z)$, of   $SL(2, \Z)$. By the second property,  it suffices to consider $ \g_\del( \tau)$ on the $\Re\tau \geq 0$ half of the fundamental domain, 
$\bH/SL(2, \Z),$ (the heavily shaded area on the Fig. \ref{fig:PoincareStrip}).  

What distinguishes the points  $\tau=e^{i\pi/2}$ and $\tau=e^{i\pi/3}$ is that they 
  are the only points in $\bH/SL(2, \Z)$, which are fixed points under the maps from $SL(2, \Z)$, other than identity, namely, under 
  \[\tau \ra - \bar\tau,\  \tau \ra -\tau^{-1}\  
 \text{   and  }\ \tau \ra 1 - \bar\tau,\ \tau \ra 1 -\tau^{-1},\]
  respectively. This fact is used to prove the third statement above.  
  



\subsection{ Goldstone spectrum}


To formulate the result we introduce some definitions and notation. 
In order to unify the exposition and 
reduce the number of related definitions, we, at the outset, rescale the equations and treat $\bar \Psi$ as an independent unknown. The former is done in order to eliminate the perturbation parameter $b$ from the spaces, 
and the latter, for the linearized problem to be complex-linear.
\paragraph{Rescaling.}
\label{sec:rescaling}

As the underlying lattice $\LAT_\om$ for  the Abrikosov lattice solution $(\Psi_\om, A_\om)$ depends on the parameter $b$, see Theorem \ref{thm:existence}, 
it is convenient (especially, in the study of the linearized problem) to rescale the problem so that the resulting lattice is $b-$independent: 
 $U_\si : (\Psi(x), A(x)) \mapsto (\si \Psi(\si x), \si A(\si x))$, where $\si=\frac{1}{\sqrt b}
 $, for a given lattice $\LATo, \om=(\tau, r),$ and a 
 magnetic flux per a fundamental cell, $b$. 

 We denote the rescaled Abrikosov lattice solution by  $(\psit, \at) := U_\si (\Psi_\om, A_\om)$.  
 It satisfies the rescaled Ginzburg-Landau equations
\begin{equation}\label{gle-resc}
\begin{cases}
    (-\COVLAP{a} - \lam)\psi + \kappa^2 |\psi|^2\psi=0, \\
    \CURL^*\CURL a - \Im(\overline{\psi}\COVGRAD{a} \psi)=0,
\end{cases}
\end{equation}
with $\lam = r^2\kappa^2=\frac{\kappa^2}{b}$ 
 and $|\psi|^2$,  $\CURL a$,  and  $ \Im(\bar{\psi}\COVGRAD{a}\psi)$ double-periodic w.r. to the lattice $\LATt:=\sqrt{\frac{2\pi}{\im\tau} }  (\Z+\tau\Z)$. 

\paragraph{Extending the system of equations.}
\label{sec:rescaling} 
We consider $\psi$ and $\bar\psi$ as independent unknowns and extend the rescaled 
  Gorkov-Eliashberg-Schmid equations 
to obtain
\begin{equation}\label{GESresc-c}
\begin{cases}
\chi \partial_{t,\phi} \psi =     (\COVLAP{a} + \lam)\psi - \kappa^2 |\psi|^2\psi, \\
\chi \overline{\partial_{t,\phi} \psi} =   (\overline{\COVLAP{a}} + \lam)\bar\psi - \kappa^2 |\psi|^2\bar\psi, \\ 
2\sigma \partial_{t,\phi} a   = -  2 \CURL^*\CURL a + 2\Im(\overline{\psi}\COVGRAD{a} \psi),
\end{cases}\end{equation}
where $\phi$ is the rescaled scalar (electric) potential and the remaining notation is explained after \eqref{gle} and \eqref{GES} (in particular,  
 $\partial_{t,\phi}\psi  =  (\partial_t + i\phi)\psi$ and  $\partial_{t,\phi}a = \partial_t a + \nabla\phi$). (The role of $2$ in the second equation will become clear later.) These equations are invariant under the rigid motions (including the reflections) and the gauge transformations, which we write out explicitly,
\begin{align}\label{gauge-transf-resc-c}
&  \psi(t,x), 
a(t,x), \phi(t,x) 
 \mapsto e^{i\gamma(t,x)}\psi(t,x), 
 a(t,x) + \nabla\gamma(t,x), \phi(t,x) - \p_t\gamma(t,x).
\end{align}

 A fully complex form of the GES equations is given in Subsection \ref{sec:GES-fully-compl}.


\paragraph{Hessian.}\label{sec:hessian}

A key role in our analysis will be played by the linearization of the equations \eqref{GESresc-c} around the rescaled Abrikosov lattice solution $\ut:=(\psit,  \bar\psit, \at) $. 
 Unlike  Section \ref{sec:pert}, 
in what follows $u$ stands for $u=(\psi, \bar \psi, a)$.  (The old designation  $u=(\Psi,  A)$ will not be used from now on.) 

Let $J(u)$ be the map on the r.h.s. of  the Gorkov-Eliashberg-Schmid equations \eqref{GESresc-c}  and $d J(u)$ $:=(d_{\psi} J_{\psi}(u), d_{\bar\psi} J_{\bar\psi}(u), d_{a} J_{a}(u))$, where $d_{\psi} J_{\psi}(u), d_{\bar\psi} J_{\bar\psi}(u), d_{a} J_{a}(u)$ are the partial complex G\^ateaux derivatives, $d_{\psi} $ is the G\^ateaux derivative in $\psi$, i.e. $d_{\psi} :=\frac12 (d_{\psi_1}-i d_{\psi_2})$, etc..  Since we vary $\psi$ and $\bar \psi$ independently, we define the operator  $d J(u)$  on a dense subset (namely, on the Sobolev space of order $2$)  of the Hilbert space  $\Hil:=L^2(\R^2;\C)\oplus L^2(\R^2;\C)\oplus L^2(\R^2;\C^2)$, with  the inner product 
\begin{equation}\label{ip-cc}   \lan v, v' \ran_{L^2} = \int \bar{\xi}\xi' + \bar{\eta}\eta' + \bar\alpha\cdot \alpha' , \end{equation}
 where  $v = (\xi, \eta, \alpha),\  v' = (\xi', \eta', \alpha')  \in \cH$.  It is  self-adjoint  on  $\Hil$ . 
 
Let $\Lt$  denote the linearization of $J(u)$ at $\ut$,   
 $\Lt:= d J(\ut)$. It is given explicitly in \eqref{Lc-expl} of Appendix \ref{sec:hess-expl}. 
 The spectrum of $\Lt$ gives the excitation spectrum of  the Abrikosov lattice solution $ (\psi_{\w}, a_{\w})$.  
In Appendix \ref{sec:hess-expl}, we compare $\Lt$ with the linearization of the standard rescaled Ginzburg-Landau equations \eqref{gle-resc}.

The operator $\Lo$ is a hessian the rescaled Ginzburg-Landau energy functional \eqref {gl-en}, 
\begin{equation} \label{gl-en-c}
 E (\psi, \bar \psi, a)\equiv   E_Q(\psi, \bar \psi, a) = 
     \int_Q \left\{ |\COVGRAD{a}\psi|^2 +  |\CURL a|^2 -  \lam |\psi|^2+ \frac{\kappa^2}{2} |\psi|^4 \right\},
\end{equation}
where, as above,  $\n_a:= \n-  i a$ and  $\psi$ and $\bar\psi$  are considered as independent variables,  at $\ut=(\psit, \bar \psit, \at)$. To explain this, we begin with the notation. Let  $ 
\n_{\psi_j}^{L^2}E (u),\ \p_{a}E (u)\equiv \n_{a}^{L^2}E (u)$ be the real, partial $L^2-$gradients (in particular, with the real $L^2-$product, $ \lan v, v' \ran_{L^2}^{\rm Re} = \re\int \bar{\xi}\xi' + \bar{\eta}\eta' + \bar\alpha\cdot \alpha' $) and  
\begin{align}\label{dE-barpsi}\p_{\bar\psi}E (u):=\frac12 ( \n_{\psi_1}^{L^2} E (u) + i \n_{\psi_2}^{L^2}E (u)),\end{align}
 and similarly for $\p_{\psi}E (u)$. 
 Then $E' (u):=(\p_{\bar\psi}E (u),\ \p_{\psi}E (u),\  \p_{a}E (u))$ is the exactly the map $J(u)$ on the r.h.s. of  the Gorkov-Eliashberg-Schmid equations \eqref{GESresc-c} and the rescaled Ginzburg-Landau equations \eqref{gle-resc} are the Euler - Lagrange equations for this functional, $\p_{\bar\psi}E (u)=0,\ \p_{a}E (u)=0$.

The linearization operator $\Lt$ 
 is the hessian of the rescaled  Ginzburg-Landau energy functional, $E(\psi, \bar\psi, a)$, with $\psi$ and $\bar\psi$  considered as independent variables, 
\begin{align}\label{E''} 
E''(u)&:= \big( d_{\vphi_i} \p_{\vphi_j}E (u)\big), 
	\end{align}
 where 
$(\vphi_1, \vphi_2, \vphi_3)= (\psi, \bar\psi, a)$, 
evaluated at $\ut$: 
\begin{align}\label{Lt} \Lt=d J(\ut) = E''(\ut). \end{align} 

\paragraph{Symmetry breaking and zero modes.} 
We describe the symmetry zero modes of the hessian \eqref{Lt}, 
due to  breaking of  the gauge, translational and rotational  symmetry of the equation \eqref{gle-resc} by the solution $(\psit,  \at)$. In the proof we use the 
 gauge and translation transformations of the rescaled and extended unknowns, 
\begin{align}\label{Tgauge}
 T^{\rm gauge}_\gamma :\ &  (\xi( x), \eta( x), \a( x)) 
  \mapsto (e^{i\gamma(x)}\xi( x), e^{- i\gamma( x)}\eta( x),  \a( x) + \nabla\gamma( x)),
\end{align}
 \begin{equation}\label{Ttrans}
    T^{\rm trans}_s : (\xi( x), \eta( x), \al ( x))  \mapsto (\xi( x+s), \eta(x+s),  \al (x+s)).
\end{equation}
Recall the notation $r^\perp:=(- r_2, r_1)$ for $r=(r_1, r_2)$. 
\begin{lemma}\label{lem:zero-modes}
	The operator $\Lt$ has the  gauge, (gauged) translational and rotational 
	 zero modes,  $\Lt  G_{ \chi} =0$,  $\Lt S_{r} =0$ and  $\Lt  R_{ \varphi}  =0$, where $G_{ \chi}$, $S_{r},\ r\in \R^2,$ and $ R_{ \varphi} $ are given in 
\begin{align} \label{gauge-modes}  & G_{\chi} := (i\chi\psit, - i\chi\bar\psit, \nabla\chi),\\
 \label{transl-modes}
 &  S_{r} = ( ({r}\cdot\n_{\at} )\psit, 
 \overline{({r}\cdot\n_{\at} )\psit}, (\curl  \at)  r^\perp) ,\\
\label{rot-zero-mode}
    &   R_{ \varphi} = { \varphi}( J x \cdot\nabla\psit,   J x \cdot\nabla\bar\psit,  - \at^\perp + x^\perp \cdot\nabla \at).
\end{align}	 	 
\end{lemma}
\begin{proof} 
To derive the relations for the gauge zero mode, substitute $T^{\rm gauge}_\gamma \ut$ into the Ginzburg-Landau equations \eqref{gle-resc} to obtain $E'(T^{\rm gauge}_\gamma \ut )=0$, where  $E'(u )$ is the r.h.s. of \eqref{GESresc-c}. 
 Then we differentiate the equation $E'(T^{\rm gauge}_\gamma \ut )=0$ w.r.to $\g$ 
 at $\g=0$, we find $d_u E'(T^{\rm gauge}_\g \ut)d_\g T^{\rm gauge}_\g \ut=0$. Since $d_u E'(\ut) =\Lt$ and $d_\g T^{\rm gauge}_\g \ut\big|_{\g=0}\g'=G_{ \g'} $, this gives $\Lt  G_{ \g'} =0$. 
 
 The situation with the translational zero mode is slightly more subtle. Proceeding as with the gauge  zero mode, we arrive at the translational zero mode $\tilde S_{r'} = ( ({r'}\cdot\n)\psit, \overline{({r'}\cdot\n)\psit}, ({r'}\cdot\n)  \at) ,\ r'\in \R^2$. However, $\tilde S_{r'}$ is not gauge covariant (or equivariant). Hence we modify it by subtracting $G_\chi$, with $\chi:=r' \cdot \at$, to obtain $\tilde S_{r'}- G_{r' \cdot \at}= S_{r'}$.
 
 The derivation of the relations for the rotational zero mode is the same as for the gauge zero mode and we omit it here.
 \end{proof} 
\DETAILS{The solution $u_\om$ of \eqref{gle} breaks also the translational and rotational 
invariance, and therefore the operator $L_\om$ has also the   translational and rotational  zero modes 
\begin{align} \label{transl-zero-mode}
 &  S_{h'} = ( ({h'}\cdot\n_{A_\om} )\Psi_\om, 
(\curl  A_\om) Jh') ,\ h'\in \R^2,\\
\label{rot-zero-mode}
    &   R_{ \varphi'} = { \varphi'}( J x \cdot\nabla\Psi_\om , - J A_\om + J x \cdot\nabla A_\om ),
\end{align}
 i.e.  $L_\om S_{h'} =0$ and  $L_\om  R_{ \varphi'}  =0$.  
 The   translational zero modes are not square integrable but  bounded and indicate about the presence of the essential spectrum at $0$.  The   rotational zero modes grows at infinity  and are not used in our analysis.}
 %
 Only $G_{ \chi}$ with $\chi\in H^1(\R^2, \R)$ are true eigenfunctions, the other zero modes are generalized ones. The translational zero modes and global gauge zero mode, 
\begin{align} \label{G-glob} \Ggl:=( \psit, - \bar\psit, 0)= - i G_{\chi=1}, \end{align}
are bounded and, as we mentioned above, are connected to the two Goldstone massless spectral bands. 

\paragraph{Result.} We say that a self-adjoint operator $L$  on a Hilbert space $\cH$ has the ($\LAT-$) {\it band decomposition} iff there are functions $\nu^j(k)$, $k\in \Om^*=\R^2/\LAT^*$, and an orthogonal decomposition $\cH=\oplus_j \cH_j$ into invariant subspaces $\cH_j$, s.t. $\s(L)=\cup_j \Ran \nu^j$ and the spectrum of $L$ on $\cH_j$ is $\Ran \nu^j$. 

  \begin{theorem}\label{thm:energyband}      
On the subspace orthogonal to $G_{ \chi},\ \chi\in H^1$,
 
 (i)  the hessian $\Lt\equiv \E''(\ut)$ 
 has has the $\LATt-$band decomposition;
 
  (ii) its two lowest bands are gapless and are of the form \eqref{gapless-branch1} - \eqref{gapless-branch2} 
  and the remaining bands have a gap  $\gs \e^2$;
  
    (iii) The lowest band is $>0$ for $k\notin \LAT^*_\tau$ (resp. its  infimum $<0$) if $(\kappa^2-\frac12)\g_{k} (\tau)>0$ for $k\notin \LAT^*_\tau$ and $\eta (\tau, \kappa)>0$ (resp. either $\inf_k \g_{k} (\tau)<0$ or  $\eta (\tau, \kappa)<0$).
 \end{theorem}
 This theorem is proven in Section \ref{sec:Lk-spec} and Appendix \ref{sec:prop-Lk-spec-pf}.

The gapless spectral branches \eqref{gapless-branch1} - \eqref{gapless-branch2} are due to breaking of the global 
translational symmetry (by $\ut$). 
 We explain this. 

 As was mentioned above, the spontaneous breaking of the global gauge and translational symmetries by $\ut$ produces  the translational and global gauge  zero modes,  $S_{r'}$ and 
$G_{\g'=1}$, which are bounded but not square integrable. 
As it turns out,  the global gauge zero mode, $G_{\g'=1}$, leads to the zero eigenvalue for all $k$. So we concentrate on the translational zero modes.

 These modes are gauge $\LAT-$translationally invariant ($\LAT-$equivariant). Hence it is natural to look for almost generalized eigenfunctions of $\Lt$ in the form of the Bloch waves, $e^{ik\cdot x}S_{r'}, 
 k\in \Omt^* $. Then, since \[e^{-ik\cdot x}\Lt e^{ik\cdot x}=\Lt+O(|k|),\] for $|k|$ small, 
 we have $\Lt e^{ik\cdot x}S_{r'} =e^{ik\cdot x} O(|k|) S_{r'} $. 
 This suggests the presence of two gapless branches of the spectrum of  $\Lt$ on $\Hil_{}$, starting at $0$. 
%
 
  This phenomenon is well known in particle physics and goes under the name of the Goldstone theorem and the resulting gapless branch of the spectrum is called the Goldstone excitations and Goldstone particles. 

\paragraph{Sketch of the proof of Theorem \ref{thm:energyband}.}
 As was mentioned above, 
  the Abrikosov lattice solution $(\psit, \at)$ is gauge-periodic (with respect to the lattice $\LATt$) in the sense that
there exist (possibly multivalued) functions $g_s : \R^2 \to \R$, $s \in \LATt$, such that
\begin{equation}\label{gauge-per-resc}
\begin{cases}	\psit(x+s) = e^{i  g_s(x)}\psit(x),\\ \at(x+s)=\at (x) + \n g_s(x),
\end{cases}\end{equation}
for  $s \in \LATt$. We can rewrite these conditions (omitting the subindex $\tau$) as
 
\begin{equation}\label{gauge-per'}
T^{\rm trans}_s (\psi, \bar\psi, a) =T^{\rm gauge}_{g_s} (\psi, \bar\psi,  a),\end{equation}
 where  $s \in \LATt$ and $T^{\rm gauge}_\gamma $ and  $T^{\rm trans}_s$ are defined in \eqref{Tgauge} and \eqref{Ttrans}, respectively.
\DETAILS{ Indeed, {\bf if state  $(\Psi, A)$ satisfies \eqref{gauge-per'}, then all associated physical quantities are $\LAT-$periodic, i.e. $(\Psi, A)$ is an Abrikosov lattice. 
In the opposite direction, if $(\Psi, A)$ is an Abrikosov lattice, then $\curl A(x)$ is periodic w.r.to $\LAT$,  and therefore   $A(x + s) = A(x) +\n g_s(x)$,  for some functions $g_s(x)$. 
Next, we write $\Psi(x)=|\Psi(x)|e^{i \phi(x)}$. Since $|\Psi(x)|$ and $J(x)= |\Psi(x)|^2 (\n \phi(x)-  A(x))$ are  periodic w.r.to $\LAT$,}  we have that   $\n \phi(x + s) = \n \phi(x) +\n \tilde g_s(x)$, which implies that   $\phi(x + s) =\phi(x) + g_s(x)$, where $g_s(x)=\tilde g_s(x)+ c_s$, for some constants $c_s$.} 
Since $\LATt$ is a group, we see that  the family of functions $g_s$ has the important cocycle property
\begin{equation}\label{cocycle-cond}
    g_{s+t}(x) - g_s(x+t) - g_t(x)\in 2\pi\Z,\ \forall s\in \LATt.
\end{equation}
This can be seen by evaluating the effect of translation by $s + t$ in two different ways. We call $g_s(x)$ the {\it gauge exponent}. 

 The key idea of the proof of the first part of Theorem \ref{thm:stability} stems from the observation that since the  Abrikosov lattice solution $\ut = (\psit, \bar\psit, \at)$ 
 is gauge periodic (or equivariant) w.r.to  the lattice $\LATt$, 
 i.e. it satisfies \eqref{gauge-per'},  the linearized map $\Lt$ commutes with gauged (magnetic) translations,  
 \begin{align}\label{rhot'}\hrho_s = (\tilde T^{\rm gauge}_{g_s})^{-1} T^{\rm trans}_s,
  \quad  \forall s\in \LATt,\end{align}
 where  $T^{\rm trans}_s$ is defined in \eqref{Ttrans}, 
and $ \tilde T^{\rm gauge}_\gamma $ is given by 
\begin{align} \label{tildeTgauge} \tilde T^{\rm gauge}_\chi : (\xi, \eta, a) \mapsto (e^{i\chi}\xi, e^{-i\chi}\eta, \a ).\end{align}
Due to \eqref{cocycle-cond}, $\hrho_s$ gives a unitary representation of the group $\LATt $. 
 Therefore $\Lt$ is unitary equivalent to a fibre integral over the dual group, $\hat\LATt$, of the group of lattice translations,  $\LATt$, which, as was already mentioned, can be identified with a fundamental domain, $\Omt^*$, of the reciprocal lattice, \begin{align} \label{fiber-deco}\Lt\approx \int_{\Oms}^\oplus \Lk \hat{d k}\ \quad \mbox{acting on}\ \quad  \int_{\Oms}^\oplus \cH_{ k} \hat {d k},\end{align}
where $\hat {d k}$ is the usual Lebesgue measure on $\Oms$ normalized so that $\int_{\Oms} \hat {d k} = 1$, $\Lk$ is the restriction of $\Lt$ to $\cH_{ k},\ k\in \Oms,$ and $\cH_{ k}$  is the set of all functions, $v$, from $L^2(\R^2;  \C\times \C)$, 
which satisfy 
	\begin{equation}\label{bc'}\hrho_s v (x) =e^{ik\cdot s} v (x),\ \forall s\in\ \LATt. 
	\end{equation} 
The  inner product in  $\cH_{ k}$  is given by $  \lan v, v' \ran_{L^2} = \frac{1}{|\Omt|}\re\int_{\Omt} \bar{\xi}\xi' + \bar{\eta}\eta' + \bar{\alpha}\cdot \alpha' ,$ where   $v = (\xi, \eta,  \alpha),\  v' = (\xi', \eta', \alpha')$, 
and in 	  $\int_{\Oms}^\oplus \cH_{ k} d k$, by $\lan v , w \ran_{\cH }:=\frac{1}{|\Oms|} \int_{\Oms}\lan v_k, w_k\ran_{\cH_k} \hat{d k}$.     (The normalization used will be useful later on.)

The 
fiber decomposition \eqref{fiber-deco} reduces the analysis of the operator $\Lt$ to that of the operators $\Lk$, which have purely discrete spectrum and therefore much easier to study.  As  $k$ varies in $\Oms$, 
 the eigenvalues of $\Lk $ sweep the spectral bands of $\Lt $. 
This shows the band structure of the spectrum of the operator $\Lt$ stated in the item (i) of Theorem \ref{thm:energyband}. 

Since each operator $\Lk$  has purely discrete spectrum, we can apply to it the standard perturbation theory in $\e$ and, for small $|k|$, in $k$. After a somewhat lengthy analysis using the Feshbach-Schur map, we obtain the statements (ii) and (iii). $\Box$ 

 \subsection{The key ideas of the proof of Theorem \ref{thm:stability}} 
 \label{sec:approach}

As was already mentioned above, the stability of the static solution $\ut$ 
 is decided by the 
 nature of the low energy the spectrum of  the (complexified) linearization operator, $\Lt$, for the map on the r.h.s. of \eqref{GES}, see \eqref{Lt}. 
 Namely, whether $ \Lt $ is non-negative or has some negative spectrum. Due to the symmetry breaking, it has always the eigenvalue $0$. Hence, if $ \Lt \ge 0$, one would like to know whether, in the former case, 
 $0$ is an isolated eigenvalue (we say the spectrum has a gap) or the continuum extends all the way to zero (the spectrum is gapless). In the latter case, the important point is played by the nature of continuous spectrum at $0$ (the dispersion relation of the gapless modes).

As Theorem \ref{thm:energyband} shows (see also the comments after it), due to the spontaneous breaking of the global gauge and translational symmetries by $\ut$, 
 the linearization operator  $\Lt$ has two gapless spectral branches, 
 %
  %
\DETAILS{ $$********$$
\paragraph{Discussion.}  $e^{ik\cdot x}\tilde{T}_i  $, for $|k|$ small, are almost zero modes of $K_k$. Indeed, $e^{-ik\cdot x}K e^{ik\cdot x}=K+O(|k|)$ and therefore $K e^{ik\cdot x}\tilde{T}_i = 
O(|k|) e^{ik\cdot x}\tilde{T}_i $. Hence there is a positive 
 branch of the spectrum of  $K_{}$ on $\mathcal{K}_{}$, starting at $0$, corresponding to translations of the lattice.
This branch  is described in Proposition \ref{prop:KkEVs}. 
By \eqref{tildeSj}, 
it originates from the subspace 
\begin{align}  
\{f_k w_{k}^0 &: f_k \in L^2(\Om^*, \C)\},\end{align}
where $w_{k }^0= w_{k \|}^0$ are described in Corollary \ref{cor:K0k-spec}.}
%
   \eqref{gapless-branch1} - \eqref{gapless-branch2}, with the remaining spectral branches are separated from $0$ by the gap  of the order $\e^2$. 
This result gives the {\it marginal} linearized (or energetic) stability of $\ut$ for all $\ta$ and $\kappa$ s.t. $\kappa^2 > \frac12$, 
$\g_k(\tau)>0,\ \forall k \notin \LATt^*$ (recall that $\g_{k}(\tau)=0,\ \forall k \in \LATt^*$) and $\eta(\tau, \kappa)>0$ and the instability, otherwise. 

If  $\ut$ turns out to be marginaly stable, then we would like to show that it is asymptotically stable. To this end, 
 we introduce the equivalence class $\mathcal{M}_\tau = \{ T^{\rm gauge}_\chi  \ut : \chi \in H^1 (\R^2, \R) \}$ of the Abrikosov $\cL-$lattice solution,  $\ut= (\psit, \bar\psit, \at)$. The tangent space, $T_{\ut}\cM$ at $\ut  \in \mathcal{M}_\tau$ is spanned by the gauge zero modes $G_{\chi}$, given in \eqref{gauge-modes}.
 We write a solution $u$ to the Gorkov-Eliashberg-Schmid equations \eqref{GESresc-c}, which is in a neighbourhood of $\mathcal{M}_\tau$ in the form
\begin{align}	\label{u-deco-gauge'} & u =T_{\g}^{\rm gauge} ( \ut +v), \quad \mbox{with}\,\  v\perp G_{\chi}\  \forall  \chi.  
\end{align}
We call $v$ the {\it fluctuation} of $ \ut$.
Now, we split the fluctuation $v$ as $v=v'+v''$, where $v'$ is the projection of $v$ onto the subspace of the  two lowest (gapless) spectral branches of $\Lt$, described above, and $v''$ is the orthogonal complement. For $v''$, we use the fact that, on  $v''$'s, $\Lt$ has a gap of oder $\e^2$. 
This allows us to use the  method of differential inequalities for the Lyapunov functionals. We define the functionals   \[\Lambda_{1}(v'')= \frac12 \lan v'', \Lt v''\ran\ \text{ and }\  \Lambda_2(v'') = \frac{1}{2}\lan \bar P x v'', \Lt\bar P x  v'' \ran_{L^2},\]  where $\bar P$ is the orthogonal projection to the orthogonal complement of the spectral subspace of the two gapless branches. By the choice of $v''$ and $\bar P$, these functionals are bounded below by Sobolev norms of the corresponding vectors. This and differential inequalities for $\Lambda_{1}(v'')$ and $\Lambda_2(v'')$ (which follow from the equation for $v$) allow us to estimate these norms. 
 (The resulting estimate of $v''$ will, of course, depend on some information about $v'$.)  

For $v'$, we use the evolution equation, which follows from the equation for $v$. In this equation, we pass to the spectral representation of the corresponding two spectral branches of $\Lt$ ($v' \leftrightarrow f$, with $f : \Om^* \ra \C^2$). The resulting equation for $f$ is of the form
\begin{align}	\label{f-eq'}
		\partial_t f = - A f + 
\cN_{v''} (f), 
	\end{align}
where $A$ is the operator of multiplication  on $L^p(\Om^*, \C^2)$, given, on vector-functions $f=(f_1, f_2) $, by \[(A f) (k) := (\nu_k^1 f_1 (k), \nu_k^2 f_2 (k)),\] $v'':= P'' v$ and $ \cN_h (f)=O(f^2)$ (as $f\ra 0$) is the nonlinearity.  

Clearly, the behaviour of the bands $\nu_k^1$ and $\nu_k^2$ plays a crucial role here.  $\nu_k^2$ is a positive function and, in the leading order at $k=0$, is a positive definite quadratic form. The behaviour of $\nu_k^1$ is determined by the functions $\g_k(\tau)$ and $\mu(\tau, \kappa)$:\\ 
$\nu_k^1$ is a positive function and, in the leading order at $k=0$, is a positive definite quadratic form iff $\g_k(\tau)>0, \forall k\notin \LATt^*,$ and $\mu(\tau, \kappa)>0$.  

In the latter case, equation \eqref{f-eq'} is roughly of the form of a nonlinear heat equation in two dimensions, with a quadratic nonlinearity.  To the authors' knowledge, the long-time behaviour of such equations for small initial data is not understood presently. However, under the condition \eqref{parity0}, the function $f$ is odd which allows us to eke out an extra decay, provided we can control derivatives of $f$ in $k$ (and assuming some information about $v''$!). Bootstrapping the estimates on $v'$ and $v''$, we arrive at the long-time estimates of $v$ and consequently, the stability result. 

The above discussion shows in particular the role of the signs of the functions $\g_k(\tau)$ and $\mu(\tau, \kappa)$ in the stability of vortex lattices.

\DETAILS{The above shows that to determine which Abrikosov lattice solutions are stable and which are not, we have to investigate the function $\g_k(\tau)$. To understand this function, we use its representation as the fast convergent series \eqref{gamk-series}.  {\bf Using Matlab, Daniel Sigal has shown that  $\g_k(\tau)>0, \forall k\in \Om^*/\{0\},$ for $\tau=e^{2\pi/3}$ and $\inf_k \g_k(\tau)< 0,$ for $|\tau-e^{2\pi/3}|\ge ??$. For $|\tau-e^{2\pi/3}|\le ??$, but $\tau\ne e^{2\pi/3}$, the function is rather flat and it is hard to determine whether $\inf_k \g_k(\tau)< 0,$ or not.}} 

\DETAILS{\paragraph{Remark 6.} It is more natural to work with the complex hessian, which, unlike real-linear $L^\om$, is complex-linear.  However, this requires another level of notation, which we would like to avoid at this stage, and some of the estimates involved become more complicated.  However, to develop the spectral theory, we do need  the complex hessian, which we introduce in Section \ref{sec:L-spec-est}. }

\medskip

Finally, we point our the relation between the gauge functions in \eqref{phik-eqs-LAT} and \eqref{gauge-per'}. To begin with, we mention that some important properties of $g_s$ entering \eqref{gauge-per'}: 
\begin{itemize}

\item[(a)]   If $(\psi, \bar\psi, a)$ satisfies \eqref{gauge-per'},
 with  $g_s(x)$, then $T^{gauge}_\chi (\psi, \bar\psi, a)$ satisfies  \eqref{gauge-per'} with  $g_s(x)$ 
replaced by $   g'_s(x) = g_s(x) + \chi (x + s) - \chi(x).$ 

\DETAILS{\item[(b)] 	
The functions  \begin{equation}\label{gs-spec-n} 
g_s(x)  = \frac{\pi n}{|\cT|} s \wedge x + c_s, \end{equation} 
  where, recall,  $\cT:=\C/\LAT$ 
  and $c_s$ are numbers satisfying 
$ c_{t+s} - c_s - c_t - \frac{\pi n}{|\cT|}  s \wedge t \in 2\pi\Z, $ 
satisfy  \eqref{cocycle-cond}.}

\item[(b)]   Every exponential  $g_s$ satisfying the cocycle condition 
\begin{equation}\label{cocycle-cond-LAT}
    g_{s+t}(x) - g_s(x+t) - g_t(x)\in 2\pi\Z,\ \forall s\in \LAT,
\end{equation}
  is gauge-equivalent to  
  \begin{equation}\label{gs-spec-LAT} 
g_s(x)  = \frac{\pi n}{|\Om|} s \wedge x + c_s, \end{equation} 
for some  $n\in \Z$,  where, recall,  
$\Om$ is a fundamental domain of the lattice $\LAT$  and $c_s$ are numbers satisfying 
  \begin{equation}\label{cs-cond-LAT} 
 c_{s+t} - c_s - c_t - \frac{\pi n}{|\Om|}  s \wedge t \in 2\pi\Z. \end{equation} 
 (The functions \eqref{gs-spec-LAT},  with $c_s$ satisfying \eqref{cs-cond-LAT}, obey  \eqref{cocycle-cond-LAT}.) Moreover, $c_s$ can be chosen to be even in $s$ (by replacing $c_s$ by $\frac12 (c_s+c_{-s})$, if necessary), so that $g_s(x)$ satisfies
\begin{equation}\label{gs-spec-even} g_{-s}(- x)  = g_s(x). \end{equation}
\item[(c)]  For $\LAT=\LATt$, we can choose 
$c_1=c_\tau=0$ and therefore (by solving \eqref{cs-cond-LAT}) 
\begin{align}\label{gs-spec} c_s=\frac{ 1}{2\im \tau}(\im s)\im(\bar \tau s)\ \text{ and }\ g_s(x):=
\frac12 s\cdot x^\perp - 
\frac{ 1}{2\im \tau}(\im s)\im(\bar \tau s).\end{align}\end{itemize} 

Indeed, the first 
statement is  straightforward.  
For  the second property, see e.g. \cite{Eil, Odeh, Takac, TS2}, though in these papers it is formulated differently. In the present formulation it was shown  by A. Weil and generalized in  \cite{Gun1}.

\begin{remark} \label{rem:cocycle} {\rm Relation \eqref{cocycle-cond} for Abrikosov lattices was used in \cite{ST2}, where it played an important role. This condition is well known in algebraic geometry and number theory (see e.g. \cite{Gun2}).}  
\end{remark} 

\subsection{Possible extensions} \label{sec:exten} 
The next step would be to extend the results to more general perturbations.
Firstly, one would like to remove the restrictive condition \eqref{parity0}. This condition simplifies the treatment of the gapless branch of the spectrum of $\Lt$ (the branch starting at $0$). 
\DETAILS{Indeed, since the solution $u_\om$ of \eqref{gle} breaks the translational invariance, the operator $L_\om$ has the translational zero mode 
\begin{align}\label{transl-zero-mode}
   S_{h'} = ( ({h'}\cdot\n_{A_\om} )\Psi_\om, 
(\curl  A_\om) Jh') ,\ h'\in \R^2,
\end{align}
i.e. $L_\om S_{h'} =0$ (see Subsection \ref{sec:zero-modes} and Supplement II). 
 Since $S_{h'}$ is only bounded and not $L^2$, we have that $0\in \s_{ess}(L_\om)$. Moreover,

(Since the solution $u_\om$ breaks also the rotational invariance the operator $L_\om$ has the  rotational zero mode, $R_{ \varphi'}$, 
i.e. $L_\om R_{ \varphi'}  =0$, - see Subsection \ref{sec:zero-modes} - but this mode is growing at infinity.)}
%

Though removing  condition \eqref{parity0} would be technically cumbersome, we expect this  would not change the result above.

Secondly, one would like to consider 
non-local perturbations,  say, 
 perturbations of the form $\tau_g u_\om +v -u_\om$, with $v\in L^\infty (\R^2, \C\times\R^2)$, where $G$ is the full symmetry group 
\begin{align}\label{G} G = H^2(\R^2;\R) \rtimes \R^2 \rtimes SO(2)\end{align} and 
  $\tau_g = \tau^{\rm gauge}_\gamma \tau^{\rm trans}_h \tau^{\rm rot}_\rho$ is  the action  of $ G$ on pairs $u=(\Psi, A)$.  Here $\tau^{\rm rot}_\rho : (\Psi (x), A (x)) \ra (\Psi (\rho^{-1} x), \rho^{-1} A ((\rho^{-1})^T x))$. (\eqref{G} is a semi-direct product, with elements $g = (\gamma, h, \rho) \in G$,  and the  composition law given by $g g' =(\g +\tau^{\rm trans}_h \tau^{\rm rot}_\rho\g' , {\rho'}^{-1} (\rho^{-1} h+h'), \rho\rho')$.) For such perturbations we would have to generalize the notion of asymptotic stability 
 by replacing $\tau^{\rm gauge}_\g $ by $\tau_g $.  
 Specifically, 

  \begin{definition}\label{def:stability-gen} We say that the Abrikosov lattice $u_\om$ is {\it asymptotically stable} under finite-energy perturbations if there is $\del>0$ s.t. for any initial condition $u_0$, whose $H^1$- distance to  the infinite-dimensional manifold   $\mathcal{M} = \{ \tau_g u_\om : g \in G \}$ is $\le \del$,
   the solution $u(t)$ of \eqref{GES}  satisfies $ \|\tau_{g(t)}^{-1} u(t)-  u_\om\|_{L^\infty} \ra 0$ as $t \ra \infty$, for some path,  $g(t)$, in $G$.  
 \end{definition}
 Next, one would like to prove the stability results for the time-depedent relativistic Ginzburg-Landau equation (see \cite{GST}). 

\subsection{Basis independent definitions} \label{sec:invar-not} 
 The definitions of $\g_k$ and $\phi_k$ above depend on the choice of 
  the basis in $\LAT_\tau$ and $\LAT_\tau^*$. The  bases independent definitions  
are given as follows.  

Let $\hat \LAT$ be  the dual group, of  the group of lattice translations,  $\LAT$, 
i.e. the group of characters, $\chi : \LAT \to U(1)$,  and let $\cT:= \R^2/\LAT$, its quotient.
 We define $\g_\chi(\LAT)$ 
  as
\begin{equation}\label{gamma-chi} 
\g_\chi(\LAT) := 2\lan|\phi_0^\LAT|^2|\phi_\chi^\LAT |^2\ran_{\cT} - |\lan (\phi_0^\LAT)^2\bar{\phi}_\chi^\LAT\bar{\phi}_{\chi^{-1}}^\LAT\ran_{\cT} | - \lan |\phi_0^\LAT|^4 \ran_{\cT}, 
 \end{equation}
where $ \lan f \ran_{\cT}:= \frac{1}{|\cT|} \int_{\cT} f$. Here the functions  $\phi_\chi^\LAT,\  \chi\in \hat \LAT,$  are unique solutions of the equations 
\begin{align}\label{phichi-eq-gen} 
(-\COVLAP{a^0}-1)\phi =0,\ \quad \phi (x+s) =e^{ig_s(x)} 
\chi(s)\phi (x),\ \quad \forall s\in \LAT,
	\end{align}
 normalized as $\lan|\phi_\chi^\LAT|^2\ran_{\cT} =1$. Here $a^0(x)$ and $g_s(x)$ are as in \eqref{phik-eqs-LAT}, but with $|\Om|$ replaced by $|\cT|$, 
  e.g.
 \begin{equation}\label{gs-spec-gen} g_s(x)  = \frac{\pi }{|\cT|} s \wedge x + c_s,\ \text{ with $c_s$ satisfying }\ c_{s+t} = c_s + c_t - \frac{\pi}{|\cT|}  s \wedge t \text{ mod }  2\pi\Z.\end{equation} 
 
Furthermore, the linearization $L^\LAT$ is unitary equivalent to a fiber integral over the dual group, $\hat\LAT$, of the group of lattice translations,  $\LAT$, or $\C/\LAT$,
\begin{align*}L^\LAT\approx \int_{\hat\LAT}^\oplus L^\LAT_{\k} \hat{d \k}\ \quad \mbox{acting on}\ \quad  \int_{\hat\LAT}^\oplus \cH_{ \k} \hat {d\k},\end{align*}
where $\hat {d\k}$ is the usual Lebesgue measure on $\hat\LAT $ normalized so that $\int_{\hat\LAT} \hat {d\k} = 1$, $L^\LAT_{\k}$ is the restriction of $L^\LAT$ to $\cH_{ \k},\ \k\in \hat\LAT,$ and $\cH_{ \k}$  is the set of all functions, $v$, from $L^2(\R^2;  \C\times \C)$, 
which are gauge-periodic,
	\begin{equation}\label{gauge-BF-gen}\hrho_s v (x) = \chi (s) v (x),\ \forall s\in\ \LAT, 
	\end{equation} 
with $\chi (s)v =(\chi(s)\xi,  \chi(s)\alpha)$, for $v =(\xi,  \alpha)$.   
The  inner product in  $\cH_{ \k}$  is given by $  \lan v, v' \ran_{L^2} = \frac{1}{|\cT |}\re\int_{\cT } \bar{\xi}\xi' + \bar{\alpha}\alpha' ,$ where   $v = (\xi,  \alpha),\  v' = (\xi', \alpha')$  and   $\cT :=\R^2/\LAT $, and in 	 
 $\int_{\hat\LAT}^\oplus \cH_{ \k} d\k$, by $\lan v , w \ran_{\cH }:=\frac{1}{|\hat\LAT|} \int_{\hat\LAT}\lan v_\k, w_\k\ran_{\cH_\k} \hat{d \k}$.     (The normalization used will be useful later on.)

As  $\chi$ varies in $\hat\LAT$, the eigenvalues of $L^\LAT_{\k} $ sweep the spectral bands of $L^\LAT$.

The dual group, $\hat \LAT$, can be identified with a fundamental cell, $\Om^*$, of the dual lattice  $\LAT^*$, and the torus   $\cT:= \R^2/\LAT$, with a fundamental cell, $\Om$, of $\LAT$. The first identification is given explicitly by  $\chi(s)\ra   \chi_k(s) = e^{ik\cdot s}\leftrightarrow k$. (See the review \cite{S} for more discussions.)   
 Identifying $\cT:= \R^2/\LAT$, with a fundamental cell, $\Om$, of $\LAT$  gives \eqref{gamk-LAT}.   
 
\subsection{Fully complex form of the GES equations} \label{sec:GES-fully-compl}
 For computations, it is convenient to pass from the real vector-fields $\al=(\al_1, \al_2)$ to the corresponding complex functions, $\al_c:=\al_1 - i \al_2$ and use the complex operators $\partial = \partial_{x_1} - i\partial_{x_2}, \bar\partial = \partial_{x_1} + i\partial_{x_2}$ and $\partial_{a_c} = \partial - i a_c$. (This definitions differ from the {\it standard}  ones by the factor $2$.)
As a result we obtain the fully complex GES equations 
 ({\it omitting the subscript} $c$): 
\begin{equation}\label{GESresc-cc}
\begin{cases}
\chi \partial_{t,\phi} \psi =     (\COVLAP{a} + \lam)\psi - \kappa^2 |\psi|^2\psi, \\
\chi \overline{\partial_{t,\phi} \psi} =   (\overline{\COVLAP{a}} + \lam)\bar\psi - \kappa^2 |\psi|^2\bar\psi, \\ 
 \sigma \partial_{t,\phi} a   =  \frac12 \p\bar\p a -  \frac12 \p^2 \bar a + \frac12  i (\overline{\partial_{a}^*\psi})\psi +\frac12  i (\partial_{a }\psi)\bar{\psi}, \\ 
 \sigma \partial_{t,\phi} \bar a   =   \frac12 \p\bar\p \bar a - \frac12  \bar\p^2 a - \frac12  i ({\partial_{a }^*\psi})\bar\psi - \frac12 i \overline{(\partial_{a }\psi)} \psi.
\end{cases}\end{equation}
We should also remember that $\COVLAP{a}=-\frac12 (\partial_{a}\partial_{a}^*+ \partial_{a}^*\partial_{a}).$ \eqref{GESresc-cc} follow from \eqref{GESresc-c} and the equations 
\begin{align}\label{compl1} 
 &(\curl^*\curl \al)_c =- \frac12 \p\bar\p\al_c +\frac12 \p^2 \bar \al_c,\ \divv \vec \al_c = \re (\bar \p \al_c),\\   
\label{compl2}  &\Im(\bar{\xi}\COVGRAD{ a}\eta)_c = \frac{1}{2 i}[(\overline{\partial_{a_c }^*\eta})\xi + (\partial_{a_c}\eta)\bar{\xi}].
\end{align}

 The fully complexified GES equations, \eqref{GESresc-cc}, 
  are invariant under 
 the rigid motions (including the reflections) and the gauge transformations, which we write out explicitly,
\begin{align}\label{gauge-transf-resc-cc}
&  \psi,\ 
a,\  \phi 
 \mapsto e^{i\chi}\psi,\ 
 a  + \p\chi,\  \phi  - \p_t\chi .
\end{align}
Here $\chi$ is a real differentiable function on the space-time. 
Since the vortex lattice solution $(\psit, \bar{\psit}, \at, \bar\at)$ breaks these symmetries, this leads to the 
 gauge and gauged translational modes 
 \begin{equation}\label{compl-gauge-mode-transf} 
 \vec G_\chi^\# = (i\chi\psit, -i\chi\bar{\psit}, \partial\chi, \bar{\partial}\chi),
    \end{equation}
\begin{equation} \label{Sj-transf} 
   S_1^\# = (\partial_{\at}\psit, -\overline{\partial^*_{a_{\tau}}\psit}, 0,  2 i b_{\tau}), \
   S_2^\# = (\partial^*_{\at}\psit, -\overline{\partial_{a_{\tau}}\psit},  2 i b_{\tau},  0), 
\end{equation}
where, as before, $b_{\tau}:=\curl  \at $ (considering $\at$ as a real vector field), and the rotational modes which we do not display here.
Clearly, $G_\chi^\#$ and $S_j^\#$ satisfy $T^{\rm transl}_s G_\chi^\# = \tilde T^{\rm gauge}_{g_s} G_\chi^\#$ and $T^{\rm transl}_s S_j^\# = \tilde T^{\rm gauge}_{g_s} S_j^\#\  \forall s\in \LAT.$ 

 The gauged translational modes are obtained as complex combinations of $(\partial_{j}\psit,$ $ \partial_{j}\bar\psit,$ $ \partial_{j}\at,  \partial_{j}\bar{a}_{\tau})- G_{a_{\tau j}}^\#,\  j=1, 2.$ \eqref{compl-gauge-mode-transf} and \eqref{Sj-transf} can be also obtained from \eqref{gauge-modes} and \eqref{transl-modes} by using the transformation $(\xi, \bar\xi, \al) \longrightarrow  (\xi, \bar\xi, \al^\C, \bar \al^\C).$ 

By differentiating the negative of the r.h.s. of the fully complexified GES equations, \eqref{GESresc-cc}, w.r.to $\psi, \bar \psi, a, \bar a$ and using that $\COVLAP{a}=-\frac12 (\partial_{a}\partial_{a}^*+ \partial_{a}^*\partial_{a})$, we obtain the (transformed) complex hessian $\Lt^\#$ (see \eqref{Lt'} below). This hessian and its `shift', rather than $\Lt$, are investigated  
in the most technically complicated Appendix \ref{sec:prop-Lk-spec-pf}. 


\subsection{Organization of the paper} \label{sec:organ}
 In Section \ref{sec:phik-etc}, we prove the existence and uniqueness of solutions of the equation \eqref{phik-eqs-LAT} and prove Proposition \ref{prop:gamma}. Results of this section are used in Appendix \ref{sec:gammaq-series} in order to prove  Theorem \ref{thm:gammak}. In Section \ref{sec:hess-spec}, we study the hessian of the energy functional \eqref{gl-en}, which is the same as the linearization of the map on the r.h.s. of the  Gorkov-Eliashberg-Schmid equations \eqref{GES}. The main result of this section is used in Section \ref{sec:Pf-StabThm} to prove main Theorem \ref{thm:stability}. Various technical computations are carried out in the appendices. 
 
 \DETAILS{We prove 
 Theorem \ref{thm:energyband} 
 in Section \ref{sec:proof-thmgamma} and Appendix \ref{sec:gammaq-comp}.  Theorem \ref{thm:stability} is proven in Section \ref{sec:Pf-StabThm}, with technical results proven in Section \ref{sec:L-spec-est} and Appendices \ref{sec:matrF-comp}-\ref{sec:hess-expl}. 
We prove Theorem \ref{thm:gammak} 
 in Section \ref{sec:proof-thm:energyband} and Appendix \ref{sec:gammaq-comp}.  Numerical investigation of functions $\g (\tau)$ and $\g_k (\tau)$ quoted in Subsections \ref{subsec:backgr} and \ref{sec:main-res} is done in Appendix \ref{sec:gammak-numer}.} 
For the reader's convenience some known facts about theta functions and the Feshbach-Schur perturbation theory are presented in Supplements I and II.

\paragraph{Notation.} In this paper we use two types of  lattices, the original one, \eqref{LATom}, and normalized one, \eqref{LATtau}. We denote by
 $\LAT_\om^*$ and $\LAT_\tau^*$, the corresponding dual lattices and by $\Omega_\om$, $\Omega^*_\om$,  $\Omega_\tau$ and $\Omega^*_\tau$, elementary cells of the corresponding lattices.   
 
 In estimates of functions and operators, we use the notation $O(1)>0$ and $o(1)>0$ to stand for some $C=O(1), C>0$ and $c=o(1), c>0$, respectively.
 \bigskip

\noindent \textbf{Acknowledgement.} 
IMS is grateful to  
Dmitri Chouchkov,  Gian Michele Graf, Lev Kapitanski, Peter Sarnak,  and Tom Spencer, and especially 
 Li Chen,  J\"urg Fr\"ohlich,  Stephen Gustafson, and Yuri Ovchinnikov, for useful discussions, and to the anonymous referee, for many constructive remarks which led to improvement of the paper.  J\"urg Fr\"ohlich  and  the anonymous referee emphasized that the gapless branch of the spectrum of the hessian are the Goldstone excitations.
%

\section{Functions  $\phi_k(x), $ and $\g_k ( \tau)$}   \label{sec:phik-etc} 
 We begin with the functions $\phi_k(x)$. Recall that these functions solve \eqref{phik-eqs-LAT}, with $\LAT=\LAT _\tau$, used in the definition \eqref{gamk-tau} of $\g_k(\tau)=\g_k(\LATt)$.
 For $\LAT=\LAT _\tau$, Eq  \eqref{phik-eqs-LAT} becomes
\begin{align}\label{phik-eqs-tau} (-\COVLAP{a^0}-1)\phi =0,\ \quad \phi (x+s) =e^{i g_s(x)} e^{i k\cdot s}\phi (x),\ \quad \forall s\in \LAT_\tau,
	\end{align}
normalized as $\lan|\phi_k|^2\ran_{\LATt} =1$.  Here 
    \begin{equation}\label{gs-spec-tau} a^0(x):= \frac12  x^\perp\   \text{ and }\  g_s(x)  = \frac{1}{2}s \wedge x + c_s, \end{equation} 
 where $x^\perp:= (- x_2,  x_1)$ and  $c_s$ are numbers satisfying 
  \begin{equation}\label{cs-cond-tau} 
 c_{s+t} - c_s - c_t + \frac{1}{2} s \wedge t \in 2\pi\Z. \end{equation} 
  
     We consider $\phi_k(x)$, for 
      $\tau$ in the entire  Poincar\'e half-plane $\bH$, rather than just for the fundamental domain $\Pi^+/SL(2, \Z)$. 
  Note that an interesting formula relating the function $\phi_k(x)$ to $\phi_0 (x)$ is proven in in Appendix \ref{sec:prop-Lk-spec-pf}, see \eqref{phik-expr}.
 
  An explicit form of the functions $\phi_k(x),\ k\in \C$, is described in the following
	 \begin{proposition} \label{prop:phik} The solution $\phi_k(x)$ of the problem \eqref{phik-eqs-tau} - \eqref{gs-spec-tau} is unique (up to a constant factor).	 Moreover, the functions  $\phi_k,\ k\in \C,$ 
are of the form
\begin{align}\label{phik-thetaq} 
\phi_k(x)= 
c' e^{\frac{\pi}{2 \im\tau}(z^2 -|z|^2)} \theta_{q}(z, \tau), \quad 
 x_1+i x_2=\sqrt{\frac{2\pi}{\im\tau}}    z,\ 
 k=- \sqrt{\frac{2\pi}{\im \tau}}   i q,   
 \end{align}
 where  $c'$ is fixed by requiting that 	 $\lan |\phi_k|^2 \ran_{\Om_\tau} = 1$,  and 
 $\theta_q$ 
are entire functions  
given by the series
	 \begin{align}\label{thetaq}\theta_{q}(z, \tau):= 
e^{ -2\pi i a z +i c_1 z} \sum_{m=-\infty}^{\infty} e^{2\pi i  ( \frac12 m^2\tau +q m + m z)+ic_\tau'm},	\end{align}
 where $a, b$ are real numbers defied by $q=-a\tau +b$ and  $c_\tau':=c_1 \tau - c_\tau$, with $c_s$, the constants entering \eqref{gs-spec-tau}. 
\end{proposition}
 Above and in the rest of this section, $b$ stands for a {\it component} of $q$, and  {\it not the average magnetic flux per cell}, which is not used here.   We begin with the following
\begin{lemma} \label{lemma:phik} The functions $\phi_k (x),\ k\in \C,$ solve the problem 
\eqref{phik-eqs-tau}  if and only if they	are  of the form \eqref{phik-thetaq}
 where  $\theta_q (z, \tau)$ are entire functions satisfying the periodicity relations  
 \begin{align}\label{thetaq-per-gen}	&\theta_q (z+s, \tau) = e_q(z , s, \tau) \theta_q (z, \tau),\ \forall s\in \LAT.	\end{align}
where $e_q (z , s, \tau) := e^{-\frac{2\pi i}{\im\tau} (  (\im s) z  -\im (\bar s  q)+\frac12  (\im s) s)  +i c_s} $ and  $c_s$ are the constants which enter \eqref{gs-spec-tau}. 
  \end{lemma}
%
\begin{proof}	  Standard methods show that  the operator $-\COVLAP{a^0}$ on $L^2(\Om;\C)$ with the periodicity conditions in \eqref{phik-eqs-tau} 
is positive self-adjoint with discrete spectrum. To find its eigenvalues, we define the harmonic oscillator annihilation and creation operators, $c$ and $c^*$, with  
	\begin{equation}	\label{annih-creat-ops}
		c = \partial^*_{a^0} 
\quad \text{ and } \quad c^* :=\partial_{a^0}, \text{ where } \partial_{a}:=\p_{a 1} -i \p_{a 2}=\partial - i  a_c.	\end{equation}
Here, recall, $a_c:=a_1-i  a_2$, the complexification of $a $  and we use the definitions differing from the {\it standard}  ones by the factor $2$: $\p:= p_1-i \p_2 \equiv  \p_{x_1}-i \p_{x_2}$.  For  $a^0(x):= \frac12  (x_2, - x_1)$, we have  $a^0_c=\frac12  (x_2 + i  x_1)=\frac12 i \bar x^c $, with $x^c: = x_1+i x_2$. 
  These operators satisfy the relations 
	\begin{align}	\label{cc*-prop} [c, c^*] = 2,\ -\COVLAP{a^0} - 1 = c^*c. 
	\end{align}
		The representation $-\COVLAP{a^0} - 1 = c^*c$ implies that $\NULL (-\COVLAP{a^0} - 1) = \NULL c$ and so we study the latter.  
	Hence $\phi_k$ satisfies the equation
	\begin{equation}	\label{cphik=0}
		c\phi_k = 0.
	\end{equation}
		The relations \eqref{phik-thetaq} and  \eqref{cphik=0}	imply  the Cauchy-Riemann equation $\bar\p\theta_q (z, \tau)= 0$, i.e. $\theta_q (z, \tau)$ are entire functions.

Next,  it is straightforward to verify the periodicity relation in \eqref{phik-eqs-tau}  implies that the functions   $\theta_q(z, \tau)$ defined in \eqref{phik-thetaq},  satisfy the periodicity relations \eqref{thetaq-per-gen}.	 
 Indeed, by \eqref{phik-thetaq} and \eqref{phik-eqs-tau}, we have 
 \[\theta_q (z+s, \tau) =e^{-\frac{\pi}{2 \im\tau}((z+s)^2 -|z+s|^2)} \phi_k (x+
  s)\]\[=e^{\frac{\pi}{\im\tau}\al_s+i c_s}e^{-\frac{\pi}{2 \im\tau}(z^2 -|z|^2)} \phi_k  (x),\] 
  where  $c_s$ are the constants which enter \eqref{gs-spec-tau} and $\a_s:= -z s +\re(\bar s z)-\frac12 s^2+\frac12|s|^2-i 
 s\cdot x^\perp+ \frac{\im\tau}{\pi} s\cdot  k$, with $x^\perp:= (-x_2, x_1)$. Using the relations  $s\cdot x^\perp=- \Im (\bar s z)$ and $k\cdot s=\Re (\bar k s)=\frac{\pi}{\im\tau}\im (\bar s q )$, 
we find furthermore, 
\[\al_s=  -z s +\re(\bar s z) - i (\im s) s+ i\im (\bar s z+2\bar s q)=-2i [  (\im s) z  +\frac12  (\im s) s-\im (\bar s q)],\] which, together with the previous relation, gives \eqref{thetaq-per-gen}. 
\end{proof}
\begin{proof}[Proof of Proposition \ref{prop:phik}] Eq   \eqref{thetaq-per-gen} implies that $\theta_q(z, \tau)$   satisfy the periodicity relations 	
	\begin{align}\label{thetaq-per1}	&\theta_q(z + 1, \tau) = e^{- 2\pi i a +i c_1 } \theta_q(z, \tau), \\
	&\theta_q(z + \tau, \tau) = e^{- 2\pi i (z+b+\frac12 \tau) + i c_\tau} \theta_q(z, \tau), \label{thetaq-per2}	\end{align}
where 
$a, b$ are real numbers defied by $q=-a\tau +b$. Consequently, 
 a standard argument (see 
 Lemma I.2 of  Supplement I) 
then shows	 that  $\theta_q(z, \tau)$ is given by 	 
\eqref{thetaq}.  \end{proof}


\begin{remark} \label{rem:theta} 
 a) The family of functions   $\theta_q(z, \tau),\ q\in   \Omega'_\tau:=\frac{\im\tau}{2\pi}  i \Omega_\tau^*$,  
 $z\in \C,\ \im\tau >0,$ defined  in \eqref{thetaq}, are the   theta functions with finite characteristics (see \cite{Mum}), appearing in number theory, while 
$\theta_{q=0}= \theta$ is the standard theta function.  
Unlike the number theory, where $a, b\in \ell^{-1}\Z$, for some positive integer $\ell$, in our case $q\in \Omega'_\tau$, 
which is a rescaled and rotated fundamental cell the dual lattice $\cL_\tau^*$. 

  One can define the  theta functions  (with finite characteristics),  $\theta_q^\LAT(z),\ q\in  \C,$ without reference to a specific representation of the underlying lattice, as the entire functions satisfying the quasi-periodicity condition \eqref{thetaq-per-gen}.


b)  
  In the terminology of Sect 13.19, eqs 10-13 of \cite{Erdel}, our theta function $\theta (z, \tau)$ is called  $\theta_{3}(z, \tau)$. The choice of the original  theta function determines the location of zeros of $\phi_k (z)$: The zeros of   $\theta_{3}(z, \tau)$ are located at the points of $\Z +\tau \Z+\frac12 + \frac12 \tau$, while the zeros of   $\theta_{1}(z, \tau)$ (in the terminology of \cite{Erdel}) are located at the points of $\Z +\tau \Z$.  To compare, $\theta_{1}(z, \tau)$ is defined as 
$\theta_{1}(z, \tau):= 
\sum_{m=-\infty}^{\infty} e^{\pi i   m} e^{\pi i (m-\frac12 )^2\tau}  e^{\pi i (2m -1)z} .$

c) Proposition \ref{prop:phik} implies that the products $|\phi_0|^2|\phi_k|^2, \phi_0^2\bar{\phi}_k\bar{\phi}_{-k}$ 
are periodic on $\Omega_\tau$.
%
%
\end{remark} 

We display the dependence of  the function  $\phi_k(x)$ on the lattice $\LAT$: $\phi_k(x) = \phi_{ k}^{\LAT}(x)$. Since,  as can be easily verified,  the function  $\psi_k(x) = \phi_{g k}^{g\LAT}(g x),\  g\in O(2),$  satisfies  \eqref{phik-eqs-tau}, the uniqueness for  \eqref{phik-eqs-tau} gives 
\begin{align}\label{phik-rot} \phi_{g k}^\LAT(g x)= \phi^{g^{-1}\LAT}_{k}(x),\ \text{ for any }\  g\in O(2). \end{align}	
 Since $g: x\ra -x$ and (assuming) $g: x\ra \s x$ , where $\s (x_1, x_2)= (-x_1, x_2)$ (or $\s x= - \bar x$, in the complex form),  leave the lattice $\LAT$ invariant,  the equation \eqref{phik-rot} implies
\begin{align}\label{phik-parity} \phi_k (-x)= \phi_{-k} (x),\ \text{ and }\ \phi_k (\s x)= \phi_{\s k} (x). 
\end{align}
  
\paragraph{Symmetries of $\g_k( \tau)$.} \label{sec:gamk}

The function $\g_{ k}( \tau)$ defined in \eqref{gamk-tau} 
 has the properties  \begin{align}\label{gamq-propert'} & \g_{k+s^*}( \tau)= \g_{ k}( \tau),\  \forall s^*\in \LAT^*, \quad  \g_{\bar k}(- \bar\tau)= \g_{k}( \tau), \quad
\g_{-k}( \tau) = \g_{k}( \tau). \end{align}
 The first property was already stated in \eqref{phik-per-k}, the second one follows from the  definition \eqref{gamk-tau} of $\g_{ k}( \tau)$ and  
  the second equation in \eqref{phik-parity} 
  and the last relation in  \eqref{gamq-propert'}, from the first equation in \eqref{phik-parity}.   
 
 These relation also follow the explicit representation \eqref{gamk-series}.   To show the second relation in  \eqref{gamq-propert'}, we notice that  the transformation $\tau\ra - \bar\tau,\ k\ra \bar k$  is equivalent to  mapping  $(m, n)\ra (m, - n)$ 
 and taking the complex conjugate of $\g_k( \tau)$. Since the sum in   \eqref{gamk-series} is invariant under the latter mapping and since $\g_k( \tau)$ is real, the r.h.s. of  \eqref{gamk-series} is invariant under the transformation $\tau\ra - \bar\tau,\ k\ra \bar k$.  
  To prove the third relation in \eqref{gamq-propert'} we observe that the r.h.s. of \eqref{gamk-series} is invariant under the transformation $k\ra - k,\ t\ra - t$.


%

\begin{proposition}\label{thm:gamk-cp2} The points  $k\in  \frac12 \cL_\tau^*$ 
   are critical points of the function $\g_{k}( \tau)$ in $k$. 
   \end{proposition}
\begin{proof} 
 We will use that, due to \eqref{gamq-propert'}, $\g_{t-k}( \tau)= \g_{ k}( \tau),$ for any $t\in  \cL_\tau^*$.  Differentiating this relation  w.r. to $\re k$ and  $\im k$ at  $k\in  \frac12  \cL_\tau^*$  and using  that the 
  points  $k\in  \frac12  \cL_\tau^*$ are fixed points under the maps $k \ra t-k$, where $t\in  \cL_\tau^*$, we find $\p_{\re k} \g_{k}( \tau)=0$ and $\p_{\im k} \g_{k}( \tau)=0$  for $k\in  \frac12  \cL_\tau^*$.\end{proof}

 %
\paragraph{Proof of Proposition  \ref{prop:gamma}.} \label{sec:gamk}

In this proof we omit the subindex 
$\del$ in $\g_\del(\tau)$. The first two properties follow directly from the corresponding properties of $\g_k (\tau)$. We summarize these properties as
 \begin{align}\label{gam-propert} &\g( \tau+1)= \g( \tau),\  \g( -\tau^{-1})= \g( \tau),\ \g (- \bar\tau)= \g( \tau). \end{align}

  To prove the third statement, we will use that the points  $\tau=e^{i\pi/2}$ and $\tau=e^{i\pi/3}$ are fixed points under the maps $\tau \ra - \bar\tau$,  $\tau \ra -\tau^{-1}$ and  $\tau \ra 1 - \bar\tau$,  $\tau \ra 1 -\tau^{-1}$, respectively. By the first and third relations in \eqref{gam-propert},    we have that $\g (n- \bar\tau)= \g ( \tau)$, for any integer $n$. Remembering the definition  $\tau =\tau_1 +i \tau_2$, differentiating the relation $\g (n- \bar\tau)= \g ( \tau)$ w.r. to $\tau_1$,  and using  that the points  $\tau=e^{i\pi/2}$ and $\tau=e^{i\pi/3}$ are fixed points under the maps $\tau \ra - \bar\tau$ and  $\tau \ra 1 - \bar\tau$, respectively, we find $\p_{\tau_1}\g ( \tau)=0$  for $\tau=e^{i\pi/2}$ ($n=0$) and  for $\tau=e^{i\pi/3}$ ($n=1$).

Next, we find the derivatives w.r. to $\tau_2$. We consider the function $\g (\tau)$ as a function of two  real variables, $\g (\tau_1, \tau_2)$. Then the relation   $\g (\tau)= \g (n-\tau^{-1})$, where $n$ is an integer, which follows from   the first two relations in \eqref{gam-propert},      can be rewritten as $\g (\tau_1, \tau_2)= \g (n-\frac{\tau_1}{|\tau|^2}, \frac{\tau_2}{|\tau|^2})$. Differentiating the latter relation w.r. to $\tau_2$, we find
\begin{align}\label{gamq-der-tau2}(\p_{\tau_2}\g )( \tau_1, \tau_2)=2\frac{\tau_1\tau_2}{|\tau|^4}(\p_{\tau_1}\g)(n -\frac{\tau_1}{|\tau|^2}, \frac{\tau_2}{|\tau|^2})+2\frac{\tau_1^2-\tau_2^2}{|\tau|^4} (\p_{\tau_2}\g)(n -\frac{\tau_1}{|\tau|^2}, \frac{\tau_2}{|\tau|^2}).\end{align} Since  $\p_{\tau_1}\g ( \tau)=0$  for $\tau=e^{i\pi/2}$ and  $\tau=e^{i\pi/3}$ and since the points  $\tau=e^{i\pi/2}$ and $\tau=e^{i\pi/3}$ are fixed points under the maps  $\tau \ra -\tau^{-1}$ and   $\tau \ra 1 -\tau^{-1}$, respectively, this gives  $(\p_{\tau_2}\g)( 0, 1)=0$  for $\tau=e^{i\pi/2}$ ($n=0$) and for $\tau=e^{i\pi/3}$ ($n=1$).     \qquad $\Box$
%

 

\section{Hessian and its spectrum} 
\label{sec:hess-spec} 

In 
this section, we study the spectrum of the hessian $\Lt$ defined in \eqref{Lt}. We fix the subindex $\tau$ and omit it at $\Lt, \LAT_\tau, \Om_\tau, \Om_\tau^*$ and the subidex $\Om$ at $\lan \cdot, \cdot \ran_\Om$, i.e. we write $L, \LAT, \Om, \Om^*$ and$\lan \cdot, \cdot \ran$ for $\Lt, \LAT_\tau, \Om_\tau, \Om_\tau^*$ and $\lan \cdot, \cdot \ran_\Om$.




\subsection{Map $J(u)$}\label{sec:J}
Let $u=(\psi, \bar\psi, a)$  ({\it recall that the scalar potential $\phi$ is determined by $u$ and is not displayed}). 
Recall that $J(u)$ denotes the map on the r.h.s. of  the Gorkov-Eliashberg-Schmid equations \eqref{GESresc-c}, let \[\p_{t, \phi}  v=((\p_{t}+i \phi)\xi, (\p_{t}-i \phi)\eta, \p_{t}\a+\n\phi),\]
 for $v=(\xi, \eta, \a)$, and assume for simplicity $\chi=1$ and $\s=\one$, for the parameters on the l.h.s. of \eqref{GESresc-c}, so that the latter equation can be written as
 \begin{align}\label{u-eq}
 &\p_{t, \phi} u=J (u).
\end{align}
(Recall that $J (u)$ is the gradient of the Ginzburg-Landau energy \eqref{gl-en-c} as explained after \eqref{dE-barpsi}.) The above symmetries follow from the covariance of this map under the corresponding transformations. 
For instance, we have
 \begin{equation}  \label{J-gauge-transl-inv}
   J( T^{\rm gauge}_\gamma u) =  \tilde T^{\rm gauge}_\gamma J(u), \quad   J( T^{\rm trans}_s u) =  T^{\rm trans}_s J(u),\end{equation}
for every $\gamma$ and $s$, where  the gauge transformations of the rescaled and extended unknowns (written for the time-independent fields), $T^{\rm gauge}_\gamma$ and $ \tilde T^{\rm gauge}_\gamma $, and the translation transformations, $T^{\rm trans}_s$, are defined in \eqref{Tgauge}, \eqref{tildeTgauge} and \eqref{Ttrans}, respectively.
\DETAILS{$T^{\rm gauge}_\gamma$ is the gauge transformations of the rescaled and extended unknowns (written for the time-independent fields),
\begin{align}\label{gauge-transf}
 T^{\rm gauge}_\gamma :\ &  (\xi( x), \eta( x), \a( x)) 
  \mapsto (e^{i\gamma(x)}\xi( x), e^{- i\gamma(t,x)}\eta( x),  \a( x) + \nabla\gamma( x)).
\end{align}
and $ \tilde T^{\rm gauge}_\gamma $ is given by 
\begin{align} \label{tildeT-gauge} \tilde T^{\rm gauge}_\gamma : (\xi, \eta, a) \mapsto (e^{i\gamma}\xi, e^{-i\gamma}\eta, \a ),\end{align}
and  $T^{\rm trans}_s$ denotes translation by $s$,
 \begin{equation}\label{Ttrans-v}
    T^{\rm trans}_s : (\xi( x), \eta( x), \al ( x))  \mapsto (\xi( x+s), \eta(x+s),  \al (x+s)).
\end{equation}}
As a demonstration, we show the gauge invariance in the following lemma needed later on:
 \begin{lemma}\label{lem:u-eq-moving-frame}
Let  $\g$ depend on $t$, and $\dot\g=\p_t \g$. If $u$ solves \eqref{u-eq} then $\tilde u=(T_{\g}^{\rm gauge})^{-1} u$ satisfies the equation
 \begin{align}\label{Tgauge-u-eq}
 &\p_{t, \phi+\dot\g}  \tilde u=J (\tilde u).
\end{align}
 \end{lemma}
\begin{proof} We write the equations  \eqref{GESresc-c} in the form
  $ \p_{t\phi }u = J(u).$ 
 We  use \eqref{J-gauge-transl-inv},  the definition $\tilde u=T_{g}^{-1} u$, introduced earlier, and
  \begin{align}  \label{dtTgaugeu}
 &  \p_{t,\phi} T^{\rm gauge}_\gamma u  =  \tilde T^{\rm gauge}_\gamma \p_{t,\phi+ \dot{\gamma}} u.\end{align}
 to obtain the equation \eqref{Tgauge-u-eq}.
\end{proof}
There are additional the reflection and `particle-hole' symmetries which play an important role in our analysis. The reflection transformation is  
 \begin{equation}\label{Trefl}
    T^{\rm refl} : (\xi( x), \eta( x), \al ( x))  \mapsto (\xi(  - x), \eta( - x), - \al ( - x)),
\end{equation}
and 
the `particle-hole' (real-linear) transformation is defined as 
\begin{equation}\label{CS}
\tilde\rho:=\mathcal{C}\cS,\ \quad	 
\text{ where }\ \cS =  \left( \begin{array}{ccc}	0 & 1 & 0  \\1 & 0 & 0  \\ 0 & 0 & 1 \end{array} \right), 
\end{equation} 
 and, recall, that $\cC$ denotes the complex conjugation. As can be easily checked from the definitions, $\tilde\rho u = u$ and the map $J(u)$  is covariant under 
 these transformations, 
\begin{align}\label{refl-tilderho-sym}
&  
T^{\rm refl} J(u)=  J( T^{\rm refl} u), \quad  
  \tilde\rho J(u)=  J( \tilde\rho u). \end{align}


The equation \eqref{refl-tilderho-sym}, 
 and the relations $T^{\rm refl} \ut=\ut$ and $\tilde\rho \ut=\ut$ imply that the hessian $L\equiv \Lt$ defined in \eqref{Lt} satisfies
\begin{align}\label{L-refl-rho-sym}
T^{\rm refl} L=  L T^{\rm refl},\ \quad   \text{ and }\ \quad \tilde\rho L=  L \tilde\rho. 
\end{align}
This can be also verified directly, using the explicit expression \eqref{Lc-expl} of Appendix \ref{sec:hess-expl} for $L$.
The first relation above is somewhat subtle as $L$ has odd as well as even entries. The latter is compensated by the fact that $T^{\rm refl}$ treats the order parameter and vector fields components differently.
\begin{remark} \label{rem:rhou}  We could have extended the GES equations so that $u$ would satisfy $\rho u=u$, instead of $\tilde\rho u=u$.
\end{remark} 

\subsection{The Bloch - Floquet - Zak decomposition for gauge-periodic operators} \label{sec:Bloch-deco}

The key tool in analyzing the hessian is to exploit the gauge-periodicity of the Abrikosov lattice. 
As a result of this periodicity, the space $\cH$ decomposes as the direct fiber integral of spaces on a compact domain in such a way that the operator $\Lt $ is decomposed as the direct integral of operators on these spaces. 

Recall the definition \begin{align}\label{rhot}\cR_s = \tilde T^{\rm gauge}_{- g_s} T^{\rm trans}_s, 
\end{align}
 for each $s \in \LAT$, where the maps  $\tilde T^{\rm gauge}_\gamma$ and  $T^{\rm trans}_s$ are defined in \eqref{tildeTgauge} and \eqref{Ttrans}, and the function  $g_s$ satisfies the co-cycle condition \eqref{cocycle-cond}. 
\DETAILS{This map can be also written as
\begin{align}\label{rhot}\rho_t = T_t \oplus \bar{T}_t \oplus T^{\rm trans}_t , 
\end{align}
where the \emph{magnetic translation operator} $T_t$ on $L^2(\R^2;\C)$ to act as
\begin{align}\label{Tt}T_t\xi(x) = e^{-ig_t(x)}\xi(x+t).\end{align}
Recall the notation $\bar A:= \mathcal{C} A\mathcal{C}$ and recall that  $T^{\rm trans}_t$ denotes translation by $t$. We now let 
$\cR_t$ be the 
  operator on $\Hil$ defined by}
Clearly, this map is related to  the \emph{magnetic translation operator} 
\begin{align}\label{Tt}T_s\xi(x) = e^{-ig_s(x)}\xi(x+s).\end{align}
\begin{proposition}
	(i) $\cR_s$ is a unitary group representation of $\LAT $ in $\cH$ (i.e. $\cR_s$ is unitary and $\cR_s \cR_t = \cR_{s+t}$, for all $s, t \in \LAT $) and (ii)  if  functions $g_s : \R^2 \to \R$, $s \in \LAT $, in \eqref{rhot} are the same as those entering  the gauge-periodicity condition \eqref{gauge-per-resc} for the Abrikosov lattice solution  $\ut$, then $\cR_s$ commutes with the operator $L$ (i.e. 
	$\cR_s L  = L \cR_s$, for all $s \in \LAT$).
\end{proposition}
\begin{proof} Clearly, the operators 
$ \tilde T^{\rm gauge}_{\chi}$ and $ T^{\rm trans}_t$ are unitary. To show the group property, we write 
\begin{align*}\cR_t \cR_s =& \tilde T^{\rm gauge}_{- g_t} T^{\rm trans}_t \tilde T^{\rm gauge}_{- g_s} T^{\rm trans}_s\\&=\tilde T^{\rm gauge}_{- g_t}  \tilde T^{\rm gauge}_{- T^{\rm trans}_t g_s}T^{\rm trans}_t T^{\rm trans}_s=T^{\rm gauge}_{- g_t- T^{\rm trans}_t g_s}T^{\rm trans}_{t+s}.\end{align*} 
Since, by the cocycle condition \eqref{cocycle-cond}, $g_t + T^{\rm trans}_t g_s-g_{t+s}\in 2\pi\Z$, this gives  $\cR_t \cR_s = \cR_{t+s}$.
Thus $T_t$ is homomorphism from $\LAT $ to the group of unitary operators on $L^2(\R;\C)$. 

To prove that  $\rho$ commutes with the operator $L$, we use \eqref{J-gauge-transl-inv} to obtain $J( T^{\rm gauge}_{- g_s} T^{\rm trans}_s u)$  $ =  \tilde T^{\rm gauge}_{- g_s} T^{\rm trans}_sJ(u)$. Differentiating this relation w.r.to $u$ and using $d ( T^{\rm gauge}_{\chi}u) ( \xi)= \tilde T^{\rm gauge}_{\chi}  \xi$ gives $d J( T^{\rm gauge}_{- g_s} T^{\rm trans}_s u)  \tilde T^{\rm gauge}_{- g_s} T^{\rm trans}_s \xi=  \tilde T^{\rm gauge}_{- g_s} T^{\rm trans}_s d J(u) \xi$, which, together with $T^{\rm gauge}_{- g_s} T^{\rm trans}_s \ut= \ut$ and the definition of $L$, implies $L \cR_t = \cR_t L$. This relation also follows by  a simple verification. \end{proof}
 We extend the character  $\chi_k(s)= e^{ik\cdot s}$ to act on $v = (\xi, \eta, \alpha)$ as the multiplication operator
\begin{align*}e^{ik\cdot s} v = 
(e^{ik\cdot s}\xi, e^{ik\cdot s}\eta, e^{ik\cdot s}\alpha). 
 \end{align*}
Note that the subspace $\Ran\pi =\{(\xi, \bar\xi, \al) \in L^2(\R^2;\C\times \C\times \C^2)\}$, we started with, is not invariant under this operator. 

We now define the direct integral Hilbert space 
$    \fK = \int^{\oplus}_{\Omega^*} \Hilk \ \hat {dk}, $ 
where $\hat {dk}$ is the usual Lebesgue measure on $\Omega^*$,  
divided by $|\Omega^*|$, and
$\cH_k$ is  the set of functions, $v$, from $L^2_{\rm loc}(\R^2;\C\times \C\times \C^2)=L^2_{\rm loc}(\R^2;\C)\times L^2_{\rm loc}(\R^2;\C)\times L^2_{\rm loc}(\R^2; \C^2)=L^2_{\rm loc}(\R^2;\C)^4$,  satisfying 
  \begin{equation}\label{gauge-per-vk} 
  \cR_s v (x) = e^{i k\cdot s}	 v (x),\ \forall s\in\ \LAT, 
  \end{equation}
   a.e. 	and   endowed with the  
    inner product 
   \begin{equation} \label{L24-inner-product}    \lan v, v' \ran_{L^2} = 
   \int_{\Om} \bar{\xi}\xi' + \bar{\eta}\eta' + \bar{\alpha}\cdot \alpha' , 
\end{equation}
where  $\Om$ is the fundamental cell  of the lattice  $\LAT$, identified with  $\cT:=\R^2/\LAT$, and $v = (\xi, \eta, \alpha),\  v' = (\xi', \eta', \alpha')$. 
The  inner product in $\fK $ is given by $\lan v , w \ran_{\fK }:=
\int_{\Omega^*}\lan v_k, w_k\ran_{\Hilk} \hat {dk}$.
 We write $f= \int_{\Omega^*}^\oplus f_k \hat {dk}$, where $f_k$ for the $k$-component of $f$, and by the symbol $ \int_{\Omega^*}^\oplus T_{k} \tilde d k$ we understand the operator $T$ acting on $\fK$  as \begin{equation*}Tf= \int_{\Omega^*}^\oplus T_{ k}f_k \hat {dk}.  \end{equation*}


\begin{remark} 
{ One can think  of $\fK$ as  $L^2_{\rm loc}(\R^2\times \R^2/\LAT^*;\C)^4$, or as the set of functions, $v_k(x)$, from $L^2_{\rm loc}(\R^2\times \R^2;\C)^4$,  satisfying \eqref{gauge-per-vk} and $v_{k+s^*}(x)=v_k(x),\ \forall s^*\in \LAT^*,$ a.e. and endowed with the $L^2(\Omega\times \Om^*;\C)^4$ inner product.}
\end{remark}

For $k \in \Omega^*$, let $L_{ k}$ be the operator $L $ acting on $\Hilk$ with the domain $\cH_k\cap H^2_{\rm loc}(\R^2;\C)^4$.  It is easy to check that the operator $L $ leaves the (gauged Bloch-Floquet) conditions \eqref{gauge-per-vk} invariant.	  Note also that $\mathcal{C} \Lk= L_{  - k}\mathcal{C}$. We have (cf. \cite{PST})
\begin{proposition} \label{prop:bloch}
		Define $U : \mathcal{H} \to \fK$ on smooth functions with compact supports by the formula
	\begin{align} \label{U}
		(U v)_k(x) = \sum_{t \in \LAT} e^{- ik\cdot t} \cR_t v(x).
	\end{align}
	Then $U$ extends uniquely to a unitary operator satisfying 
	\begin{equation}\label{K-decomposition}
		UL U^{-1} = \int_{\Omega^*}^\oplus \Lk \hat {dk}.
	\end{equation}
Each $\Lk$ is a self-adjoint operator with compact resolvent (and therefore purely discrete spectrum), and
	\begin{equation}\label{KKkspecrelat}
		\sigma(L ) = \overline{\bigcup_{k\in\Omega^*} \sigma(\Lk)}.
	\end{equation}
\end{proposition}
\begin{proof} We begin by showing that $U$ is an isometry on smooth functions with compact domain. Using Fubini's theorem and the property $|\Om ||\Om^*|=1 $, we calculate
	\begin{align*}
		\| Uv \|_{\fK}^2
		&= 
		\int_{\Omega^*} \left\| (U v)_k \right\|_{L^2} \hat {dk} = \int_{\Omega^*} \int_{\Omega}
			\left| \sum_{t \in \LAT} e^{- ik\cdot t} \cR_t v(x) \right|^2 dx \hat {dk} \\
		&=  \int_{\Omega} \left(
			\sum_{t,s \in \LAT} \cR_t v(x) \overline{\cR_s v(x)}
			\int_{\Omega^*} e^{- ik\cdot t} e^{ ik\cdot s} \hat{dk} \right) dx.
	\end{align*}
We compute, after writing $k$ and $s-t$ in the coordinate form,  $\int_{\Omega^*} e^{- ik\cdot t} e^{ik\cdot s} \tilde d k=\del_{s, t}$. Using this, we obtain furthermore
	\begin{align*}
		\| Uv \|_{\fK}^2
		&=  \int_{\Omega} \sum_{t \in \LAT} |\rhoc_t v(x)|^2 dx = \int_{\R^2} |v(x)|^2 dx = \| v \|_{\mathcal{K}}^2.
	\end{align*}	
	Therefore $U$ extends to an isometry on all of $\Hil$. To show that $U$ is in fact a unitary operator we define $U^* : \fK \to \Hil$ by the formula
	\begin{align} \label{U*}
		 (U^* g)(x) = 
		 \int_{\Omega^*} e^{ ik\cdot t}  (\rhoc_t^* g_k)(x) \hat {dk},	\end{align}
for any 
 $x \in \Om +t$  and $t \in \LAT$. Straightforward calculations show that $U^*$ is the adjoint of $U$ and that it too is an isometry, proving that $U$ is a unitary operator.

 For completeness we show that  $UU^*=\id$. Using that the definition of  implies that $\rhoc_t (U^* g)(x) = \int_{\Omega^*} e^{ ik\cdot t} g_k(x) dk$, we compute
 \begin{align*}
		(U U^* g)_k(x) = \sum_{t \in \LAT} e^{- ik\cdot t} \int_{\Omega^*} e^{ ik' \cdot t} g_{k'}(x) \hat {dk}'.
	\end{align*}
Furthermore, using  the Poisson summation formula,
\begin{align}\label{Poisson'}
\sum_{t \in \LAT}  e^{- ik\cdot t}  
= |\Om^*| 
\sum_{t^*\in \LAT^*}  \del(t^*- k),	\end{align}
we find, on a dense set of functions $g_k(x)$ vanishing on the boundary $\p\Omega^*$,
\begin{align*}
		(U U^* g)_k(x) = \int_{\Omega^*}   |\Om^*| 
		 \sum_{t^*\in \LAT^*}  \del(t^*- k+k') g_{k'}(x) \hat {dk}' =g_{k}(x).
	\end{align*}

	Next, we show that  $(Uv)_k$ satisfies the gauge-periodicity conditions \eqref{gauge-per-vk}: 
		\begin{align*}
		\rhoc_t (Uv)_k(x)
			&= \sum_{s \in \LAT} e^{- ik\cdot s} \rhoc_t \rhoc_s v(x) = \sum_{s \in \LAT} e^{- ik\cdot s} \rhoc_{t+s} v(x) \\
			&= e^{ ik\cdot t} \sum_{s \in \LAT} e^{- ik\cdot (s+t)} \rhoc_{t+s} v(x),
	\end{align*}
	which gives that
	\begin{align*}
		\rhoc_t (Uv)_k(x)	&= e^{ ik\cdot t} (U v)_k(x).
	\end{align*}
  Treating $\Lk$ as a differential expression applied to differentiable function of $x$, we compute	
  \begin{align*}
		(\Lk (U v)_k)(x)
		&= \sum_{t \in \LAT} e^{- ik\cdot t} \Lk\rhoc_t v(x) 
		= \sum_{t \in \LAT} e^{- ik\cdot t} \rhoc_t L  v(x)	\end{align*}
	and therefore
	\begin{align*}
		(\Lk (U v)_k)(x)		&= (U L v)_k(x),
	\end{align*}	which establishes \eqref{K-decomposition}.

The self-adjointness of the operators   $\Lk$ and the  compactness of their resolvents follow by standard arguments. We now turn to the relation \eqref{KKkspecrelat}. 
We first prove the $\supseteq$ inclusion. Suppose that $\lambda \in \sigma(\Lk)$ for some $k \in \Omega^*$. Then there exists a smooth eigenfunction $v \in \mathcal{D}(\Lk)$ solving  $\Lk v = \lambda v$. 
By the definition of  $L $, the function $v$ solves the equation  $L v = \lambda v$. Therefore, by Schnol-Simon theorem (see e.g. \cite{GS}), $\lambda$ must be in the essential spectrum of $\Lt$. 
	
	As for the $\subseteq$ inclusion, suppose that $\lam \not\in \overline{\bigcup_{k\in\Omega^*} \sigma(\Lk)}$. Then the operators $(\Lk - \lambda)^{-1}$ are uniformly bounded, and therefore $(L  - \lam)^{-1} = \int^{\oplus}_{\Omega^*} (\Lk - \lambda)^{-1} \hat {dk}$ is also bounded and therefore $\lam \not\in \sigma(L )$.
\end{proof}
  We call the map $U$  the Bloch - Floquet - Zak (or BFZ) operator. 
 We collect general statements related to the map  $U$, 
  we use below.  An additional property of $U$ is given in Appendix \ref{U-prod-transf}.   Let  $\n' =\n_k $ and $D$ is either $\nat \oplus \one\oplus \one$ or  $\one \oplus \overline{\nat}\oplus \one$ or $\one\oplus \one \oplus \n$. We have 
\begin{equation}\label{Uv-k-period} (U v)_{k+s^*}(x) =
	(U v)_k(x),\ \forall s^*\in \LAT^*, \end{equation}
	\begin{align} \label{Trefl-tilderho-fibration}
 T^{\rm refl}(U v)_{k}(x)=	(U T^{\rm refl} v)_{- k} (x),\  \tilde\rho(U v)_{k}(x)= (U \tilde\rho v)_{-k}(x),
	\end{align}  
 \begin{align} \label{k-deriv-Uv} &(i \n' + x ) U =   U x,\ \quad  D U  =  U D,\\ \label{norm-rel1} &
\| D^m U v\|_{L^2}  
= \|D^{m}  v \|_{L^2},\ \quad  \|\p_k^m U v\| \ls \sum_{m'\le m}\| x^{m'} v \|_{L^2},\\ 
\label{norm-rel2}   &\| 
D^m U^{-1} g\|_{L^2}  
=\|D^{m} g \|_{\sH}. \end{align}  
Above  $T^{\rm refl}$ and $\tilde\rho$ act  
on $\cH$ and from $\cH_k$ to $\cH_{-k}$ (we use the same notation in both cases). 
 The property \eqref{Uv-k-period}  follows from the definition of $U$ and the fact that $e^{-(k+s^*)\cdot t}=e^{- k\cdot t}, \forall t\in \LAT$. 
 The property \eqref{Trefl-tilderho-fibration}  follows from the property  $g_{-s}(-x)=g_{s}(x)$ (see \eqref{gs-spec-even}). 

Next, by \eqref{U}, we have $\n_k (U v)_k(x) =- i  \sum_{t \in \LAT} e^{- ik\cdot t} t  (\rhoc_t  v)(x)$. Now, using $(x+t) \rhoc_t= \rhoc_t x$ gives \[\n_k (U v)_k(x)  =- i  \sum_{t \in \LAT} e^{- ik\cdot t}  (\rhoc_t x v) (x) + i x \sum_{t \in \LAT} e^{- ik\cdot t}  (\rhoc_t  v)(x),\]
  which implies the first relation in \eqref{k-deriv-Uv}. To derive the second relation in \eqref{k-deriv-Uv} we use that 
 $(\n_{\at (x)} \oplus \one \oplus \one) \rhoc_t= \rhoc_t (\n_{\at (x-t) +\n_x g_t (x-t)} \oplus \one\oplus \one)$, $\at (x-t) = \at (x) + i\n_x g_{-t} (x)$, $g_{-t} (x)=- g_{t} (x)$ and $g_{t} (x-t)=- g_{t} (x)$. Similarly we deal with $\one \oplus \overline{\nat}\oplus \one$.  The proof for $\one\oplus \one \oplus \n$ in \eqref{k-deriv-Uv} is even simpler.

Using  the relations \eqref{k-deriv-Uv} and the fact that $x$ on the l.h.s. of \eqref{k-deriv-Uv} stands for the multiplication operator by $x\in \Om$ (while on the r.h.s., by $x\in \R^2$) and  the unitarity of $U$, we find \eqref{norm-rel1} - \eqref{norm-rel2}.

Finally, the equations \eqref{L-refl-rho-sym} and \eqref{Trefl-tilderho-fibration} imply that $\Lk$ has the following symmetries: 
	\begin{align} \label{commut-Trefl-tilderho-Lk}
	 T^{\rm refl}\Lk=	L_{- k} T^{\rm refl},\  \quad \tilde\rho \Lk= L_{- k}  \tilde\rho.
	\end{align} 
The symmetries $\tilde\rho$ and $T^{\rm refl}$ do not fibre, i.e. do not descends to $\cHk$, 
but their combination 
\begin{align}\label{rho}
&\rho: =T^{\rm refl} \tilde\rho \equiv T^{\rm refl} \mathcal{C}  \cS,
  \end{align} 
 does. Indeed, let 
  $\refl: f(x)\ra f(-x)$. Using $\refl e^{i g_s}= e^{i g_{- s}} \refl$ (due to \eqref{gs-spec-even}) and $\refl \overline{e^{i k \cdot x}}= e^{i k \cdot x} \refl$ and remembering the definitions of  $\tilde T^{\rm gauge}_\chi$, 
 we have that 
 \[ T^{\rm transl}_{s} T^{\rm refl} = T^{\rm refl} \tilde T^{\rm transl}_{-s},\ \tilde T^{\rm gauge}_{g_s} T^{\rm refl} = T^{\rm refl} \tilde T^{\rm gauge}_{g_{- s}}\ \text{ and  }\ \refl\cC e^{i k \cdot x}= e^{i k \cdot x} \refl\cC.\] 
 This implies that if $v$ satisfies \eqref{gauge-per-vk}, then so does $\rho v$, i.e.  $\rho$ maps $\cH_k$ into itself. Thus we have
 \begin{align}\label{commut-L-Lk-rho}
&[ L, \rho ] = 0,\  \quad  [ \Lk, \rho ]=0.
  \end{align}

We conclude this subsection with a general property of the eigenvalues of $\Lk$, which is used below.
 \begin{lemma}\label{lem:evLk-smooth-k} 
Simple eigenvalues of $\Lk$ are smooth in $k$. \end{lemma} 
\begin{proof} Consider the operator $\tilde \Lk:= e^{- i k\cdot x} \Lk e^{i k\cdot x}$ defined on $\cH_{k=0}$. Since the map $v\ra e^{i k\cdot x} v$ maps $\cH_{k=0}$ unitarily into $\cH_{k}$, it has the same eigenvalues as the operator $\Lk$. It is easy to show, using the relation  $ e^{- i k\cdot x} \n_a e^{i k\cdot x}=\n_{a-k}$ that the operator $\Lk$ depends on $k$ smoothly (say, as an operator from $H^2$ to $L^2$). Hence the statement follows from the standard perturbation theory (see e.g. \cite{RSIV, GS}). \end{proof}

\begin{remark}  
 While the global gauge and translational zero modes  \eqref{G-glob} and \eqref{transl-modes} are generalized eigenfunctions of the operator $\Lt$, they are now in the space $\cH_{k=0}$ and are standard eigenfunctions of $L_{k=0}$.
 \end{remark}  
 

\subsection{The fiberization of the gauge zero modes} 
 \label{sec:fiberiz}

Define the Bloch - Floquet (-  Fourier) transform $U^{\rm bf} : L^2(\R^2) \to L^2(\Om^*, L^2(\Om, \C))$ on smooth functions with compact supports by the formula
	\begin{align} \label{U}
		(U_{\rm bf} \chi)_k(x) = \sum_{t \in \LAT} e^{- ik\cdot t} T^{\rm trans}_t \chi(x),
	\end{align}
where $T^{\rm trans}_t$ acts now on scalar functions, and its adjoint
	\begin{align} \label{U*}
			 (U_{\rm bf}^* \eta)(x) = 
		 \int_{\Omega^*} e^{ ik\cdot t_x}  (T^{\rm trans}_{-t_x} \eta_k)(x) \hat {dk},	\end{align}
where $t_x \in \LAT$ is defined by the relation $x - t_x \in \Om $. Note that the  Bloch - Floquet (-  Fourier) transform $ \chi_k$ of  $ \chi$ satisfies $T^{\rm trans}_t \chi_k= e^{i k\cdot t} \chi_k$. 
Thus $\chi_k\in H^1_k$, where 
\begin{align} \label{H1k-space} H^1_k:=\{\chi_k\in H^1: \chi_k(x+s) = e^{ik\cdot s}\chi_k(x), \forall s\in \LAT, \chi_k(- x)= \chi_{ -k}(x)\}. \end{align} 
 (The condition $\chi_k(- x)= \chi_{ -k}(x)$ comes from the fact that the original functions $\chi (x)$ are real.) We begin with 
\begin{proposition}  \label{prop:Gk} Let $\chi_k$ be the Bloch - Floquet-Fourier transform of $\chi$. We have
 \label{prop:Lk-spec}\begin{align} \label{UG} & (U G_\chi)_k(x) = G_{\chi_k}=(i \chi_k\psit, - i \chi_k\overline{\psit}, \n \chi_k),\\
 \label{Lk-Gk}  &\Lk G_{\chi_k}=0,\  \cR_s G_{\chi_k} = e^{i k\cdot s} G_{\chi_k},\ \forall s\in\ \LAT,\\  
 &  \label{Trefl-rho-Gk} T^{\rm refl} G_{\chi_k}= G_{t^{\rm refl} \chi_k},\  \tilde\rho G_{\chi_k}=- G_{\bar\chi_k},\  
 \rho G_{\chi_k}=- G_{t^{\rm refl} \bar\chi_k}.\end{align}
\end{proposition}
\begin{proof} We apply the BFZ transform,
\begin{align} \label{U}
		(U v)_k(x) = \sum_{t \in \LAT} e^{- ik\cdot t} \cR_t v(x),
	\end{align}
	 to the gauge zero mode $G_\chi$, with $\chi\in H^1$. Using that $T^{\rm trans}_t (\chi \psit)=e^{ i \g_t}T^{\rm trans}_t (\chi) \psit$, we obtain \eqref{UG}.
	In opposite direction, we have, $G_\chi=i \int^\oplus_{\Om^*} G_{\chi_k} dk$ for $\chi\in H^1$.
Furthermore, since $\Lt G_\chi=0,\ T^{\rm trans}_t \chi_k= e^{i k\cdot t} \chi_k$, we have that \eqref{Lk-Gk} - \eqref{Trefl-rho-Gk} follow.  
\end{proof}

\begin{proposition}[The fiberization of the gauge orthogonality]  \label{prop:gauge-orth-cond-expl} 
\begin{equation}\label{Gk-orthog-expl} v \perp  {G}_\chi  \quad \forall \chi\in H^1  \quad \Longleftrightarrow  \quad \G^* v_k\equiv - i \bar\psit \xi_k +  i \psit \eta_k - \divv \al_k =0 \quad  \forall \ k\in \Om^*, \end{equation}
where  $v_k= ( \xi_k, \eta_k, \al_k)$ is the BFZ transform of $v$. 
\end{proposition}
\begin{proof}
The orthogonal projection operator, $P_g$, onto the space spanned by the gauge zero modes is given explicitly  as 
\begin{equation}\label{gauge-modes-proj} P_g = \G  h^{-1}\G^* ,\ \text{   where   }\ \G \gamma': =G_{\gamma'}\ \text{   and  }\ \G^* v=- i \bar\psit \xi +  i \psit \eta - \divv \al, 
\end{equation}
 for $v=(\xi, \eta, \al)$, and $h:=- \Delta + 2 |\psit|^2$  (see the proof of Proposition \ref{prop:decomp-gauge}).  	The latter operators 
 satisfy $\G^* \G=h $ and  $ \G G_{\chi}=G_{h \chi}$ and therefore $P_0 G_{\chi}=G_{\chi}$.  Using that $t^{\rm trans}_t (\chi \psit)=e^{ i \g_t}t^{\rm trans}_t (\chi) \psit$, we obtain 
\begin{align} \label{GamU}
		\G^* U^* = U^*  \G^* ,\  U  \G^* =\G^* U.	\end{align}

Eq. \eqref{gauge-modes-proj} implies that the condition $P_g v =0$ is equivalent to the condition $\G^* v =0$. Moreover, by \eqref{GamU}, this relation can be used fiberwise, \eqref{Gk-orthog-expl}. 
 \end{proof}

\subsection{Spectral properties of $L_{ k}$ and their consequences} 
\label{sec:Lk-spec}

We now turn to analysis of the spectrum of the fibre operators $\Lk$.  We define the perturbation parameter
  \begin{equation}\label{eps} 
	\e = \sqrt{\frac{\kappa^2 - b}{\kappa^2[(2\kappa^2 - 1)\beta(\tau) + 1]}}.
\end{equation}
  The term $(2\kappa^2 - 1)\beta(\tau) + 1$ in the denominator of \eqref{eps} is necessary in order to have a positive expression under the square root and to regulate the size of the perturbation domain.

 Let $H^1_k$ be the space of $H^1_{\rm loc}(\R^2)$ functions satisfying $T^{\rm trans}_t \chi_k= e^{i k\cdot t} \chi_k$. 
 By Proposition \ref{prop:Gk} below,   $0$ is an eigenvalue of  $\Lk$  of infinite multiplicity with the eigenspace $\cV_k  :=\{G_{\chi_k}: \chi_k\in H^1_k\}$. 
 (Recall that $G_{\chi_k}, \chi_k\in H^1_k$ appear as fibers of the  gauge zero modes, $G_\chi$ of $\Lt$.)   
We will see below that this zero eigenspace is eliminated as the solution we are seeking is orthogonal to $\int^\oplus_{\Om^*} \cV_k \hat{d k}$. Hence we consider $\Lk$ on the orthogonal complement (in the Sobolev space $\cH^2$ based on $\cH$, see below) of the space, $\cV_k $, 
 Thus we introduce $\Lk^\perp:= \Lk\big|_{\cV_k^\perp}$.

 The  following proposition is the main result of this section. 
 \begin{proposition} \label{prop:Lk-spec'} The operator $\Lk^\perp$ is self-adjoint and has the following properties
 
 
 (A)  
$\Lk^\perp$   has purely discrete spectrum, with all, but three, eigenvalues 
 are  $\gs 1$;
 
 (B)  
 $\Lk^\perp$ has three eigenvalues, $\nu_{k}^i\equiv \nu_{k}^i(\tau), i=1, 2, 3,$ of order $o(1)$; these eigenvalues are simple and of the form 
\begin{align} 
 \label{nuk1-form}&\nu_{k}^1  = c_1 (\tau) \frac{\e ^2|k|^{2}}{\e ^2 + |k|^{2}} [(\kappa^2-\frac12) \g_k (\tau)+ \eta(\tau, \kappa)\e^2] + O(\e^4 |k|^{2}),\\ 
 \label{nuk2-form}&\nu_k^2 = c_2 (\tau)  |k|^2+ O(\e^2|k|^2),\ \text{ for } |k|\ll 1,\  \text{ and }\ \nu_k^2  \gs |k|^2,\\ 
\label{nuk12-est}	
& |\p^\al \nu_{k }^j| \ls |k|^{2-|\al|},\ j= 1, 2,\ |\al|\le 2,\\
	\label{nuk3-form}&\nu_k^3 \ge c_3 (\tau) \e^2 ,\end{align} 
 where $c_i (\tau)\gs 1$,  and  $\g_{k}(\tau)$ and $\eta(\tau, \kappa)$ are given in \eqref{gamk-LAT} and \eqref{gamk-tau} and Lemma \ref{lem:tilde-gam12-expan}, respectively; 

  
(C)   The eigenfunctions  $v_{k }^j$   of $\Lk$,   corresponding to the eigenvalues $\nu_{k}^j, j=1, 2, 3$, satisfy   
 \begin{align}
   \label{vkj-est} & v_{k  }^j(x)\in C^\infty(\R^2 \times \R^2, \C^4),\ \forall j.   \end{align} \end{proposition}


This proposition is proven in Appendix  \ref{sec:prop-Lk-spec-pf}. 
 Propositions \ref{prop:bloch} and \ref{prop:Lk-spec'} imply Theorem \ref{thm:energyband}.

Proposition \ref{prop:Lk-spec'} shows that the hessian $\Lo$ has exactly two gapless spectral branches (bands) $\nu_{k }^i, i=1, 2.$ 
 The gapless band eigenfunctions, $v_{k }^1$ and $v_{k }^2$, originate from the longitudinal and transverse translational zero modes, $S_{\hat k^\perp}$ and $S_{\hat k }$, and are due to breaking the translational symmetry. 

The proposition above implies inequalities on the quadratic form of $\Lk$, which 
 play an important role in our analysis below. We begin with defining  the Sobolev space of order $1$ defined by the covariant derivatives, i.e., $ H^1_{\textrm{cov}}:=\{v\in L^2(\R^2, \C\times \C\times \C^2)\ |\ \|v\|_{H^1}< \infty \}$, where the norm $\|v\|_{H^1}$ is determined by
 the covariant inner product
\begin{align*}
    \langle v, v' \rangle_{H^1} = \Re \int \bar{\xi}\xi' &+  \bar{\eta}\eta' + \overline{\nat\xi} \cdot \nat\xi' + \overline{\nat\eta} \cdot \nat\eta' + \bar\alpha \cdot \alpha' + \sum_{i=1}^2 \nabla\bar \alpha_i\cdot\nabla\alpha'_i, 
\end{align*}
where $v=( \xi, \eta, \al)$ and $v'=( \xi', \eta', \al')$ (cf. Subsection \ref{sec:pert}).

 Let  $P_g$ be the orthogonal projection operator onto the space, $\cV_g:=\{G_\chi: \chi\in H^1\}$, spanned by the gauge zero modes, given explicitly in \eqref{gauge-modes-proj} of Appendix \ref{sec:fiberiz}, and $\bar P_g :=\one  - P_g$  the orthogonal projection onto  the orthogonal complement, $\cV_g^\perp$, of the subspace spanned by the gauge zero modes. We define the orthogonal projections
 \begin{equation} \label{P'P''} 
 P':= U^{-1}\int^\oplus_{\Om^*} P_k' \hat{dk}U, \quad 
  P'' :=\bar P_g - P',\end{equation}    
 where $P_k'$ is the orthogonal projection onto the span of the eigenspaces of the operator $\Lk$, corresponding to the eigenvalues $\nu_{k }^i, i=1, 2, $ described in Proposition \ref{prop:Lk-spec}.   These projections form  the partition of unity $P'+P''=\bar P_g$. 
 Proposition \ref{prop:Lk-spec} implies 
\begin{equation}	\label{L-L2-lower-bound}
	\lan v, \Lo v \ran_{L^2} \ge \mu_\tau \|v\|_{L^2}^2 \gs \e^2 \|v \|_{L^2}^2,\  \quad   \forall v\in \Ran P''.
	\end{equation}
where $\mu_\tau:=\inf_k \nu_{k}^3(\tau) \gs \e^2$. The bound  \eqref{L-L2-lower-bound} implies
\begin{align}\label{L-lowbnd2}  \| \Lt  v \|^2\ge \mu_\tau \lan v,  \Lt v \ran \gs \e^2  \|v \|_{L^2}^2,\  \forall v\in \Ran P''.\end{align} 
Indeed, ones observes that  \eqref{L-L2-lower-bound} gives $\mu_\tau\|  v \|^2 \lan  v,  \Lo v  \ran\le  \lan v,  \Lo v  \ran^2\le \|  v \|^2 \| \Lo  v \|^2 $, which implies \eqref{L-lowbnd2} (the second inequality follows from \eqref{L-L2-lower-bound}).

  The lower bound \eqref{L-L2-lower-bound} is upgraded  to the one  involving the $H^1_{\textnormal{cov}}$ norm on the r.h.s., which together with the standard elliptic upper bound, gives
 \begin{corollary}\label{lem:H1-bound'}
        \begin{align}        \label{L-H1-lowbnd}
&   \|v \|_{H^1}^2 \gs 	\lan v, \Lo v \ran_{L^2} \gs 
\e^2 \|v \|_{H^1}^2,\  \forall v\in \Ran P'',\\ 
        \label{H2-L-est}
& \|v\|_{H^2}\ls \| \Lo v \|_{L^2},\ 
\quad 
\forall v\in \Ran P''. \end{align} 
\end{corollary}
\begin{proof}
To begin with, using the explicit expression for $\Lo$ in \eqref{Lc-expl}, integration by parts and Schwartz and Sobolev inequalities, 
 we show that  
       \begin{equation}    \label{H1-lower-bound}
	\langle v, \Lo v \rangle_{L^2} \geq  \|v\|_{H^1}^2 - C \|v\|_{L^2}^2,
    \end{equation}
    for some positive constant $C$.

    Now let $\delta \in [0,1]$ be arbitrary and denote $\mu\equiv \mu_\tau$.     We combine \eqref{H1-lower-bound} with the bound
   \eqref{L-L2-lower-bound}, to obtain 
      \begin{align*}
        \lan v, \Lo v \ran_{L^2}
        &= (1 - \delta) \lan v, \Lo v \ran_{L^2} + \delta \lan v, \Lo v \ran_{L^2} \\
        &\geq (1 - \delta) \mu \|v\|_{L^2}^2 + \delta (\|v\|_{H^1}^2 - C \|v\|_{L^2}^2) \\
        &= ((1 - \delta) \mu - \delta C ) \|v\|_{L^2}^2 + \delta   \|v\|_{H^1}^2.
    \end{align*}
    \eqref{L-H1-lowbnd} 
    now follows by choosing $\delta = \frac{\mu }{1 +\mu + C}$ and using that $\mu \gs \e^2$.

 Next, 
 the bound \eqref{H2-L-est} follows from $\| \Lo v \|_{L^2} \ge \del \| \Lo v \|_{L^2}+ (1-\del) \mu \| v \|_{L^2} \ge \del \| \n v \|_{L^2}-  \del C \|  v \|_{L^2} + (1-\del) \mu \| v\|_{L^2}$. 
\end{proof}

Now we consider $\Ran P'$. Let $\nu_k^j$ be the eigenvalues, described in Proposition \ref{prop:Lk-spec'}, which are $\LAT^*-$periodic functions on $\R^2$. We define the operator of multiplication  $A$ on $L^p(\Om^*, \C^2)$, given by 
\begin{align}   \label{A}(A f) (k) := (\nu_k^1 f_1 (k), \nu_k^2 f_2 (k)),\end{align}
 where $f=(f_1, f_2) $. Furthermore, recall $\rf f(k)= f(-k)$. The following lemma describes the spectral decomposition map for the operator $\Lt P'$, 
\begin{lemma}\label{lem:V} There are maps $V: L^2(\Om^*, \C^2)\ra L^2(\R^2;\C\times \C\times \C^2)$ and $V^{*}: L^2(\R^2;\C\times \C\times \C^2)\ra L^2(\Om^*, \C^2)$, s.t. (a)  $V^{*} V=\one$, (b)  $V V^{*}=P'$, (c)  $V^*  \Lt =A V^{*} $,  (d)  $\|V f\|_{H^1}\ls \|f\|_{L^2}$.  Moreover, $V$ and $V^*$ satisfy  
\begin{align}	\label{V-par} V^* T^{\rm refl} v = - \rf  V^* v,\  T^{\rm refl} V f =-V \rf  f,  \end{align}
 \begin{align}   \label{V-bnd5} 
&\|V f\|_{L^\infty}\ls  \| f\|_{L^1},\  \|x V f\|_{L^\infty}\ls   \| f\|_{L^1}+\|\n f\|_{L^1}.\end{align} 
 \end{lemma} 
Note that (i) \eqref{V-par} implies that, if $v$ is even, as in our case, see \eqref{v-parity},  then the function $f:=V^* v$ must be odd and 
 (ii) $V^{*}$ gives the spectral decomposition of the operator $ \Lo P' $, namely, $\Lo P' =V A V^{*}$.
 
   The maps $V$ and $V^{*}$ give the isomorphism 
 between the spaces $\Ran P$ 
and $L^2(\Om^*, \C^2)$.
We define these maps explicitly. 
For the eigenfunctions $v_{k }^j$  of $\Lk$,  corresponding to the eigenvalues $\nu_{k}^j, j=1, 2, 3$, described in Proposition \ref{prop:Lk-spec'}, we define $\uvk:=(v_{k }^1, v_{k }^2)$. We extend our standard operators to act on $\uvk$. The facts that $T^{\rm refl}\Lk = L_{-k} T^{\rm refl}$  and that the eigenvalues $v_{k }^j$ are simple
 and equation 
 \eqref{vkj-est} imply that these functions have the following properties
\begin{equation} \label{hwk-parity} 
  T^{\rm refl} \hw_{k }=  \hw_{-k }\ \text{ and }\  \hw_{k}\in H^m(\Om^* \times \Om, \C^2),\ \forall m.  
 \end{equation}	

 Now, we define the maps $V: L^2(\Om^*, \C^2)\ra L^2(\R^2;\C\times \C\times \C^2)$ and $V^{*}: L^2(\R^2;\C\times \C\times \C^2)\ra L^2(\Om^*, \C^2)$ by 
 \[V^{*}: v\ra \int_{\Omega^*}^\oplus (\lan v_{k 1}, (U v)_k\ran, \lan v_{k 2}, (U v)_k\ran) \hat{dk} \text{ and  } V: f \ra U^{-1}\int_{\Om^*}^\oplus \hat{dk}  f_k \cdot \hw_k.\]  
\begin{proof}[Proof of Lemma \ref{lem:V}] The properties  (a) - (c)  of the maps $V$ and $V^*$ stated in Lemma \ref{lem:V} follow from the definition (in particular, to see that $V$ is an isometry it is useful to represent it as  $V=U^{-1} W$, where $W f:=\int_{\Om^*}^\oplus \hat{dk}  f_k \hw_k$, and use the unitarity of $U$ and $W$). 

Recall $\tilde T^{\rm gauge}_\g := e^{ i\g} \oplus e^{- i\g} \oplus \one$. To show 
\eqref{V-par}, we use that, by the definition of $V$, we have the formula
\begin{align}\label{V} 
(V f) (x+t)= \int_{\Om^*} dk   \chi_k(t) f_k  \tilde T^{\rm gauge}_{\g_t} \hw_k(x),\ x \in \Om.\end{align}
Since $\g_{-t}(-x)= \g_{t}(x)$ and $T^{\rm refl} \hw_{k }= \hw_{- k}$, we have 
\[(T^{\rm refl} V f )(x+t)= \int_{\Om^*} \hat{dk}    \chi_k(-t) f_k   \tilde T^{\rm gauge}_{\g_t} T^{\rm refl} \hw_{k } \]\[=- \int_{\Om^*} \hat{dk}    \chi_{-k}(t) f_k   \tilde T^{\rm gauge}_{\g_t} \hw_{-k} = -\int_{\Om^*} \hat{dk}    \chi_{k}(t) f_{-k}   \tilde T^{\rm gauge}_{\g_t} \hw_{k} .\] This gives the second relation in \eqref{V-par}. The first relation in \eqref{V-par} is proven similarly. 


Before proceeding to the remaining statements in Lemma \ref{lem:V}, we prove the following  properties of the maps $V$ and $V^*$: 
\begin{align} \label{xV-Vk-deriv} 
 x_j V = V i \n'_j +  V_j,\   \qquad V^* x_j = i \n'_j V^* -   V^*_j, \end{align}
where $V_j$ and $V_j^*$ are adjoint operators satisfying  $\| V_j f\|_{L^2} \ls  \|f \|_{L^2}$ and $ \| V_j^* v\|_{L^2} \ls  \|v \|_{L^2}$, 
 By \eqref{k-deriv-Uv}, we have $ x V f = U^{-1}(i \n_k  + x )\int_{\Om^*}^\oplus \hat{dk}  f_k  \hw_k =\U^{-1}\int_{\Om^*}^\oplus \hat{dk}[\n_k f_k  \hw_k+   f_k (i \n_k  + x )\hw_k]$. This gives the first relation in \eqref{xV-Vk-deriv}, with  
$ V_j$ and $V_j^*$ adjoint operators  given by 
\begin{align}\notag  V_j f:= U^{-1}\int_{\Om^*}^\oplus \hat{dk}  f_k  (i \p_{k_j}  + x_j ) \hw_k,\ \text{ and }\ 
 V_j^* v:=\int_{\Om^*}^\oplus \hat{dk} \lan  (i \p_{k_j}  + x_j )\hw_k,   (U v)_k\ran _\Om.\end{align}
    The second relation is adjoint of the first. It can be also derived independently, similarly to the first one. 
The boundedness of $V_j$ and $V_j^*$  follows from the second relation in \eqref{hwk-parity}.

Next, using the representation \eqref{V} and the estimates on $\hw_k$ in \eqref{hwk-parity} 
 and writing $\sup_{x\in \R^2}=\max_{t\in \LAT} \sup_{x\in \Om+t}$, we find the first estimate in \eqref{V-bnd5}. 
  Using  the first estimate in \eqref{V-bnd5} and and  the second relation in \eqref{xV-Vk-deriv}, we find the second estimate in \eqref{V-bnd5}. 

Proceeding similarly as in the previous paragraph and 
 and using the second relation in \eqref{k-deriv-Uv} and  the fact that  the Sobolev space $H^1$ is defined in terms of the covariant derivatives, $D$, we obtain
\begin{align}   \label{V-bnd6} 
&\|D^m  V f\|_{L^\infty}\ls   \| f\|_{L^2}, 
\end{align}  
where, recall that $D$ is either $\n_{\at} \oplus \one$ or $\one \oplus \n$. This inequality implies 
 the property  (d).

Finally, we prove the estimates, used in the estimates of the nonlinearity in Appendix \ref{sec:nonlin}: 
\begin{align}	\label{V*est}  
 \left\|V^* v\right\|_{L^\infty}\ls \| U v\|_{H^{-1}_{x}L^\infty_{k}},\end{align} 
where the notation $\| w\|_{X_{k}Y_{x}}$, for a function $w_k(x)$, stands for the norm $\| \| w_k\|_{Y_{x}}\|_{X_{k}}$, and 
 \begin{align}   \label{V-bnd-L2} 
 &\|D^m  V f\|_{L^2}\ls   \| f\|_{L^2},\ 
\|x V f\|_{L^2}\ls   \| f\|_{L^2}+\|\n f\|_{L^2},\end{align}

\eqref{V*est} follows from the definition of $V^*$ given above and the properties \eqref{hwk-parity} of $\hw_k$. Indeed, we have  \begin{align*}\left\|V^* v\right\|_{L^\infty}&\le \sup_{k\in\Omega^*} |\lan \hw_k,  (U v)_k\ran| \\ & \le \sup_{k\in\Omega^*} \|\hw_k\|_{H^{1}}  \sup_{k\in\Omega^*} \|(U v)_k\|_{H^{-1}}  \ls \| U v\|_{H^{-1}_{x}L^\infty_{k}}.\end{align*}
The inequalities in
\eqref{V-bnd-L2} follow from the first relation in \eqref{xV-Vk-deriv}, the definition of $V$ and \eqref{norm-rel2}. \end{proof} 


\section{Asymptotic stability: Proof of Theorem \ref{thm:stability}}\label{sec:Pf-StabThm} 

\subsection{Reparametrization of solutions} 
\label{sec:deco}

Our goal in this subsection is to reparametrize a neighbourhood of the equivalence class $\mathcal{M}_\tau = \{ T^{\rm gauge}_\chi  \ut : \chi \in H^1 (\R^2, \R) \}$  of the Abrikosov $\cL-$lattice solution,  $\ut= (\psit, \bar\psit, \at)$. The tangent space, $T_{\ut}\cM$ at $\ut  \in \mathcal{M}_\tau$ is spanned by the 
gauge zero modes $G_{\chi}$, 
given in \eqref{gauge-modes}.

Below we consider orthogonal complements the tangent spaces, $T_{\ut}\cM$. 
 We define, for $\delta > 0$, its tubular neighbourhood, 
\begin{equation} \label{tubular-nbhd}
    U_\delta = \{ T^{\rm gauge}_\chi  (\ut + v) : \chi \in  H^1 (\R^2, \R),\  v \in H^1_{\textrm{cov}},\  \|v\|_{H^1} < \delta \}
\end{equation}
and prove the following decomposition for $u$ close to the manifold $\mathcal{M}$. 
Recall the notation $u= (\psi, \bar\psi, a)$. 
\begin{proposition}\label{prop:decomp-gauge}
	There exist $\delta_* > 0$ (depending on $\e$) 
	and a map $\eta : U_{\delta_*} \to H^2(\R^2, \R)$ such that $v:=(T^{\rm gauge}_{\eta(u)})^{-1} u - \ut \perp G_{\chi}$ for all $\chi\in H^2(\R^2, \R)$. Moreover, if $u$ satisfies $T^{\rm refl} u=u$, then $\eta(u)$ satisfies $\rf \eta(u)  = \eta(u)$, where $\rf: \chi(x) \ra \chi(-x)$, and therefore $v:=(T^{\rm gauge}_{\eta(u)})^{-1} u - \ut \perp G_{\chi}$ satisfies $T^{\rm refl} v= v$.
\end{proposition}
\begin{proof}
We omit the superindex $"gauge"$ in $T^{\rm gauge}_\gamma $.   Our goal is to  solve the equation   $\lan G_{\chi},  u - T_{-\g}\ut \ran_{L^2}=0,\ \forall \chi \in  H^1 (\R^2, \R),$ for $\g$ in terms of $u$.   Define the affine space $X = \ut + H^1_{\textrm{cov}}$ and let $\G \chi: =G_{\chi}$. Then,  $\lan G_{\chi},  u - T_{-\g}\ut \ran_{L^2}=\lan \chi,   \G^*  (u - T_{-\g}\ut) \ran_{L^2}=0,\ \forall \chi \in  H^1 (\R^2, \R),$ is equivalent to $\G^*  T_\g^{-1} (u - T_{\g} \ut) =0$. Hence our problem can be reformulated as solving the equation $f(\g, u) =0$ for $\g$ in terms of $u$, where
 the map $f$     is given by 
  \begin{equation*}
        f(\g, u) =  \G^*  T_\g^* (u - T_{\g} \ut) .
    \end{equation*}
Here we used that $T_\g^{-1} =T_\g^{*} $. To solve this equation, we use the Implicit Function Theorem. 
From the definition, it is clear that  $f : H^2(\R^2, \R) \times X \to L^2(\R^2, \R)$,  is a $C^1$ map and  $f(\g, T_{\g} \ut) = 0$.  
       Finally, we   calculate the linearized map:
$      d_\g f(\g, u) |_{u=T_{\g} \ut}  \chi= - \G^*  T_\g^*  d_\g   T_{\g} \ut  \chi$
and $d_\g   T_{\g} \ut    \chi= T_\g g_{\dot\gamma }\ut \ =T_\g \G\chi ,  $  where $ \chi\in H^1(\R^2, \R)$, which gives $      d_\g f(\g, u) |_{u=T_{\g} \ut}  \chi= - \G^* \G  \chi.$     We compute $\G^* v=- i \bar\psit \xi +  i \psit \eta - \divv \al$ for $v=(\xi, \eta, \al)$. Hence, we have $\G^* \G=- \Delta + 2 |\psit|^2$. The last two relations give   \begin{equation*} d_\g f(\g, u) \big |_{u=T_{\g} \ut}=  \Delta - 2 |\psit|^2.  \end{equation*} 
For $|\psit|^2$ periodic, $-\Delta + 2 |\psit|^2$ is self - adjoint and, as easy to see using uncertainty principle near zeros of $ |\psit|^2$, is strictly positive, $-\Delta + 2 |\psit|^2\ge \del>0$, with $\del$ depending on $\e$. Therefore it is invertible.
    The Implicit Function Theorem then gives us a neighbourhood $V$ of $\ut$ in $X$ and a neighbourhood $W$ of $0$ in $H^2(\R^2, \R)$ and a map $H : V \to W$ such that $f(\g, u) = 0$ for $(\g, u) \in W \times V$ if and only if $\g = H(u)$. We can always assume that $V$ is a ball of radius $\delta_0$.

    We can now define the map $\eta$ on $U_\delta$ for $\delta < \delta_0$ as follows. Given $u \in U_\delta$, choose $\g$ such that $u = T_\g(\ut + v)$ with $\ut + v \in V$. We define $\eta(u) =   \g H(\ut + v)$.
To show that $\eta$ is well defined, we first show that if $g$ is sufficiently close to the identity, then $H(T_\g u) = \g H(u) $. To begin with, we note for all $\g$, $T_\g (V) \subset V$. One can easily verify, by the definition of $f$, that $f(\g \del, T_\g u) = f(\del, u)$. Indeed, we have $f(\g h, T_\g u) = \G^*  T_{\g \del}^* (T_\g u - T_{\g \del} \ut)  = \G^*  T_{\del}^* ( u - T_{\del} \ut) = f(\del, u)$.   Hence $f(\g H(u), T_\g u) = f(H(u), u) = 0$, and therefore by the uniqueness of $H$, it suffices to show that $\g H(u)  \in W$, but this can easily be done by taking $V$ to be smaller if necessary.

    Suppose now that we have also $u = T_{\g'}(\ut + v')$. Then $T_{\g^{-1}\g'}(\ut + v') = \ut + v$. Therefore, by the relation $H(T_\g u) = \g H(u) $, we  have
    \begin{equation*}
       \g H(\ut + v)   =\g H(T_{{\g}^{-1} \g'}(\ut + v'))=\g {\g}^{-1}\g' H(\ut + v')
        ={\g'}  H(\ut + v'),
    \end{equation*}
    so $\eta$ is well-defined. Finally, since $ T^{\rm refl} \G = \G r$ and $r\G^* = \G^*  T^{\rm refl} $, we have that the function $f(\g, u)$ satisfies $f(r\g, T^{\rm refl} u)=f(\g, u)$, which implies the second statement follows and the proof is complete.
\end{proof}

\begin{remark} \label{rem:gaug-fix}  Going from $u$ to $u':=(T^{gauge}_{\eta(u)})^{-1} u$, with $\eta(u)$ s.t. $u' - \ut \perp G_{\chi},\ \forall \chi$, fixes the gauge. Another gauge fixing would be choosing $\eta(u)$ so that $\divv a' =0$, where $a'=a-\n \eta(u)$. 
 One cannot do both as they are incompatible.
Both ways to fix the gauge coincide in the leading order in $\eps$. So in the leading order, one can take $\divv a =0$, which eliminates the null space of $\curl$. 
\end{remark}

\subsection{GES equations in the moving frame}\label{sec:mov-frame}

 Since, by the definition $\tilde \rho u=u$ and $\tilde \rho T_{\chi}^{\rm gauge}  \ut=T_{\chi}^{\rm gauge} \ut$, we have that $v$ in Proposition \ref{prop:decomp-gauge} satisfies $\tilde \rho v=v$.    With this in mind, we reformulate the result of Proposition \ref{prop:decomp-gauge} 
        as
\begin{align}	\label{u-deco-gauge} & u =T_{\g}^{\rm gauge} ( \ut +v), \quad \mbox{with}\,\  v\perp G_{\chi}\  \forall  \chi\in H^1, \,\   \mbox{and} \quad  \tilde \rho v=v,\  \mbox{and some}\  \gamma,\\ 
 \label{gam-parity} &\rf\gamma   = \gamma\  \text{ and }\  T^{\rm refl} v=v,\ \text{ if, in addition,  $u$ satisfies }\  T^{\rm refl} u=u.  
\end{align}
\begin{proposition}\label{prop:v-eq}
If 
$u=(\psi, \bar\psi, a)$ is a solution to the Gorkov-Eliashberg-Schmidt equation \eqref{u-eq} (which is equivalent to  \eqref{GESresc-c}), then 
$\g$ and $v$, defined in the equation \eqref{u-deco-gauge}, satisfy the equation
	\begin{align}	\label{v-eq}
		\partial_t v = - \Lo  v -V_{\s}    v- \Nt (v) - G_{\s},\  \quad \s:=\phi+\dot\gamma, 
		\end{align}	
 and $ \tilde \rho v=v$.	Here 
$\Lt=d J(\ut)=\E''(\ut)$ is the hessian defined in \eqref{Lt}, 
$\dot\gamma=\p_t \g$, 

\begin{equation} \label{Vgam} 
  V_{\chi} v:= (i\chi\xi, - i\chi\eta, 0),
 \end{equation}
for $v = (\xi, \eta, \al)$, and $\Nt (v)$ is the  nonlinearity,
\begin{align}	\label{N} \Nt(v)=J(\ut+v)- J(\ut)- \Lt v.\end{align}
 (The terms $\Lt$ and  $\Nt(v)$ 
     are given explicitly by expressions \eqref{Lc-expl} and   \eqref{N-expl} of Appendices \ref{sec:hess-expl} and \ref{sec:nonlin}.) 

In the opposite direction, if $\g$ and $v$ satisfy the equations \eqref{v-eq} and 
   $ \tilde \rho v=v$, then $u$, defined by \eqref{u-deco-gauge},  is a solution to \eqref{GESresc-c}. 
   
   In addition, $T^{\rm refl} u=u \longleftrightarrow T^{\rm refl} v = v,\  \rf\s = \s$. \end{proposition}
\begin{proof} Plugging $\tilde u :=(T_{\g}^{\rm gauge})^{-1} u=\ut + v$ into \eqref{Tgauge-u-eq}, using $\p_{t, \s}  (\ut + v)=\partial_t v + V_{\s}    v +G_{\s}$, with $\s:=\phi+\dot\g$, and expanding $J(\ut + v)$ in $v$ and using that $d J(\ut)=E''(\ut)=\Lt$, where $d J(u)=(d_{\bar\psi}J (u),\ d_{\psi}J (u),\  d_{a}J (u))$ (see Section \ref{sec:hessian} for the notation), 
 gives \eqref{v-eq}.

Proceeding in the reverse order, one can easily prove that $u$ satisfies the generalized version of the equation \eqref{u-eq}, i.e. with $J$ extended to functions of the form $(\psi, \theta, a)$, which becomes \eqref{u-eq} (which is, of course, equivalent to  \eqref{GESresc-c}), once we know that $ \tilde \rho u=u$.
\end{proof}
Eq \eqref{v-eq} is for the unknowns $v\  \mbox{and}\ \g$. For $\g=-\eta (u)$ and $v$, defined in Proposition \ref{prop:decomp-gauge}, this equation is supplemented by the conditions $ v\perp G_{\gamma'},$  $\forall  \gamma'$.  
Projecting it onto  the subspace spanned by the gauge zero modes (tangent vectors) and its orthogonal complement leads to two coupled equations for  $v\  \mbox{and}\ \g$. To derive these equations we need the following definitions. Recall that $\bar P_g :=\one  - P_g$ be  the orthogonal projection onto  the orthogonal complement, $\cV_g^\perp$, of the subspace, $\cV_g$, spanned by the gauge zero modes. 
Recall that the real-linear operator $\rho$ is defined in \eqref{rho}.  We have 
 \begin{proposition} \label{prop:v-sig-eqs}  If  $u(t) = (\psi(t), \bar\psi(t), a(t))$ satisfies  \eqref{u-eq}, then the functions $v$ and $ \s:=\phi+\dot\gamma$, defined by $u$ in \eqref{u-deco-gauge}, satisfy 
 the following equations
\begin{align}	\label{v-eq-proj}
		\partial_t v = - \Lt^\perp  v - \bar P_g (V_{\s}    v+  \Nt(v)), 
		\end{align}	
	 \begin{align}\label{sig-eq} 
(-   \Delta+ 2 |\psit|^2 &+ \bar\psit\xi+ \psit\eta ) \s  =- i \bar\psit  N_\xi(v) - i \psit  N_\eta(v)+\divv  N_\al(v),
\end{align}
 where $\Lt^\perp:=\Lt \bar P_g$,  $ \Nt (v)= (N_\xi(v),   N_\eta(v),  N_\al(v))$, and $\tilde\rho v_{}= v$. Moreover, if $T^{\rm refl} u=u$, then $v$ and $\s$ obey
\begin{equation} \label{v-parity-gam} 
T^{\rm refl} v=v\  
 \quad \text{ and }\  \quad  \rf \s=\s. 
 \end{equation} 
 
 Conversely,  given  $v$ and $\s$  solving the equations \eqref{v-eq-proj} and \eqref{sig-eq} and  satisfying $\tilde\rho v_{}= v$, the function $u$, defined by \eqref{u-deco-gauge}, satisfies \eqref{u-eq} (and $\tilde\rho u= u$). Moreover, if $v$ and $\s$ obey  \eqref{v-parity-gam}, then  $T^{\rm refl} u=u$.\end{proposition}
\begin{proof}
By Proposition \ref{prop:v-eq}, Eq \eqref{u-eq}  is equivalent to Eqs  \eqref{v-eq} and $\tilde \rho v_{}= v$. Projecting the equation \eqref{v-eq} onto  the orthogonal complement of the subspace spanned by the gauge zero modes (tangent vectors) 
gives \eqref{v-eq-proj}.

Now, we project \eqref{v-eq} onto  the tangent vectors, $G_{\chi}$  $\forall  \chi$ (see \eqref{gauge-modes}).
Multiplying \eqref{v-eq} scalarly by $G_{\chi} = (i\chi\psit, - i\chi\psit, \nabla\chi)$ and using $\lan G_{\chi}, v\ran=0,\ \Lt G_{\chi}=0$ and that $\Lt$ is a symmetric operator,  we find
\begin{align}\label{projGchi}\lan G_{\chi}, V_{\s} v+\Nt(v)+ G_{\s}\ran=0. \end{align}
 Remembering the definitions \eqref{Vgam} and  \eqref{N},  
and using that  $\lan G_{\chi}, V_{\g'} v \ran=
\lan \chi, (\bar\psit\xi+ \psit\eta) 
\g'\ran$, $\lan G_{\chi},  \Nt (v)\ran=\lan \chi, i \bar\psit  N_\xi(v)+ i \psit  N_\eta(v) 
-\divv N_\al(v)\ran$ and 
\begin{align}\label{GG-ip}\lan G_{\chi}, G_{\g}\ran=\lan \chi, (-\Delta +2 |\psit|^2) \g\ran,\end{align} 
 we see that \eqref{projGchi} can be rewritten in the form $\lan \chi, f\ran=0$, 
where  $f:=(-   \Delta+ 2 |\psit|^2 + 
\bar\psit\xi+ \psit\eta) \s + i \bar\psit  N_\xi(v)+ i \psit  N_\eta(v)-\divv  N_\al(v)$, which,  since $\chi\in H^2$ is arbitrary and $\divv \at=0$, 
implies the equation \eqref{sig-eq} (see also \eqref{gauge-modes-proj} of Appendix \ref{sec:fiberiz}).

By the last part of Proposition \ref{prop:v-eq}, $T^{\rm refl} u=u$ iff  $v$ and $\s$ obey \eqref{v-parity-gam}. 

In the opposite direction,  since $J(\tilde \rho u)=\tilde \rho J(u)$ (see \eqref{refl-tilderho-sym} and \eqref{rho}), we see that \eqref{v-eq} and therefore also \eqref{v-eq-proj} and \eqref{sig-eq} are 
 invariant under the transformation $v\ra \tilde \rho v$. Hence, it follows that,  given  $v$ and $\s$  solving the equations \eqref{v-eq-proj} and \eqref{sig-eq}, the function $u$, defined by \eqref{u-deco-gauge}, satisfies \eqref{u-eq}. Moreover, by \eqref{u-deco-gauge}, if $v$ and $\s$ satisfy \eqref{v-parity-gam}, then  $u$, defined by \eqref{u-deco-gauge}, satisfies $T^{\rm refl} u=u$. \end{proof}

\DETAILS{In conclusion of this subsection, we mention that the orthogonal projection operator, $P_0$, onto the space spanned by the gauge zero modes is given explicitly  as 
\begin{equation}\label{gauge-modes-proj} P_0 = \G  h^{-1}\G^* ,\ \text{   where   }\ \G \gamma': =G_{\gamma'}\ \text{   and  }\ \G^* v=- i \bar\psit \xi +  i \psit \eta - \divv \al, 
\end{equation}
 for $v=(\xi, \eta, \al)$, and $h:=- \Delta + 2 |\psit|^2$  (see the proof of Proposition \ref{prop:decomp-gauge}).  	The latter operators 
 satisfy $\G^* \G=h $ and  $ \G G_{\chi}=G_{h \chi}$ and therefore $P_0 G_{\chi}=G_{\chi}$. }

%
\DETAILS{ Note that $G_{\gamma'} =  g_{\gamma'} u_\om$, where  $\gamma'$ are tangent vectors to the group $ H^2(\R^2;\R)$ and $g_{\gamma'}$, are the infinitesimal gauge transformation (i.e. the G\^{a}teaux derivative  of the transformation \eqref{gauge-sym} in $\g$), defined, for  $u= (\psi, \bar\psi, a)$, as 
\begin{equation} \label{ggam'} 
  g_{\gamma'} u:= (i\gamma'\Psi, \nabla\gamma').
 \end{equation}}
 %


\subsection{Asymptotic stability}\label{sec:as-stab}

In this subsection, 
assume that  $\e$ is sufficiently small and $\tau$ and $\kappa$ are such that
 \begin{align}	\label{posit-cond}
 \g_{k} (\tau) >0 \qquad \forall k\in (\R^2/\LAT^*)/\{0\},\ \kappa> 1/\sqrt 2,\ \eta (\tau, \kappa)>0. 
 \end{align}
These inequalities and \eqref{nuk1-form} -  \eqref{nuk2-form}  imply that the first two eigenvalues of $\Lk$ are positive for $k\ne 0$ and satisfy the estimate 
 \begin{align} \label{nuk12-lwb}\nu_{k}^i\equiv \nu_{k}^i(\tau) \gs |k|^2, i=1, 2. \end{align}
The remaining eigenvalues are always $ \gs \e^2$.

\DETAILS{Multiplying the third equation in \eqref{GESresc-c} by $\sigma^{-1}$
and taking $\divv$ of the result, we obtain
\begin{equation}\label{Phi-eq}
\Delta \phi=
\divv\sigma^{-1}
[\im(\overline{\psi}\nabla_a \psi)- \CURL^2 a].
\end{equation}
This gives an equation for
$\phi$, which can be easily solved. The solution is determined up to a \textit{harmonic function} on $\R^2$. We \textit{fix the solution} (up to a constant) by demanding that $\phi$ is bounded.
In what follows we always
assume that $\phi$ is a bounded solution of the equation \eqref{Phi-eq} and, in
particular, is a function of $\psi$ and $a$, and we do not list it
among unknowns and use the notation $u=(\psi, \bar\psi, a)$.}

 We  consider the Gorkov-Eliashberg-Schmidt equations \eqref{GESresc-c} (or \eqref{v-eq}) 
 with an initial condition 
  $(\psi_0, \bar\psi_0, a_0)\in U_{\delta_0}$, with $\del_0<\del_*$ ($U_{\delta}$ is defined in \eqref{tubular-nbhd} and $\del_*$ is given in Proposition \ref{prop:decomp-gauge}), 
 satisfying 
 \begin{equation} \label{parity0-resc}
(\psi_0 (x), \bar\psi_0 (x), a_0 (x))=(\psi_0 (- x), \bar\psi_0 (- x), - a_0 (- x)).	 	
\end{equation} Then, by the local existence there $T>0$ s. t.  \eqref{GESresc-c} has a solution, 
$(\psi(t), \phi(t), a(t)) \in C^1([0,T]; U_\delta)$ for some $\delta < \delta_*$. Moreover,  by the uniqueness, this solution 
satisfies 
\begin{equation} \label{parity}
(\psi (t, x), \bar\psi (t, x), a (t, x))=(\psi (t, - x), \bar\psi (t, - x), - a (t, - x)).	 		 
\end{equation}

 Since $u(t) = (\psi(t), \bar\psi(t), a(t)) \in C^1([0,T]; U_\delta)$ for some $\delta < \delta_*$, by  
   Propositions \ref{prop:v-eq} and \ref{prop:v-sig-eqs}, 
if  $u(t)$ satisfies   the Gorkov-Eliashberg-Schmidt equations \eqref{u-eq}, then  
 the functions $v$ and $\s$, defined by $u$ in \eqref{u-deco-gauge} and by $\s:=\dot\g+\phi$, satisfy the equations \eqref{v-eq-proj}, \eqref{sig-eq} and  \eqref{v-parity-gam}.  

It turns out that the equations \eqref{v-eq-proj} and \eqref{sig-eq} for $v$ and $\s$ are more convenient for  analysis than \eqref{GESresc-c} and we concentrate on them. 
 We supplement the latter equations with the initial condition $v_0$. 
Since, by the uniqueness,  the Abrikosov lattice solutions $\ut = (\psit, \bar\psit, \at)$ satisfy $T^{\rm refl} \ut=\ut$ and, trivially, $T^{\rm refl} u=u$, where $T^{\rm refl}$ is defined in \eqref{Trefl}, 
 $v_0$ 
  must satisfy
 \begin{equation} \label{v0-parity}
	T^{\rm refl} v_0=v_0,\ \quad 
	\text{ and }\  \quad  \rho v_{0}= v_{0}.	 
\end{equation} 
By the reflection invariance of the equation \eqref{v-eq},  Proposition \ref{prop:decomp-gauge}, 
the invariance of \eqref{v-eq-proj} under the transformation $v\ra \rho v$, and the condition \eqref{v0-parity}, the solutions $v$ and $\s$ to  the equations \eqref{v-eq-proj} and \eqref{sig-eq} satisfy
\begin{equation} \label{v-parity}
T^{\rm refl} v=v,\ \quad 
\rf\s =\s,\  \quad \text{ and }\  \quad  \rho v= v.	 
\end{equation}
(We can also appeal to Proposition \ref{prop:v-sig-eqs}.) 

Finally, by Proposition \ref{prop:decomp-gauge}, 
 we can assume from now on that $v$ belongs to the space
\begin{equation} \label{H1-perp}
	H^1_\perp:=\{v\in H^1_{\rm cov}: v\perp G_{\chi}\  \forall  \chi\}\equiv (\Ran P_g)^\perp.
\end{equation}

Recall that $P'$ and $P''$ are the projections defined in \eqref{P'P''} and satisfying $P'+P''=\bar P_g$.  Since by Proposition \ref{prop:v-eq}, $\bar P_g v=v$, we can write $v=v'+v''$, where $v':= P' v$ and $v''=P''   v$, 
 and split the  the equation \eqref{v-eq-proj} into the two equations
	\begin{align}	\label{v'-eq}
		\partial_t v' &= - \Lo'  v' +\tilde N' (v),\\ 
	\label{v''-eq}	\partial_t v'' &= - \Lo''  v'' +\tilde N''(v), 	\end{align}	
where $\Lo' $ and $\Lo''$ are the restrictions of the operator $\Lo$ to the subspaces $\Ran P'$ and $\Ran P''$ 
 and   
$\tilde N' (v):=P' \tilde N(v)$ and  $\tilde N'' (v):=P''  \tilde N (v)$, with   $\tilde N (v):=V_{\s} v + \Nt (v)$. (Note that $P' \bar P_g=P' $ and $P'' \bar P_g=P'' $, where, recall, $\bar P_g:=\one - P_g$.) 

Relation \eqref{L-L2-lower-bound} implies that the 
restriction of the operator $\Lo$ to the subspace $ \Ran P''$   has a gap $\gs \e^2$ in the spectrum, 
 and therefore $v''$ can be estimated in terms of $v'$ using the differential inequalities for appropriate Lyapunov functionals. The 
restriction of the operator $\Lo$ to the subspace $ \Ran P'$ is the multiplication operator by   $\nu_{k }^1\oplus \nu_{k }^2$, which is of a simple form and behave as $O(|k|^2)$ at $k=0$ and therefore the equation \eqref{v'-eq} can be handled directly by standard techniques.  Of course, outcoming estimates in one subspace are incoming into the other. 

Our first goal is to prove {\it a priori bounds} on $v'$ and $v''$. In what follows, we denote   \[<t>:=(1+t)^{1/2}\ \text{ and }\ 
\|f\|_{X \cap Y} := \| f\|_{X}+\| f\|_{Y}.\] 

%
\DETAILS{Clearly, the structure of  the nonlinearity $\tilde N (v)$ plays an important role in our analysis.  It is shown in Appendix \ref{sec:est-RN} that it can be  written as $\tilde N (v)= \tilde N_2 (v, v)+ \widehat N (v)$, where $N_2 (v, v)$ is bilinear contribution to $N_\om (v)$, while $\widehat N (v)$ collects 
 trilinear  and higher order terms and enjoys better estimates.  To concentrate on the essentials and keep notation from running amok, we 
present estimates, proven in Appendix \ref{sec:est-RN},   only for the quadratic term $N_2 (v, v)$. 
The term $\widehat N (v)$ is treated similarly. (The  estimates for $\tilde N_2 (v, v)$ hold also for $\widehat N (v)$, provided the norms involved are less than some constant.)}
 %
  To concentrate on the essentials and keep notation from running amok, in what follows, we keep only the second order terms in the nonlinear estimates, omitting thus the third order ones. These terms give the main contributions if the norms involved are less than some constant.

\paragraph{Control of $v''$: the Lyapunov functionals.}\label{sec:as-stab.}

 We introduce the norms $\|v''(t)\|_{H^1_1}^2 := \|v''(t)\|_{H^1}^2 + \|P'' x v''(t)\|_{H^1}^2$, $\| v' \|_{L^\infty_1}^2:= \| v' \|_{L^\infty}^2+\|x v' \|_{L^\infty}^2$ and 
$\|v'\|_{X_{\del T}'}:=\sup_{t\in [0, T]}[ <t>^{\frac12\del} \|v'(t)\|_{H^1}+ <t>^{\del-\frac12} \| v'(t)\|_{L^\infty_1}].$ 
 With these definitions, 
 the main result of this paragraph is the following 
\begin{lemma}\label{lem:v''-control}
Let $\mu=\mu_{\tau}$ be the same as the one in \eqref{L-L2-lower-bound}. There is $\e > 0$ such that if $
\|v(t)\|_{H^1} \le  \e,\ \forall t \in [0,T]$, then, for all $t \in [0,T]$, 
 \begin{align}\label{v''-control}\|v''(t)\|_{H^1_1} &\lesssim e^{- \frac{1}{4}c \mu t} \|v_0''\|_{H^1_1}+  <t>^{- 11/8}\|v'\|_{X_{3/2 T}'}^{3/2}.
 \end{align} 
\end{lemma}
%
%
\begin{proof} 
We begin with some auxiliary statements. 
First, we define the Lyapunov functional 
	$\Lambda_1(v) = \frac{1}{2}\lan v, \Lo v \ran_{L^2}$ 
 and derive for it a differential inequality.  We compute $\partial_t \Lambda_1 (v)	=\lan \Lo v,  \partial_t  v\ran$.  Now using the equation \eqref{v''-eq} to express $\partial_t  v$, 
	we obtain
\begin{align}\label{Lam-eq} \partial_t \Lambda_1 (v'') = & \lan \Lo v'', - \Lo v'' + \tilde N'' (v)\ran.	\end{align}
 Recall that, due to the assumptions $\kappa > 1/\sqrt 2$ and  \eqref{posit-cond}, $\mu=\mu_{\tau}$ entering \eqref{L-L2-lower-bound}, is positive, $\mu>0$. 
 \DETAILS{ we use following bound 
 \begin{align}\label{L-lowbnd2}  \| \Lo  v'' \|^2\ge \mu \lan v'',  L_\om v''  \ran \ge c \mu^2 \|v''(t)\|_{L^2}^2.\end{align} 
To derive  \eqref{L-lowbnd2}, ones observes that  \eqref{L-L2-lower-bound} gives $\mu\|  v \|^2 \lan  v,  L_\om v  \ran\le  \lan v,  L_\om v  \ran^2\le \|  v \|^2 \| L_\om  v \|^2 $, which implies \eqref{L-lowbnd2} (the second inequality follows from \eqref{L-L2-lower-bound}).}

We begin with the estimate of the second  term on the r.h.s..  We use the rough inequality $|\lan \Lo v'',  \tilde N'' (v)  \ran|\le  \| \Lo v'' \|_{L^2}\| \tilde N'' (v))\|_{L^2}$. To treat the second factor on the r.h.s., we 
use the  estimate 
 \begin{align}\label{tildN-est1} 
\| \tilde N (v)\|_{L^2}\ls   \|v'\|_{H^1} \|v'\|_{L^\infty}+ \|  v'' \|_{H^2}  \|v\|_{H^1},\end{align}
which follows from \eqref{tildeN-est'}, shown in Appendix \ref{sec:nonlin}, if one assumes $ \|v\|_{H^1}\ls 1$. (This assumption is shown later to be superfluous and is made to simplify the expressions involved; it cab also built into our spaces.) 
Next, we use  the bounds \eqref{L-H1-lowbnd} and \eqref{H2-L-est}, together with the estimates above
\DETAILS{need the bounds
\begin{align}\label{H2-L-est}
 \|v''\|_{H^2}\ls \| \Lo v'' \|_{L^2},\ 
\quad \|v''\|_{H^1}\ls \lan v'',  L_\om v''  \ran, \end{align} 
the first of which follows from $\| \Lo v'' \|_{L^2} \ge \del \| \Lo v'' \|_{L^2}+ (1-\del) \mu \| v'' \|_{L^2} \ge \del \| \n v'' \|_{L^2}-  \del C \|  v'' \|_{L^2} + (1-\del) \mu \| v'' \|_{L^2}$ and the second one is derived similarly.}
 and use  $\| \Lo v'' \|_{L^2}\ \|v'\|_{H^1} \|v'\|_{L^\infty} \le  \frac14 \| \Lo v'' \|_{L^2}^2 + C \|v'\|_{H^1}^2 \|v'\|_{L^\infty}^2$, to find
\begin{align*} |\lan \Lo v'',  \tilde N'' (v)  \ran|\le   \frac14 \| \Lo v'' \|_{L^2}^2 + 
C\| \Lo v'' \|_{L^2}^2\|v\|_{H^1}+ C \|v'\|_{H^1}^2 \|v'\|_{L^\infty}^2. \end{align*}
 This estimate, together with  the relations \eqref{Lam-eq} 
 and $ \| \Lo v'' \|_{L^2}^2\ge  \frac12 \| \Lo v'' \|_{L^2}^2+  \mu  \Lambda_1(v'')$ (this follows from \eqref{L-lowbnd2} and $\lan v'',  \Lo v''  \ran=2 \Lam_1(v'')$), 
gives, for some $ C>0$, 
\begin{align}\label{Lyap-diff-ineq2}
\partial_t \Lambda_1(v'') &\leq -  \mu  \Lambda_1(v'') - (\frac{1}{4}  - C\|v\|_{H^1} )\| \Lo v'' \|_{L^2}^2  +
C  \|v'\|_{H^1}^2 \|v'\|_{L^\infty}^2.\end{align}

Now, to control  $ \| P'' x v''(t)\|_{H^1}^2$, we use the Lyapunov functional
\begin{equation}	\label{Lambda2}
	\Lambda_2(v'') = \frac{1}{2}\lan \Lo P'' x v'', P'' x  v'' \ran_{L^2}. 
	\end{equation}		
Similarly to \eqref{Lam-eq}, we derive a differential inequality for $\Lambda_2(v'')$. As in \eqref{Lam-eq}, on the first step, we  obtain
\begin{align}\label{dtLambda2}\partial_t \Lambda_2 (v'') = & \lan  \Lo  P'' x v'',  P'' x (- \Lo v'' + \tilde N'' (v)) 
 \ran. 
 \end{align}
\DETAILS{ We estimate the terms on the r.h.s.. We use that $\lan  \Lo  \bar P x v'',  \bar P x L v''  \ran= \lan \Lo   \bar P x v'',  \Lo \bar P x v''\ran +\lan    v'',  ([\bar P x,  [\bar P x,   L_\om]] v''\ran$ and use $|\lan    v'',  [\bar P x,  [\bar P x,   L_\om]] v''\ran|\le \|   v''\|_{H^1}^2 ??$. {\bf (use that $xP \approx P x$)} Next, by  \eqref{L2-lower-bound'}, we have $\mu \|  \bar P x v'' \|^2_{H^1}\le  \lan  \bar P x v'', \Lo \bar P x v''  \ran$ and $\mu \lan  \bar P x v'', \Lo \bar P x v''  \ran\le  \| \Lo \bar P x v''  \|^2$. 
Finally, using $\tilde N''_\om(v) =\tilde N''_\om(v')+ d \tilde N''_\om(v')v''+\dots$, $\| \bar P x \tilde N''(v')\|_{H^{-1}}\ls \|x^{1/2} v'\|_{H^1}^2+\dots $ and $\| \bar P x d \tilde N''_\om(v')v''\|_{H^{-1}}\ls \|x^{1/2} v'\|_{H^1} \|x^{1/2} v''\|_{H^1}+\dots $, 
 to obtain 
$$|\lan   \bar P x \vpp,  \bar P x \tilde N''_\om(v) \ran|  \le \frac\mu4 \| \bar P x \vpp\|_{H^1}^2 +\| \bar P x \tilde N''(v')\|_{H^{-1}}^2+ \| \bar P x \vpp\|_{H^1}\|   v''\|_{H^1}\| v'\|_{H^1}.$$}
 We estimate the terms on the r.h.s.. We use 
  $ P'' x \Lo = \Lo P'' x + P''  [ x,   \Lo] $  and the estimates $\| [ x,   \Lo] v''\|_{L^2}\ls  \|   v''\|_{H^1}$ and  \eqref{L-lowbnd2} with $v''\ra  P'' x v''$, to obtain 
\begin{align}\label{ineq1}\lan  \Lo  P'' x v'',  P'' x \Lo v''  \ran &\ge \lan \Lo   P'' x v'',  \Lo P'' x v''\ran- \|    \Lo   P'' x v''\|_{L^2}\|   v''\|_{H^1}\notag\\  
&\ge  \mu \Lambda_2(v'')+ \frac14 \|\Lo P'' x v''\|^2-\|   v''\|_{H^1}^2. \end{align}
Let $\| v\|_{L^p_1}:=\| v\|_{L^p} + \| x v\|_{L^p}$. In Appendix \ref{sec:nonlin}, 
we prove the following estimate 
 \begin{align}\label{xtildNv'-est} \| x \tilde N(v)\|_{L^2}\ls  \| v'\|_{L^\infty_1}\| v'\|_{H^1} 
&+\|v'\|_{H^1} \|P'' x v''\|_{H^2} \notag\\  &+\|v'\|_{L^\infty} \|P'' x v''\|_{H^1}. 
\end{align}
By the definition \eqref{P'P''}, $[x, P'']=-[x, P'] - [x, P_g]$, where, recall, $P_g$ is the orthogonal projection onto the span of the gauge modes, defined in \eqref{gauge-modes-proj}.  Furthermore, the property (b) ($V V^{*}=P'$) 
  and the relation \eqref{xV-Vk-deriv}  imply the relations 
\begin{align} \label{Px}  P' x_j =x_j P' - V_j V^* - V V_j^*. 
  \end{align}
Then the relations  \eqref{Px} and \eqref{R-bnds} of Appendix  
 \ref{sec:nonlin} show that  $[x, P'']$ is a bounded operator on $L^2$. We use this fact,  the relation $x \tilde N'' (v) = P''x \tilde N (v) 
  + [x, P'']\tilde N (v)$, 
and the estimates \eqref{tildN-est1} and \eqref{xtildNv'-est}, 
 we find 
\begin{align} \label{xN''v-est} 
&\| x \tilde N''(v)\|_{L^2}\ls  \| v'\|_{L^\infty_1}\| v'\|_{H^1} 
+\|v'\|_{H^1} \|\bar P x v''\|_{H^2}+\|v'\|_{L^\infty} \|\bar P x v''\|_{H^1}. 
\end{align}
 Using this and using $|\lan  \Lo \bar P x v'',  \bar P x \tilde N'' (v) \ran|  \le \|\Lo \bar P x v''\|_{L^2} \|x \tilde N'' (v)\|_{L^2}$, the triangle inequality $ab\le \frac18 a^2 + 2 b^2$ and  $ \|\bar P x v''\|_{H^2}\ls \|\Lo \bar P x v''\|_{L^2}$,  we obtain 
 $$|\lan  \Lo \bar P x v'',  \bar P x \tilde N'' (v) \ran|  \ls
\|\Lo \bar P x v'' \|_{L^2} \big(\|x v'\|_{L^\infty}\| v'\|_{H^1}+  \|\bar P x v''\|_{H^2}\|v\|_{H^1}\big)$$
$$ 
\qquad \qquad  \le  \frac18 \|\Lo \bar P x v''\|_{L^2}^2 +C\|x v'\|_{L^\infty}^2 \| v'\|_{H^1}^2+ C\|v\|_{H^1} \|\Lo \bar P x v''\|_{L^2}^2,$$ 
where, recall, $\|f\|_{X \cap Y} := \| f\|_{X}+\| f\|_{Y}$.
This together with \eqref{dtLambda2} and \eqref{ineq1}, implies 
\begin{align}\label{Lyap2-diff-ineq2}
\partial_t \Lambda_2(v'') \leq -  \mu   \Lambda_2(v'') - (\frac18 &- C_1 \|v\|_{H^1}) \|\Lo \bar P x v''\|_{L^2}^2 
\notag \\ &+  \|   v''\|_{H^1}^2+ C\|x v'\|_{L^\infty}^2 \|v'\|_{H^1}^2. 
\end{align}
Adding this inequality times $\del$ 
to \eqref{Lyap-diff-ineq2} and denoting $\Lambda (v'')=\Lambda_1(v'') + \del\Lambda_2(v'')$ and using the estimate $\| v''\|_{H^{1}}^k\le \| v''\|_{H^{1}}^2\| \Lo v'' \|_{L^2}^{k-2}$ 
 and choosing $\del>0$ so that $\del \|   v''\|_{H^1}^2\le \frac\mu2  \Lambda_1(v'') $
 gives
\begin{align}\label{Lyap-diff-ineq3}
\partial_t \Lambda (v'') \leq -   \mu  \Lambda (v'') &- (\frac{1}{4}  - C\|v\|_{H^1} )\| \Lo v'' \|_{L^2}^2 
 \notag\\
& -  (\frac18- C_1 \|v\|_{H^1}) \|\Lo \bar P x \vpp\|_{L^2}^2  +C_2 \| v'\|_{L^\infty_1}^2\|v'\|_{H^1}^2. 
\end{align}

Now, to complete the proof, we pick $\e$ so that $ C\e \le \frac{1}{8}  $ and  $C_1\e\le \frac{1}{16}$, 
where $C$ and $C_1$ are the same as in \eqref{Lyap-diff-ineq3}. Then 
the assumption  $\|v'(t)\|_{H^1\cap L^\infty} < \e$, for all $t \in [0,T]$, and the inequality \eqref{Lyap-diff-ineq3}  imply that
  $\partial_t \Lambda (v'') 	\leq -  \mu  \Lambda (v'') +C 
g(v')$,    where $g(v) := \| v\|_{L^\infty_1}^2 \|v \|_{H^1}^2$,  which yields $\Lambda (v'') \lesssim e^{- \frac{1}{2} \mu t} \Lambda (v_0'')+ \int_0^t e^{- \frac{1}{2} \mu (t-s)}C  g(v'(s)) ds$, where $v_0=v|_{t=0}$, and therefore, by  the estimate  $\|v''(t)\|_{H^1_1}^2\ls \Lambda (v'')$ (recall the definition $\|v''(t)\|_{H^1_1}^2:=\|v''\|_{H^1}^2+ \|\bar P x v''(t)\|_{H^1}^2$), we have 
\begin{align}\label{v''-control'}\|v''(t)\|_{H^1_1}^2 &\lesssim e^{- \frac{1}{2}c \mu t} 
 \|v_0''\|^2_{H^1_1}+  \int_0^t e^{- \frac{1}{2} \mu (t-s)}g(v'(s)) ds, 
 \end{align} 
for all $t \in [0,T]$,   where, recall, $g(v) := \| v \|_{L^\infty_1}^2 \|v \|_{H^1}$. 
Using the definition of the spaces $X_\del'$, we bound  $ g(v'(s))  \ls   <s>^{- 11/4} \|v'\|_{X_{3/2 T}'}^3$, which gives for the second term on the r.h.s. the estimate \[ \int_0^t e^{- \frac{1}{2} \mu (t-s)}g(v'(s))  ds\ls  \int_0^t e^{- \frac{1}{2} \mu (t-s)}<s>^{- 11/4} \|v'\|_{X_{3/2 T}'}^3ds\]\[\ls <t>^{- 11/4}\|v'\|_{X_{3/2 T}'}^3.\]
This, together with \eqref{v''-control'}, the statement of the lemma. \end{proof}

\paragraph{Equation \eqref{v'-eq}.}\label{sec:v'-eq}	Now we obtain bounds on $v'$. To this end  we investigate the equation \eqref{v'-eq} for $v'$, in which we pass to the spectral representation for the corresponding bands of $\Lo$, described in Lemma \ref{lem:V}. 
\DETAILS{We have
\begin{lemma}\label{lem:v'-control} Assume $v''$ is even and obeys the estimate  $\|v''\|_{H^1_{1/2}}\ls <t>^{-3/4}$.  Then $v'$ obeys the estimates  
\begin{align}   \label{v'-control} <t>^{3/4}\|v'\|_{H^1}+  <t>\| v' \|_{L^\infty_1}\ls \|v'_0\|_{H^1}+  \| v'_0 \|_{L^\infty_1}+ \| v''\|_{Y_{3/4}}^2.\end{align} 
 $\|v'\|_{H^1}\ls <t>^{-3/2}$. $\| v'\|_{X_{3/2 T}'}\ls \|f_0\|_{W^{\infty, 1} \cap L^1}  + \| h\|_{Y_{3/4 T}}^2$
\end{lemma}
To prove this lemma, we investigate the equation \eqref{v'-eq} for $v'$.}
Below, we say that a function $v: \R^2 \ra \C\times  \C\times  \C^2$ is even/odd iff $T^{\rm refl} v = v/T^{\rm refl} v = -v$. 

Recall the definition of the operators $A,\ V$ and $V^{*}$, given in \eqref{A} and in Lemma \ref{lem:V}. We  define $f:=V^* v$. Applying the map $V^* $ to \eqref{v'-eq} and using that $V^*  \Lo =A V^{*} $, we rewrite the equation \eqref{v'-eq} as 
\begin{align}	\label{f-eq}
		\partial_t f = - A f + 
\cN_h (f),\ \text{ where } 
h=v'':= P'' v\ \text{ and }  \cN_h (f)=V^{*} \tilde N (Vf + h).	\end{align}

 Let  $W^{p, s}=\{f \in L^p(\Om^*, \C^2):\ \p_k^s f \in L^p(\Om^*, \C^2) \}$. (For $s<1$,  $W^{\infty, s}$ are identified with the H\"older spaces.)  We consider this equation in the Banach spaces $X_{\del T}:=\{f: [0, T] \ra \C^2: f \text{ is odd and }\ \|f\|_{X_\del T}<\infty \}$, with the norm 
\begin{align} \label{Xdel}\|f\|_{X_{\del T}}:=\sup_{t\in [0, T]}\sum_{s=0, 1}[ <t>^{\del-s/2} \|f(t)\|_{W^{1, s}}+ \| f(t)\|_{W^{\infty, s}}].\end{align}
Now, for $v''$, we take the space  $Y_\del:=\{h: [0, T]\ra H^1(\R^2, \C\times\R^2): \|h\|_{Y_\del}<\infty \}$, with the norm $\|h\|_{Y_\del}:=\sup_{t\in [0, T]} <t>^{\del} \|h(t)\|_{H^1_{1/2}}$. Recall the notation $\|f\|_{X \cap Y} := \| f\|_{X}+\| f\|_{Y}$. 
In this paragraph we prove 
\begin{lemma}\label{lem:f-control} Assume $v''$ is even (i.e. $T^{\rm refl} v'' = v''$) 
 and obeys the estimate  $\|v''\|_{H^1_{1/2}}\ls <t>^{-3/4}$. Then  
any solution to the equation 
\eqref{f-eq},  with an initial condition $f_0\in W^{\infty, 1} \cap L^1$ which is odd, 
 satisfies the following bound 
\begin{align} \label{f-control}\| f\|_{X_{3/2 T}} & \ls\|f_0\|_{W^{\infty, 1} \cap L^1} + \| f\|_{X_{3/2 T}}^2 + \| h\|_{Y_{3/4 T}}^2. 
\end{align}
\end{lemma}
\begin{proof} 
In this proof we set $h=v''$.  By the construction 
\begin{align} \label{f-eq-rhs}- A f + \cN_h (f)= V^* P'J(\ut + V f +h)=: J_h (f).\end{align} 
The equations \eqref{refl-tilderho-sym} and \eqref{V-par} and the relations $T^{\rm refl} P'= P' T^{\rm refl}$ and $T^{\rm refl}  (\ut + V f +h) = \ut  - V \rf f +T^{\rm refl}  h$ show that 
\[\rf V^* P'J(\ut + V f +h)=-  V^* T^{\rm refl} P'J(\ut + V f +h)\]\[=-  V^* P'J( \ut  - V \rf f +T^{\rm refl}  h),\] 
which implies that 
\begin{align} \label{f-eq-rhs-parity} J_{T^{\rm refl} h} (- \rf f)=- J_h (f).\end{align}
 
By the equations \eqref{f-eq-rhs} and \eqref{f-eq-rhs-parity},  the equation \eqref{f-eq} has the parity symmetry in the sense that if $f$ is a solution, then so is $-\rf f$.  Thus if it has a unique solution  and $f_0$ is odd, then this solution is odd.

 Next, we address the nonlinearity $\cN_{\vpp } (f)$.  Let $H^1_{s}=\{v \in H^1(\R^2;\C\times \R^2),\  \bar P |x|^s v \in H^1(\R^2;\C\times \R^2) \}$.   
   We claim that 
   the map $\cN_\vpp (f)=\int_{\Om^*}^\oplus \hat{dk}   \cN_{\vpp, k} (f) $ satisfies  
\begin{align} \label{cNk-even} & 
\cN_{T^{\rm refl} h, k} (- T^{\rm refl}_k f)=\cN_{h, -k} (f). 
\end{align}
To show \eqref{cNk-even}, we use that the second relation in \eqref{V-par}  implies
\begin{align}\label{Vf-h-parity}   T^{\rm refl} (V f + \vpp) = -V T^{\rm refl}_k f+ T^{\rm refl} h.\end{align}
Furthermore, by the  definitions   $\tilde N (v):=V_{\s} v + \Nt (v)$ and $ \Nt(v)=J(\ut+v)- J(\ut)- \Lt v$ (see \eqref{N}) and the relations $T^{\rm refl} J(u) = J(T^{\rm refl} u), T^{\rm refl} V_{\s} = V_{\s} T^{\rm refl}, T^{\rm refl} \Lt = \Lt T^{\rm refl} $ and $T^{\rm refl} \ut = \ut$,
the nonlinearity $ \tilde N  (v)$ satisfies 
\begin{equation} \label{N-parity}T^{\rm refl}  \tilde N  (v)= \tilde N  (T^{\rm refl} v).\end{equation}
 By the definition, we have $\cN_\vpp(f)=V^* \tilde N (Vf+ h)$. 
 This together with \eqref{Vf-h-parity} and  \eqref{N-parity} gives \eqref{cNk-even}.

 By \eqref{cNk-even}, if $h$ is even, then $\cN_{\vpp, k} (f)$ is odd in the sense that $\cN_{\vpp, k} (- T^{\rm refl}_k f)=\cN_{\vpp, -k} (f)$. Furthermore,  if $\cN_{\vpp, k} (f)$ is differentiable in $k$ and if $h$ is even and $f$ is odd, then  $\cN_{\vpp, k=0} (f)=0$ and is of the form $\cN_{\vpp, k} (f)=k \cdot \cN_{\vpp, k 1} (f)$, where $\cN_{\vpp, k 1} (f):=\n_k  \cN_{\vpp, \bar k} (f)$, for some $\bar k \in \Om^*$. 

In what follows we use the following estimate on the nonlinearity 
  \begin{align}  \label{cN-est}
&\left\|\p_k^m\cN_{h, k } (f)\right\|_{L^\infty}\ls 
 \| f\|_{H^{m/2}}^2  +  \| h\|_{H^1_{m/2}}^2,\ m=0, 1, 
 \end{align}
which follows from Proposition \ref{prop:sec:nonlin}, shown in Appendix \ref{sec:nonlin}, if one assumes $\| f\|_{H^{m/2}}\ls 1$ and $\| h\|_{H^1_{m/2}}\ls 1$ (see the parenthetical remark after \eqref{tildN-est1}). 

Now, using the Duhamel principle, we rewrite \eqref{f-eq} as
\begin{equation} \label{Duham-eq}
f(t)=e^{- tA}f_0+\int_0^t e^{-(t-s)A}\cN_{h(s)}  (f(s))\, ds.
\end{equation}
We estimate the propagator $e^{- tA}$. In the rest of the proof we omit the subindex $T$ in $X_{\del T}$ and $Y_{\del T}$. We claim that, for $0 <s\le 1$ and $\frac12 <\del\le \frac34$, 
 \begin{equation}\label{semigrbnd} 
\left\|e^{-tA}   f\right\|_{X_{1+s/2}}\ls \|f\|_{W^{\infty, s} \cap L^{1}}, 
\end{equation}
where, recall, $\|f\|_{X \cap Y} := \| f\|_{X}+\| f\|_{Y}$, 
and
\begin{align} \label{Duham-est}
\left\|\int_0^t e^{-(t-s)A}\cN_{h(s)}  (f(s))\, ds   \right\|_{X_{2\del}} & \ls \| f\|_{X_{2\del}}^2 + \| h\|_{Y_{\del}}^2. 
\end{align}
The desired estimate \eqref{f-control} follows from  the estimates \eqref{semigrbnd} and \eqref{Duham-est} and the integral equation \eqref{Duham-eq}. To be specific we take $s= 1$ and $\del = \frac34$, which suffices for our purposes. 

To  prove the $L^1-$part of  the first bound, we consider first $t\ge 1$. We write $f_k=k \cdot f_{k 1}$, where $f_{k 1}:=\n_k  f_{\bar k}$, for some $\bar k \in \Om^*$, and estimate 
\begin{equation}\label{L1-semigrbnd'} 
\left\|e^{-tA} f\right\|_{L^1}\le\sum_i \|\p_k f_i \|_{L^\infty} \int_{\Om^*}  e^{-t \nu_k^i} |k| \hat{dk} .
\end{equation}
Using the estimate \eqref{nuk12-lwb} on $\nu_k^i$ and changing the variable of integration as $k'=\sqrt t  k$ gives  
\begin{align} \label{Int-est} \int_{\Om^*}  e^{-t\nu_k^i}|k|^\g  \hat{dk}\ls t^{-1-\g/2}\int_{\sqrt t\Om^*}  e^{-|k|^2/2}|k|^\g  \hat{dk}\ls t^{-1-\g/2},\ \g > -2,\ i=1, 2.\end{align}
 This, together with the previous estimate,  gives  
$\left\|e^{-tA} f\right\|_{L^1}\ls t^{-3/2} \|\p_k f\|_{L^\infty} .$ 
Next,  the $L^1-$part of \eqref{semigrbnd}  for  $t\ls 1$ follows from  the elementary estimate
\begin{equation}\label{semigrbnd'''} 
\left\|e^{-tA} f\right\|_{L^1}\le \max_i\sup_{k\in\Om^*} ( e^{-t\nu_k^i})  \|f\|_{L^1} \ls  \|f\|_{L^1} .
\end{equation} 
The last two estimates give
\begin{equation}\label{L1-semigrbnd} 
\left\|e^{-tA} f\right\|_{L^1}\ls <t>^{-3/2} (\|\p_k f\|_{L^\infty} +\| f\|_{L^1} ).
\end{equation}

 Now, using $\p_k e^{-(t-s)\nu_k^i} =-t(\p_k \nu_k^i )e^{-t\nu_k^i} +  e^{-t\nu_k^i}\p_k $ and $| \p_k \nu_k^i|\ls  |k|$ 
 and using the representation $f_k=k \cdot f_{k }'$ and \eqref{Int-est}, we obtain $\left\|\p_k e^{-tA} f\right\|_{L^1}\ls t^{-1} \|\p_k f\|_{L^\infty}.$
We estimate $\left\|\p_k e^{-tA} f\right\|_{L^1}$ for $t\ls 1$ similarly to \eqref{semigrbnd'''}. 
 This shows
 \begin{equation}\label{L1-semigrbnd2} 
\left\|\p_k e^{-tA} f\right\|_{L^1}\ls 
<t>^{-1} \|\p_k f\|_{L^\infty} +  \|f\|_{L^1},
\end{equation}
 which completes the proof of the $L^1-$part of  \eqref{semigrbnd}. 

Similarly,  we have $\left\|e^{-tA} f\right\|_{L^\infty}\le 
\max_i \sup_{k\in\Om^*} ( e^{-t\nu_k^i})  \|f\|_{L^\infty} $ and similarly  $\left\|\p_k e^{-tA} f\right\|_{L^\infty}\ls \left\|t|k|^2 e^{-tA} f'\right\|_{L^\infty}+ \left\|e^{-tA} \p_k f\right\|_{L^\infty}\ls \left\| \p_k f\right\|_{L^\infty}$, where $f':=\int^\oplus_{\Om^*} f_{k }' \hat{dk}$, and therefore the $L^\infty-$part of \eqref{semigrbnd}, i.e. $\left\|e^{-tA} f\right\|_{W^{\infty, 1} }$ $ \ls   \|f\|_{W^{\infty, 1} }$, holds. 

Next, we show the estimate \eqref{Duham-est}.
Using that  $\cN (f, h)$ is odd 
 and using \eqref{L1-semigrbnd}, 
 we obtain
\begin{align} 
\left\|\int_0^t e^{-(t-s)A}  \cN_{h(s)}  (f(s))\, ds\right\|_{L^1} & \ls \int_0^t <t-s>^{-3/2}\|\cN_{h(s)}  (f(s))\|_{W^{\infty, 1}}\, ds. \notag \end{align}
Using this estimate, the bounds \eqref{cN-est} and   
$ \|f\|_{L^2_{1/2}}^2\le \|f\|_{L^\infty_{1}}\|f\|_{L^1_{}}$, 
  we find 
\begin{align} \label{cN-est'}  \|\cN_{h(s)}  (f(s))\|_{W^{\infty, 1} }\ls &\left\| f(s)\right\|_{L^\infty_1}\|f(s)\|_{L^1}+  \| h(s)\|_{H^1_{1/2}}^2\notag \\
&  \ls  <s>^{-3/2} (\left\| f\right\|_{X_{3/2}}^2 +  \| h\|_{Y_{3/4}}^2).\end{align} 
This gives 
\begin{align} 
\bigg\|\int_0^t e^{-(t-s)A} & \cN_{h(s)}   (f(s))\, ds\bigg\|_{L^1}  
\notag \\
& \ls \int_0^t <t-s>^{-3/2} <s>^{-3/2}\, ds (\left\| f\right\|_{X_{3/2}}^2 +  \| h\|_{Y_{3/4}}^2). \notag\end{align}
Next, using \eqref{semigrbnd} and \eqref{cN-est}, 
we have
\begin{align} 
\left\|\int_0^t \p_k e^{-(t-s)A}  \cN_{h(s)}  (f(s))\, ds\right\|_{L^1} 
& \ls \int_0^t <t-s>^{-1}\|\p_k\cN_{h(s)}  (f(s))\|_{L^\infty}\, ds\notag . 
\end{align}
\DETAILS{\begin{align} 
\left\|\int_0^t \p_k^s e^{-(t-s)A}  \cN (f(s), h(s))\, ds\right\|_{L^1} & \ls \int_0^t <t-s>^{-1- r/2}\|\p_k^{r+s}\cN_{ 1} (f(s), h(s))\|_{L^\infty}\, ds \notag \\
& \ls \int_0^t <t-s>^{-1- r/2}(\left\|\p_k^{s/2} f(s)\right\|_{L^\infty}\|\p_k^{s/2}f(s)\|_{L^1}+  \| h(s)\|_{H^1}^2)\, ds \notag \\& \ls \int_0^t <t-s>^{-1- r/2} <s>^{-1- r/2}\, ds (\left\| f\right\|_{X_{1+ r/2, s}}^2 +  \| h\|_{Y_{3/4}}^2). \notag \end{align}}
By \eqref{cN-est'} of Appendix \ref{sec:nonlin} and $\int_0^t <t-s>^{-\al} <s>^{-\beta}\, ds\ls <t>^{-\al}$, for $\beta \ge \al$ and $\beta >1$, this gives 
$$\left\|\int_0^t e^{-(t-s)A}\cN_{h(s)}  (f(s))\, ds   \right\|_{W^{\infty, 1} }  \ls <t>^{-3/2} (\left\| f\right\|_{X_{3/2}}^2 +  \| h\|_{Y_{3/4}}^2),$$ which is the $L^1-$part of the estimate in \eqref{Duham-est}. Similarly to above, using that  $\cN_{\vpp, k } (f)=k \cdot \cN_{\vpp, k, 1} (f)$, we have
\begin{align} 
\bigg\|\int_0^t   e^{-(t-s)A} & \cN_{h(s)}  (f(s))\,  ds\bigg\|_{L^\infty} \notag\\
& \le \int_0^t \max_i  \sup_{k\in\Om^*} ( e^{-(t-s)\nu_k^i})\|\cN_{h(s)}  (f(s))\|_{L^\infty}\, ds. 
\end{align}
Now, using 
\eqref{cN-est'}, we conclude
\begin{align} 
\left\|\int_0^t  e^{-(t-s)A}  \cN_{h(s)}  (f(s))\, ds\right\|_{L^\infty} &   \ls  \left\| f\right\|_{X_{3/2}}^2 +  \| h\|_{Y_{3/4}}^2,  \notag 
\end{align}
Next, we have
\begin{align} 
\bigg\|\int_0^t \p_k e^{-(t-s)A} & \cN_{h(s)}  (f(s))\, ds\bigg\|_{L^\infty} \notag \\
& \le \int_0^t  \max_i \sup_{k\in\Om^*} ( e^{-(t-s)\nu_k^i}(t-s)|k|^2)\|\cN_{h(s), 1}  (f(s))\|_{L^\infty}\, ds \notag \\
&+ \int_0^t \max_i  \sup_{k\in\Om^*} ( e^{-(t-s)\nu_k^i})\|\p_k\cN_{h(s)}  (f(s))\|_{L^\infty}\, ds, \notag 
\end{align}
which, by \eqref{cN-est'},  gives $ \left\|\int_0^t e^{-(t-s)A}\cN_{h(s)}  (f(s))\, ds    \right\|_{L^\infty_1}  \ls   \left\| f\right\|_{X_{3/2 }}^2 +  \| h\|_{Y_{3/4}}^2$. This is the $L^1-$part of the estimate in  \eqref{Duham-est}.
\end{proof}

\begin{proof}[Proof of asymptotic stability.]  Let $V X_{\del T}:=\{v'=V f: f\in X_{\del T}\}$. We consider the equations \eqref{v'-eq}-\eqref{v''-eq} in the spaces $V X_{3/2 T}\times Y_{3/2 T}$ for $(v', v'')$. By the standard local theory, there is $T>0$ s.t.  these equations are well-posed on the interval $[0, T]$.

Now, using that  $v'=V f$,  we show that
\begin{align}   \label{X-X'-bnd}\| v'\|_{X_{\del T}'}\ls \| f\|_{X_{\del T}}.\end{align} 
Indeed,  by the statement (d) of Lemma \ref{lem:V},  
 we have
 \[\|v'\|_{H^1}\ls \|f\|_{L^2}\ls (\| f\|_{L^{\infty}}\| f\|_{L^{1}})^{1/2}\le <t>^{-\del/2}\| f\|_{X_{\del T}}.\] Next, 
the inequality \eqref{V-bnd5}
implies \[\| v' \|_{L^\infty_1}\ls \| f\|_{W^{1, 1}}\le <t>^{-\del+\frac12}\| f\|_{X_{\del T}},\] 
which gives \eqref{X-X'-bnd}. 

 By taking initial condition sufficiently small, we can attain that $ \|v'\|_{X_{3/2 T}'}+\|v''\|_{Y_{3/4 T}}\le \e$, for given $\e>0$. 
Next, we observe $\sup_{0\le t\le T}\|v\|_{H^{1}} \le \|v'\|_{X_{2\del T}'}+\|v''\|_{Y_{\del T}}$ for any $\del\ge 0$. Hence the estimates of Lemmas \ref{lem:v''-control} and \ref{lem:f-control} hold. Combining these estimates and using \eqref{X-X'-bnd}, we obtain
\begin{align} \label{combin-est}\| f\|_{X_{3/2 T}}+\|v''\|_{Y_{11/8 T}} & \ls\|f_0\|_{W^{\infty, 1} \cap L^1}+  \|v_0''\|_{H^1_1} + \| f\|_{X_{3/2 T}}^{3/2} + \| v''\|_{Y_{3/4 T}}^2. 
\end{align}
This estimate, for  sufficiently small $\|f_0\|_{W^{\infty, 1} \cap L^1} + \|v''_0\|_{H^1_1}$, implies that $\| f\|_{X_{3/2 T}} + \|v''\|_{Y_{11/8 T}} \ls\|f_0\|_{W^{\infty, 1} \cap L^1} + \|v''_0\|_{H^1_1}$. This bound can be bootstrapped to 
  $T=\infty$.
\end{proof}

\subsection{Instability}

Recall, $\g_{  \del }(\tau):=\inf_{\dist(k, \LATt^*)\ge \del}\g_{ k }(\tau)$, with $\g_{ k }(\tau)$ given in \eqref{gamk-tau}.  Assume that 
either $\kappa^2<1/2$ or $\g_{ \del }(\tau_*) <0$ for some  $\tau_*$ and $\del>0$. 
 Then there is $k_*\in \Om^*_{\tau_*}$ in whose neighbourhood, $(\kappa^2- 1/2)\g_k (\tau_*) <0$. (From the definition,  $\sup_k\g_k (\tau) > 0$.) Now, by Theorem \ref{thm:energyband} (see also Proposition \ref{prop:Lk-spec}), the lowest spectral branch, $\nu^2_k$, of 
  the hessian $\Lt$ is negative for $\tau = \tau_*$ and these $k$'s, provided $\e$ is sufficiently small. 
Then  the energetic instability of $\ut$ for such a $(\tau_*, \kappa_*)$ follows directly from  the definition. 
$\Box$

\appendix

\section{Product transformation under $U$} \label{U-prod-transf}  
In this section, we develop the product transformation formula for the magnetic Bloch-Fourier-Zak transform $U$, which we use in estimating the nonlinearities in Appendix \ref{sec:nonlin}.  

We consider maps 
$p_k: \underbrace{\C^4 \times \dots \times \C^4}_{k  \text{ times}} \times \C\ra  \C^4$, written as $p_k(v_1, \dots, v_k, \psi)$, with 
  the properties that 
\[\text{$p_k$ is linear in the first $k$ arguments,} \]  \[p_k(\tilde t^{\rm gauge}_\chi v_1, \dots, , \tilde t^{\rm gauge}_\chi v_k, e^{i \chi} \psi)=\tilde t^{\rm gauge}_\chi p_k(v_1, v_2, \psi),\] 
where  $\tilde t^{\rm gauge}_\chi : (\xi, \eta, \al) \ra (e^{i \chi}\xi, e^{- i \chi}\eta, \al)$, with $\chi\in \R$ (cf. \eqref{tildeTgauge}), and 
\[|p_k(v_1, \dots, v_k, \psi)|\ls (1+|\psi|)|v_1| \dots |v_k|.\] 
 Examples of such products are  $p_2(v_1, v_2, \psi)= (v_1^\al v_2^\xi, \bar\psit  v_1^\xi v_2^\eta, \bar\psit  v_1^\al v_2^\xi)$,  $p_2(v_1, v_2, \psi)= (\bar\psit v_1^\xi v_2^\xi , v_1^\eta v_2^\al, \bar v_1^\xi v_2^\xi)$ and $p_3(v_1, v_2, v_3)= (v_1^\al v_2^\al v_3^\xi,  v_1^\xi v_2^\eta  v_3^\eta, v_1^\al v_2^\xi v_3^\eta)$.  
 Here and below, we use the super-indices $\xi$ and $\al$ to distinguish the $\xi-$ and $\al-$ components of the vectors in $\cH:=L^2(\R^2, \C)\oplus L^2(\R^2, \C)\oplus L^2(\R^2, \C)$ and the operators acting on these components.
  
\DETAILS{ We also consider  the trilinear maps $p_3: \C^4 \times \C^4 \times \C^4\ra  \C^4$, written as $p_3(v_1, v_2, v_3)$, with 
  the properties 
\[\text{$p_3$ is linear in each argument,} \]   
  \[p_3(\tilde t^{\rm gauge}_\chi v_1, \tilde t^{\rm gauge}_\chi v_2, t^{\rm gauge}_\chi v_3)=\tilde t^{\rm gauge}_\chi p_3(v_1, v_2, v_3),\]
\[|p_3(v_1, v_2, v_3)|\ls |v_1| |v_2| |v_3|,\] 
such as  $p_3(v_1, v_2, v_3)= (v_1^\al v_2^\al v_3^\xi,  v_1^\xi v_2^\eta  v_3^\eta, v_1^\al v_2^\xi v_3^\eta)$. } 
 
\begin{lemma}\label{lem:BFZT-prod-est} If $T^{\rm trans}_{t} \psi=e^{ i g_s} \psi$ and $\|\psi\|_{L^\infty}<\infty$, then
\begin{align} \label{BFZT-prod2-est}
\|U{p_k(v_1, \dots,  v_k, \psi)} \|_{L^\infty_k L^r_x}\ls  |\Om^*| (1+\|\psi\|_{L^\infty_x}) \prod_j\|U v_{j} \|_{L^{p_j}_k L^{q_j}_x}, 
 \end{align}
for $ 
\sum_j p_j^{-1} =1,\ \sum_j q_j^{-1}  =r^{-1},\ 1\le r \le \infty$.
\DETAILS{, and  
\begin{align} \label{BFZT-prod3-est}\|U{p_3(v_1, v_2, v_3)} \|_{L^\infty_k L^r_x}\ls  |\Om^*|  \|U v_{1} \|_{L^p_k L^q_x}  \|U v_{2} \|_{L^{p'}_k L^{q'}_x} \|U v_{3} \|_{L^{p''}_k L^{q''}_x}, \end{align}
for $ p^{-1}+ (p')^{-1} + (p'')^{-1} =1,\  q^{-1} + (q')^{-1}  + (q'')^{-1} =r^{-1},\ 1\le r \le \infty$.} \end{lemma} 
\begin{proof} Denote $\hat v:= U v$. For the sake of simplicity of notation, take $k=2$. Writing $v_i = U^* \hat v_{i}$ and using the definition of the map $U^*$ in \eqref{U*}, the bi-linearity of $p_2$, the property $T^{\rm trans}_{t} \psi=e^{ i g_s} \psi$ and the second property of $p$, we find 
\[p_2(v_1, v_2, \psi) (x)=p(\int_{\Omega^*}   \hat {dk'}  e^{ i k' \cdot t}  T^{\rm trans}_{- t} \tilde T^{\rm gauge}_{g_t}\hat v_{1 k'} (x), \int_{\Omega^*}   \hat {dk''}  e^{ i k'' \cdot t}  T^{\rm trans}_{- t} \tilde T^{\rm gauge}_{g_t}\hat v_{2 k''} (x), \psi (x))\]
\[=\int_{\Omega^*}   \hat {dk'}  e^{ i k' \cdot t}   \int_{\Omega^*}   \hat {dk''}  e^{ i k'' \cdot t}  T^{\rm trans}_{- t} p( \tilde T^{\rm gauge}_{g_t}\hat v_{1 k'} (x), \tilde T^{\rm gauge}_{g_t}\hat v_{2 k''} (x), T^{\rm trans}_{t} \psi (x))\]
\[=\int_{\Omega^*}   \hat {dk'}   \int_{\Omega^*}   \hat {dk''}  e^{ i (k'+k'') \cdot t}  T^{\rm trans}_{- t} p( \tilde T^{\rm gauge}_{g_t}\hat v_{1 k'} (x), \tilde T^{\rm gauge}_{g_t}\hat v_{2 k''} (x), e^{ i g_s}  \psi (x))\]
\[=\int_{\Omega^*}   \hat {dk'}   \int_{\Omega^*}   \hat {dk''}  e^{ i (k'+k'') \cdot t}  T^{\rm trans}_{- t}  \tilde T^{\rm gauge}_{g_t} p_2(\hat v_{1 k'} (x), \hat v_{2 k''} (x),\psi (x)).\] 
This, together with  the definition of the map $U$ in \eqref{U},  gives
\begin{align*}(U{p_2(v_1, v_2, \psi)})_k (x) = \int_{\Omega^*}   \hat {dk'}   \int_{\Omega^*}   \hat {dk''} \sum_{t \in \LAT}  e^{ i (k'+k''-k) \cdot t}   p_2(\hat v_{1 k'} (x), \hat v_{2 k''} (x),\psi (x)).
\end{align*} 
Using the Poisson summation formula 
$\sum_{t \in \LAT} e^{- ik\cdot t} = |\Om^*| \del(k)$ on periodic functions, this gives 
\begin{align*}(U{p(v_1, v_2, \psi)})_k (x) =  |\Om^*|  \int_{\Omega^*}   \hat {dk'}   p(\hat v_{1 k'} (x), \hat v_{2 k-k'} (x),\psi (x)). \end{align*}
Now, by the third property of $p$ and by the H\"older and Hausdorf-Young inequalities, we have 
\eqref{BFZT-prod2-est} for $k=2$. 
The general $k$ is done in exactly the same way.  \end{proof}
\section{The shifted hessians $\Lt^{\rm shift}$ and $\Kt$ and their fibers} 
 \label{sec:hessLshift}

In computation of the spectrum of $\Lt$,  the infinite dimensional subspace, $\cV_g:=\{G_\chi: \chi\in H^1\}$, of zero modes $G_{\chi},\ \chi\in H^1$, presents a considerable headache. To eliminate this subspace, we pass to the operator 
 \begin{align} \label{Lshift}\Lt^{\rm shift}:= \Lt + \frac12 \G \G^*,\end{align}
  where the operators $\G$ and $\G^*$ are defined in \eqref{gauge-modes-proj}. 
 We call $\Lt^{\rm shift}$ the {\it shifted} hessian. 
A standard result yields that it is self-adjoint.  It has the two advantages: (i) $G_{\chi},\ \chi \in H^1$ are not zero modes of $\Lt^{\rm shift}$ anymore, while  $\Lt^{\rm shift}\big|_{\cV_g^\perp} =  \Lt\big|_{\cV_g^\perp}$ (since, as can be readily checked, $\G^*= 0$ on $\cV_g$); (ii) $\Lt^{\rm shift}$ has a simpler explicit form than $\Lt$. 
 
 To elaborate (i), 
 we  define the subspace $\cH_\perp:=\{v\in\mathcal{H}: v\perp G_\gamma\}$ 
 and denote by $\cH_\perp^s$  the corresponding Sobolev spaces. Then 
\begin{equation}\label{specLsh-L}  \Lt^{\rm shift}\big|_{\cH^2_\perp} =  \Lt\big|_{\cH^2_\perp}, 
\end{equation}
 \begin{align} \label{LshGchi}
   \Lt^{\rm shift}  G_\chi = G_{h\chi}, \quad \mbox{where }\ \quad h\equiv \hr:=-\Delta +|\psit|^2.\end{align}
This leads to the following 
\begin{lemma}\label{lem:spec-Lshift-L} \begin{align}\label{specLsh-L}  \s(\Lt^{\rm shift})\cap (-\infty,  c\e^2)=  \s(\Lt)\cap (-\infty, c\e^2),\ c\gs 1. 
\end{align}  
\end{lemma}
\begin{proof} Since $ \lan  G_\chi, \Lt^{\rm shift}    G_{\chi'} \ran= \lan \chi,  h\chi' \ran$, the equation \eqref{LshGchi} implies 
\begin{equation}\label{GLshG} \lan  G_\chi, \Lt^{\rm shift}   G_\chi \ran= \| h\chi\|^2. \end{equation} 
Since $h\gs \e^2$ and, by the above, $\| G_\chi\| = \| \sqrt h \chi\|$, this gives $\lan  G_\chi, \Lt^{\rm shift}   G_\chi \ran \gs \e^2 \| G_\chi\|^2$ and therefore $\Lt^{\rm shift}\big|_{\cV_g}\subset [c \e^2, \infty),\ c\gs 1$. Hence, by the invariance of $\cV_g^\perp$ under $ \Lt $ and $ \Lt^{\rm shift}$, \eqref{specLsh-L} holds. \end{proof} 
The translational modes $S_j$  are not zero modes of $\Lt^{\rm shift} $ but certain of their combinations with the gauge modes $G_\chi$  gives zero modes.
Indeed, we have the following 

\begin{lemma}\label{lem:transl-zero-mode} The operator $\Lt^{\rm shift} $ has the following `gauged' translational zero modes 
\begin{align} \label{Tj} 
T_j  :=  S_{j}  - G_{h^{-1}\g_j},  \quad \mbox{where}\ 
 \g_j:= - i \bar\psit  (\n_{\at j} )\psit + i\psit  \overline{(\n_{\at j} )\psit} - \n^\perp_j  (\curl  \at). \end{align}
Here, recall, $h:=-\Delta +|\psit|^2 $,  $\n^\perp:=(-\n_2, \n_1)$, $ S_{j}$ and $  G_{\g}$ are the translational and gauge modes given above. The zero modes $T_i $ are gauge periodic w.r.to $\cL$. 
\end{lemma}
\begin{proof}  First,  by the definition and the equation $\Lt S_{j} = 0$, we have $\Lt^{\rm shift} S_{j} =\G \G^*  S_{j} = G_{\g_j}$, where $\g_j:=\G^*  S_{j}$ (see \eqref{Lshift} and \eqref{gauge-modes-proj}). 
Computing $\g_j:=\G^*  S_{j}$, we find the expression in \eqref{Tj}. %
\DETAILS{The definitions give 
 \begin{align} \label{LshS} 
   \Lt^{\rm shift}  S_{j} = G_{\g_j},  \quad \mbox{where}\  \g_j:= - i \bar\psit  (\n_{\at j} )\psit + i\psit  \overline{(\n_{\at j} )\psit} - \n^\perp_j  (\curl  \at), 
\end{align} 
where $\n^\perp:=(-\n_2, \n_1)$.  Using \eqref{LshS}} Using $\Lt^{\rm shift}  S_{j} = G_{\g_j}$ and \eqref{LshGchi},  we obtain $\Lt^{\rm shift}  (S_{j} -  G_{h^{-1}\g_j}) = G_{\g_j}  - G_{\g_j} =0$.
\end{proof}

Similarly to the fiber decomposition of the operator $\Lt$, given in Section \ref{sec:Bloch-deco}, one constructs the fiber decomposition,  $U \Lt^{\rm shift} U^{-1} = \int_{\Omega^*}^\oplus \Lk^{\rm shift} \hat {dk}$, of the operator $\Lt^{\rm shift}$.  By a standard result, the fibers are self-adjoint.  Repeating the arguments above, one sees that the relation analogous to \eqref{specLsh-L} holds also for $\Lk^{\rm shift}$:
 \begin{equation}\label{specLksh-L}  \s(\Lk^{\rm shift})\cap (-\infty,  \e^2)=  \s(\Lk)\cap (-\infty, c \e^2),\ c\gs 1. 
\end{equation}
Eq \eqref{specLksh-L} shows that, if we are interested in the spectrum of $\Lk^{\rm shift}$ in the interval $(-\infty, c \e^2),\ c\gs 1$, then it suffices to find the spectrum of $\Lk^{\rm shift}$ in this interval. Because of the property (i) this task is much simpler than the original one for $\Lk$ and is done  in Appendix \ref{sec:prop-Lk-spec-pf}.

\begin{remark}\label{rem:spurious-spec}   By \eqref{LshGchi}, the spurious spectrum of $\Lk^{\rm shift}\big|_{\cV_g}$ is given by 
\[\s(\Lk^{\rm shift}\big|_{\cV_g})=\s(h\big|_{H^1_k})\]
where $H^1_k$ is defined in \eqref{H1k-space}.  
 The corresponding eigenfunctions of $\Lk^{\rm shift}$ are given by $G_{\chi_j}$, where $\chi_j$ are the eigenfunctions of $h$ on $H^1_k$.
\end{remark}
One might be able to distinguish the spurious eigenfunctions, $G_{\chi_j}$, by their symmetry (see  \eqref{commut-Trefl-tilderho-rho-Kk}). However, this is not necessary, since the corresponding eigenvalues $\gs \e^2$ and play no separate role in our analysis.

\medskip

We note that since the zero modes $T_i $ are gauge periodic w.r.to $\cL$, they belong to $\cH_{k=0}$ and therefore are eigenfunctions of  the fiber operator $L_{k=0}^{\rm shift}$ with the eigenvalue $0$. 

Finally,  we mention the following general result:
  \begin{lemma}\label{lem:evLk-smooth-k} 
The operators  $\Lk^{\rm shift}$ have purely discrete spectrum. Simple eigenvalues of $ \Lk^{\rm shift}$ are smooth in $k$, with the corresponding eigenfunctions, $v_{k  }^j$, smooth in $x$ and $k$,
 \begin{align}   \label{vkj-est} & v_{k  }^j \in C^\infty(\R^2 \times \R^2, \C^4),\ \forall j.  \end{align}   \end{lemma} 
\begin{proof} The first statement is a standard result. A hands on proof is given in Proposition \ref{prop:Kk-spec-gen}below. 

Consider the operator $\tilde  \Lk^{\rm shift}:= e^{- i k\cdot x}  \Lk^{\rm shift} e^{i k\cdot x}$ defined on $\cH_{k=0}$. Since the map $v\ra e^{i k\cdot x} v$ maps $\cH_{k=0}$ unitarily into $\cH_{k}$, it has the same eigenvalues as the operator $\Lk^{\rm shift}$. It is easy to show, using the relation  $ e^{- i k\cdot x} \n_a e^{i k\cdot x}=\n_{a-k}$ that the operator $\Lk^{\rm shift}$ depends on $k$ smoothly (say, as an operator from $H^2$ to $L^2$). Hence the statement follows from the standard perturbation theory (see e.g. \cite{RSIV, GS}). 

 Estimate \eqref{vkj-est} follows by  the elliptic regularity and the perturbation theory arguments. \end{proof}



In what follows, we denote the perturbations (or fluctuations) of the magnetic potential $a$ by $\vec\al =(\al_1, \al_2): \R^2\ra \C^2$ and the corresponding full fluctuations, by 
$\vec v = (\xi, \phi, \vec\alpha)$. 
For computations, it is convenient to pass from $\vec\al$ to the complex vector-fields $\vec\al^\#$ defined by
$ \vec\alpha =(\alpha_1, \alpha_2) 
\longrightarrow \vec \al^\#:=(\alpha, \beta),\ \alpha:=\alpha_1 - i\alpha_2,\ \beta:=\alpha_1 + i\alpha_2.$ This leads to the transformation
\begin{align}\label{vec-v-transf} \#: \vec v = (\xi, \phi, \vec\alpha) \longrightarrow \vec v^\# = (\xi, \phi, \alpha, \beta),\ \alpha:=\alpha_1 - i\alpha_2,\  \beta:=\alpha_1 + i\alpha_2.\end{align}
 Now, we reserve the notation $v$ for the vectors $\vec v^\#: v \equiv  \vec v^\# = (\xi, \phi, \alpha, \beta)$. The corresponding space is denoted by $\mathcal{K}=L^2(\R^2;\C\times \C\times \C\times \C)\equiv L^2(\R^2;\C)^4$.
On this space, we consider the usual $L^2$ inner product {\it (different from \eqref{ip-cc}??)}
\begin{equation*}
    \lan v, v' \ran_{L^2} = \frac12 \int \bar{\xi}\xi' + \bar{\phi}\phi' + \bar{\alpha}\alpha' + \bar{\omega}\omega',
\end{equation*}
where  $v = (\xi, \phi, \alpha, \omega),\  v' = (\xi', \phi', \alpha', \omega')  \in L^2(\R^2;\C)^4$. 

 The transformed shifted hessian 
 is denoted by $\Kt$. 
\DETAILS{To pass to the hessian $\Kt$ 
 acting on 
 vectors of the form $v = (\xi, \phi, \alpha, \omega)\in \mathcal{K}$, we perform the transformation
\begin{align}\label{vec-al-transf} \vec\alpha =(\alpha_1, \alpha_2) \longrightarrow \vec\alpha^\# = (\alpha, \beta),\ \alpha:=\alpha_1 - i\alpha_2,\  \beta:=\alpha_1 + i\alpha_2.\end{align} 
 For 
 $v = (\xi, \phi, \vec\alpha)$, we denote 
  $v^\# = (\xi, \phi, \alpha, \beta)$, where $\alpha:=\alpha_1 - i\alpha_2$ and $ \beta:=\alpha_1 + i\alpha_2$,
 and define $\Kt$ by}
 Formally, it is defined by $\Kt \vec v^\#=(\Lt^{\rm shift} \vec v)^\#$. One can show that 
 \begin{align} \label{Kt}\Kt =(\Lt^{\rm shift})^\# =(\Lt^\#)^{\rm shift},\ \text{ where }\ A^\# :=\#  A \#^{-1}.\end{align}
    The hessian $\Kt$ can be also obtained by differentiating the fully complexified Ginzburg-Landau (or Gorkov-Eliashberg-Schmidt) equations and then performing the shift, see Subsection \ref{sec:GES-fully-compl} and Appendix \ref{sec:hess-expl}. 

Denote by $\Kk$ the $k-$fibers of $\Kt$. They are related to $\Lk^{\rm shift}$ and $\Lt^\#$ as in \eqref{Kt}.  As for $\Lt$ and $\Lk$, the operators $\Kt$ and $\Kk$ have the same form (given explicitly in Appendix \ref{sec:hess-expl}) and differ only by the constraints on the vectors on which they are defined. 
By the definition, we have  that 
\begin{align} \label{K-Lsh-equiv} 
\Kt/\Kk\ \text{ is unitary equivalent to }\ \Lt^{\rm shift}/\Lk^{\rm shift}. 
\end{align}
Relations  \eqref{specLksh-L} and \eqref{K-Lsh-equiv} imply in particular that
 \begin{equation}\label{specK-L}  \s(\Kk^{\rm shift})\cap (-\infty,  \e^2)=  \s(\Kk)\cap (-\infty, c \e^2),\ c\gs 1. 
\end{equation}

As for $\Lt^{\rm shift}$, the (transformed) gauge and gauged translational modes $G_\chi^\#$ and $ S_j^\#$, given in \eqref{compl-gauge-mode-transf} and \eqref{Sj-transf}, are not zero modes of $\Kt$, but their combinations, as in \eqref{Tj}, give the zero modes, 
\begin{align} \label{Tj-transf} 
K_\# T_j^\#  = 0,\ T_j^\#  :=  S_{j}^\#  - G_{h^{-1}\g_j}^\#,  \quad K_\# =\Kt K_0, 
\end{align}
where, recall, $h:=-\Delta +|\psit|^2 $ and (with $b_\tau;= \curl \at$) \[\g_1= - \bar\g_2 = -  i (\bar \psit \pat \psit +  \psit \overline{\pat^* \psit} + 2\p b_\tau).\]   
The zero modes $T_i^\# $ are gauge periodic w.r.to $\cL$  and therefore belong to $\cK_{k=0}$ and are eigenfunctions of  the fiber operator $K_{k=0}$ with the eigenvalue $0$. 

Moreover, we have the following asymptotical behaviour 

\begin{align}  \label{Tj-exp} 
 T_{j}^\#  =  S_{j}^\# + O(\e^2),\  S_{j}^\#  =  S_{j}^0 + O(\e),\ 
  S_{1}^0 = ( 0, 0, 1, 1),\    S_{2}^0 = ( 0, 0, i, -i).
\end{align} 
Indeed, using the expansion $\psit = O(\e ), \	 \at = a^0 + O(\e^2),$ where $a^0:=\frac{1}{2} x^\perp,\  x^\perp:= (-x_2, x_1)$, 
we obtain  \eqref{Tj-exp}. 

Finally, we translate some important notions and statements to the new representation.  The $\#-$transformations of  the operators $ \G$ and $\G^*$ defined in \eqref{gauge-modes-proj} are given by
\begin{equation}\label{Gam-Gam*}  \G^\# \chi: =G_{\chi}^\#\ \text{   and  }\ (\G^\#)^* v=- i \bar\psit \xi +  i \psit \eta  - \bar \p \al  -\p \beta, 
\end{equation}
 for $v=(\xi, \eta, \al, \beta)$.

For the transformed symmetry operations,  
the reflection transformation is   \begin{equation}\label{Trefl}
    T^{\rm refl} : (\xi( x), \eta( x), \al ( x), \beta ( x))  \mapsto (\xi(  - x), \eta( - x), - \al ( - x), - \beta ( - x)),
\end{equation}
and 
the `particle-hole' (real-linear) transformation is defined as 
\begin{align}\label{CS}
&\tilde\rho:=\mathcal{C}\cS\ 	 \text{ and }\ \rho: =T^{\rm refl} \tilde\rho \equiv T^{\rm refl} \mathcal{C}  \cS,\\ 
 &	\text{ where }\ \cS = \left( \begin{array}{cc}	\s   & 0 \\ 0 & \s \end{array} \right),\ \s = \left( \begin{array}{cc}	0 & 1  \\ 1 & 0\end{array} \right), 
\end{align} 
 and, recall, that $\cC$ denotes the complex conjugation. As can be easily checked from the definitions, $T^{\rm refl} \ut=\ut$ and $\tilde\rho \ut=\ut$ and
	\begin{align} \label{commut-Trefl-tilderho-rho-Kk}
	 T^{\rm refl} \Kk= K_{- k} T^{\rm refl},\  \quad \tilde\rho  \Kk=  K_{- k}  \tilde\rho,\   \quad  [ \Kk, \rho ]=0.
	\end{align} 
The vectors $S_{j}^\#,  G_{\chi_k}^\#,\  \chi_k\in H^1_k$ satisfy
\begin{align}\label{SG-parity} \rho S_{j}^\# = - S_{j}^\#,\ \rho G_{\chi_k}^\# = - G_{\chi_k}^\#,\    T^{\rm refl} S_{j}^\# = - S_{j}^\#,\ T^{\rm refl} G_{\chi_k}^\#=  G_{\chi_k}^\#.\end{align}

  Since the maps $\rho$, $\tilde\rho$ and $T^{\rm refl}$ satisfy $\rho^2 = 1$, $\tilde\rho^2 = 1$ and $(T^{\rm refl})^2 = 1$, their spectra consist of the eigenvalues $\pm 1$. The multiplication by $i$ maps between the subspaces $\{\rho = 1\}$ and $\{\rho = - 1\}$.  In view of \eqref{commut-Trefl-tilderho-rho-Kk}, we can restrict $\Kk $ to the invariant subspace $\{\rho = - 1\}$. 

\section{Explicit expressions of 
 various hessian}\label{sec:hess-expl}

In this appendix we present the explicit expressions for various hessians. 
  In what follows, $(A \psi)$ denotes the function obtained by applying the operator $A$ is applied to the function $\psi$, while, $A \psi$,  the product of two operators one of which is multiplication by $\psi$, etc.

To begin with, it is a straightforward to show that the complex hessian $\Lt$, 
defined in \eqref{Lt} is given explicitly by
 \begin{equation}\label{Lc-expl} 
	 \Lt = \left( \begin{array}{ccc} \hta'' & \kappa^2 \psit^2  & e \\
\kappa^2 \bar{\psi}_{\tau}^2 & \overline{\hta''}  & 
 \bar e \\ 
 e^* &  
 \bar e^* &  2 \hta'''
		\end{array} \right),
\end{equation} 
where,  for an operator $A$, we let $\bar{A} := \mathcal{C} A \mathcal{C}$, with $\mathcal{C}$ standing for the complex conjugation, and 
    \begin{align}\label{h'} \hta'' &:= -\COVLAP{a_{\w}}  + 2\kappa^2|\psi_{\w}|^2 - \lambda_{\w},\  \hta''' := \curl^*\curl  +  |\psi_{\w}|^2,\\ 
    \label{e}   e &:=
 i((\COVGRAD{a_{\w}}\psi_{\w})\cdot + \COVGRAD{a_{\w}} \psi_{\w} \cdot)= 
 2 i (\COVGRAD{a_{\w}}\psi_{\w})\cdot + 
 i \psi_{\w}\n  \cdot .\end{align} 
  In the last equality, we used that $\COVGRAD{a_{\w}}\cdot \psi_{\w}=(\COVGRAD{a_{\w}}\psi_{\w}) +\psi_{\w}\n  \cdot$.   The operator $L^{ c} $ is symmetric. 
A straightforward calculation shows that
\begin{align}\label{compl3}  	& 2i\vec\alpha\cdot\COVGRAD{\vec a}\psi
		= -i(\partial_{a}^*\psi)\alpha + i(\partial_{a}\psi)\bar{\alpha}.
\end{align}

Using relations \eqref{compl1}, \eqref{compl2} and \eqref{compl3}, we transform $\Lt $ according to \eqref{vec-v-transf} to find the explicit expression 
 for  $\Lt^\#$ (see \eqref{Kt}):  
  \begin{align}\label{Lt'}
	\Lt^\# = & \left( \begin{array}{cccc}
\hta'' & \kappa^2 \psit^2  & b & c\\
\kappa^2 \bar{\psi}_{\tau}^2 & \overline{\hta''} & \bar c & \bar b \\
b^* & \bar c^* & \tilde\hta  &  \frac12 \p^2 \\
 c^*  & \bar b^* &  \frac12 \bar\p^2 &   \tilde\hta 
		\end{array} \right),
\end{align}
where $\hta'' := -\COVLAP{\at}  + 2\kappa^2|\psit|^2 - \lamt$, $ \tilde\hta:=-\frac12 \Delta +  |\psit|^2, $    $\p:= \p_1-i \p_2 \equiv  \p_{x_1}-i \p_{x_2}$ (differing from the {\it standard}  ones by the factor $2$, see  \eqref{annih-creat-ops})
\[b:= -\frac12 i (\partial_{\at}^*\psit) - \frac12 i \partial_{\at}^*\psit =- i(\partial_{\at}^*\psit) + \frac12 i \psit \bar \p,\]
\[ c:=  \frac12  i(\partial_{\at}\psit) + \frac12 i \partial_{\at}\psit=   i(\partial_{\at}\psit) + \frac12 i \psit\p. \] 
Here $\partial_{a}:=\p_{a 1} -i \p_{a 2}=\partial - i  a_c$, with $a_c:=a_1-i  a_2$, the complexification of $a $,   and we used, in passing to the latter equations, $\p^*=-\bar \p$ and
\begin{align}\label{der-rel}  	& \pa \psi +(\pa \psi) = 2 (\pa \psi) + \psi \p,\  \quad \pa^* \psi +(\pa^* \psi) =  2 (\pa^* \psi) + \psi \p^*. 
\end{align}
(Not to confuse  $c$ and $c^*$ in \eqref{Lt'}, which are used only in this section, with the annihilation and creation operators, $c$ and $c^*$,  given in \eqref{annih-creat-ops}, and which are not used in this section.)

Note that the transformed hessian $\Lt^\#$ can be also obtained by differentiating the negative of the r.h.s. of the fully complexified GES equations, 
 \eqref{GESresc-cc}, w.r.to $\psi, \bar \psi, a, \bar a$ and using that $\COVLAP{a}=-\frac12 (\partial_{a}\partial_{a}^*+ \partial_{a}^*\partial_{a})$.

Finally, we find the explicit form of the operator $\Kt$. 
  \begin{lemma}\label{lem:K-expl} 
The operator $\Kt$ is given explicitly as 
\begin{equation}\label{K}
	\Kt = \left( \begin{array}{cccc}
	\hta'   & (\kappa^2 - \frac12)\psit^2  & -i(\partial_{\at}^*\psit) & i(\partial_{\at}\psit) \\
(\kappa^2 - \frac12)\bar{\psi}_{\tau}^2 & \overline{\hta' } & -i(\overline{\partial_{\at}\psit}) & i(\overline{\partial_{\at}^*\psit}) \\
i(\overline{\partial_{\at}^*\psit}) & i(\partial_{\at}\psit) & \hta & 0 \\
		-i(\overline{\partial_{\at}\psit})
			& -i(\partial_{\at}^*\psit)
			& 0 & \hta
		\end{array} \right),
\end{equation}
where $\hta' := -\COVLAP{\at}  + (2\kappa^2 + \frac12 )|\psit|^2- \lamt $ and $\hta:= -\Delta +  |\psit|^2$.    
  \end{lemma} 
\begin{proof} 
We can do this by using either formula  $\Kt =(\Lt^{\rm shift})^\#$, or  $\Kt  =(\Lt^\#)^{\rm shift}$, where, recall $\Lt^{\rm shift}$ and $\Lt^\#$ are defined in \eqref{Lshift} and \eqref{Kt}, respectively. We proceed in the second way. We note that $\bar{\partial} \alpha = \DIV \vec \alpha - i\CURL \vec \alpha$ and $\curl\vec \al = \frac12 \p\al +\frac12 \bar\p \beta$. 
 (In addition, we have $\curl^*\curl = -\Delta + \n\divv$ and $(\n\divv\vec \al)^\C =\frac12 \p\bar\p\al +\frac12 \p^2 \bar \al$.) 

Now, using expressions, \eqref{Gam-Gam*} we compute 
\begin{align}\label{GamGam*}
	\G \G^* = & \left( \begin{array}{cccc}
|\psit|^2 & - \psit^2  & -i\psit \bar\p  & -i\psit \p  \\
-\bar{\psi}_{\tau}^2 & |\psit|^2 &  i\bar \psit \p   &  i\bar\psit \p  \\
  -i \p \bar\psit &   i\p\psit  & -\bar \p\p & -\p^2 \\
  -i \p \bar\psit &   i \p\psit  & -\bar\p^2 & -\bar \p\p
		\end{array} \right).
\end{align}
Since $\Kt  =(\Lt^\#)^{\rm shift}=\Lt^\# +\frac12 \G \G^*$, \eqref{Lt'}, \eqref{compl3}, \eqref{GamGam*} and the relation $-\bar \p\p=\Delta$ give \eqref{K}. 
\end{proof}  
%

\section{Proof of Proposition \ref{prop:Lk-spec'}}\label{sec:prop-Lk-spec-pf}
 Proposition \ref{prop:Lk-spec'} follows from   the fact that $\Lk^{\rm shift}$ is unitary equivalent to $\Kk$, see, in particular, relation  \eqref{specK-L} and the following two propositions: 
 \begin{proposition} \label{prop:Kk-spec-gen} 
  The operator $\Kk$ is self-adjoint and   has purely discrete spectrum, with all but four eigenvalues $\gs 1$. \end{proposition}

  \begin{proposition} \label{prop:Kk-spec}    Let $\kappa^2>\frac12$. Then  the four lowest eigenvalues, $\nu_{k}^i\equiv \nu_{k}^i(\tau), i=1, 2, 3, 4,$  of the operator $\Kk$ (the ones of the order $o(1)$), 
  are simple and of the form \eqref{nuk1-form} - \eqref{nuk12-est} and 
\begin{align} 
\DETAILS{ \label{nuk1-form'}&\nu_{k}^1  =c_1 (\tau) \frac{\e ^2|k|^{2}}{\e ^2 + |k|^{2}} [(\kappa^2-\frac12) \g_k (\tau)+ \eta(\tau, \kappa)\e^2] + O(\e^4 |k|^{2}),\\
\label{nuk2-form'}&\nu_k^2  \gs |k|^2\text{ and, if } |k|\ll 1,\  \nu_k^2= c_2 (\tau)  |k|^2+ O(\e^2|k|^2) ,\\
	\label{nuk12-est'}& |\p_k^\al\nu_k^i|  \ls |k|^{2-|\al|}, i= 1, 2, |\al|\le 2,\\}
	\label{nuk34-est}&\nu_k^i \ge c_i (\tau) \e^2 ,\ c_i (\tau)\gs 1,\  i=3, 4.  
		\end{align}
 \end{proposition}

 \DETAILS{First, we introduce the functions \eqref{v012} and \eqref{v034}  \begin{equation}
\label{v012}	 v_{k 1}:=(\phi_{k}, 0, 0, 0),\ \quad  v_{k 2}:=(0,  \rf\bar{\phi}_{k}, 0, 0),	
	\end{equation}
where $\phi_{k}\in \NULL (-\COVLAP{a^0}-1)$ 
are described in Proposition \ref{prop:phik}, 
 and   
\begin{equation}
	\label{v034}	  v_{k 3}:=(0, 0,  e^{i k\cdot x}, 0),\ \quad    v_{k 4}:=(0,  0, 0,  e^{i k\cdot x}),	
			\end{equation}
which turn out to be  eigenfunctions of the operator $ \Kk^0:= \Kk \big|_{\e=0}$. Next,}

By Lemma \ref{lem:K0k-spec} below, the  operator $ \Kk^0:= \Kk \big|_{\e=0}$ has purely discrete spectrum containing the eigenvalues $0$ and $|k|^2$, with the remaining eigenvalues $\gs 1$. 
 To determine the fate under the perturbation of these two eigenvalues, we use the Feshbach-Schur map argument (see e.g. \cite{BFS, GS2} and Supplement I). 

Let  $\Pk$ be the orthogonal projector, onto $\NULL \Kk^0 \oplus \NULL (\Kk^0 -|k|^2)$ and $\bar P_k:= \one - \Pk$, and let $\bar \Kk:=\bar \Pk \Kk \bar \Pk$.  by a standard perturbation theory, the operators $\bar \Kk-\lam$ are invertible  for $\lam <c$, where $c:=\frac14 \min \{1, |q|^2: q\in \LAT^*/\{0\}\}$, and $\e$ sufficiently small. Then the  Feshbach - Schur map
\begin{equation}\label{fesh-k}
\mathcal{F}_{ k }(\lambda) := \left[\Pk \Kk \Pk - \Pk \Kk \bar \Pk ( \bar \Kk - \lambda)^{-1}\bar{\Pk} \Kk \Pk\right]\big |_{\RANGE \Pk},
\end{equation}
is well defined for $\lam <c$ and $\e$ sufficiently small. Therefore, by the isospectrality theorem for the Feshbach-Schur maps (see \cite{BFS, GS2} and Supplement I),
\begin{equation}\label{isosp}
 	\lam \in \sigma(\Kk)\cap (-\infty, c)\  \textnormal{ if and only if  }\ \lam \in \sigma(\cF_k (\lam))\cap (-\infty, c).
\end{equation}
By Lemma \ref{lem:K0k-spec}, the functions $v_{j k}, j=1, 2, 3, 4,$ defined in  \eqref{v012} and \eqref{v034}, form the basis in $\Ran \Pk = \NULL \Kk^0 \oplus \NULL (\Kk^0 -|k|^2)$. 
We define the matrix elements 
\[f_{ij }= \lan v_{i k }, \mathcal{F}_{ k }(0)v_{j k }\ran,\ \forall i, j =  1, 2, 3, 4,\]
 of $\mathcal{F}_{ k }(\lambda)$ in this basis. 
  We begin with  a general result which gives $\nu_k^1$  in an essentially close form suitable for 
  for the computation of $\eta(\tau, \kappa)$ appearing in \eqref{nuk1-form}. 
  \begin{proposition} \label{prop:nuk12-exact}  Let $\kappa^2>\frac12$. Then, (i) for $|k|\ll \eps$, the operator $\Kk$ has only two eigenvalues $\ll \e^2$ and  these eigenvalues are simple and of the form 
  \begin{align}\label{nuk12-exact} &\nu^{i}_k 
   =c_{i} |k|^2 \tilde \g_k^i (\tau, \kappa, \e) 
   + O(\e^4 |k|^2 ),\  i=1, 2,\\
\label{tilde-gamk12}  & \tilde  \g_k^{1/2} (\tau, \kappa, \e) :=  a_k\mp |b_k| -|c_k|^2/d_k, \end{align} 
 where $c_{i}\gs 1$,  $a_k:=f_{ 11}/\e^2$, $b_k:=f_{ 12}/\e^2$ , $c_k:=f_{ 13}/(\e|k|)$ and $d_k:=f_{ 33}/|k|^2;$ 
 
 (ii) for  $|k|\gs \e$, the operator $\Kk$ has only one eigenvalue $\ll \e^2$ and  this eigenvalue is simple and of the form
  \begin{align} \label{nuk1-exact} \nu^{1}_k & =c \e^2 \tilde \g_k^1 (\tau, \kappa, \e),\ 
  c\gs 1, \end{align}
 if $| \tilde \g_k^1 (\tau, \kappa, \e)|\ll 1$, and has all eigenvalues $\gs \e^2$, otherwise.\end{proposition}
\begin{remark} \label{rem:remnu1-gen} (a) We can also relate the remainder in \eqref{nuk12-exact} to the matrix elements $f_{ ij}(\lam)$, if one wishes to go to this order. 

(b)   In the domain  $|k|\sim \e$, expressions \eqref{nuk12-exact} and \eqref{nuk1-exact} for $\nu_k^1$ match (the factor $\frac{\e ^2|k|^{2}}{\e ^2 + |k|^{2}}$ in \eqref{nuk1-form} interpolates between the factors $|k|^2$ and $\e^2$ in \eqref{nuk12-exact} and \eqref{nuk1-exact}).   \end{remark}  
\DETAILS{ We show in Lemma \ref{lem:Fk-expan4} below that 
  \begin{align}\label{ak'} 
	&  a_k =(\kappa^2  - \frac{1}{2}) a_k' +   |c_k'|^2+  a_k''\e^2 + O(\e^4),\     a_k':=2\lan|\phi_0|^2|\phi_k|^2\ran  -\beta (\tau),\\ 
\label{bk'} & b_k=(\kappa^2  - \frac{1}{2})b_k'+ b_k''\e^2 + O(\e^4),\ b_k':=\lan \phi_0^2 \bar{\phi}_{- k}\bar{\phi}_{k} \ran ,\\ 
\label{ck'} & c_k = c_k' +  c_k''\e^2 + O(\e^4),\     
c_k' = - e^{-|k|^2/2}\hat k,\\  
\label{dk'} &  
d_k=1  + O(\e^2   |k|^2). 
\end{align}
 This gives}
\eqref{nuk12-exact} and \eqref{nuk1-exact} give the exact form of the two gapless spectral branches of $\Lt$.  Lemma \ref{lem:Fk-expan4} below implies 
\begin{lemma} \label{lem:tilde-gam12-expan}  Let $\eta(\tau, \kappa) := \p^2_\e  \tilde \g_k^1 (\tau, \kappa, \e)\big|_{\e=0, k=0} $. Then
\begin{align}\label{tilde-gam1-expan} 
&\tilde \g_k^1 (\tau, \kappa, \e)= (\kappa^2-\frac12) \g_k(\tau)+ \eta(\tau, \kappa)\e^2+O(|k|^{2}\e^2)+O(\e^{4}),\\
\label{tilde-gam2-expan} &\tilde \g_k^2 (\tau, \kappa, \e)= 2 (\kappa^2-\frac12) \beta (\tau)+O(\e^{2})+O(|k|^{2}).\end{align} 
\end{lemma} 
Proposition \ref{prop:nuk12-exact} and Lemma \ref{lem:tilde-gam12-expan} imply the statements \eqref{nuk1-form}, \eqref{nuk2-form} and \eqref{nuk34-est} of Proposition \ref{prop:Kk-spec}. The statement \eqref{nuk12-est} is straightforward. This proves 
Proposition  \ref{prop:Kk-spec}.

\begin{proof}[Proof of Proposition \ref{prop:Kk-spec-gen}] 
The proof of self-adjointness of $\Kk$ and discreteness of its spectrum are standard. 

To prove the remaining statement, 
 we write $\Kk= \Kk^0 +O(\e)$, where $ \Kk^0:= \Kk \big|_{\e=0}$, and use the perturbation theory in $\e$.
As follows from the explicit expression \eqref{K}, 
the operator $ \Kk^0$ is given by 
\begin{equation}\label{K0}
	\Kk^0 = \left( \begin{array}{cccc}
		-\COVLAP{a^0} - 1 & 0 & 0 & 0 \\
		0 & -\overline{\COVLAP{a^0}} - 1 & 0 & 0 \\
		0 & 0 & -\Delta & 0 \\
		0 & 0 & 0 & -\Delta
		\end{array} \right),
\end{equation}
 defined  on $\cK:=L^2(\R^2;\C)^4$, with the periodicity 
 conditions \eqref{gauge-per-vk}. The next lemma describes the spectrum of $ \Kk^0$. 
   \begin{lemma}\label{lem:K0k-spec} The operator   $ \Kk^0 $ is non-negative, has purely discrete spectrum (which can be described explicitly) and its lowest eigenvalues are $0$ (due to  the operator $-\COVLAP{a^0}-1$) with the eigenfunctions 	
   \begin{equation}
\label{v012}	 v_{k 1}:=(\phi_{k}, 0, 0, 0),\ \quad  v_{k 2}:=(0,  \rf\bar{\phi}_{k}, 0, 0),	
	\end{equation}
where $\phi_{k}\in \NULL (-\COVLAP{a^0}-1)$ 
are described in Proposition \ref{prop:phik}, 
related to the global gauge zero modes,
 and  $|k|^2$ (due to  the operator $-\Delta$), with the eigenvectors  
 \begin{equation}
	\label{v034}	  v_{k 3}:=(0, 0,  e^{i k\cdot x}, 0),\ \quad    v_{k 4}:=(0,  0, 0,  e^{i k\cdot x}),	
			\end{equation} 
related to the translational zero modes.  The rest of the eigenvalues of $K^0$ are $\gs 1$. 
 \end{lemma}
\begin{proof}   
  The operator   $\Kk^0$ is a direct sum the operators  $-\COVLAP{a^0} - 1,\ -\overline{\COVLAP{a^0}} - 1$ and $-\Del$ on $L^2(\R^2;\C)$, with the periodicity 
 conditions $e^{-\frac{i}{2}s\wedge x}\phi(x+s) = e^{ik\cdot s}\phi(x)$  and $\alpha(x + t) = e^{ik\cdot t}\alpha(x)$, respectively. By a standard theory (see e.g. \cite{RSII, GS2}), the latter  operators are self-adjoint, and have discrete spectrum, and therefore so is and does the operator $\Kk^0$. The relation between the spectra of  $\Kk^0$ and  $\COVLAP{a^0}$ and $\Del$ is given by
 \begin{equation*}\sigma(\Kk^0) = \sigma(-\COVLAP{a^0} - 1) \cup \sigma(-\Del). \end{equation*}
In fact, the spectra of the operators $\COVLAP{a^0}$ and $\Del$ and therefore of $\Kk^0$ can be found explicitly and are given in the lemma below, which completes the proof of the proposition. \end{proof}
  \begin{lemma}\label{lem:specDelta}
	 For each $k\in \Om^*$, the operators  $-\COVLAP{a^0}$ and $-\Del$ are  self-adjoint with discrete spectra given by
	\begin{equation*}
		\sigma(-\COVLAP{a^0} - 1) = \left\{ 0, 2, 4, \dots \right\},\
		\sigma(-\Del) = \left\{|k + q|^2 : q \in \LAT^* \right\}.
	\end{equation*}
	Moreover, $\Null(-\COVLAP{a^0}-1)$  is spanned by the function $\phi_k$,  
	described in Proposition \ref{prop:phik},  while the eigenfunctions of $-\Del$ corresponding to the eigenvalues $|k + q|^2 : q \in \LAT^*$ are given by 
	\[e_{k+q }(x) = 	e^{i(k+q)\cdot x},\  q \in \LAT^*.\] 
	\end{lemma}
\begin{proof}
	To describe the spectra of $-\COVLAP{a^0}$ and $-\Del$, we first consider the operator $-\COVLAP{a^0}$ on $L^2(\Omega;\C)$ with boundary conditions,
	$e^{-\frac{i}{2}s\wedge x}\phi(x+s) = e^{ik\cdot s}\phi(x)$ 
(see \eqref{phik-eqs-tau}). In Subsection \ref{sec:phik-etc}, we obtained  the representation	 $-\COVLAP{a^0} - 1 = c^*c$, where $c$ and $c^*$ are   the 
 annihilation and creation operators,	 introduced in \eqref{annih-creat-ops} and satisfying the commutation relations 
		 $[c, c^*] = 2$. This representation implies that the spectrum of $-\COVLAP{a^0}$ consists of the simple eigenvalues $2 n + 1, n= 0, 1, 2, \dots$, with the eigenfunctions $(c^*)^n \phi_k$, where $\phi_k$ solves $c\phi_k=0$ subject to $e^{-\frac{i}{2}s\wedge x}\phi(x+s) = e^{ik\cdot s}\phi(x)$. By Proposition \ref{prop:phik}, the latter problem has a unique (up a multiplicative constant) solution $\phi_k$, which we normalize as $\lan |\phi_k|^2\ran=1$. 
		
	We now turn to  the operator $-\Del$ acting on $L^2_k(\Omega; \C)$, which is $L^2(\R^2;\R^2)$, with the periodicity 
	$\alpha(x + t) = e^{ik\cdot t}\alpha(x)$. Standard methods show that this is a non-negative self-adjoint operator with discrete spectrum. 
Using the  orthonormal basis in $L^2_k(\Omega;\C)$, given by 		$e_{k+q}(x) = e^{i (k+q)\cdot x},\  q \in \LAT^*,$ 
one can show that the spectrum of $- \Del$ 
consists of the eigenvalues $|k+q|^2$ 
with the eigenvectors $e_{k+q }(x) =e^{i(k+q)\cdot x},\   q \in \LAT^*$.
  \end{proof}
Now, using the standard perturbation theory, we conclude the proof of Proposition \ref{prop:Kk-spec-gen}. 
\end{proof}

\DETAILS{\begin{proof}[Proof of Proposition \ref{prop:Kk-spec-gen}(B)]  
 Estimate \eqref{vkj-est} follows by  the elliptic regularity argument. \end{proof}}

As was mentioned after Lemma \ref{lem:tilde-gam12-expan}, Proposition  \ref{prop:Kk-spec} follows from Proposition \ref{prop:nuk12-exact}. 
 Proposition \ref{prop:nuk12-exact} follows from  \eqref{isosp}  and the following result, together with the implicit function theorem,
  \begin{proposition} \label{prop:Fk-spec} 
  Let $\kappa^2>\frac12$. Then for each $ \lam \in (-\infty, c),$ the operator $\cF_k(\lam)$ has four eigenvalues, $\nu_{k}^i(\tau, \lam), i=1, 2, 3, 4,$ these eigenvalues are simple and satisfy $\p_\lam \nu_{k}^i(\tau, \lam)=O(\e^2)$. In addition, the first two satisfy $\nu_{k}^i \equiv \nu_{k}^i(\tau, 0)$ of the form \eqref{nuk12-exact} - \eqref{nuk1-exact}, 
   while the remaining eigenvalues $\gs \e^2$. 
   \end{proposition}

\subsection{Proof of Proposition \ref{prop:Fk-spec}}\label{sec:Pf-Fk-spec} 
Before proceeding to the proof of Proposition \ref{prop:Fk-spec}, we describe general properties of the family $\cF_k(\lam)$. The map $\cF_k(\lam)$ inherits the $\rho-$ and $T^{\rm refl}-$symmetries of the operator $\Kk$ (see \eqref{commut-Trefl-tilderho-rho-Kk} and \eqref{commut-L-Lk-rho}):  
\DETAILS{in any basis $\{v^0_{k 0},  v^0_{k 1},   v^0_{k 2},  v^0_{k 3}\}$ satisfying 
\begin{align}\label{vtilde-parity}
\Span\{\tilde v_{k j}\}  \text{ is invariant under }  \rho\  \text{ and is mapped to } \Span\{\tilde v_{-k j}\}  \text{ by }  T^{\rm refl}.\end{align}}
%
 \begin{lemma}\label{lem:F-gen} 
\begin{align} \label{cFk-refl-sym}
&   \rho \mathcal{F}_{ k }(\lambda) =  \mathcal{F}_{ k }(\lambda) \rho 
 \ \quad \text{ and }\ \quad 
  T^{\rm refl} \mathcal{F}_{ k }(\lambda) =  \mathcal{F}_{- k }(\lambda) 
  T^{\rm refl}.
\end{align}
\end{lemma}
 \begin{proof} \DETAILS{Since $T^{\rm refl}$ and $ \cC\cS$ commute with $\Lt$ so does $\rho$.  Since $\rho$ maps $\cHk$ into itself, it commutes also with $\Lk$.        
 Due to \eqref{vtilde-parity}, the operator $P_k$, and therefore $\bar P_k$, commutes with $\rho$. 
  This and the definition \eqref{fesh-k} imply the first relation in \eqref{cFk-refl-sym}.  
  
Furthermore, since $r e^{i g_s}= e^{i g_{-s}} r$ and $r e^{i k \cdot x}= e^{i (- k) \cdot x} r$, we have $T^{\rm refl} L_{ k } =  L_{- k } T^{\rm refl}$ and $T^{\rm refl} P_{ k } =  P_{- k } T^{\rm refl}$,  which implies the second relation in \eqref{cFk-refl-sym}.}
By the equations \eqref{commut-Trefl-tilderho-Lk} and \eqref{commut-L-Lk-rho},  $T^{\rm refl}$ and  $\rho$ map $\cKk$ into  $\cKk$ and  $\cK_{-k}$, respectively, and have the following commutation relations $T^{\rm refl} K_{ k } =  K_{- k } T^{\rm refl}$ and  $[\rho, \Kk]=0$. 
The fact that the operator $P_k$ commutes with $\rho$ and satisfies  $T^{\rm refl} P_{ k } =  P_{- k } T^{\rm refl}$ (clearly, $\bar P_k :=\one -P_k$ has the same properties) and the definition \eqref{fesh-k} imply \eqref{cFk-refl-sym}.  \end{proof}
\DETAILS{Since the maps $\rho$ and $T^{\rm refl}$ satisfy
	$\rho^2 = 1$ and $(T^{\rm refl})^2=1$, the eigenfunctions of $ \cF_k$ with simple eigenvalues must satisfy 
	 \begin{align} \label{rho-Trefl-vk'}\rho v_{ k}=\pm v_{ k},\  \text{ and }\ T^{\rm refl} v_{ k}=\pm v_{- k}.\end{align}}	


 Equation \eqref{cFk-refl-sym} and the relation $\lan \rho v, \rho w\ran = \overline{\lan v, w\ran}$ imply that the matrix $F_k(\lambda)$ of $\mathcal{F}_{ k }(\lambda)$ in an orthonormal basis  $v_{k 1},  v_{k 2},   v_{k 3},  v_{k 4}$, 
  satisfying
\begin{align}\label{v0-parity} \rho v_{k i/i+1} = v_{k i+1/i},\  i=1, 3,\ \quad 
T^{\rm refl} v_{k i}=  v_{- k i},\end{align} 
 has the following properties  
\begin{align}\label{Fk-rho'-rho''-sym} & \rho' F_k(\lam)=  F_k(\lam) \rho'  \ \quad \text{ and }\ \quad  F_k (\lambda) = F_{-k}(\lambda), 
\end{align}
where $\rho'$ 
is the matrix of $\rho$ 
in the basis used, 
\begin{align}\label{rho'}
& 
 \rho': =
  \mathcal{C}  \cS',\  
 \text{ where }\   \cS' := \left( \begin{array}{cccc}	-1 & 0 & 0  & 0 \\0 & -1 & 0  & 0 \\ 0 & 0  & 0 &  1  \\ 0 & 0  & 1 &  0\end{array} \right).  
  \end{align}

\DETAILS{ The general structure of the matrix $F_{ k }(\lambda):=(f_{ ij}(\lam))$, where  $f_{k ij}(\lambda):= \lan v^0_{k i}, \mathcal{F}_{ k }(\lambda)v^0_{k j}\ran,$ $\forall i, j = 0, 1, 2, 3$, can be deduced from the symmetries of $\cF_{ k }(\lambda)$ (or of $F_{ k }(\lambda)$) inherited from the operator $K$ 
   and the fact that the  map $\cF_{ k }(\lam)$ is hermitian 
	 giving}
 Eqs \eqref{Fk-rho'-rho''-sym} and \eqref{v0-parity} and the fact that the  map $\cF_{ k }(\lam)$ is hermitian imply	  (omitting the argument $\lam$): 
\begin{align}  \label{fij-prop1} & f_{k i2}=-\overline{f_{ki3}}, i=0, 1,\   
   f_{k 22}= f_{k 33},\\   
      \label{fij-prop2} &
  f_{k i j}=\overline{f_{k j i}},\  \forall i, j,\\  
   \label{fij-prop3} &  f_{k i j}={f_{- k i j}},\  \forall i, j. \end{align}

Recall that $\Ran P_k=$ span $\{v_{j k}, j=1, 2, 3, 4\}$. Note that $v_{k j}$ given in Lemma \ref{lem:K0k-spec} satisfy \eqref{v0-parity} and therefore Eqs \eqref{fij-prop1} and \eqref{fij-prop3} hold.  
Introduce the matrix elements 
\[f_{ij }(\lam)= \lan v_{i k }, \mathcal{F}_{ k }(\lam)v_{j k }\ran,\ \forall i, j =  1, 2, 3, 4,\]
 of $\mathcal{F}_{ k }(\lambda)$ in the basis $\{v_{j k}, j=1, 2, 3, 4\}$.
%
The following lemma  plays a key role in our analysis. 
\begin{lemma}\label{lem:Fk-expan4} 
Let $ k:=k_1+i k_2,\ \hat k:= k/ |k|$ and $\lam=O(\e^2)$. Then the matrix elements $f_{ij }(\lam),\ \forall i, j =  1, 2, 3, 4,$ are of the form (omitting the argument $\lam$)
\begin{align}\label{f12}
	& f_{ 11 } =  f_{ 22} =   \e^2  a_k,\      f_{ 12} = \bar f_{ 21} =   \e^2  b_k, \\
		\label{f13} & 
 f_{ 13} =    f_{ 24}  
  = 0,\	 f_{ 14}  = - \bar f_{ 23} 
  =   \e |k| c_k,\\   
\label{f33}
&f_{ i j  } =
\big(1  + \e^2   f(|k|^2) \big) |k|^2\del_{ij}   + O(\e^4|k|^2 ),\   f(|k|^2)=O(1),\   i, j=3, 4, \\  
\label{ak} 
	&  a_k =(\kappa^2  - \frac{1}{2}) a_k' +   |c_k'|^2+ O(\e^2),\             a_k':=2\lan|\phi_0|^2|\phi_k|^2\ran  -\beta (\tau),\\ 
\label{bk} & b_k=(\kappa^2  - \frac{1}{2})b_k'+ O(\e^2),\ b_k':=\lan \phi_0^2 \bar{\phi}_{- k}\bar{\phi}_{k} \ran ,\\ 
\label{ck} & c_k = c_k' + O(\e^2),\     
c_k' :=- \lan \overline{\phi_{k}}\phi_0 e^{i k\cdot x}\ran \bar{\hat k} =- e^{-|k|^2/2}\bar{\hat k}.
\end{align}
    Furthermore,  $a_k,  b_k, c_k $ (which should not be confused with the similar objects introduced in Proposition \ref{prop:nuk12-exact})  are even in $k$ and
\begin{align}\label{dfij}\p_\lam f_{ i j}(\lam)= O(\e^{2}), \forall i, j. 
\end{align} 
\DETAILS{ and 
 \begin{align}       \label{tk}   t_k:=   \lan \bar\phi_0\phi_k e^{-i k\cdot x} \ran= e^{-|k|^2/2}. 
    \end{align}}   
    \end{lemma} 
    This lemma is proven in Appendix \ref{sec:Fk-expan4} below. 
 It shows that (the first statement follows, in fact, from Lemma \ref{lem:Fk-sym})
\begin{corollary}\label{cor:Fk-k0} For $k=0$, 
 (a)  the matrix $F_{ k=0}(\lambda)$ is block-diagonal with the $22-$block being identically zero and  the $11-$block having non-zero eigenvalues; (b) the $11-$block of the matrix $F_{ k=0}(\lambda)$ is invertible with a uniform in $\e$ bound. 
\end{corollary}


 \begin{proof}[Proof of Proposition \ref{prop:Fk-spec}] We have to find the eigenvalues of the $4\times 4$ matrix $F_k(\lam)$. 
To this end, we consider separately three regimes:
\begin{align} \label{3regimes} & 
 (i) \quad |k|\ll \eps,\   \quad  (ii) \quad \eps \ls  |k|\ls \eps,\   \quad  (iii) \quad \ |k|\gg \eps. \end{align}
In each of these regimes we apply the Feshbach-Schur methods but with a different projection. In each case we exploit the separation of the eigenvalues into two groups.

\medskip


We begin with the domain $|k|\ll \eps$ and the following precise result  (cf. Proposition \ref{prop:nuk12-exact})   \begin{proposition} \label{prop:nuk12-exact'}  Let $\kappa^2>\frac12$ and $|k|\ll \eps$. Then the two lowest eigenvalues, $\nu_{k}^i( \lam)$, of the matrix $F_k(\lam)$ satisfy $\p_\lam \nu_{k}^i(\lam)=O(\e^2)$, with $\nu_{k}^i \equiv \nu_{k}^i(0)$ of the form \eqref{nuk12-exact} - \eqref{tilde-gamk12}.   The remaining eigenvalues $\gs \e^2$.
 \end{proposition}
\begin{proof} 
 We  write the matrix  $F_k (\lam):=(f_{ ij}(\lam))$ 
   in the block form
\begin{align}\label{F-matr} 
	 F_{ k }(\lambda)&:= \left( \begin{array}{cc}     F_{ k 11}(\lambda) &  F_{k12}(\lambda)\\
        F_{k 21}(\lambda)   &   F_{ k 22}(\lambda)   \end{array} \right), 
	\end{align} 
where $F_{k ij}(\lambda)$ are $2\times 2-$matrices.  By Lemma \ref{lem:Fk-expan4},   
  the matrix $F_{ k 11}(\lambda)/\e^2$  has  the eigenvalues 
  \[\ak\pm |\bk| =(\kappa^2  - \frac{1}{2}) (\ak'\pm |\bk'|) +   |\ck'|^2+ O(\e^2)\]\[=(\kappa^2  - \frac{1}{2})\g_k^{1/2}(\tau) +   |\ck'|^2+ O(\e^2),\]
 where \begin{align}\label{gamk-pm}\g_k^{1/2}(\tau):= 2\lan|\phi_0|^2|\phi_k|^2\ran  -\beta (\tau) \pm \lan \phi_0^2 \bar{\phi}_{- k}\bar{\phi}_{k} \ran. \end{align} 
  Note that $\g_k^1 (\tau)= \g_k (\tau)$, where $\g_k (\tau)$ is defined in \eqref{gamk-tau}. {\it We assume for simplicity} 
 \begin{itemize} \item $ 
(\kappa^2  - \frac{1}{2})\g_k (\tau) +   e^{-|k|^2}>0$, \end{itemize}
which is surely satisfied in the regime we are interested in. Then 
 $F_{ k 11}(\lam) \gs \e^2$ and, for $\nu \ll \eps^2$, we can use  the Feshbach-Schur map with the projection on the $22-$block.	
This 
reduces the problem to finding the eigenvalues of the matrix 
\begin{align}\label{cV22}\mathcal{V}_{ k 22}(\lam):=F_{ k 22}(\lam) -F_{ k 21}(\lam) F_{ k 11}(\lam)^{-1} F_{ k 12}(\lam).\end{align} 

By Lemma \ref{lem:Fk-expan4}, the blocks $F_{k ij}(\lambda)$ are of the form 

\begin{align}\label{F11exp}
	& F_{ k 11}(\lambda) =  \e^2 
A_k,\ A:= \left( \begin{array}{cc}   a_k
	&  b_k\\  \overline{b_k}
             &  a_k   \end{array} \right),\\ 
             	\DETAILS{	\mathcal{F}_{2 k}:=(\kappa^2  - \frac{1}{2}) & \left( \begin{array}{cc}   2\lan|\phi_0|^2|\phi_k|^2\ran_\Om -\beta (\tau)
	&  \lan \phi_0^2 \bar{\phi}_{- k}\bar{\phi}_{k} \ran_\Om\\   \lan \bar\phi_0^2  {\phi}_{- k} {\phi}_{k} \ran_\Om
             &  2\lan|\phi_0|^2|\phi_k|^2\ran_\Om-\beta (\tau)   \end{array} \right) +
              |\lan \bar\phi_0\phi_k e^{-i k\cdot x} \ran_\Om|^2 \one ,} 
	\label{F12exp} & F_{ k 12}(\lambda)= F_{ k 21}^*(\lambda)=\e  |k| C,\  
	C:=\left(\begin{array}{cc}  	0 &   c_k\\  -  \overline{c_{k}}  & 0 \end{array} \right),\\ 
\label{F22exp}
&F_{ k 22}(\lambda) 
 = D  |k|^2,\  D:=d_k  \one + O(\e^4),\\
  \label{dk}
& d_k:=1  +   O(\e^2). 
 \end{align} 
Recalling  \eqref{cV22}, we compute $A^{-1}= \frac{1}{\det}\left( \begin{array}{cc}   a_k & - b_k\\  -\overline{b_k}    &  a_k   \end{array} \right),$  where $\det:=\det A = a_k^2-|b_k|^2$, and
$    C^* A^{-1} C=\frac{1}{\det}\left( \begin{array}{cc}   a_k  |c_k|^2	& - \overline{b_k}  c_k^2  \\  -b_k \overline{c_k^2 }             &  a_k  |c_k|^2 \end{array} \right) |k|^2.$
Keeping in mind  \eqref{F11exp} - \eqref{F22exp}, we insert this into \eqref{cV22} to obtain
\[\cV_{ k 22}(\lam) = 
D |k|^2 - \frac{1}{a_k^2-|b_k|^2}\left( \begin{array}{cc}   a_k  |c_k|^2	& - \overline{b_k}  c_k^2 \\  -b_k \overline{c_k^2 }             &  a_k  |c_k|^2  \end{array} \right) |k|^2.\]
Recalling the expression for $D$ in \eqref{F22exp} and \eqref{dk}, we compute the eigenvalues of this matrix and set $\lam=0$ to obtain  \eqref{nuk12-exact} - \eqref{tilde-gamk12}.    The remaining eigenvalues $\gs \e^2$.

Estimate \eqref{dfij} shows that $\p_\lam \nu^{ i }(\lam)= O(\e^{2})$.\end{proof}
Applying the implicit function theorem proves the first part of Proposition \ref{prop:nuk12-exact}.

Now we consider the domain  $|k|\gs \eps$. For simplicity, we assume $|k|\ll 1$. The case $|k|\gg \eps$ is much simpler, in this case, we use the Feshbach-Schur map with the orthogonal projection on the $11-$block, since the block $\mathcal{F}_{ k 22}(\lam)$ has the eigenvalues $ \gs |k|^2 \gg \e^2$ and therefore is invertible (for $\nu \ll \eps^2$).
The domain $\eps \sim  |k|$ is most cumbersome as here, three of the eigenvalues are expected to be of the same order, $O(\e^2)$ (while the fourth one is  $\ll \e^2$).
	 
\begin{remark} \label{rem:tech1}  Under the condition $a_k'  + \re b_k' \gs 1$, where $a_k' , b_k' $ are defined in \eqref{ak} and \eqref{bk}, which can be verified numerically, the treatment below can be extended to the domain $\eps \ls  |k|\ls 1$.
\end{remark}  


  \begin{proposition} \label{prop:nuk1-exact'}  Let $\kappa^2>\frac12$ and $|k|\gs \eps$. Then the lowest eigenvalue, $\nu_{k}^1 ( \lam)$, of the matrix $F_k(\lam)$ satisfy $\p_\lam \nu_{k}^1(\lam)=O(\e^2)$, with $\nu_{k}^1 \equiv \nu_{k}^1 (0)$ of the form \eqref{nuk1-exact}.  The remaining eigenvalues $\gs \e^2$. \end{proposition}
\begin{proof} 

Thus we consider the domain $\eps \ls  |k|\ll 1$. In this case, we use  the Feshbach-Schur map with the orthogonal projection on the vector $\frac{1}{\sqrt 2} (1, -1, 0, 0)$ 
 (related to  the eigenvectors of $\nu^1 _k$).

 It is a straightforward but tedious computation to show that the complementary $3\times 3$ matrix 
 $ \gs \e^2$. To do this, we transform the matrix $F_k$ as $V F_k V=: F'$, where $V$ is the orthogonal matrix given by
 \[V:= 
 \left( \begin{array}{cccc}
		1/\sqrt 2 & 1/\sqrt 2 & 0 & 0 \\
		1/\sqrt 2 & -1/\sqrt 2 & 0 & 0 \\
		0 & 0 & 1 & 0 \\
		0 & 0 & 0 & 1
		\end{array} \right).\]
As a result, we obtain 
 \[ F' = \left( \begin{array}{cccc}
		\al & \beta & -\bar\s & \s \\
		\bar \beta & \mu & \bar\s & \s  \\
		-\s  & \s  & \del & 0 \\
		\bar\s  &\bar\s  & 0 & \del
		\end{array} \right),\]
where (not to confuse the entry $\beta$ used only in this proof with the Abrikosov function $\beta(\tau)$)
\begin{align} \label{al-bet} 
 &\al:= f_{11} + \re f_{12},\  \beta:= - i \im f_{12},\ \mu:=  f_{11} - \re f_{12},\ \s:= f_{14}/\sqrt 2,\ \del:=  f_{33}. \end{align}
Now, we apply to $F'$ the Feshbach-Schur map with the orthogonal projection, $Q$, on the vector $ (0,  1, 0, 0)$. To show that the  Feshbach-Schur map is well defined we have to check that $\bar F':= \bar Q F' \bar Q\big|_{\Ran \bar Q}$ is invertible. 
 Clearly,
 \[\bar F' = \left( \begin{array}{ccc}
		\al &  -\bar\s & \s \\
		-\s    & \del & 0 \\
		\bar\s    & 0 & \del
		\end{array} \right).\]
The determinant of this matrix is $\det \bar F'= \del(\al \del - 2|\s|^2)$ and the eigenvalues are 
\[\lam_\pm := \frac12 (\al +\del)\pm \sqrt{\frac14 (\al +\del)^2 - (\al \del - 2|\s|^2)}\]\[\ge \frac{\al \del - 2|\s|^2}{\al +\del}.\]
Remembering the definitions of $\al,\ \s,\ \del$ above and using \eqref{ak} - \eqref{ck}, we find in the leading order in $|k|$, for  $|k|\ll 1$,
 \[\lam_+> \lam_- \ge \e^2 |k|^2 (\kappa^2  - \frac{1}{2})\beta/\xi+O( \e^4)+O( \e^2 |k|^2),\] 
 where $\xi:=\e^2 ( (\kappa^2  - \frac{1}{2})2 \beta +1) + |k|^2$.  Hence, for $\kappa > 1/\sqrt {2}$ and $|k|\gs \e$,  we have $\lam_\pm \gs \e^2$. Therefore $\bar F' -\nu$ is invertible for $\nu \ll \e^2$ and we can apply to  $\tilde F$ the Feshbach-Schur map with the orthogonal projection, $Q$, on the vector $ (0,  1, 0, 0)$. This gives that $F'$ has three eigenvalues satisfying $\gs \e^2$ and one, $\nu_k$, satisfying 
\begin{align} \label{nu}\nu = \mu -  \bar w \cdot (\bar F' -\nu)^{-1} w = \mu-\mu' +O(|\mu'|\nu),\end{align} 
where $\mu':= \bar w \cdot  \bar (F')^{-1} w$ and $w:=(\beta, \s, \bar\s)$. We compute
 \[\bar F' = \frac{1}{\del \al -2|\s|^2}\left( \begin{array}{ccc}
		\del &  \bar\s & - \s \\
		\s    & \al - |\s|^2/\del & - \s^2/\del \\
		- \bar\s    &  - \bar\s^2/\del & \al - |\s|^2/\del
		\end{array} \right),\]
which gives
\begin{align} \label{mu'}
&\mu' =\frac{1}{\del \al -2|\s|^2} (\del |\beta|^2 +2 \al |\s|^2 -4 |\s|^4/\del). 
\end{align}
Combining \eqref{nu} and \eqref{mu'}, recalling the definition of $\al,\  \beta,\ \mu,\ \s,\ \del$ in \eqref{al-bet} and $a_k$, $b_k$ , $c_k$ and $d_k$ in Proposition \ref{prop:nuk12-exact'} and  Eq. \eqref{dk} and omitting the subindex $k$, we find
\begin{align} \label{nu0}
\notag  \nu =\e^2[a-\re b &-\frac{|c|^2}{d}-\frac{|\im b|^2}{a+\re b -|c|^2/d}] +O(\e^2)\nu\\
& =\e^2\frac{(a -\frac{|c|^2}{d})^2-|b|^2}{a+\re b -|c|^2/d} +O(\e^2)\nu.\end{align}
Comparing this with the definition of $\tilde \g_k^1 (\tau, \kappa)$ in \eqref{tilde-gamk12}, we arrive at
\begin{align} \label{nu} \nu_k & =\e^2c \tilde \g_k^1 (\tau, \kappa) +O(\e^2)\nu_k,\end{align}
where $c:=\frac{a -\frac{|c|^2}{d} +|b|}{a+\re b -|c|^2/d}$.  Using  \eqref{ak} - \eqref{ck}, we find $c\gs 1$.
This gives \eqref{nuk1-exact}.
  
 Moreover, $\nu_k\ll \e^2$, by the condition $|k|\ll 1$, and it is the only eigenvalue which is $\ll \e^2$, the remaining three eigenvalues are $\gs \e^2$. \end{proof}
  This proves the second part of Proposition \ref{prop:nuk12-exact} 
  and with it 
 Proposition \ref{prop:Kk-spec} (see the paragraph after Lemma \ref{lem:tilde-gam12-expan}). 
\end{proof}



\subsection{Proof of 
Lemma \ref{lem:Fk-expan4}}\label{sec:Fk-expan4}

We begin with a general result. We write $P=P_1\oplus P_2$, where $P_1$ and $P_2$  are the orthogonal projections onto the subspaces 
$\NULL K^0_{k }$ and $\NULL (K^0_{k }-|k|^2)$, respectively, 
	   and write the operator $ \cF_{ k }(\lambda)$ 
	   in the block form
\begin{align}\label{cF-matr} 
	 \cF_{ k }(\lambda)&:= \left( \begin{array}{cc}    \cF_{ k 11}(\lambda) &  \cF_{k12}(\lambda)\\
          \cF_{k 21}(\lambda)   &   \cF_{ k 22}(\lambda)   \end{array} \right), 
	\end{align}
where $\cF_{k ij}(\lambda):= P_i \cF_{k}(\lambda) P_j$. 
 \DETAILS{Specifically,  the  $11-$block is defined as  $F_{ k 11}(\lambda) 
 =\big[ \lan v^0_{k i}, \mathcal{F}_{ k }(\lambda)v^0_{k j}\ran\big]$, etc.} 
These blocks have the following general properties
\begin{lemma}\label{lem:Fk-sym}
\begin{align}\label{F-k-sym} & \mathcal{F}_{ -k ii}(\lambda) =   \mathcal{F}_{ k ii}(\lambda),\  \quad
i=1, 2,\  \quad \mathcal{F}_{ -k 12}(\lambda) = - \mathcal{F}_{ k 12}(\lambda),\\  
\label{F22-k0}&\mathcal{F}_{ k=0 22}(\lambda=0)=0. 
\end{align} 
\end{lemma}   
 \begin{proof} 
To prove \eqref{F-k-sym}, we recall that $\rf$ is the reflection operator on scalar functions, changing $x$ to $-x$ (so that $T^{\rm refl}=\rf \oplus  \rf \oplus - \rf \oplus - \rf$) and let \[R=\rf \oplus  \rf \oplus  \rf \oplus  \rf.\] We write $K=K^{\rm even}+W^{\rm odd}$, where  $K^{\rm even}$ and $W^{\rm odd}$ satisfy  
\[R K^{\rm even}=K^{\rm even} R\ \text{ and }\ R W^{\rm odd}=- W^{\rm odd} R,\] and use the expansion (for $\lambda =o(1)$) 
\begin{align}\label{Rbar-expan} 
\bar P ( \bar P K_{ k }\bar P  - \lambda)^{-1}\bar{P}=\sum_{n=0}^\infty  \bar R_{ k}^{\rm even}(\lambda) [- \bar R_{ k}^{\rm even}(\lambda)\ W^{\rm odd}]^n \bar P ,\end{align}
  where $ \bar R_{ k}^{\rm even}(\lambda):=\bar P ( \bar P K_{ k }^{\rm even}\bar P - \lambda)^{-1}\bar{P}$.

Let $\pi_1$ and $\pi_2$ be the projections onto the upper and lower two entries of $v$, respectively. We show below that the projections $P_j$ have the following properties: 
\begin{align}\label{Pj-prop} \pi_i P_j =0,\ i \ne j\ \text{ and }\ R P_i = (-1)^{i-1} P_i, i=1, 2. \end{align}

To show the the first relation in \eqref{F-k-sym}, we use the fact that $W^{\rm odd}$ maps the first two components into the last two and vice versa, while $L^{\rm even}$ maps the first two components into the first two components and  the last two, into the last two, to obtain 
\begin{align}\label{odd-even-prop} \pi_i W^{\rm odd}\pi_i =0\ \text{ and }\ \pi_i  \bar R_{ k}^{\rm even}(\lambda)=\pi_i  \bar R_{ k}^{\rm even}(\lambda)\pi_i,\ i=1, 2.\end{align}
Then we use the expansion \eqref{Rbar-expan} in the definition of $\mathcal{F}_{ k 11}(\lambda)$ and use that 
for $n$ even, the properties \eqref{Pj-prop}, the relation $\pi_i W^{\rm odd}[- \bar R_{ k}^{\rm even}(\lambda)\ W^{\rm odd}]^n \pi_i=0$ 
(which follows from \eqref{odd-even-prop}), to find
\begin{align}\label{Fiiexpan}\mathcal{F}_{ k ii}(\lambda) =
  P_i [ K_{ k}^{\rm even} +\sum_{n\ \text{odd}}  W^{\rm odd}[- \bar R_{ k}^{\rm even}(\lambda)\ W^{\rm odd}]^n] P_i. \end{align}
Using the reflection symmetries mentioned above and the relations 
$R P_{1 k} = P_{1 -k}$ and $R P_{2 k} = - P_{2 -k}$, we obtain the first relation in \eqref{F-k-sym}. 

 To prove the second relation in \eqref{F-k-sym}, we use \eqref{Rbar-expan}, \eqref{Pj-prop} and \eqref{odd-even-prop} to expand   
\begin{align}\label{Fijexpan}\mathcal{F}_{ k ij}(\lambda) =\sum_{n\ \text{even}}  
P_i W^{\rm odd}[- \bar R_{ k}^{\rm even}(\lambda)\ W^{\rm odd}]^n P_j,\ i \ne j. \end{align}
 Since moreover $R W^{\rm odd} =- W^{\rm odd} R $ and $R \bar R_{ k}^{\rm even}=\bar R_{- k}^{\rm even} R$, we conclude from \eqref{Fijexpan} that the second relation in \eqref{F-k-sym} holds.

 To show 
  \eqref{F22-k0}, we observe that by the second relation in \eqref{F-k-sym},   the operator $\mathcal{F}_{ k= 0}(\lambda)$ is block - diagonal, $\mathcal{F}_{ k= 0}(\lambda)= \mathcal{F}_{ k= 0 11}(\lambda)\oplus \mathcal{F}_{ k= 0 22}(\lambda)$.  
  
  Furthermore, by   \eqref{Tj-transf}, 
  the operator $K_{k=0}$ has the doubly-degenerate eigenvalue $0$ with the eigenfunctions satisfying  $v^j_{0  }=v_{0 j}+O(\e), j=3, 4$, where $v_{k j }$ satisfy $\pi_1 v_{k j } = 0$. 
 Therefore the operator $\mathcal{F}_{ k=0}(\lambda=0)$ has also the doubly degenerate zero eigenvalue and the corresponding eigenfunctions are  $P_{ k= 0} v^j_{0  }=v_{0 j}+O(\e), j=3, 4$. Hence, since, by \eqref{F-k-sym}, $\mathcal{F}_{ k=0}(\lambda)$ is block diagonal, the $22-$block 
   of $\mathcal{F}_{ k=0}(\lambda=0)$ vanishes identically, 
  $\mathcal{F}_{ k=0, 22}(\lambda=0)=0$, and therefore \eqref{F22-k0}   holds.  \end{proof} 

 From now on we consider only fiber operator for a fixed $k-$fiber 
  and so we omit the subindex $k$ and write e.g. $K\equiv \Kk,\ P\equiv P_k,\ \cF\equiv \cF_k,\ F\equiv F_k,\ v_{  j} \equiv v_{ k j}$. 
(Recall that  the  vectors $v_{K j}, j=1, 2, 3, 4,$ 
 are defined in \eqref{v012} 
  and \eqref{v034}.)  We also omit the subindex $\LAT$ in the expectations and the inner products over a fundamental cell $\Om$.

In what follows, the magnetic potentials, $a$ are assumed to be in their {\it complex form}, $a=a_1-i a_2$. We derive the expansion of $\Kt$ in $\e$.  To this end, we use expansion of the solution branch $(\psit, \at, \lamt)$ 
obtained in \cite{TS}
\begin{equation}\label{psi-a-expan} 
	\begin{cases}
	\psit = \e \psi^0 
	+ O(\e^3), \\
	\at = a^0 + \e^2 a^1 + O(\e^4), \\
\lamt = 1 + \e^2 \lambda^1 + O(\e^4),
	\end{cases}
\end{equation}
where the first two remainders are in the sense of the norm in $H^2(\Om, \C)$ 
 and $a^0:=- \frac{1}{2} i (x_1-i x_2)$ (comming from $\frac{1}{2} x^\perp)$,  $\psi^0$ satisfies 
 \eqref{phik-eqs-tau} with $k=0$ and is normalized as $\lan |\psi^0|^2 \ran  = 1$, so that
  $\psi^0=\phi_0$ (see Proposition \ref{prop:phik}),  
 and $a^1$ satisfies 
\begin{equation} \label{curla1}
	i\bar\p a^1	= \frac{1}{2} (\lan |\phi_0|^2 \ran - |\phi_0|^2),\   \Delta a^1 = \frac{i}{2}\bar\phi_0 (\partial_{a^0}\phi_0),\end{equation}
\begin{equation}
  a^1(x+s) = a^1(x),\ \forall s\in \cL,
\end{equation}
and $\lam^1$ is given by
   \begin{equation} \label{lam1}
\lambda^1=\left[\frac{1}{2} 
+ \big(\kappa^2-\frac{1}{2}\big)\beta(\tau)\right]\langle |\phi_{0} |^2 \rangle .
\end{equation}
(The latter follows from the definitions $\lamt:=\frac{\kappa^2}{b}$ and \eqref{eps}.)

The expansion \eqref{psi-a-expan}, the explicit expression for the operator $K $, given in \eqref{K}, and the relation $\partial_{a^0}^*\phi_0 = 0$ imply 
  the following expansion of the operator $K$, 
\begin{equation}\label{K-expan} K  = K^0  + W^{\rm even}+ W^{\rm odd},\   W^{\rm odd}= \e W^1   + O(\e^3),\   W^{\rm even}= \e^2 W^2 + O(\e^4),\end{equation} 
 where  $O(\e^j)$ is understood in the sense of the operator norm, 
 $K^0$ is given in \eqref{K0} and the remaining terms are given by 

\begin{align} \label{W1}
&	W^1
	= \left( \begin{array}{cccc}
		0 & 0 & 0 & i(\partial_{a^0}\phi_0) \\
		0 & 0 & -i(\overline{\partial_{a^0}\phi_0}) & 0 \\
		0 & i(\partial_{a^0}\phi_0) & 0 & 0 \\
		-i(\overline{\partial_{a^0}\phi_0}) & 0 & 0 & 0
		\end{array} \right),\\
\end{align} \begin{align}
\label{W2}
	& W^{2} = \left( \begin{array}{cccc}
			B^0 -\lambda^1 & (\kappa^2 - \frac{1}{2})\phi_0^2 & 0 & 0 \\
			(\kappa^2 - \frac{1}{2})\bar{\phi}_0^2 & \overline{B}^0 -\lambda^1 & 0 & 0 \\
			0 & 0 & |\phi_0|^2 & 0 \\
			0 & 0 & 0 & |\phi_0|^2
		\end{array}	\right).
\end{align}
\DETAILS{
\begin{equation} \label{W3-2}
	W^3
	= \left( \begin{array}{cccc}
		0 & 0 & 0 & \tilde\phi_1\\
		0 & 0 &  \overline{\tilde\phi_1} & 0 \\
		0 & \tilde\phi_1 & 0 & 0 \\
		\overline{\tilde\phi_1}  & 0 & 0 & 0
		\end{array} \right).
\end{equation}
\begin{equation}\label{W4}
	W^{4} = \left( \begin{array}{cccc}
			B^1 -\lambda^2 & (\kappa^2 - \frac{1}{2})2\phi_0\phi_1 & 0 & 0 \\
			(\kappa^2 - \frac{1}{2})2\bar{\phi}_0\bar \phi_1 & \overline{B}^1 -\lambda^2 & 0 & 0 \\
			0 & 0 & 2\re (\bar\phi_0\phi_1) & 0 \\
			0 & 0 & 0 & 2\re (\bar\phi_0\phi_1)
		\end{array}	\right).
\end{equation} }
Here $\p_a$ is defined in \eqref{annih-creat-ops}, $\lam^1$ is given by \eqref{lam1} with $\lan |\phi_0|^2 \ran = 1$,
 and
\begin{equation}
	B^0 =  (2\kappa^2 + \frac{1}{2})|\phi_0|^2 - ia^1 \partial^*_{a^0}+ i\bar{a}^1 \partial_{a^0}. 
\end{equation}

\DETAILS{\paragraph{Remark.}  The operators $K^{\rm even} := K^0  + W^{\rm even}$ and $W^{\rm odd}$ are given explicitly as 
\begin{align}\label{K}
	&K^{\rm even}= \left( \begin{array}{cccc}
	\hta'   & (\kappa^2 - \frac12)\psit^2  & 0 &0 \\
(\kappa^2 - \frac12)\bar{\psi}_{\tau}^2 & \overline{\hta' } & 0 &0 \\
0 & 0 & \hta & 0 \\
		0 & 0			& 0 & \hta
		\end{array} \right),\\ 
&	W^{\rm odd} = \left( \begin{array}{cccc}
	0   & 0  & -i(\partial_{\at}^*\psit) & i(\partial_{\at}\psit) \\
0 & 0 & -i(\overline{\partial_{\at}\psit}) & i(\overline{\partial_{\at}^*\psit}) \\
i(\overline{\partial_{\at}^*\psit}) & i(\partial_{\at}\psit) & 0 & 0 \\
		-i(\overline{\partial_{\at}\psit})
			& -i(\partial_{\at}^*\psit)
			& 0 & 0
		\end{array} \right).
		\end{align} 
 
\medskip}

We begin with the block $F_{11} (\lambda)= (\lan v_{ i}, \cF (\lambda) v_{  j}\ran, i, j =1, 2).$ 
Denote $\bar  K^0:= \bar P  K^0\bar P $. By the spectral theorem, $\|(\bar  K^0- \lambda)^{-1}\| \lesssim 1$, provided $\lambda =o(1)$. Using this relation, 
   and the relations $K^{\rm even}= K^0 + W^{\rm even}$ and \eqref{K-expan} and assuming $\lambda =O(\e)$, we expand 
\begin{align}\label{Rbar0-expan}\bar P ( \bar K^{\rm even}  - \lambda)^{-1}\bar P  =\bar P ( \bar  K^0)^{-1}\bar P +O(\e^2).\end{align}
Furthermore, since $P$ is an eigenprojection of $K^0$, we have 
\[P K^{\rm even}  P= P (K^0  + \e^2 W^2)  P  + O(\e^4),\ P W^{\rm odd}  \bar P= \e P W^1  \bar P  + O(\e^3).\] This together with \eqref{Fiiexpan} gives
  \begin{equation}\label{Fiiexpan'}
	\mathcal{F}_{ii} (\lambda) =  P_i (K^0  + \e^2 W^2) P_i - \e^2 P_i W^1  \bar P  ( \bar  K^0)^{-1}\bar P W^1 P_i  +O(\e^4). \end{equation}
\DETAILS{By \eqref{Rbar-expan} and \eqref{K-expan}, we have the resolvent expansion
\begin{align}\label{resolv-expan}\bar R_{ k}^{\rm even}(\lambda) &=\bar R_{ k 0}(\lambda)  -\bar R_{ k 0}(\lambda)  W^{\rm even}\bar R_{ k 0}(\lambda)   +O(\e^4)\notag \\ & 
=\bar R_{ k 0} (\lambda) -\bar R_{ k 0}(\lambda)  \e^2 W^2\bar R_{ k 0}(\lambda)  +O(\e^4),\end{align} 
Using this expansion and assuming that $\lambda = o(1)$ and using  
 \eqref{Lv-exp2},  we find, similarly to \eqref{F11exp'}, 	
\begin{align}\label{F0ij} \lan w_{ k i}^0, F^0(\lambda) w_{ k j}^0\ran &= |k|^2 \del_{i j} +\e \lan w_{k i}^0, W^1 w_{k j}^0 \ran  
  -\e^2 \lan w_{ k i}^0, W^2 w_{ k j}^0\ran \notag\\ & -\e^2  \lan w_{ k i}^0, W^1\bar R_{ k 0}  (\lambda) W^1 w_{ k j}^0\ran + O(\e^4|k|^2). 
	\end{align}
Since, as evidenced by the definition \eqref{W1}, the operator $W^1$  switches the first two entries of the vectors it acts on with the last two and since the first two entries of $ w_{k i}$ are zero, we have that  $\lan w_{k i}^0, W^1 w_{k j}^0 \ran = 0$.
	
We compute the $\e^2$ order matrix elements in \eqref{F0ij}.} 
\DETAILS{Let $F^0(\lambda)$ be the matrix with the matrix elements $f_{ ij}(\lambda):=  \lan v_{i  }^0, \Lk v_{j }^0 \ran -  \e^2 \lan v^0_{j }, W^1   \bar{P}  \bar R_{ 0} \bar{P}   W^1 v^0_{j } \ran$ 
 $\forall i, j = 1, 2, 3, 4$. This $4\times 4$ matrix should be translated into the $3\times 3$ matrix using that 
\begin{align} \label{vke3'}&v_{3 k}= (0, 0,  e^{i k\cdot x}, - \bar k k^{-1}  e^{i k\cdot x})= v^0_{k, 3} -  \bar k k^{-1} v^0_{k, 4}.\end{align}

 The basis vectors $v^0_{ i}, i=1, 2, 3, 4$, defined in \eqref{v012} - \eqref{v034}, are orthonormal and satisfy \eqref{v0-parity}. Hence,  the matrix elements $f_{  ij}(\lambda)$   satisfy \eqref{fij-prop1} - \eqref{fij-prop3}. ??? } 
 Next, we have 
 \begin{align} \label{K0W2-expan}\lan v_{ i}, (K^0  + \e^2 W^2)  v_{  j}\ran=|k|^2 \del_{i, j\ge 3} + \e^2 \lan v_{ i}, W^2 v_{ j}\ran.\end{align}  
 
 To compute $\lan v_{ i}, W^2v_{ j}\ran, i, j =1, 2$, we use the definitions \eqref{W2} and \eqref{v012}, 
to obtain
        \begin{align}\label{F2firstterm1}
        \lan v_{ 1},  W^2v_{ 1}\ran  &=\overline{\lan v_{- k 2}, W^2 v_{-  2}\ran}= -\lambda^1\lan |\phi_k|^2\ran  \notag\\
        & + (2\kappa^2 +\frac{1}{2})\lan|\phi_0|^2|\phi_k|^2\ran   + \lan\bar{\phi}_k (i\bar{a}^1\partial_{a^0}-i a^1\partial_{a^0}^*)\phi_k\ran ,\\
            \lan v_{ 1}, W^2v_{ 2}\ran   & = 
            \overline{\lan v_{ 2}, W^2 v_{ 1}\ran}=    (\kappa^2  - \frac{1}{2})  \lan \phi_0^2 \bar{\phi}_{- k}\bar{\phi}_{k} \ran.
            \end{align}
   We integrate by parts and use the fact that $\phi_k$ satisfies \eqref{cphik=0}, which can rewritten explicitly as
	\begin{equation}	\label{phik-eq}
		\partial_{a^0}^*\phi_k = 0,
	\end{equation}   
 together with the identity \eqref{curla1} ($i\bar\p a^1	= \frac{1}{2} (\lan |\phi_0|^2 \ran - |\phi_0|^2)$), to obtain
    \begin{align}    \label{F2identity}
 \lan\bar{\phi}_k (i\bar{a}^1\partial_{a^0}-i a^1\partial_{a^0}^*)\phi_k\ran &= \lan\bar{\phi}_k i\bar{a}^1\partial_{a^0}\phi_k\ran = - \lan i |\phi_k|^2 \partial\bar{a}^1 \ran   \\
 &      = -\frac{1}{2}\lan |\phi_0|^2|\phi_k|^2 \ran + \frac{1}{2}\lan |\phi_0|^2 \ran \lan |\phi_k|^2 \ran .
	\end{align}
 The equations \eqref{F2firstterm1},  
 \eqref{F2identity}, the expression \eqref{lam1} for $\lambda_1$  and the normalization $\lan |\phi_k|^2 \ran = 1$    give
	\begin{align}  \label{W2matrelem-ii}
	 &  \lan v_{ i}, W^2 v_{  i}\ran   = -(\kappa^2 - \frac12)\beta + 2\kappa^2 \lan|\phi_0|^2|\phi_k|^2\ran, i=1, 2,\\  
 \label{W2matrelem-12} & \lan v_{ 1}, W^2v_{ 2}\ran   = \overline{  \lan v_{ 2}, W^2 v_{ 1}\ran}   = (\kappa^2  - \frac{1}{2})\lan \phi_0^2 \bar{\phi}_{- k}\bar{\phi}_{k} \ran.	\end{align}
	

    We compute 
    $g_{ij}:=\lan v_{ i}, W^{1} \bar{P} (\bar{P} K^{0 } \bar{P})^{-1} \bar{P} W^1 v_{  j} \ran,\ i, j\le 2$. Recall that   $P$ is the orthogonal projection onto $\Span \{v_j,  j=1, 2, 3, 4\}$ and  $ \bar P =\one - P $.  $\bar P $ can be written as  $ \bar P = \bar q_k \oplus  \bar q_k' \oplus  \bar p_k\oplus  \bar p_k$, where $ \bar q_k,\  \bar q_k' ,\  \bar p_k$ are the orthogonal projections in the space of scalar functions, $f$,  satisfying $f(x+s) = e^{ik\cdot s}f(x)$,  onto the orthogonal complements of $\phik,$ of $\refl\bar\phik,$ of $e^{ik\cdot s}$, respectively.        
    Using that \eqref{phik-parity}, \eqref{K0} and \eqref{W1}, we calculate
\begin{align}    \label{F2second'}
&g_{11}        = -\lan  \bar{\phi}_k ( \partial_{a^0}\phi_0) \Delta^{-1}\bar p_k ((\overline{\partial_{a^0}\phi_0}) \phi_k  ) \ran \\
                &g_{22}
        =  -\lan   \phi_{- k}  (\overline{\partial_{a^0}\phi_0}) \Delta^{-1} \bar p_k ((\partial_{a^0}\phi_0)  \bar{\phi}_{- k} ) \ran ,\\
        &g_{12}=g_{21} =0.
    \end{align}
        Now note that by \eqref{phik-eq} and the definition of $\bar\partial_{a^0}$ in \eqref{annih-creat-ops}, we have $\bar{\partial}\phi_k = -\frac12 (x_1+ix_2)\phi_k= i\bar{a}^0\phi_k$ and therefore
    \begin{equation*}
    	\bar{\partial}(\bar\phi_0\phi_k) = \bar\phi_0\bar{\partial}\phi_k + \phi_k\bar{\partial}\bar\phi_0 = i\bar{a}^0\bar\phi_0\phi_k + \phi_k\bar{\partial}\bar\phi_0
    	= \phi_k (\overline{\partial_{a^0}\phi_0}).
    \end{equation*}
   This gives
	\begin{align}\label{F2secondterm'}
 g_{11}	&= - \lan \phi_0\bar{\phi}_k \bar{\partial}^* \Delta^{-1} \bar p_k \bar{\partial}(\bar\phi_0\phi_k) \ran .
	\end{align}
Note that  the function $f:=\bar{\phi}_0\phi_k$ satisfies $f(x+s) = e^{ik\cdot s}f(x), \forall s\in \LAT$. Since $ \bar{\partial}^*  \bar{\partial}= -\Delta $, we have 
  \begin{equation} \label{LaplRel'}
	\bar{\partial}^* \Delta^{-1} \bar{\partial}=\Delta^{-1} \bar{\partial}^* \bar{\partial} =-\one. 
\end{equation}
 This  and the definition of $p_k$  yield $ g_{11}	 =   \lan \phi_0\bar{\phi}_k  \bar p_k  (\bar\phi_0\phi_k) \ran$ and therefore 
\begin{align}\label{F2secondterm'}
			g_{11}	&= \lan |\phi_0|^2|\phi_k|^2 \ran - |\lan \bar\phi_0\phi_k e^{-i k\cdot x} \ran |^2. 
	\end{align}
 Similarly we treat $g_{22} $ to obtain finally 
	\begin{align}    \label{gij1}
&g_{ij} =  (\lan |\phi_0|^2|\phi_k|^2 \ran - |\lan \bar\phi_0\phi_k e^{-i k\cdot x} \ran |^2) \del_{ij},\ i, j=1, 2.    \end{align}
Equations \eqref{Fiiexpan'}, \eqref{W2matrelem-ii}, \eqref{W2matrelem-12} and \eqref{gij1} give \eqref{f12}, with \eqref{ak} - \eqref{bk}. 


%
Now, we compute the blocks $F_{12} (\lambda)= (\lan v_{ i}, \cF (\lambda) v_{  j}\ran, i =1, 2, j=3, 4)$ and $F_{21} (\lambda)= (\lan v_{ i}, \cF (\lambda) v_{  j}\ran, i =3, 4, j=1, 2)$.  
By  \eqref{Fijexpan} and \eqref{K-expan}, we have 
\begin{align}\label{F12exp''}\mathcal{F}_{ i j}(\lambda)  =\e P_i W^{1} P_j + O(\e^3 |k|),\ i \ne j, \end{align}
where to obtain the right dependence of the error estimate on $k$, we used the second equation in \eqref{F-k-sym}. (We can also obtain this error estimate by showing that for $n$ even, $(P_i W^{\rm odd}[- \bar R^{\rm even}(\lambda)\ W^{\rm odd}]^n P_j)_k$$ = (P_i W^{\rm odd}[- \bar R^{\rm even}(\lambda)\ W^{\rm odd}]^n P_j)_{-k}$.) 
\DETAILS{Next, recall that $\rf$ denotes the reflection operator, changing $x$ to $-x$. By   the parity property \eqref{phik-parity} and the definitions of $v_{ k i}^0$, we have  $v_{ k i}^0=\rf v_{- k i}^0$ and $w_{ k i}^0=\rf w_{- k i}^0$. Since moreover $\rf W^3 =- W^3 \rf$ and $\rf \bar R_{ k 0}=\bar R_{- k 0} \rf$, we conclude 
 that $\lan v^0_{i,  k}, W^3  w_{j, k}^0\ran = - \lan v^0_{i, - k}, W^3  w_{j, - k}^0\ran$, which implies that $P W^3 P\big|_{k=0}=0$ and  therefore $P W^3 P=O(|k|)$. (Here we used that $a_{\om'}(-x)= -a_{\om'}(x)$ and therefore $a^{0}(-x)= -a^{0}(x)$ and $a_1(-x)= -a_{1}(x)$.)
This contributes to the second term on the r.h.s. of \eqref{F12exp}. 
This shows \eqref{F12exp}. 
}

Next,  
we have 
\begin{align}\label{viW1vj''}\lan v_{ i  }, W^1  v_{i +2 }\ran= 0,\ i=1, 2, \end{align} and 
 \[\lan v_{ 1 }, W^1  v_{4 }\ran= \lan \phi_{k}, i\p_{a^0}\phi_0e^{i k\cdot x}\ran.\]
			 We use the parity relation in \eqref{phik-parity} 
 to obtain
 \[\lan v_{ 2 }, W^1  v_{3 } \ran= - i \overline{\lan \phi_{-k},\p_{a^0}\phi_0 e^{- i k\cdot x}\ran}\]\[\qquad =i  \overline{\lan \phi_{k},\p_{a^0}\phi_0 e^{i k\cdot x}\ran}=- \overline{\lan v_{ 1 k}, W^1  v_{4 k} \ran}.\] 
  Now, integrating by parts, we find 
\[\lan v_{ 1 }, W^1  v_{4 } \ran= i \lan \p_{a^0}^*\phi_{k}, \phi_0e^{i k\cdot x}\ran+ i\lan \phi_{k}, \phi_0  (-\p) e^{i k\cdot x}\ran.\] Now, using \eqref{phik-eq} and $-(\partial_{x_1} - i\partial_{x_2}) e^{i k\cdot x}= - i (k_1-i k_2) e^{i k\cdot x}$, we find 
			\[\lan v_{ 1 }, W^1  v_{4 } \ran= - \overline{\lan v_{ 2 }, W^1 v_{3 } \ran}=(k_1-ik_2)\lan e^{i k\cdot x}\overline{\phi_{k}}\phi_0\ran.\]  
This together with \eqref{viW1vj''} and the notation $k=k_1 +i k_2$ gives \eqref{f13}, with \eqref{ck}. 


 Now, we compute the block 
 $F_{22} (\lambda)= (\lan v_{ i}, \cF (\lambda) v_{  j}\ran,  i, j =3,4)$. 
 Our starting point is the equations \eqref{Fiiexpan'} and \eqref{K0W2-expan}. 
 To compute $\lan v_{  i}, W^2 v_{ j} \ran, i, j =3,4$, we use the explicit form of $W^2$ and  $
   v_{ i}, i=3, 4,$ given in \eqref{W2} and \eqref{v034}, and  the fact that $\partial_{a^0}^*\phi_k = 0$, the normalization  $\lan |\phi_k|^2 \ran = 1$ and \eqref{lam1}, to calculate
        \begin{align}\label{wi-wj'}
        \lan v_{i  }, W^2 v_{i }\ran & = \lan |\phi_0|^2\ran \del_{ij},\ i, j=3, 4. 
    \end{align}
  
Let $ f(|k|^2) := e^{-|k|^2} \sum_1^\infty \frac{|k|^{2 (n-1)}}{n!}  [1 - (3 + \frac12 |k|^2 n^{-1}) 2^{- n} ]$. 
 We show below that
 \begin{align}    \label{wi-wj''}
&\lan v_{i }, W^{1} \bar R_{ 0}  W^1 v_{j }  \ran  = \{ 
1 -  |k|^2 f(|k|^2)\}  \lan |\phi_0|^2\ran \del_{ij},\ i, j = 3, 4.
    \end{align}
  This together with the last two computations and \eqref{Fiiexpan'} and \eqref{K0W2-expan} and the normalization condition $\lan |\phi_0|^2\ran=1$ gives
 \begin{align}\label{Fijexpan4}
\lan v_{i }, \cF(\lambda) v_{j }\ran =\big(1  + \e^2   f(|k|^2) \big) |k|^2 \del_{ij} 
+ 	O(\e^4),\ i, j = 3, 4.
	\end{align}     
 
 Next, by \eqref{F22-k0}, we have $\mathcal{F}_{ k=0 22}(\lambda=0)=0$, which, together with \eqref{Fijexpan4}, gives 
 \begin{align}\label{Fijexpan5}
\lan v_{i }, \cF(0) v_{j }\ran = \big(1  + \e^2   f(|k|^2) \big) |k|^2\del_{ij} 
+ 	O(\e^4 |k|^2),\ i, j = 3, 4.
	\end{align} 
This implies  \eqref{f33}. The estimates on the derivatives in $\lam$ are proven similarly but are much simpler. 
 
 To  complete the proof of Lemma \ref{lem:Fk-expan4}, it remains to prove  the equality $ \lan \overline{\phi_{k}}\phi_0 e^{i k\cdot x}\ran =  e^{-|k|^2/2}$ in \eqref{ck} and  \eqref{wi-wj''}. 
   
Proof of \eqref{wi-wj''}.  
To compute   $g_{ij}(k):=\lan v_{k i}, W^{1} \bar R_{ 0}  W^1 v_{k j }  \ran,\ i, j = 3, 4$, 
 we use the parity property \eqref{phik-parity}, 
  and  \eqref{K0}, \eqref{v034} and \eqref{W1} to calculate
\begin{align}    \label{g22} 
                &g_{22}(k)    = \overline{ g_{11} (-k)}    =  \lan  e^{- i k\cdot x} (\overline{\partial_{a^0}\phi_0}) (-\COVLAP{a^0} - 1)^{-1} P^\perp_0 (\partial_{a^0}\phi_0) e^{ i k\cdot x}\ran ,\\
 \label{g12}       &g_{12}(k)=g_{21}(k)=0, 
    \end{align}
where $P^\perp_0=\one - P_0$ and $P_0$ is the orthogonal projection onto the subspace spanned by $\phi_{ k}$. Now, we compute the r.h.s. of \eqref{g22}. Let $g_{\del} :=  \lan  e^{- i k\cdot x} (\overline{\partial_{a^0}\phi_0}) (-\COVLAP{a^0} - 1+\del)^{-1} P^\perp_0 (\partial_{a^0}\phi_0) e^{ i k\cdot x}\ran.$
Using $P^\perp_0=\one - P_0$, we write $g_{\del}$ as $g_{\del}  =  A_k' - \frac1\del A_k'',$ 
where 
 \begin{align}\label{A'1}
& A_k' : =  \lan    \partial_{a^0} \phi_0, e^{- i k\cdot x} (-\db-1+\del)^{-1} e^{i k\cdot x}    \partial_{a^0} \phi_0\ran,\\ 
& A_k'':=\lan    \partial_{a^0} \phi_0, e^{- i k\cdot x} P_0 e^{i k\cdot x}    \partial_{a^0} \phi_0\ran.\end{align} 
For $A_k'$, 
 we use that $e^{-i k\cdot x} \COVLAP{a^0} e^{i k\cdot x} =\COVLAP{a^0 - k}= \COVLAP{a^0} +2k\cdot i \n_{a^0} -|k|^2$, \[-\COVLAP{a^0} - 1= c^*  c,\  k\cdot \nao=\frac12 (k c^* - \bar k c),\ k \equiv k^c:=k_1+ i k_2,\] where $c := \partial^*_{a^0}$ and $c^*:=\partial_{a^0}$ are the annihilation and creation operators given in \eqref{annih-creat-ops}. \DETAILS{Recall the {\bf creation and annihilation operators},  $c^* :=\partial_{a^0}$ 
 and $c := \partial^*_{a^0}$, introduced in \eqref{annih-creat-ops}. } 
 This gives 
\begin{align}\label{A'2}
&e^{-i k\cdot x} (- \COVLAP{a^0} - 1) e^{i k\cdot x}  =  c^*  c+ (k c^* - \bar k c)+ | k|^2.
\end{align}
Since $e^{-\vphi}   c e^{\vphi}=c+i k$ and $e^{\vphi} c^*   e^{-\vphi}=c^* -i \bar k $, where $\vphi:= \frac i2 (k c^* + \bar k c)$,  this gives  $e^{-i k\cdot x} (\COVLAP{a^0} - 1) e^{i k\cdot x}=e^{-\vphi} c^*  c e^{\vphi}$. 
This gives
\begin{align}\label{A'3}
& A_k'  =  \lan    e^{\vphi} c^*  \phi_0,  (c^*  c+\del)^{-1} e^{\vphi} c^*  \phi_0\ran.\end{align} 
Using $e^{\vphi} c^*   e^{-\vphi}=c^* -i \bar k $ again, we find 
\begin{align}\label{A'4}
e^{\vphi} c^*  \phi_0=(c^* -i \bar k) e^{\vphi}   \phi_0=e^{-|k|^2/2} (c^* - i \bar k) \sum_0^\infty \frac{1}{2^n n!} (i k c^*)^n  \phi_0.
\end{align}
Furthermore, using 
$[ c, c^*]=2$, we find, by a standard computation,  $c^* c (c^*)^n\phi_0$ $=2 n (c^*)^n\phi_0$ and therefore $(c^* c+\del)^{-1} (c^*)^n\phi_0=(2n+\del)^{-1} (c^*)^n\phi_0$. This, together with \eqref{A'3}, gives

 \begin{align}\label{A'5}
A_k'  = e^{-|k|^2}  &\sum_0^\infty \frac{|k|^{2n}}{(2^n n!)^2}\{  \lan  ( c^*)^{n+1} \phi_0, (c^*  c+\del)^{-1} ( c^*)^{n+1} \phi_0\ran\notag\\
&+|k|^2\lan  ( c^*)^n \phi_0, (c^*  c+\del)^{-1} ( c^*)^n \phi_0\ran\}\notag\\
& + 2 e^{-|k|^2}  \sum_1^\infty \frac{2 n |k|^{2n}}{(2^n n!)^2}  \lan  ( c^*)^{n} \phi_0, (c^*  c+\del)^{-1} ( c^*)^{n} \phi_0\ran\notag\\
 =& e^{-|k|^2}  \sum_0^\infty \frac{|k|^{2n}}{(2^n n!)^2}\{(2(n+1)+\del)^{-1} \|( c^*)^{n+1} \phi_0\|^2 \notag\\
&  + |k|^2(2n+\del)^{-1} \|( c^*)^{n} \phi_0\|^2\}\notag\\
& + 2 e^{-|k|^2}  \sum_1^\infty \frac{2 n |k|^{2n}}{(2^n n!)^2} (2 n+\del)^{-1} \|( c^*)^{n} \phi_0\|^2. \end{align}
 Now, using $\|( c^*)^{n} \phi_0\|^2 = 2^n n! \|\phi_0\|^2$ and $\|\phi_0\|=1$, we obtain furthermore

 \begin{align}\label{A'}
A_k'  =e^{-|k|^2}&\{\sum_0^\infty \frac{|k|^{2n}}{2^n n!} [(2(n+1)+\del)^{-1} 2 (n+1) + |k|^2(2n+\del)^{-1} ] \notag\\
&+ 2  \sum_1^\infty \frac{ |k|^{2n}}{2^{n} n!} 2n (2 n+\del)^{-1}\} 
\notag\\
=e^{-|k|^2}&\{1  + |k|^2 \del^{-1} \ +\sum_1^\infty \frac{|k|^{2n}}{2^{n} n!} [3  + |k|^2(2n)^{-1} ]\} +O(\del). 
\end{align}

Next, using $e^{i k\cdot x} c^*   e^{-i k\cdot x}=c^* - i  k $, we compute
\begin{align}\label{A''0}
& \lan \phi_k, e^{ i k\cdot x} c^* \phi_0\ran=  \lan \phi_k, (c^* - i  k) e^{ i k\cdot x}  \phi_0\ran= - i  k \lan \phi_k, e^{ i k\cdot x}  \phi_0\ran\end{align} 
and therefore, by \eqref{ck} (see \eqref{ck-comp'} below) and $\|\phi_0\|=1$, 
\begin{align}\label{A''}
& A_k''=|\lan \phi_k, e^{ i k\cdot x} c^* \phi_0\ran|^2=|k|^2 |\lan e^{- i k\cdot x} \phi_k,  \phi_0\ran|^2 = |k|^2 e^{-|k|^2}.\end{align} 
The last relation, together with  
$g_{\del}  =  A_k' - \frac1\del A_k''$  and \eqref{A'},
 shows that the terms in front of $\frac1\del$ cancel 
  and therefore, after taking $\del\ra 0$ and using that $\sum_1^\infty \frac{|k|^{2n}}{2^{n} n!} n^{-1}=\frac12 |k|^2\sum_0^\infty \frac{|k|^{2n}}{n!}(n+1)^{-2}$, we arrive at
\begin{align}\label{a-y}
 g_{22}(k)=&e^{-|k|^2}\{1   +   \sum_1^\infty \frac{|k|^{2 n}}{2^{n} n!} [3  +  |k|^2 (2n)^{-1}  ]\}. 
  \end{align} 
Using the definition  $g_{ij}(k):=\lan v_{k i}, W^{1} \bar R_{ k 0}  W^1 v_{k j}  \ran
$, the relations \eqref{g22} and \eqref{g12} and that, by the computation above, $\overline{ g_{22}(- k)}= g_{22}(k)$, we conclude 
 that \eqref{wi-wj''} holds. $\Box$ 

Proof of the third equality in \eqref{ck}, i.e. of 
\begin{align}\label{ck-comp'}\lan \bar\phi_0\phi_k e^{-i k\cdot x} \ran = e^{-|k|^2/2}.\end{align}
Recall the definition of the annihilation and creation operators in \eqref{annih-creat-ops}. To begin with, we derive the following
  representation of the functions $\phi_k$ used below: 
		 \begin{proposition}\label{prop:L0k-spec} \begin{align}\label{phik-expr}
\phi_k=e^{i k\cdot x} e^{-i \vphi_k}  \phi_0,
\end{align}
where $\vphi_k:= \frac 12 (k c^* + \bar k c)$ (since $e^{i k\cdot x} e^{-i \vphi_k} $ is unitary it preserves the normalization). 
(Note that \eqref{phik-expr} is defined up to a constant phase factor $e^{i \al}$.) \end{proposition}
\begin{proof}   
 This formula should be compared with the Bloch theorem, which says that the functions $\phi_k$ can be written as $\phi_k=\tilde \phi_k e^{i k\cdot x} $, where $\tilde \phi_k$ is a gauge periodic function, i.e. the one satisfying the gauge-periodicity condition 
 $ \phi (x+s) =e^{i g_s(x)} 
 \phi (x),\ \quad \forall s\in \LAT_\tau, $ 
 with $g_s(x):=\frac12 s\wedge x+ c_s,$ with $c_s$ numbers satisfying $ c_{s+t} - c_s - c_t + \frac12 s \wedge t \in 2\pi\Z$ (see \eqref{gs-spec-tau}). 
  Our representation gives a detailed information about the function $\tilde \phi_k$.  

\DETAILS{Using the relations $e^{-i k\cdot x} \COVLAP{a^0} e^{i k\cdot x} =\COVLAP{a^0 - k}= \COVLAP{a^0} +2k\cdot i \n_{a^0} -|k|^2$, 
{\bf $-\COVLAP{a^0} - 1= c^*  c$ 
 and $ k\cdot \nao=\frac12 (k^c c^* - \bar k^c c)$,  where, recall,} $k^c:=k_1+ i k_2$, 
we obtain (omitting the superscript $c$) $e^{-i k\cdot x} (\COVLAP{a^0} - 1) e^{i k\cdot x}  =  c^*  c+ i (k c^* - \bar k c)+ | k|^2$. Furthermore,}
The relation  $e^{-i \vphi_k}   c e^{i \vphi_k}=c+i k$ gives $c^*  c+ i (k c^* - \bar k c)+ | k|^2=e^{-i \vphi_k} c^*  c  e^{i \vphi_k}$. This relation together with \eqref{A'2} implies 
 \begin{align}\label{cov-lapl-rel}
&e^{-i k\cdot x} (-\COVLAP{a^0} - 1) e^{i k\cdot x}  
=e^{-i \vphi_k}  (-\COVLAP{a^0} - 1)  e^{i \vphi_k}.
\end{align}
The operator $\COVLAP{a^0} - 1$ on the l.h.s. acts on functions satisfying the gauge-periodicity condition 
from the definition of $\phi_k$ in \eqref{phik-eqs-tau}, 
\begin{align}\label{gp'} 
\phi (x+s) =e^{i g_s(x)} e^{i k\cdot s}\phi (x),\ \quad \forall s\in \LAT_\tau,  \end{align}
 with  $ g_s(x)$ the same as above, 
 while on the r.h.s., satisfying 
 the same gauge-periodicity condition  with $k=0$. 

Denote by $h_k$ the operator $\COVLAP{a^0} - 1$ acting on functions satisfying the gauge-periodicity condition \eqref{gp'}. 
Then the relation \eqref{cov-lapl-rel} can be rewritten as $  e^{i \vphi_k} e^{-i k\cdot x} h_k e^{i k\cdot x} e^{-i \vphi_k} = h_{k=0}.$ Applying the latter relation to the function $\phi_0$ and using that $h_0\phi_0=0$ and that $\phi_k$ is the unique zero eigenvector of $h_k$, we conclude that   \eqref{phik-expr} holds.
\end{proof}

The Cambell-Baker-Housdorff relation $e^{X+Y}=e^{X} e^{Y}e^{-\frac12 [X, Y]}$, provided $[X, Y]$ is a multiple of the identity, implies $ e^{i\vphi_k}   \phi_0 =e^{-|k|^2/2} e^{ik c^*}   \phi_0$, and therefore $ e^{i\vphi_k}   \phi_0 =e^{-|k|^2/2} \sum_0^\infty \frac{1}{2^n n!} (i k c^*)^n  \phi_0$. Using this, we rewrite \eqref{phik-expr} as the series 
 \begin{align}\label{phik-expan}\phi_k =e^{-|k|^2/2} e^{i k\cdot x} \sum_0^\infty \frac{1}{2^n n!} (- i k c^*)^n  \phi_0.\end{align} 
 Using this series and the fact that $\lan \phi_0, (c^*)^n  \phi_0 \ran= 0$ for $n\ge 1$, we obtain \eqref{ck-comp'}. 
$\Box$ 

This completes the proof of Lemma \ref{lem:Fk-expan4}. \qquad $\Box$  
\section{Estimates of 
the nonlinearity $\tilde N (v)$}\label{sec:nonlin}

In this appendix, we analyze 
the nonlinearity $\tilde N (v):= \Nt (v)+V_{\s}v$, 
where $\Nt (v)$, $V_{\gamma'}$ and $\s$ are defined in 
Proposition \ref{prop:v-eq} and used in  Subsection \ref{sec:as-stab}.
The main result 
here is the following
\begin{proposition} Let $v=v'+v''$ as in Subsection \ref{sec:as-stab}. We have the estimates 
	\DETAILS{\begin{align}	
	\label{N-bnd}
	&\|N_\om(v)\|^2 \ls (\|v\|_{H^1}^3+\|v\|_{H^1}^5)\|v\|_{H^2},
\end{align}}
\begin{align}\label{tildeN-est'} &\| \tilde N (v)\|_{L^2} \ls 
 \e^{-2} \big( \sum_1^2\|v'\|_{H^1}^k \|v'\|_{L^\infty}+ \sum_3^4\|v'\|_{H^1}^k 
 + \|  v'' \|_{H^2} 
\sum_1^3 \|v\|_{H^1}^k\big),\\ 
 \label{xtildeN-est'} &\| x \tilde N(v)\|_{L^2}\ls   \e^{-3} \big(\| v'\|_{L^\infty_1}\sum_1^3\| v'\|_{H^1}^k 
 + \|\bar P x v''\|_{H^2} \sum_1^3\| v\|_{H^1}^k 
 \big),\\  
\label{UtildeN-est}  & \| U \tilde N(v) \|_{H^{-1}_{x}L^\infty_{k}}\ls \| U \Nt\|_{H^{-1}_{x}L^\infty_{k}}.\end{align} 
 \end{proposition}
\begin{proof}
	We first note that, due to the diamagnetic inequality for $a \in L^2_{loc}(\R^2)$ (see \cite{LL}), $|\nabla|f|| \leq |\COVGRAD{a}f|$, and by the standard Sobolev embedding theorem $H^1_{\textrm{cov}}$ is continuously embedded in $L^p(\R^2;\C\times\R^2)$ for all $p \in [2,\infty)$. 
We also note that since $\psit$ is a gauge-periodic smooth function, it is bounded. 
\DETAILS{So we have
	\begin{align*}
		\int |\xi|^2|\Re(\bar{\Psi}_\om\xi)|
		\leq \int |\xi|^3|\Psi_\om|
		\lesssim \int |\xi|^3
		\lesssim \|v\|_{H^1}^3.
	\end{align*}
	We also have
	\begin{align*}
		\int |\alpha \cdot \Im(\bar{\xi}\COVGRAD{A_\om}\xi)|
		\leq \left( \int |\alpha|^4\right)^\frac{1}{4} \left( \int |\xi|^4\right)^\frac{1}{4} \left( \int |\COVGRAD{A_\om}\xi|^2 \right)^\frac{1}{2}
		\lesssim \|v\|_{H^1}^3.
	\end{align*}
	The other terms of $R_\om$ are handled similarly.

\begin{align*}
	\|\bar{\xi}\n\xi\|^2 \leq \|\xi\|_{p}^2\|\n\xi\|_{q}\|\n\xi\|_{2},
\end{align*}
with $p^{-1}+(2q)^{-1}=4^{-1}$. Now, as long as $p, q < \infty$, by the Sobolev embedding theorem in dimension $2$, we have $\|\xi\|_{p}\lesssim \|\xi\|_{H^1}$ and $\|\n\xi\|_{q}\lesssim \|\xi\|_{H^2}$ and therefore
\begin{align*}
	\|\bar{\xi}\n\xi\|^2 \leq \|\xi\|_{H^1}^3\|\xi\|_{H^2}.
\end{align*}}

We begin with the  estimate 
$\|  \Nt (v)\|_{L^2} \ls 
 \e^{-2} ( \|v\|_{H^1} \|v\|_{L^\infty}+ \sum_2^3\|v\|_{H^1}^k).$ 
To establish 
this bound, we first write the explicit expression for
 the nonlinearity  $\No(v)$. 
Using the definition \eqref{N} of $\No(v)$ and the definition of $J(u)$ as the r.h.s. in the equation \eqref{GESresc-c} and letting $v=(\xi, \eta, \al)$, we find 	\begin{equation}\label{N-expl}
		 \Nt (v) = \ThreeByOne{-2i\alpha\cdot\nat\xi - |\alpha|^2\psit - |\alpha|^2\xi - i\xi\DIV\alpha - \kappa^2(\bar{\psit}\xi^2+2 \psit\xi\eta +\xi^2 \eta)}{2i\alpha\cdot\overline{\nat}\eta - |\alpha|^2\psit - |\alpha|^2\eta + i\eta\DIV\alpha - \kappa^2({\psit}\eta^2+2 \bar{\psit}\xi\eta +\eta^2 \xi)}{\eta\nat\xi +\xi\overline{\nat}\eta - 2\al(\bar{\psit}\xi+\eta\psit +\xi\eta)}.
	\end{equation}
  
We consider the worst terms 
  $\eta\nat\xi$ and $\xi^2 \eta$ in \eqref{N-expl}. We have trivially, $\|\eta\nat\xi\|_{L^2} \leq \|\eta\|_{L^\infty}\|\nat\xi\|_{L^2}\leq \|\eta\|_{L^\infty}\|\xi\|_{H^1}$. For the second term, we have, by the Sobolev embedding theorem in dimension $2$,
\begin{align*}
	\|\xi^2 \eta\|_{2} 
	\ls \|\xi\|_{H^1}^2 \|\eta\|_{H^1}.
\end{align*}
The remaining terms are simpler and treated similarly. 
Similarly, we show  the estimate 
\begin{align}\label{tildeN-est1'} &\| \Nt (v)-\Nt (v')\|_{L^2} \ls  \|  v'' \|_{H^2} \sum_1^2
 \|v\|_{H^1}^k.\end{align}
Indeed, e.g. considering the contribution of the term $\eta\COVGRAD{A_\om}\xi$, we have $\eta\nat\xi-\eta'\nat\xi'= \eta\nat\xi''+ \eta''\nat\xi'$, where $\xi'':=\xi' - \xi$ and $\eta'':=\eta' - \eta$, which is estimated as $\|\eta\nat\xi-\eta'\nat\xi'\|_{L^2} \leq \|\eta\|_{L^4}\|\nat\xi''\|_{L^4}+\|\eta''\|_{L^\infty}\|\nat\xi'\|_{L^2}$\  $ \leq \|\eta\|_{H^1}\|\xi''\|_{H^2}+ \|\eta''\|_{H^1}\|\xi'\|_{H^2}$.
Now, writing $\Nt (v)=\Nt (v')+ ( \Nt (v)- \Nt (v'))$ and using the above estimates, we arrive at 
\begin{align}\label{N-est1} 
&\|  \Nt (v)\|_{L^2} \ls \|v'\|_{H^1} \|v'\|_{L^\infty}+ 
\|v'\|_{H^1}^3 
+   \|  v'' \|_{H^2} \sum_1^2 \|v\|_{H^1}^k.\end{align}
This proves \eqref{tildeN-est'} for $ \Nt (v)$. 
 

%
\DETAILS{
Next, we obtain the  bounds 
 involving $V_
{\dot\g}$ and $G_{\dot\gamma }$ used in  Section \ref{sec:as-stab}.
\begin{proposition} We have the estimates 
\begin{align}	&\lan L_\om v,  V_{\dot\g} v  \ran  \ls \e^{-2} (\|v\|_{H^1}^2+\|v\|_{H^1}^4)\|v\|_{H^2}^2, \label{LV-bnd-gauge}
\end{align}
 \end{proposition}
\begin{proof}}
To estimate the term $V_{\s} v$, we begin with 
an estimate $\s:=\dot \g+\phi$.  

\begin{lemma} \label{lem:est-dotgam}  For $\s$ satisfying the equation \eqref{sig-eq}, we have the estimates  
	\DETAILS{\begin{align}	\label{dotgam-bnd1}
	\|\dot \g\|_{H^1} \lesssim \e^{-2} \big(\|v'\|_{H^1} \|v'\|_{L^\infty}+ \|v'\|_{H^1}^3 
+   \|  v'' \|_{H^2} \sum_1^2 \|v\|_{H^1}^k\big). 
	\end{align}}
\begin{align}	\label{dotgam-bnd1}\|\s\|_{H^s} \lesssim  \|\Nt (v)\|_{H^{s-1}}.\end{align}
\end{lemma}
\begin{proof}
We use the equation \eqref{sig-eq} for $\s$ and therefore first we have to show that the operator $-   \Delta+ |\psit|^2 + \bar{\psit}\xi + \psit\eta$ in \eqref{sig-eq}, considered from $H^{s}$ to  $H^{s-2}$, is invertible, for $\|\xi\|_{L^2}$ sufficiently small and $s=0, 1, 2$.  
To show this, we use the following 
\begin{lemma} \label{lem:h-low-bnd}  We have, for some constant $c>0$ independent of $\e$,
\begin{align}	\label{h-low-bnd} -   \Delta+ |\psit|^2 \ge c \e^2. \end{align} 
\end{lemma}
\begin{proof} 
First, we notice first that $h:=-   \Delta+ |\psit|^2\ge 0$. 
To show that this operator has a gap at $0$, we 
pass from $h$ to  its Bloch-Floquet representation $\int^\otimes_{\Om^*} h_k \hat{d k}$, acting on $ \int^\otimes_{\Om^*} \cH_k \hat{d k}$, where $\cH_k:=\{\al\in L^2(\R^2, \C^2): \al(x+s)= e^{ik\cdot x}\al(x), \forall s\in \LAT\}$, with the inner product $\lan \al, \beta\ran:=\int_\Om \bar\al\cdot\beta$, for some some elementary cell $ \Om$ of $\LAT$. Then, using the map $e^{ik\cdot x}: L^2( \Om^*, L^2_{\rm per})\ra \int^\otimes_{\Om^*} \cH_k \hat{d k}$, where $L^2_{\rm per}:=\{\al\in L^2(\R^2, \C^2): \al(x+s)= \al(x), \forall s\in \LAT\}$, with the inner product as above, further to the operator $$g: f_k(x)\ra [(-i \n -k)^2+ |\psit|^2] f_k(x)$$ on $L^2( \Om^*, L^2_{\rm per})$. We consider $g$ as a perturbation of the operator $g_0:  f_k(x)\ra (-i \n -k)^2 f_k(x)$, whose spectrum consist of the bands, $\{|\nu-k|^2: k\in \Oms\}, \nu \in\LAT^*$, with the band eigenfunctions, $e^{i \nu \cdot x}$. We apply to $g$ the Feshbach-Schur map with the projection, $P$, on the eigenspace of the operator $g_0$ corresponding to  its lowest spectral branch, $\{|k|^2: k\in \Oms\}$, i.e. $P f_k(x)\ra \lan 1, f_k\ran =\lan f_k\ran_{\Om}\in L^2(\Om)$. 
The corresponding the Feshbach-Schur map is the multiplication operator
$$f(\mu):=\lan (-i \n -k)^2+ |\psit|^2 - |\psit|^2 \bar r_k(\mu)|\psit|^2\ran_{\Om},$$
 where $\bar r_k(\mu):= \bar P (\bar P g_k \bar P -\mu)^{-1}\bar P$, with $g_k:=(-i \n -k)^2+ |\psit|^2$ and $\bar P:=\one - P$, acting on $L^2(\Om^*)$. It can be rewritten as $f(\mu)= | k|^2+\lan |\psit|^2\ran_{\Om} -\lan |\psit|^2 \bar r_k(\mu)|\psit|^2\ran_{\Om}$ and, since $\lan |\psit|^2\ran_{\Om}\gs \e^2$ and $ |\psit|^2=O(\e^2)$, it satisfies $f(\mu)\gs \e^2$. Since $\mu\in \s(f(\mu))\Leftrightarrow \mu \in \s(h)$, the last estimate proves \eqref{h-low-bnd}.
\DETAILS{(Another way to prove this bound is to expand 
\begin{align}\label{R-expan} (-   \Delta+ |\psit|^2& + \re(\bar\Psi_\om\xi) )^{-1}\notag\\
&=\sum_0^\infty (-   \Delta+ |\psit|^2)^{-1} [-(-   \Delta+ |\Psi_\om|^2)^{-1} \re(\bar\Psi_\om\xi) ]^n\end{align}
 and using that $ (-   \Delta+ |\psit|^2)^{-1}$ is bounded on $H^s, s=0, 1$ (by $C\e^{-2}$).)}
\end{proof} 
Now,  the bound \eqref{h-low-bnd} and the condition  $\|\xi\|_{L^2}\ll \e^2$ imply that $-   \Delta+ |\psit|^2 + \bar{\psit}\xi + \psit\eta \ge \frac12 c \e^2$ and therefore $-   \Delta+ |\psit|^2 + \bar{\psit}\xi + \psit\eta$ is invertible and its inverse is bounded as $\|(-   \Delta+ |\psit|^2 + \bar{\psit}\xi + \psit\eta)^{-1}\|_{H^{s-2}\ra H^s} \ls \e^{-2}$.

   Since  the operator $-   \Delta+ |\psit|^2 + \bar{\psit}\xi + \psit\eta$ is invertible,  we can rewrite \eqref{sig-eq} as 
\begin{align}\label{dotgam-eq'} \s  =R b,\ \text{ where }\ R:=(-   \Delta+ |\psit|^2 + \bar{\psit}\xi + \psit\eta)^{-1}\end{align}  
  and $b(v):=- i \bar\psit  N_\xi(v) - i \psit  N_\eta(v) 
  +\divv  N_\al(v)$ 
   (with $ \Nt (v)= ( N_\xi(v),   N_\eta(v), N_\al(v))$).
  Then  \eqref{dotgam-eq'}, together with the  estimate $R\ls \e^{-2}$,  gives that 
\begin{align}	\label{dotgam-bnd2}\|\s\|_{H^s} \lesssim  \e^{-2}(\| i \bar\psit  N_\xi(v) + i \psit  N_\eta(v)-\divv  N_\al(v)\|_{H^{s-2}}).
\end{align}
This 
implies  
\eqref{dotgam-bnd1}.  
   \end{proof}
%
\DETAILS{ The definitions of $ N_\om(v)$ and $G_{\dot\gamma }$ and integration by parts give	\begin{align}			
	&\lan L_\om v,  V_{\dot\g}v  \ran  \ls \|\dot \g\|_{H^2} \|v\|_{H^1}^2.\label{LV-bnd}
\end{align} 
\eqref{LV-bnd} - \eqref{NF-bnd} follow from \eqref{LV-bnd} - \eqref{NF-bnd} and Lemma \ref{lem:est-dotgam}.}

Now, we prove the bound \eqref{tildeN-est'} for $V_{\s}v$.  
\DETAILS{Due to the decomposition $V_{\dot\g (v)}v=V_{\dot\g (v')}v'+(V_{\dot\g (v)}v-V_{\dot\g (v')}v')$, this bound is equivalent to the following estimates 
\begin{align}\label{Vdotg-est}  
&\| V_{\dot\g}v \|_{L^2} \ls \e^{-2}(\|v\|_{L^\infty} \sum_1^3\|v\|_{H^1}^k+  \|v\|_{H^1}^3),\\
\label{Vdotg-est2} & \|V_{\dot\g (v)}v- V_{\dot\g (v')}v' \|_{L^2}\ls \e^{-2} \sum_1^4\|v\|_{H^1}^k\|v''\|_{H^1},\end{align} 
which we now prove.}
  By 
the definition of $V_{\s}$ (namely, $V_{\s} v:= (i\s\xi, - i\s\eta, 0)$, for $v = (\xi, \eta, \al)$, see \eqref{Vgam}), we have $\| V_{\s}v \|_{L^2}\ls \|  \s\|_{H^1}\|v\|_{H^1}+ \|\s\|_{L^2}$, which, together with \eqref{dotgam-bnd1} and \eqref{N-est1}, gives the estimate  \eqref{tildeN-est'} for $V_{\s}v$. This together with \eqref{N-est1}, completes the proof of \eqref{tildeN-est'}. 

Next, we prove the estimate \eqref{xtildeN-est'}. 
 By the explicit form of  $\Nt (v)$, given in  \eqref{N-expl}, $\Nt (v)$ can be  written as 
 \begin{align}\label{Ntau-deco}\Nt (v)= N_2 (v, v)+ N_3 (v, v, v),\end{align}
  where $N_2 (v, v)$ and $N_3 (v, v, v)$ are bilinear and trilinear contributions to $\Nt (v)$. We agree that it is  the first arguments in $N_2 (v, v)$ and $ N_3 (v, v, v)$ which have no derivatives. Then $x \Nt (v)= N_2 (x v, v)+ N_3 (x v, v, v)$. 
 Now, using the estimate similar to \eqref{N-est1} and recalling the definition of the space $L^\infty_1$, we arrive at the estimare 
 \begin{align} \label{xN-est} &\| x  N(v)\|_{L^2}\ls   \| v'\|_{L^\infty_1}\sum_1^2\| v'\|_{H^1}^k  + \|\bar P x v''\|_{H^2} \sum_1^2\| v\|_{H^1}^k .\end{align}

Now we turn to the term $x V_{\s} v$. 
 By the definition  $V_{\s} v:= (i\s\xi, - i\s\eta, 0)$, for $v = (\xi, \eta, \al)$, we have  
 $\| x V_{\s} v|_{L^2}\ls   \|x v\|_{L^\infty}\| \s\|_{L^2}$, which implies  the desired result. 
 
 For the second component, we write $x_i\n_j\s= \n x\s-\del_{i, j}\s$. Furthermore,  
using the equation \eqref{dotgam-eq'} and the relation $[x, R]=-R 2\n R$,  which follows from $[x, -\Delta]=2\n$,  we find  $x\s=R x b-R 2\n R b$. To bound the latter expression, we use the third of the following bounds 
\begin{align} \label{R-bnds} &\|R \| \ls \e^{- 2},\ \quad \|\n R \|  \ls \e^{- 2}, \quad \|R \n R\|\ls \e^{-3},\end{align}  first of which follows from \eqref{h-low-bnd}  the second,  from  
the first one as  $\|\n R f\|^2=\lan f, R (-\Delta) R f\ran =\lan f, R (R^{-1} +O(\e^{2})) R f\ran \ls \e^{- 2} \| f\|^2$, and the third, from the first two. 
Remembering the definition of $b(v)$ 
 and using  the above bound on $R \n R$ and the estimate \eqref{dotgam-bnd1}, we arrive  at the estimate 
\begin{align} \label{xVv-est1} &\| x\n\s\|_{L^2}\ls 
 \e^{-2} \|x \Nt (v)\|_{H^{-1}} + \e^{-3}  \|  \Nt (v)\|_{H^{-1}}.\end{align}
This, together with the estimate of $\s x\xi$ and the estimates \eqref{N-est1} 
 and \eqref{xN-est}, gives
\begin{align} \label{xVv-est2} &\| x V_{\s} v\|_{L^2}\ls   \e^{-3} \big(\| v'\|_{L^\infty_1}\sum_1^3\| v'\|_{H^1}^k  + \|\bar P x v''\|_{H^2} \sum_1^3\| v\|_{H^1}^k  \big).\end{align}
 This completes the proof of  the estimate \eqref{xtildeN-est'}. 
\DETAILS{Next, we prove the estimates \eqref{xtildeNv'-est'} - \eqref{xtildeN1-est'}. 
Since we keep only the bi-linear term, $\tilde N_2 (v, v)$, in $\tilde N (v)$, we have $x \tilde N (v) = \tilde N_2(x v, v) +\dots$, where the dots stand for the terms arising from pulling $x$ through various differential operators and their inverses, these terms have stronger estimates. (Indeed, for the more singular term $V_{\dot g} v$,  to estimate $x \dot g$, we use the equation \eqref{dotgam-eq'} and pull $x$ through $R$. Since $[x, -\Delta]=2\n$, this leads to the operator $R 2\n R$, 
which is bounded (by $\e^{-3}$). Consequently, $x$ drops out and the resulting term is estimated similarly to $\Nt$.) 
This implies the estimates \eqref{xNv'-est'} - \eqref{xN1-est'}}
 %
 
 Now, we prove the estimate \eqref{UtildeN-est}.
\DETAILS{ Indeed,  by \eqref{dotgam-eq}, we can write  $ \dot g  =R b,$  where, recall,  $R:=(-  2 \Delta+ |\Psi_\om|^2 + \re(\bar\Psi_\om\xi) )^{-1}$ and $b(v):=(-\im(\bar\Psi_\om   N_\xi(v))+\divv  N_\al(v))$ (with $ N_\om(v)= ( N_\xi(v),   N_\al(v))$). Expanding \[(-  2 \Delta+ |\Psi_\om|^2 + \re(\bar\Psi_\om\xi) )^{-1}=\sum_0^\infty (-  2 \Delta+ |\Psi_\om|^2)^{-1} [-(-  2 \Delta+ |\Psi_\om|^2)^{-1} \re(\bar\Psi_\om\xi) ]^n\]}
Using the expansion 
\begin{align}\label{R-expan} (-   \Delta+ |\psit|^2& + \bar{\psit}\xi + \psit\eta)^{-1}\notag\\
&=\sum_0^\infty (-   \Delta+ |\psit|^2)^{-1} [-(-   \Delta+ |\psit|^2)^{-1} (\bar{\psit}\xi + \psit\eta) ]^n\end{align}
and using that, by the second relation in \eqref{k-deriv-Uv} (or the first relation in \eqref{norm-rel1}), $U (-   \Delta+ |\psit|^2)^{-1}  =(-   \Delta+ |\psit|^2)^{-1}  U$ and that $ (-   \Delta+ |\psit|^2)^{-1}$ is bounded on $H^{-1}_{x}L^\infty_{k}$ ( which is shown similarly to \eqref{h-low-bnd}), 
we conclude that $U(-   \Delta+ |\psit|^2+ \bar{\psit}\xi + \psit\eta)^{-1}  =\hat R  U$, where $\hat R$ is a bounded operator on $H^{-1}_{x}L^\infty_{k}$. 
   This, together with the equation \eqref{dotgam-eq'} and the definition of $\tilde N$, gives \eqref{UtildeN-est}.
\end{proof}

\begin{proposition} \label{prop:sec:nonlin} The nonlinearity  $\cN_\vpp(f)=V^* \tilde N (Vf+ h)$ satisfies the esimates
 \begin{align}  \label{cN-est'}
&\left\|\p_k^m\cN_{h, k } (f)\right\|_{L^\infty}\ls 
\sum_{r=2}^3 ( \| f\|_{H^{m/2}}^r  +  \| h\|_{H^1_{m/2}}^r),\ m=0, 1. 
 \end{align}
 \end{proposition} \begin{proof} 
\DETAILS{Here we prove the weaker and simpler estimate
\begin{align} \label{dcN-bnd'} \|\p_k^m \cN_{k} (f, h)\|_{L^\infty}\ls  
\|f\|_{H^{1}}\|f\|_{H^{m}} +\|v''\|_{H^1}\|v''\|_{H^m}  +\|\bar P |x |^{m/2} v''\|_{H^1}^2,\ m=0, 1, \end{align}
which implies the decay $<t>^{-5/4}$ and $<t>^{-3/4}$, rather than $<t>^{-3/2}$ and $<t>^{-1}$, encoded in  the Banach space $X_{3/2}$, but which suffices for our purposes (it leads to  the space $X_{5/4}$). The estimate \eqref{cN-est}, which requires more explicit information about the nonlinearity, is proven in Appendix \ref{sec:cN-est}. }
%
 To concentrate on the essentials and keep notation from running amok, we present the estimates  only for the bilinear part $N_2 (v, v)$ of $\Nt (v)$, defined after \eqref{Ntau-deco}. The cubic term $N_3 (v, v, v)$ is treated similarly. (The  estimates for $N_2 (v, v)$ hold also for $N_3 (v, v, v)$, provided the norms involved are less than some constant.)

We observe that the bilinear part, $N_2 (v, v)$, of $\Nt (v)$ can be written as a sum of products -- in the sense of the definition in Appendix \ref{U-prod-transf}  -- of two vectors. 
 For instance, one of  the two most singular contributions to $\Nt (v)$ is $( \al\cdot  \n_{\at} \xi, \al\cdot  \overline{\n_{\at}} \eta,  \eta \n_{\at} \xi)$ and it can be written as a product of the vectors $v_1=( \n_{\at} \xi, \eta, \al)$ and $v =(\xi, \eta, \al)$.  

For $m=0$, 
by virtue of \eqref{V*est} and \eqref{UtildeN-est}, we have to estimate $\| U\Nt\|_{H^{-1}_{x}L^\infty_{k}}$, i.e. $\| UN_2 (v, v)\|_{H^{-1}_{x}L^\infty_{k}}$, in our case. 
Using the equation 
  \eqref{BFZT-prod2-est} , 
we find that
\begin{equation} \label{UN-est}
\| U  N_2 (v, v)\|_{H^{-1}_{x}L^\infty_{k}} \ls  \|\hat v_{1} \|_{L^2}  \|\hat v  \|_{L^2}, 
\end{equation}
 where $v_1=( \n_{\at} \xi, \eta, \al)$ and $v =(\xi, \eta,\al)$ 
and $\hat v_{1} =U v_{1}$ and $\hat v =U v$. 
 The unitarity of $U$ gives $\|\hat v_{2} \|_{L^2}=\|v_{} \|_{L^2}$. Next, by 
 the second equation in \eqref{k-deriv-Uv}, 
we have $\hat v_{1}=(\n_{\at}\oplus \one) U v$. By  the second relation in \eqref{k-deriv-Uv} or in \eqref{norm-rel1}, this gives $\|\hat v_{1} \|_{L^2}\le \|(\n_{\at}\oplus \one)  v\|_{L^2}$. Now, 
\DETAILS{$v_2=v=V f+v''= \U^{-1}  \int_{\Om^*}^\oplus \hat{dk}  f_k  \hw_k +v''$
 and therefore, $g_{1}=\int_{\Om^*}^\oplus \hat{dk}  f_k (\n_{A_\om}\oplus \one)   \hw_k+\U(\n_{A_\om}\oplus \one)  v''$. This, by the differentiability property \eqref{hwk-parity} of $\hw_k$  and the second relation in \eqref{norm-rel1},} 
 %
$v_2=v=V f+v''$ and therefore, by the first relation in \eqref{V-bnd-L2},  implies $\|\hat v_{1} \|_{L^2}\le \|f \|_{L^2} +  \|v_{}'' \|_{H^1}$, which, together with \eqref{UtildeN-est} and \eqref{UN-est}, gives
\begin{equation} \label{UtildeN-est'}\| U \tilde N_2 \|_{H^{-1}_{x}L^\infty_{k}} \ls  \|f \|_{L^2}^2 +  \|v_{}'' \|_{H^1}^2, \end{equation}
where $ \tilde N_2$ is obtained from $ \tilde N$ by dropping the trilinear terms in $\Nt$. Using the definition  $\cN_\vpp(f)=V^* \tilde N (Vf+ h)$ and the estimates \eqref{V*est} and \eqref{UtildeN-est'}, 
 leads to the quadratic part of  \eqref{cN-est'} (or \eqref{cN-est'} for $ \| f\|_{H^{m/2}}^r \ls 1$ and $ \| h\|_{H^1_{m/2}}^r\ls 1$) for $m=0$.

 Now, to prove differentiability of $\cN_{\vpp, k} (f)$ in $k$, we use the definition $\cN_\vpp(f)=V^* \tilde N (Vf+ h)$ and the first relation in \eqref{xV-Vk-deriv}, which gives $ i \n'_j V^* =V^* x_j  - V_j^*$, to compute $\p_{k_j} \cN_{\vpp} (f)= (V^* x_j  - V_j^*)  \tilde N (Vf+ h)$. 
 Now, letting $v=Vf+ h$, we have 
 \begin{align} \label{Nk-deriv1} &\p_{k_j} \cN_{\vpp} (f)= (V^* x_j  - V_j^*)  \tilde  N (v) . \end{align} 
Using \eqref{V*est} and proceeding as with the estimate \eqref{UtildeN-est}  above, 
to obtain 
\begin{align} \label{Nk-deriv2} \|V^* x  \tilde  N_2 (v, v)\|_{L^\infty}\ls   \| U x \tilde  N_2 (v, v)\|_{H^{-1}_{x}L^\infty_{k}} 
 \ls   \| U  x N_2 (v, v)\|_{H^{-1}_{x}L^\infty_{k}}. 
 \end{align}
By the explicit expression for the nonlinearity, \eqref{N-expl}, 
	we have 	$ x_j N_2 (v, v)= N_2 (x_j^{1/2} v, x_j^{1/2}v) +$ a simpler term, where the simpler term comes from pulling $x_j^{1/2}$ through $\n$ and is of the same form as $N_2 (v, v)$. 
\DETAILS{\begin{align} \label{Nk-deriv} 
&\p_{k_j} \cN_{\vpp} (f)= (V^* x  - \tilde V^*)  \tilde  N_2 (v, v)    
 = V^*  \tilde N_2 (x_j^{1/2} v, x_j^{1/2} v) - \tilde V^* \tilde N_2 (v, v) + \text{a simpler term}. \end{align}}
Now, using \eqref{Nk-deriv1} and \eqref{Nk-deriv2} and proceeding as in \eqref{UN-est}, we arrive at the estimate 
 \begin{align} \label{dcN-bnd}\|\n_k \cN_{k 2} (f, h)\|_{L^\infty}\ls 
 \|\hat v_{1} \|_{L^2}  \|\hat v_{2} \|_{L^2}+  \|\hat v_{3} \|_{L^2}  \|g_{4} \|_{L^2}, \end{align}
where $\cN_{k 2} (f, h)$ is obtained from $\cN_{k } (f, h)$ by dropping the trilinear terms in $\Nt$, $\hat v_{i} =U v_{i}  $ and $v_1$ and $v_2$ are as above and $v_3= x_j^{1/2} v_1$ and $v_4= x_j^{1/2} v_2$. The first summand on the r.h.s. is the same as in \eqref{UN-est} and leads eventually to the same contributions to the final estimate. To estimate the second summand on the r.h.s.,  
we have to analyze the function $x v (x)$. 

\DETAILS{
 Furthermore, we use the definition of $P$, which we write as   $P= \U^{-1} W$, where $Wv:= \int_{\Omega^*}^\oplus \lan \hw_k, (\U v)_k \ran \hw_k\hat {dk}$, and the equation \eqref{k-deriv-Uv} below, which implies that $[x, \U^{-1}]=U^{-1}i \n'$, where $ \n'=\n_k$, to obtain $[x, P]=\U^{-1}i \n' W+ \U^{-1} [x, W]$. Next, by \eqref{k-deriv-Uv}, $W x v=\int_{\Omega^*}^\oplus \lan \hw_k, (\U x v)_k \ran \hw_k\hat {dk}=  \int_{\Omega^*}^\oplus \lan \hw_k, (i \n_k + x )(\U  v)_k \ran \hw_k\hat {dk}=  \int_{\Omega^*}^\oplus i \n_k\lan \hw_k, (\U  v)_k \ran \hw_k\hat {dk}-  \int_{\Omega^*}^\oplus \lan i \n_k\hw_k, (\U  v)_k \ran \hw_k\hat {dk}+  x\int_{\Omega^*}^\oplus \lan \hw_k, (\U  v)_k \ran \hw_k\hat {dk}=i\n' W+ x W -  \int_{\Omega^*}^\oplus \lan i \n_k\hw_k, (\U  v)_k \ran \hw_k\hat {dk}$.  Collecting these equalities, we find  $[x, P]=\U^{-1}\{i \n' W+  x W-  i\n' W-  x W +  \int_{\Omega^*}^\oplus \lan i \n_k\hw_k, (\U  v)_k \ran \hw_k\hat {dk}\}=\U^{-1}\ \int_{\Omega^*}^\oplus \lan i \n_k\hw_k, (\U  v)_k \ran \hw_k\hat {dk}$.
Since $\hw_k$ satisfies the estimates \eqref{hwk-parity}, we see that $[x, P]$ is a bounded operator on $L^2$.}

Using $v=v'+v''=Vf +v''$, we write $x v=xVf+P x v''+\bar P x v''$, where $f:=V^* v$. Recall that by \eqref{xV-Vk-deriv}, $x_jVf= (V i \n'_j + V_j) f$, where, recall, $\n' =\n_k $.  The term $\bar P x v''$ is fine as we have shown in the previous subsection how to control it. To treat the remaining term, we use \eqref{Px} and $Pv''=0$ to obtain  $P x_j v''=V  V^*_j v''$. Collecting these expression, we find 
\begin{align} \label{x v-deco} x_j v= (V i \n'_j + V_j)  f+V V^*_j v''+\bar P x_j v'',\end{align}
By the relations $\|V f\|_{L^2} =  \|f \|_{L^2}$,  $\|   V_j f\|_{L^2} \ls  \|f \|_{L^2}$ and $ \|  V^*_j v\|_{L^2} \ls  \|v \|_{L^2}$ (see Lemma \ref{lem:V} and the line after \eqref{xV-Vk-deriv}), this gives  
\begin{align}\notag \|x^{m} v \|_{L^2} \ls \|f\|_{H^{m}}  +\|v''\|_{L^2}  +\|\bar P |x |^{m2} v''\|_{L^2},\end{align}
 for $m=0, 1$,  which by the interpolation gives also $m=1/2$.
Then remembering \eqref{dcN-bnd} and  proceeding as in the paragraph after the equation \eqref{UN-est}, we arrive at  the quadratic part of  \eqref{cN-est'} (or \eqref{cN-est'} for $ \| f\|_{H^{m/2}}^r \ls 1$ and $ \| h\|_{H^1_{m/2}}^r\ls 1$) for $m=1$. 
\end{proof}
\paragraph{Remark 16.}  In the estimates \eqref{V*est} and \eqref{UtildeN-est}, we can replace $H^{-1}_{x}L^\infty_{k}$ and $H^{-1}_{x}L^\infty_{k}$ by $L^{1}_{x}L^\infty_{k}$ and  $L^{1}_{x}L^\infty_{k}$.
This would effect the estimate \eqref{Nk-deriv2} giving $\|V^* x  \tilde  N (v)\|_{L^\infty}\ls   \| U x \tilde  N (v)\|_{H^{-1}_{x}L^\infty_{k}\cup L^{1}_{x}L^\infty_{k}}$ $ \ls   \| U  x \Nt (v)\|_{H^{-1}_{x}L^\infty_{k}\cup L^{1}_{x}L^\infty_{k}}$. Moreover, the estimate \eqref{UN-est} can be replaced by the estimate $\| U  \Nt (v) \|_{L^{r}_{x}L^\infty_{k}} \ls   \|\hat v_{1} \|_{L^p_k L^q_x}  \|\hat v_{2} \|_{L^{p'}_k L^{q'}_x},$ $ p^{-1}+ (p')^{-1} =1,\  q^{-1}+ (q')^{-1} =r^{-1}$.

\section{Proof of Theorem \ref{thm:gammak} 
} \label{sec:gammaq-series}

In this appendix we 
prove the explicit representation \eqref{gamk-series} of the functions $\g_k(\tau)$. 
We introduce the functions  
 \begin{equation}\label{vphiq} \vphi_q (z)= \phi_k(x),\ \quad  x_1+i x_2=\sqrt{\frac{2\pi}{\im\tau}}    z,\ 
  k=- \sqrt{\frac{2\pi}{\im \tau}}   i q 
 \end{equation} 
   Now, 
    rewrite the functions $\g_k(\tau)$, 
defined in  \eqref{gamk-tau}, in terms of the functions $\vphi_q$, introduced in \eqref{vphiq}:
\begin{equation}\label{gammak-varphiq}
  \g_k( \tau) := 2\lan|\vphi_0|^2|\vphi_q|^2\ran_{\Om_\tau} + |\lan \vphi_0^2\bar{\vphi}_q \bar{\varphi}_{-q}\ran_{\Om_\tau} | - \lan |\vphi_0|^4 \ran_{\Om_\tau} .
	\end{equation}
Theorem \ref{thm:gammak} follows from \eqref{gammak-varphiq}, 
Proposition \ref{prop:gammaq-comp}  below 
and fact that $|\Om_\tau |= \Im\tau $. $\Box$

\DETAILS{ It is straightforward to verify 
 the periodicity relation in \eqref{phik-eqs}  implies that the functions $\vphi_q (z)$ 
 satisfy  the periodicity relations 
 \begin{equation}\label{vphiq-per} \vphi_q (z+s) =e^{\frac{i\pi}{\im\tau}(\im (\bar s z)+2\im (\bar s  q))+i c_s}\vphi_q (x),\ \forall s\in 
 \Z+\tau \Z, \end{equation}
where $s\in ??$, $c_s$ are the constants which enter \eqref{gs-spec} (for $\LAT=\LAT_\tau$??). 
Indeed, by \eqref{vphiq} and \eqref{phik-eqs}, we have $\vphi_q(z+s) =\phi_k(x+
 s)=e^{i\al}\phi_k (x)$, where $\al:= 
 \frac{\pi}{\im\tau}  s\cdot J x+ s\cdot  k$. Using the identifications in \eqref{vphiq}, we find furthermore, $\al:=  \frac{2\pi}{\im\tau}  \im (\frac12\bar s z+\bar s q)$, which, together with the previous relation, gives \eqref{vphiq-per}. }
 %

By \eqref{vphiq}    and Proposition \ref{prop:phik},  the functions  $\vphi_q (z)$ are related to the theta-functions $\theta_{q}(z, \tau)$ as 
\begin{align}\label{vphiq-thetaq} \vphi_q(z)=c 
e^{\pi i (a^2\tau-a b)} e^{\frac{\pi}{2 \im\tau}(z^2 -|z|^2)} \theta_{q}(z, \tau). 
\end{align} 
Here $c$ is such that  
	 $\lan |\vphi_q|^2 \ran_{\Om^\tau} = \lan |\vphi_0|^2 \ran_{\Om^\tau} =1$, where $\Om^\tau$ is a fundamental domain of the lattice $\Z+\tau \Z$, and, since the function $\vphi_q(z)$ is defined up to a phase factor, we inserted the factor $e^{\pi i (a^2\tau-a b)}$,  which makes some expressions below simpler.

To  compute the integrals entering \eqref{gammak-varphiq}, we use the relation \eqref{vphiq-thetaq} and the explicit series representation  \eqref{thetaq} for  $\theta_{q}(z, \tau)$.
As was mentioned above, {\bf we can set} $c_1 = c_\tau=0$ in \eqref{thetaq}, which we do from now on. 
	
\begin{proposition}\label{prop:gammaq-comp}  Recall $q=-a\tau +b$ and $
\im\tau=\tau_2$ and  let $\Om=\Omega^\tau:=\{u_1 +\tau u_2 :  - 1/2 \le u_1, u_2\le 1/2\} $. We have 
\begin{align} 
\label{phi0phik-int0}&\int_\Om \bar\vphi_0 \vphi_q e^{\frac{2\pi i}{\tau_2}\im (\bar q z)} dz=\frac{{c}^2}{2} \sqrt{\tau_2} e^{-\frac{\pi }{2\tau_2}   |- a \tau+b| ^2},\\
\label{phi0phik-int1}
&\int_\Om|\vphi_0|^2|\vphi_q|^2 dz = \frac{c^4}{ 2} 
 \sum_{p, n=-\infty}^{\infty} e^{-\frac{\pi}{\tau_2}   |n-  p   \tau-  a   \tau+b| ^2},\\ 
&\int_\Om \bar\vphi_0^2 {\vphi}_q {\vphi}_{-q} dz=  \frac{c^4}{ 2} 
 e^{-2\pi  i   ab}     \sum_{p, n=-\infty}^{\infty}  
      e^{-\frac{\pi}{\tau_2}   |n-  p   \tau-  a   \tau+b| ^2   - 2\pi  i   [b p- n a]},\label{phi0phik-int2}
\end{align}
where 
$c$ is the constant given in \eqref{vphiq-thetaq}, which, by \eqref{phi0phik-int0} with $q=0$, is $c= \sqrt{2}(\tau_2)^{1/4}$. 
\end{proposition}
\begin{proof} 

Introduce the function $f_{q}(z):=e^{\frac{2\pi i}{\tau_2}\im (\bar q z)}(\bar\vphi_0{\vphi}_q)(z)$. 
The functions   $f_{q} $ 
are periodic functions w.r.to the lattice $\LAT$.   To convert this to standard periodicity (w.r.to  the square lattice), we write 
$z =z_1+i z_2=u_1+ u_2\tau : -1/2\le u_i\le 1/2,\ i=1, 2$, or in coordinates, $z=z(u)$, given by
 \begin{align}\label{zu} z_1= u_1+  \tau_1 u_2\ \quad  \mbox{and}\  \quad    z_2=  \tau_2 u_2.\end{align}
Then the functions   $ f_{q}(z(u))\equiv f_{q}(u_1+ u_2\tau)$ are periodic functions w.r.to $u_i$ with the period $1$.
We begin with 
\begin{lemma} \label{lem:fq}  
The function $f_{q}(z(u))$ is of the form
\begin{align}\label{fq} f_{q}(z(u)):=e^{\frac{2\pi i}{\tau_2}\im (\bar q z)}\bar\vphi_0\vphi_q=  \sum_{m, n=-\infty}^{\infty}e^{2\pi\al_{m, n}(u)}, \end{align}
  where $\al_{m, n}(u)$ is given by 
 \begin{align}\label{almm'22}  
\al_{m, n}(u)&= -  \tau_2  (u_2+r)^2+   i [(n-   a)\tau_1 +b](u_2 +r) 
 +  i n u_1+ \beta_{n},
	\end{align}
	  with  
	  $r= m-\frac12 n-\frac12  a$, 
	   and
	  \begin{align} \label{betap}  
\beta_{n }&=    i  \frac12 b n  -\frac14 ( n-a)^2\tau_2.	\end{align}
\end{lemma}
\begin{proof} First, \eqref{vphiq-thetaq} and the series representation  \eqref{thetaq} for $\theta_q(z, \tau)$ and \eqref{zu} 
yield  \begin{align}f_{q}(z(u)) &:=e^{\frac{2\pi i}{\tau_2}\im (\bar q z)}\bar\vphi_0\vphi_q=  \sum_{m, m'=-\infty}^{\infty}e^{2\pi\al_{m, m'}'(u)},\ 
 \text{  with}\\
\label{almm'12} 
\al_{m, m'}'(u)&=\frac{ i}{\tau_2}\im (\bar q z) -    \frac{1}{\tau_2}  z_2^2 -i a z+ i  q m +\frac12 i (m^2\tau-m'^2\bar\tau)\notag \\
&  + i (m z- m' \bar z)
+\frac12 i (a^2\tau  - ab). 
	\end{align}
We use that
$\tau:=\tau_1  +i\tau_2 $ and $q=-a \tau +b$ and \eqref{zu}, to obtain
\begin{align}\label{eq3} \frac{ i}{\tau_2}\im (\bar q z) -  i a  z&= \frac{  i}{\tau_2} ((-a\tau_1 +b) z_2+a\tau_2 z_1) -   a   (i z_1- z_2)\notag \\
&=\frac{1}{\tau_2}(i(-a\tau_1 +b)+ 
 a\tau_2)z_2=i \frac{1}{\tau_2} q z_2
= i  q  u_2. \end{align}
Next,  we use the relations \eqref{zu} and   
 \begin{align}\label{mz} m  z- m' \bar z &=(m -m' ) z_1 
  + i (m +m') z_2\notag \\   
  &=(m-m')(u_1+  \tau_1 u_2) + i (m+m')\tau_2 u_2,  
    \end{align}
and  the notation  $n=m-m'$ and $r= \frac12 (m+m')-\frac12  a= m-\frac12 n-\frac12  a$, so that $m=r + \frac12 n+\frac12  a$ and $m'=r - \frac12 n+\frac12  a$ (and $r=m-\frac12 n-\frac12  a$), 
 to obtain
   \begin{align}\label{12}
\frac{ i}{\tau_2}\im (\bar q z)-  &  \frac{1}{\tau_2}  z_2^2 -i a z + i (m z- m' \bar z)= -  \tau_2  u_2^2 + a\tau_2 u_2+i (-a\tau_1+b) u_2\notag \\
&  \qquad+ i n(u_1+  \tau_1 u_2) -  2 (r +\frac12  a)\tau_2 u_2\notag \\
&= -  \tau_2 [ u_2^2 -a u_2+ 2 r u_2] +  i n u_1+   i [(n-   a)\tau_1 +b]u_2\notag \\
&= -  \tau_2  (u_2+r)^2 +  i n u_1+   i [(n-   a)\tau_1 +b]u_2\notag \\
& \qquad  +\tau_2\frac14 (2 r)^2.	\end{align}
Now,  we use that 
  \begin{align} 
&  \frac12  i( m^2\tau-(m')^2\bar\tau) = i(m^2-{m'}^2)\tau_1 - (m^2+{m'}^2)\tau_2 \label{eq2'}\\
&= in (r+\frac12  a)\tau_1 - ((r+\frac12  a)^2  + \frac14 n^2)\tau_2.
\end{align}
 The last three relations together with \eqref{almm'12} give
 \begin{align*} 
\al_{m, m-n}'(u)&= -  \tau_2  (u_2+r)^2+   i [(n -   a)\tau_1 +b](u_2  +r) 
 +  i n u_1+ \beta_{n},
	\end{align*}
which is the same as \eqref{almm'22},	with
\begin{align} 
\beta_{n }&= -  i  [(n-   a)\tau_1  + b] r +\tau_2 r^2 \notag\\
&+[ i  (-a\tau_1+b) +  a\tau_2] m +  i n (r+\frac12  a) \tau_1 - (\frac14 n^2+(r+\frac12  a)^2)\tau_2+\frac12 i  (a^2\tau  - ab). 
  \notag 
  \end{align}
Using $i q   = i  (-a\tau_1+b) +  a\tau_2 $, we simplify this expression to \eqref{betap}. Now, we change the summation variable from $m'$ to $n=m-m'$ and use that 
\[\sum_{m, m'=-\infty}^{\infty}g_{m, m'}=\sum_{m=-\infty}^{\infty} \sum_{m'=-\infty}^{\infty}g_{m, m'}=\sum_{m=-\infty}^{\infty} \sum_{n=-\infty}^{\infty}g_{m, m-n}=\sum_{m, n=-\infty}^{\infty}g_{m, m-n}\] to find \eqref{fq} - \eqref{betap}.
 \end{proof} 

Let 
$c_{n_1, n_2} (f)$ denote the Fourier coefficients of a function $f$  (w.r.to $u_i$), i.e. $c_{\al} (f):=\int_{[0, 1]^2} f(u) e^{2\pi i \al\cdot u} du$, $\al=(n_1, n_2)\in \Z^2$.

 To compute the Fourier coefficients $c_{n_1, n_2}(f_{q})$, we compute the FTs of    $f_{q}$ on the entire $\R^2$ and use 
  that for any function $f(u)$ of the period $1$, we have
  \begin{align}\label{FT-FS}\int_{\R^2} f(z(u)) e^{2\pi i \xi \cdot u} du = \sum_{n\in \Z^2} 
  c_n(f)\del( \xi_1-n_1)\del( \xi_2-n_2).\end{align}

Now, the Fourier transform of $f_{q}$ is given by  \[FT(f_{q})(\xi)=  {c'}^2\sum_{m, n=-\infty}^{\infty}\int_{\R^2} du e^{2\pi ( i\xi_1u_1+ i\xi_2u_2 +\al_{m, n}(u))}.\]
 Using \eqref{almm'22} and  passing to the new variables $y_1=u_1$ and $y_2=u_2 -\frac12 a+r$, 
 we find
  \begin{align}\label{FT22}
 i\xi_1u_1+ i\xi_2u_2 +\al_{m, n}(u) &=  i\xi_1y_1+ i\xi_2(y_2-r)   -  \tau_2  y_2^2+  i  (n-   a)  \tau_1 y_2 \notag \\ &\qquad +i b y_2+  i  p y_1+\beta_{n }\notag\\
&= -i\xi_2 (r+\frac12 a) +\frac12 i \xi_2 a + i(\xi_1+  n) y_1 -  \tau_2  y_2^2  \notag \\ &\qquad + i(\xi_2+ (n-   a)  \tau_1+b) y_2+\beta_{n}. \end{align}
Now, changing $n$ to $-n$, remembering that $r=m-\frac12 n-\frac12 a$, with $m, n\in \Z$, 
and using the standard formulae, the first of which is  the Poisson summation formula,
\begin{align}\label{Poisson}&\sum_{m=-\infty}^{\infty}  e^{- 2\pi i m \xi_2} =
\sum_{n'=-\infty}^{\infty}  \del(\xi_2- n'),\\	
\label{FT-eq1}
&\int_{\R} dy_1 e^{2\pi i \xi y_1}=  
 \del(\xi),\\ 
 \label{FT-eq2}
& \int_{\R} dy_2 e^{ 2\pi ( -  \tau_2  y_2^2 +  i \eta y_2)}  =\frac{1}{\sqrt{2\tau_2}}  e^{-2\pi \frac{1}{4\tau_2} \eta^2 } ,\end{align}
 we obtain
 \begin{align}\label{FT32}
{c'}^{-2} FT(f_{q})&=  \sum_{n, n'=-\infty}^{\infty}\int_{\R^2} dy e^{2\pi ( i(\xi_1-  n) y_1 -  \tau_2  y_2^2 + i(\xi_2- (n+   a)  \tau_1+b) y_2 +i\frac12\xi_2 (a -n) +\beta_{- n })}  \del(\xi_2- n')\notag\\
&=  \sum_{n, n'=-\infty}^{\infty}\int_{\R} dy_2 e^{ 2\pi ( -  \tau_2  y_2^2 +   i[\xi_2- (n+   a)  \tau_1+b] y_2 -i\xi_2\frac12  (n-a) +   \beta_{- n })}  \del(\xi_2- n')  \del(\xi_1- n)\notag\\
&=  \frac{1}{\sqrt{\tau_2}}\sum_{n, n'=-\infty}^{\infty} e^{-2\pi \Phi}  \del(\xi_2- n')  \del(\xi_1- n).	 \end{align}
where $\Phi:=\frac{1}{4\tau_2}  [n'- (n+   a)  \tau_1+b] ^2   + i\frac12 n  (n-a) -   \beta_{- n }$. Remembering \eqref{betap} , we simplify this expression as
\begin{align}\label{comput}
\Phi &=\frac{1}{4\tau_2}  [n'- (n+   a) \tau_1+b] ^2   +\frac12  i   [ n' (n-  a) + b n]  + \frac14 ( n+a)^2\tau_2\notag\\
&= \frac{1}{4\tau_2}  |n'- (n+   a) \tau+b| ^2   +\frac12  i   [ n' n- n' a  +  b n]. 
\end{align}
The latter expression and \eqref{FT-FS} and the relation $e^{-\pi i n' n}=(-1)^{n' n}$, gives the Fourier coefficients for $f_q$
 \begin{align}\label{Fcoef2}
c_{n, n'}(f_{q}) &=  \frac{1}{\sqrt{\tau_2}} (-1)^{n' n} e^{-\pi (\frac{1}{2\tau_2}   |n'- (n+   a)  \tau+b| ^2   + i   [b n- n' a  ]) }. 
     \end{align}
Passing to the variables $u_i$ with the Jacobian  $ \det\left( \begin{array}{cc} 1 & \tau_1 \\ 0 & \tau_2 \end{array} \right) = \tau_2$, we obtain  \[\int_\Om \bar\vphi_0 \vphi_q e^{\frac{2\pi i}{\tau_2}\im (\bar q z)} dz=\int_\Om f_{q}(z) dz  =\tau_2\int_{[0, 1]^2}f_{q}(z(u)) du=\tau_2c_{00}(f_{q}).\] 
This, together with \eqref{Fcoef2}, 
 gives \eqref{phi0phik-int0}.
 
To compute the integral \eqref{phi0phik-int1}, we use that it can be written,  in terms of the function $f_q$, as $\int_\Om |\vphi_0|^2|\vphi_q|^2  dz=\int_\Om \bar f_{q} f_{q} dz$.   Using the change of variables and the Plancherel theorem, we find
  \begin{align}\label{int1}
&
\int_\Om |\vphi_0|^2|\vphi_q|^2  dz =  \tau_2\sum_{n, n'=-\infty}^{\infty} \bar c_{n, n'}(f_{q}) c_{n, n'}(f_{q}) .\end{align} 
This equation, together with \eqref{Fcoef2}, yields
\DETAILS{\begin{align}\label{int1'}
\int_\Om \bar f_{q} f_{q} dz &= 
{c'}^{4} \sum_{p, n=-\infty}^{\infty}  e^{-2\pi \frac{1}{2\tau_2}   |n- (p+   a)    \tau+b| ^2} ,\end{align}
which in turn gives} \eqref{phi0phik-int1}. 

Now, we compute the integral \eqref{phi0phik-int2}, which in terms of the function $f_q$ can be written as $\int_\Om \bar\vphi_0^2 \vphi_q {\vphi}_{-q} dz=\int_\Om f_{-q} f_{q} dz$.   Using the change of variables and the Plancherel theorem, we find
  \begin{align}\label{int2}\int_\Om 
 f_{-q} f_{q} dz &  =  \tau_2\sum_{n, n'=-\infty}^{\infty} c_{-n, -n'}(f_{-q}) c_{n, n'}(f_{q}) .\end{align}
This equation, together with \eqref{Fcoef2}, yields
\begin{align}\label{int32}
\int_\Om f_{-q} f_{q} dz &= 
{c'}^{4} \sum_{n, n'=-\infty}^{\infty}  e^{-2\pi (\frac{1}{2\tau_2}   |n'- (n+   a)    \tau+b| ^2   + i    [b n - n' a ])} ,\end{align}
which in turn gives \eqref{phi0phik-int2}. \end{proof}

\section{Energy of fluctuations}\label{sec:fluct-en}

\begin{lemma}\label{lem:En-fluct-expr}
	For all $v \in H_{\textnormal{cov}}^1$, we have
	\begin{equation}	\label{Lambda-expres'}
 \lim_{Q\ra\R^2} \big(\E_{Q}(u_\om + v) - \E_{Q}(u_\om)\big) 
 = \frac{1}{2}\lan v, L_\om v \ran_{L^2} + R_\om(v),
	\end{equation}
	where $R_\om$ is given by
	\begin{equation}	\label{remainder'}
	R_\om(v) = \int (|\alpha|^2 + \kappa^2|\xi|^2)\Re(\bar{\Psi}_\om\xi) - \alpha \cdot \Im(\bar{\xi}\COVGRAD{A_\om}\xi) + \frac{1}{2}(|\alpha|^2 + \frac{\kappa^2}{2}|\xi|^2) |\xi|^2.
	\end{equation} 
\end{lemma}
\begin{proof}
	We first considers smooth $v$ with compact support. Choose any $Q \subset \R^2$ that is bounded and contains the support of $v$.
	Using definition \eqref{gle}, we expand $\E_Q(u_\om + v)$ in $v$, collecting terms that are linear in $v$, $\frac{1}{2} \int_Q 2\Re(\overline{\COVGRAD{A_\om}\xi}\cdot\COVGRAD{A_\om}\Psi_\om + i\alpha\cdot\bar{\Psi}_\om\COVGRAD{A_\om}\Psi_\om)
				+ 2(\CURL A_\om)(\CURL\alpha)
				- 2\kappa^2(1 - |\Psi_\om|^2)\Re(\bar{\xi}\Psi_\om)$, and integrate by parts (boundary terms vanish due to $v$) to see that they are equal to
	\begin{align*}
		\Re \int_Q \bar{\xi}(-\COVLAP{A_\om}\Psi_\om - \kappa^2(1 - |\Psi_\om|^2)\Psi_\om)
			+ \alpha\cdot(\CURL^*\CURL A_\om - \Im(\bar{\Psi}_\om\COVGRAD{A_\om}\Psi_\om)) = 0,
	\end{align*}
	since $(\Psi_\om, A_\om)$ is a solution of the Ginzburg-Landau equations.
	For the quadratic terms we again integrate by parts and use the fact that the terms vanish outside $Q$ to see that they give us exactly $\frac{1}{2}\lan v, L_\om v \ran_{L^2}$. Similarly the higher order terms give us $R_\om(v)$. We thus have
	\begin{equation*}
		\E_Q(u) - \E_Q(u_\om) = \frac{1}{2}\lan v, L_\om v \ran_{L^2} + R_\om(v).
	\end{equation*}
	The right hand side is independent of $Q$, and therefore taking the limit proves \eqref{Lambda-expres'} for smooth compactly supported $v$. The general result follows from the fact that $\Lambda_\om$ is a continuous functional on the space $H^1_{\textrm{cov}}$.
\end{proof}
 \section*{Supplement I. Feshbach-Schur perturbation theory} 

In this appendix we present for the reader's convenience the main result of the Feshbach-Schur perturbation theory. Let $P$ and $\oP$ be orthogonal projections
on a separable Hilbert space $X$, satisfying $P + \oP = \bfone$.
Let $H$ be a self-adjoint operator on $X$.
We assume that $\Ran P \subset D(H)$, that
$H_{\oP} := \oP H \oP \restriction_{\Ran \oP}$ is invertible, and
\begin{equation}  \tag{I.1} 
  \| R_{\oP} \| < \infty \ , \quad\quad
  \| PHR_{\oP}\| < \infty \ \hbox{\quad and \quad }
  \| R_{\oP} H P \| < \infty \ ,
\end{equation}
where $R_{\oP} = \oP H_{\oP}^{-1} \oP$.
We define the operator
\begin{equation} \tag{I.2} 
  F_P(H) := P(H - HR_{\oP} H)P \restriction_{\Ran\, P} \ .
\end{equation}
The key result for us is the following:

\medskip \noindent {\bf Theorem I.1.}(Isospectrality theorem for the Feshbach-Schur maps)  
{\em
Assume (I.1) 
hold. Then
\begin{equation} \tag{I.3} 
0 \in \sigma(H) \Leftrightarrow 0 \in \sigma(F_P(H))  \end{equation}
and
\begin{equation} \tag{I.4} 
  H\psi = 0 \Leftrightarrow F_P(H)\phi = 0,
\end{equation}
where $\psi$ and $\phi$ are related by $\phi = P\psi$ and
$\psi = Q\phi$, with the (bounded) operator $Q$ given by}
\begin{equation} \tag{I.5} 
  Q = Q(H) := P - R_{\oP} H P.
\end{equation}
\begin{proof} Both relations are proven similarly so we prove only the second one which suffices for us. 
First, in addition to (I.5), we define the operator
\begin{equation} \tag{I.6} 
  Q^\# = Q^\#(H) := P-PH R_{\oP}.
\end{equation}
The operators $P$,  $Q$ and $Q^\#$ satisfy
\begin{equation} \tag{I.7} 
  H Q = P H' \ , \qquad
  Q^\# H = H' P,
 \end{equation}
 where $H' = F_{P}(H)$.  
Indeed, using the definition of $Q$, we transform
\begin{equation} \tag{I.8} 
\begin{split}
  H Q &= HP - H \oP H^{-1}_{\oP} \oP H P  \\
  &= PHP + \oP HP - PH \oP H^{-1}_{\oP} \oP HP
  - \oP H \oP H^{-1}_{\oP} \oP H P \\
  &= P HP - P H\oP H^{-1}_{\oP} \oP HP \\
  &= P F_P (H).
\end{split}
\end{equation}
Next, we have
\[
\begin{split}
  Q^\# H &= PH-PH \oP H^{-1}_{\oP} \oP H \\
  &= PHP + PH\oP - PH\oP H^{-1}_{\oP} \oP HP - PH\oP H^{-1}_{\oP} \oP H\oP \\
  &= PHP - PH\oP H^{-1}_{\oP} \oP HP \\
  &= F_P (H) P.
\end{split}
\]
This completes the proof of (I.7). 

Now, we show
\begin{equation} \tag{I.9} 
   \Null Q \cap \Null H' = \{ 0 \}\quad \hbox{ and } \quad
   \Null P \cap \Null H = \{ 0 \},
\end{equation} The first relation in (I.9) follows from the fact that the projections $P$ and $\oP$ are orthogonal, which implies the inequality
\[
  \| Qu \|^2 = \| Pu \|^2 + \| R_{\oP} H Pu \|^2 \geq \| Pu \|^2,
\]
and the relation $\Null P\subset \Null H'$, which follows from the definition of $H'$.
To prove the second relation  in (I.9)
we use the equation $P + \oP = \bfone$ and the definitions $H_{\oP}  = \oP H\oP$ and (I.5) to obtain
\begin{equation} \tag{I.10} 
  \bfone = QP + R_{\oP} H \ ,
\end{equation}
which, in turn, is implies  the second 
relation  in (I.9). Indeed, applying~(I.10) to a vector
$\phi \in \Null P \cap \Null H$, we obtain
$\phi = QP\phi + R_{\oP} H \phi = 0$.

Now the statement (I.4) follows from relations~(I.7) and (I.9).  \end{proof}


\end{document}